%% file: main.tex
\documentclass[onecolumn,superscriptaddress,longbibliography]{revtex4-2}

\usepackage{graphicx}
\usepackage{multirow}
\usepackage{amsmath,amssymb, amsthm}
\usepackage{mathtools}
\usepackage{soul}
\usepackage{verbatim}

\usepackage{comment}
\newtheorem{theorem}{Theorem}
\newtheorem{lemma}[theorem]{Lemma}
\newtheorem{corollary}[theorem]{Corollary}
\usepackage{bm}
\usepackage[colorlinks=true, linkcolor=blue, citecolor=gray]{hyperref}
\usepackage[table]{xcolor}
\usepackage{tikz}
\usepackage{xspace}
\usepackage[capitalize]{cleveref}
\usepackage{import}
\usepackage[caption=false]{subfig}
\usepackage{bbold}
\usepackage[braket, qm]{qcircuit}
\usepackage{chemformula}
\usepackage{adjustbox}
\DeclarePairedDelimiterX\braket[2]{\langle}{\rangle}{#1 \delimsize\vert #2}

\usepackage{epstopdf}
\AppendGraphicsExtensions{.tga}
\newcolumntype{C}[1]{>{\centering\arraybackslash}m{#1}}
\newcommand{\abs}[1]{\left\lvert#1\right\rvert}
\newcommand{\norm}[1]{\left\lVert#1\right\rVert}

\makeatletter
\newcommand\footnoteref[1]{\protected@xdef\@thefnmark{\ref{#1}}\@footnotemark}
\makeatother

\newcommand{\ketbra}[2]{\ket{#1}\!\bra{#2}}
\def\bf#1{\mathbf{#1}}

\newcommand{\eq}[1]{(\ref{eq:#1})}
\newcommand{\thm}[1]{\hyperref[thm:#1]{Theorem~\ref*{thm:#1}}}
\newcommand{\defn}[1]{\hyperref[defn:#1]{Definition~\ref*{defn:#1}}}
\newcommand{\lem}[1]{\hyperref[lem:#1]{Lemma~\ref*{lem:#1}}}
\newcommand{\prop}[1]{\hyperref[prop:#1]{Proposition~\ref*{prop:#1}}}
\newcommand{\fig}[1]{\hyperref[fig:#1]{Figure~\ref*{fig:#1}}}
\newcommand{\tab}[1]{\hyperref[tab:#1]{Table~\ref*{tab:#1}}}
\renewcommand{\sec}[1]{\hyperref[sec:#1]{Section~\ref*{sec:#1}}}
\newcommand{\app}[1]{\hyperref[app:#1]{Appendix~\ref*{app:#1}}}
\newcommand{\cor}[1]{\hyperref[cor:#1]{Corollary~\ref*{cor:#1}}}
\newcommand{\obs}[1]{\hyperref[obs:#1]{Observation~\ref*{obs:#1}}}

\newcommand{\paragraphsentence}[1]{}
\renewcommand{\paragraphsentence}[1]{\emph{#1}}

\newcommand{\arccot}{\,{\rm arccot}}

\begin{document}

\newcommand{\Google}{\affiliation{%
Google Research, Venice, CA 90291, United States}}

\newcommand{\Columbia}{\affiliation{%
Department of Chemistry, Columbia University, USA}}
\newcommand{\Toronto}{\affiliation{%
Department of Computer Science, University of Toronto, Canada M5S 1A4}}

\newcommand{\PNNL}{\affiliation{%
Pacific Northwest National Laboratory, Richland WA, USA 99354}}

\newcommand{\BI}{\affiliation{%
Quantum Lab, Boehringer Ingelheim, 55218 Ingelheim am Rhein, Germany}}

\newcommand{\BIMedChem}{\affiliation{%
Boehringer Ingelheim Pharma GmbH \& Co KG, Birkendorfer Strasse 65, 88397 Biberach, Germany}}

\newcommand{\Macquarie}{\affiliation{%
Department of Physics and Astronomy, Macquarie University, New South Wales 2109, Australia}}

\title{Efficient quantum computation of molecular forces and other energy gradients}

\author{Thomas E.~O'Brien}
\email{teobrien@google.com}
\Google

\author{Michael Streif}
\email{michael.streif@boehringer-ingelheim.com}
\BI

\author{Nicholas C.~Rubin}
\email{nickrubin@google.com}
\Google

\author{Raffaele Santagati}
\email{raffaele.santagati@boehringer-ingelheim.com}
\BI

\author{Yuan Su}
\Google

\author{William J.~Huggins}
\Google

\author{Joshua J.~Goings}
\Google

\author{Nikolaj Moll}
\BI

\author{Elica Kyoseva}
\BI

\author{Matthias Degroote}
\BI

\author{Christofer S.~Tautermann}
\BIMedChem

\author{Joonho Lee}
\Google
\Columbia

\author{Dominic W.~Berry}
\Macquarie

\author{Nathan Wiebe}
\email{nawiebe@cs.toronto.edu}
\Toronto
\PNNL

\author{Ryan Babbush}
\email{babbush@google.com}
\Google

\begin{abstract}
While most work on the quantum simulation of chemistry has focused on computing energy surfaces, a similarly important application requiring subtly different algorithms is the computation of energy derivatives. Almost all molecular properties can be expressed an energy derivative, including molecular forces, which are essential for applications such as molecular dynamics simulations. Here, we introduce new quantum algorithms for computing molecular energy derivatives with significantly lower complexity than prior methods. Under cost models appropriate for noisy-intermediate scale quantum devices we demonstrate how low rank factorizations and other tomography schemes can be optimized for energy derivative calculations. We perform numerics revealing that our techniques reduce the number of circuit repetitions required by many orders of magnitude for even modest systems. In the context of fault-tolerant algorithms, we develop new methods of estimating energy derivatives with Heisenberg limited scaling incorporating state-of-the-art techniques for block encoding fermionic operators. Our results suggest that the calculation of forces on a single nucleus may be of similar cost to estimating energies of chemical systems, but that further developments are needed for quantum computers to meaningfully assist with molecular dynamics simulations.
\end{abstract}

\maketitle

\tableofcontents

\section{Introduction}

Quantum chemistry is widely regarded as one of the most promising areas of application for quantum computers.
This is due to the relative ease of mapping the electronic structure problem onto a quantum device~\cite{Lloyd1996,Aspuru-Guzik2005,Cao2019}, its difficulty in simulating classically, and its high relevance  to industry.
Interest in such applications has been steadily increasing following initial beyond-classical quantum computing demonstrations~\cite{Google19Quantum,Zhu21Quantum} and experimental~\cite{Chen21Exponential,Egan21Fault,RyanAnderson21Realization} demonstrations of hardware near the fault-tolerant threshold for quantum error-correction~\cite{Fowler12Surface}.
A significant body of work has emerged in recent years optimizing quantum algorithms for chemistry, both for fault-tolerant quantum computers~\cite{Babbush2017LowStructure,Motta2018,Berry19Qubitization,lee2021even,Reiher2017,vonBurg2020,Kim2021Fault} and current NISQ devices~\cite{Mcclean16Theory,Grimsley19Adaptive,Huggins19Efficient,Bonet20Nearly}, including various experimental implementations~\cite{PeruzzoNC2013,OMalley2016,Kandala17Hardware,Santagati18Witnessing,Hempel18Quantum,Google20Hartree,Huggins21Unbiasing}.
Predominantly this work has focused on estimating energies of ground states of the electronic structure problem, perhaps the most natural property to extract from a quantum chemistry simulation.
However, ground state energies are not a quantity typically measured in the lab, and further processing of energy data is required to obtain properties of relevance to industry.
Thus, quantum algorithms to estimate properties other than ground state energies are of high interest as we progress towards larger NISQ or future fault-tolerant devices.

The calculation of forces (the derivative of energies with respect to nuclear positions) is a subroutine in most modern computational approaches to navigate molecular potential energy surfaces. Beyond identifying minima and other stationary points, reaction path following is also based on the determination of gradients \cite{Hratchian:2005}.
The knowledge of low energy stationary points allows the generation of conformational Boltzmann ensembles, calculation of reaction rates, and prediction of tautomer equilibria.
Forces are also an essential ingredient of molecular dynamics (MD) simulations, which are invaluable for studying macroscopic thermodynamic properties.
This covers highly diverse applications such as the description of heterogeneous processes on surfaces including catalysis \cite{Li:2021}, observation of phase transitions, such as nucleation processes for water \cite{Matsumoto:2002}, and maybe of highest importance for pharmaceutical research, the interaction of drugs with their targets in the human body \cite{Salo-Ahen:2021}.
By free energy calculations based on MD simulations, these interactions can be quantified, allowing the prediction of compound affinities \cite{Cournia:2017, Malone2021}, which are eventually linked to therapeutic doses.
Beyond that, MD simulations of drug-target systems enable the observation of conformational changes of the target.
The corresponding drug-induced active and inactive states or ligand bias is a ligand-dependent selective signaling pattern \cite{Fleetwood:2021}, which is especially useful to avoid drug-induced side effects \cite{Nivedha:2018}.
Other decisive parameters such as drug residence time can nowadays be determined through specialized MD simulations \cite{Nunes-Alves:2020}.
These powerful techniques render MD simulations one of the most broadly used and powerful tools in drug design, and hence make forces a clear target for quantum computing.

Though some research on quantum algorithms for force and gradient estimation has been performed previously, efforts to accurately cost algorithms in a fault-tolerant or NISQ setting have been limited.
The suggestion to estimate nuclear forces on a quantum device was first suggested by~\cite{KassalJCP2009}, which studied estimation via the Hellman-Feynman theorem and via the quantum gradient estimation algorithm of~\cite{Jordan_2005}.
This topic was then relatively untouched by the quantum community for a decade until it was revived by~\cite{Obrien19Calculating,Mitarai2019,Parrish2019}.
Ref.~\cite{Obrien19Calculating} studied force estimation in both a NISQ and FT framework and performed the first experimental force calculation, but only found loose asymptotic bounds of $N^{7}-N^{15}$ to estimate a single force component.
Ref.~\cite{Mitarai2019} put the mathematical formulation of force estimation in NISQ on a significantly stronger footing, combined this with gradient estimation for the optimization of variational quantum eigensolvers, but only considered the cost of estimating all $N^4$ terms in the fermionic 2-reduced density matrix to constant precision.
Ref.~\cite{Parrish2019} firmed the theoretical chemistry behind force estimation on a quantum device, presenting a detailed derivation in a Lagrangian formalism focusing on an \emph{ab initio} exciton model, and stressed the importance of including full response. The paper presented explicit formulas and circuits based on the parameter-shift rule~\cite{mitarai2018,Schuld2019} but did not provide asymptotic costs for the estimation of forces on a quantum device.
These works were followed by small experimental demonstrations of molecular dynamics simulations for various applications~\cite{Magann21Digital,Fedorov20Abinitio,Sokolov:2021}, and theoretical studies extending gradient calculations to the derivatives of energies beyond the ground state~\cite{Arimitsu21Analytic,Yalouz21Analytical,Parrish21Analytical}.
However, many possible optimizations remain for both NISQ and fault-tolerant algorithms to estimate forces.
Furthermore, little work has been done to estimate the magnitude of force operator quantities that are relevant for quantum algorithm resource requirements (e.g. induced $1$-norms).

In this work we optimize and cost methods for estimating forces and other first-order energy gradients for NISQ and fault-tolerant quantum computers.
We study the required tolerance on the error in a force estimation for molecular dynamics and geometry optimization, finding a relevant figure for accurate estimation of the pair correlation function of a moderate-sized water simulation being a root mean square (RMS) error of no more than $6.4$~mHa/Å per single derivative component.
We optimize tomography methods for NISQ quantum devices, where the relevant cost model is the number of repeated experiments required to achieve a target $2$-norm in the error vector.
We find that all methods have similar or even slightly better asymptotic costs to estimate an entire force vector to a given accuracy compared to the cost of estimating energies to a similar accuracy.
We study methods for block-encoding force operators for future fault-tolerant algorithms, and present the first investigation of the induced $1$-norm of a force operator with the system size (a critical property for fault-tolerant quantum algorithms).
We find that for state-of-the-art techniques block encodings of derivatives are at most constant or polylog factors more costly than block-encodings of the corresponding Hamiltonian, and that in practice they may be significantly cheaper.
We finally detail and cost three separate Heisenberg-limited fault-tolerant algorithms for force estimation: a semi-classical higher-order finite difference algorithm drawing energy estimates at different configurations from the quantum device, an application of the overlap estimation algorithm to gradient estimation, and an extension of the new gradient-based expectation value estimation algorithm of~\cite{Huggins2021Nearly}.
We determine asymptotic costings for these three algorithms on hydrogen chains and water clusters, and for plane wave systems in first quantization.
Surprisingly, due to difficulties to parallelize the overlap estimation algorithm and the need to perform Hamiltonian simulation as part of the reflection subroutine, we find that in some cases the finite difference method will be preferable (or at worst competitive) compared to the gradient-based expectation value estimation algorithm, which strictly asymptotically dominates the overlap estimation algorithm.
Our results suggest that while force estimation in NISQ may be somewhat cheaper than energy estimation (ignoring the overhead of needing to optimize the variational preparation of a quantum state), force estimation in FT is at best asymptotically the same cost, and in some cases significantly worse. Ultimately, we do not see a useful beyond-classical molecular dynamics simulation to be tractable in a NISQ or FT quantum computing setting, as such a calculation would require many millions of force estimations to be performed~\cite{hollingsworth2018molecular}, each of which would be at least as costly as estimating the energy of a system.
However, for applications such as geometry optimization, coupling parameter estimation or spectral prediction, which do not require such a high number of repeat derivative estimations, our methods appear feasible for early fault-tolerant devices.

\subsection{Outline}
We begin this work in Sec.~\ref{sec:energy_gradient} with a review of \emph{ab initio} electronic structure theory and how energy derivatives may be estimated as the expectation value of a derivative operator through the Hellman-Feynman theorem.
Though we focus on atomic forces (i.e. derivatives with respect to nuclei positions) for the majority of this work, we briefly detail here how these methods may be immediately extended to other first-order properties of a molecular system.
In Sec.~\ref{sec:force_operators_atomic_basis} we present a simple, calculable derivation of the force operator in second quantization for an atomic-centered basis orbital set, based on the orbital connection theory of Helgaker and Alml\"of~\cite{Helgaker1984Secondquantization}.
Then, in Sec.~\ref{sec:force_operators_plane_waves} we derive the exact form of the force operator in a plane wave basis in first- and second-quantization, and demonstrate that this operator is diagonalized by the quantum Fourier transform with aliased frequencies and the fermionic fast Fourier transform respectively.
In Sec.~\ref{sec:error_tolerance}, we estimate the error tolerance on a force vector required for geometry optimization and molecular dynamics simulations.
Based on radial distribution function calculations for a system of $216$ water molecules, we estimate that it is relevant for molecular dynamics and geometry optimization to target a RMS of the error in one force component below $0.6$~mHa/Å.

We then turn in Sec.~\ref{sec:nisq} to the optimization and costing of state tomography for the estimation of force vectors in molecular systems.
We overview a general scheme for low-cost NISQ tomography methods of arbitrary operators, and in Sec.~\ref{sec:basis_rotations} we review previous work on choices of basis rotation to implement this scheme.
In Sec.~\ref{sec:NISQ_DirectMeasurements}, we extend previous work on importance sampling to directly target the $2$-norm error in a force vector, and demonstrate the importance of parallelization of measurements where possible.
In Sec.~\ref{sec:fermionic_shadow_tomography}, we review the fermionic shadow tomography scheme of Refs.~\cite{huang2020predicting,Zhao20Fermionic}, and calculate the relevant bound on the cost of the number of measurements to estimate a constant $2$-norm error here as well.
Then, in Sec.~\ref{sec:numerics_nisq}, we find bounds on the costs of the different methods above, both analytically and numerically for hydrogen chains, and estimate the asymptotic costing of each.

A critical piece of fault-tolerant quantum computation is block encoding, so before giving fault-tolerant algorithms for force estimation we study the cost of block-encoding a force operator.
We review general block encodings in Sec.~\ref{sec:block_encodings}.
We give explicit methods to block-encode Hamiltonians and force operators in second quantization for factorized methods (Sec.~\ref{sec:block_encoding_second_quantization}), and in first quantization for plane waves (Sec.~\ref{sec:BEfirst}).
In molecular systems the cost of simulating block-encoded derivative operators is found to be at most a constant factor worse (as one may differentiate the Hamiltonian in its factorized form), while in plane-wave systems the rescaling factor of the block encodings is found to be identical and the circuit cost only $\mathrm{poly}(\log(N))$ worse.
In Sec.~\ref{sec:numerics} we study the rescaling factors for block encoding in atomic orbital bases, finding that when using sparse simulation methods the cost of simulating forces is similar to the cost of simulating Hamiltonians, but for factorized methods the cost is clearly asymptotically lower.

We finish this work in Sec.~\ref{sec:FT} by designing three new algorithms for force estimation on fault-tolerant quantum computers, and estimating their asymptotic costs on various chemical systems using results from previous sections.
In Sec.~\ref{sec:Num_diff_FD} we use higher-order difference formulas to estimate gradients with a fault-tolerant quantum computer as a subroutine to estimate the energy at different atomic configurations.
We optimize the importance sampling, choice of finite difference order and step size, and consider the efficiency of reusing the state register on the quantum device between different calls to the subroutine.
In Sec.~\ref{sec:HeisenbergLimit}, we use the overlap estimation algorithm of \cite{Knill06Optimal} to estimate gradients via the Hellman-Feynman theorem at the Heisenberg limit.
We optimize this algorithm for general block-encoded Hermitian operators by a factor $4$, and optimize importance sampling over the gradient terms.
We further optimize the choice of reflection operator using techniques from \cite{Lin20Near}, and demonstrate the ability to perfectly recycle the state register on the quantum device between calls to the amplitude estimation subroutine.
However, due to the need to reflect about the ground state (which requires Hamiltonian simulation), we find that the overlap estimation algorithm can only outperform finite difference estimation when state preparation is the dominant cost of estimation in both routines.
Finally, we implement force estimation using the new gradient estimation technique of \cite{Huggins2021Nearly}, and compare it to both previous methods.
We find it achieves a strict asymptotic improvement over the overlap estimation algorithm, which implies in turn that it may often be better than a semi-classical finite difference method.
We conclude in Sec.~\ref{sec:conclusion}, where we summarize our results, discuss the implications for the field of molecular dynamics, and suggest paths for further improvement.

\section{Energy derivative calculation in electronic structure}
\label{sec:energy_gradient}

The goal of \emph{ab initio} electronic structure theory is to solve the time-independent Schrödinger equation,
\begin{equation}
    H_{\mathrm{tot}}|\Psi \rangle = E_{\mathrm{tot}} |\Psi\rangle,
\end{equation}
for a given molecular system. In most applications, it is sufficient to consider the non-relativistic, time-independent molecular Hamiltonian $H_{\mathrm{tot}}$, as we do here. Moreover, within the context of the Born-Oppenheimer approximation, one only needs to solve for the electronic Hamiltonian $H$, given by
\begin{align} \label{eq:elec_hamil}
   H &= -\frac{1}{2} \sum_{i}^{\eta} \nabla^2_i - \sum_{i}^{\eta} \sum_{A }^{N_{\mathrm{a}}} \frac{Z_A}{ r_{iA}} + \sum_{i>j}^{\eta} \frac{1}{r_{ij}} \nonumber \\ 
   & = \sum_i^{\eta} h_i(\mathbf{r}) + \sum_{i>j}^{\eta} \frac{1}{ r_{ij}},
\end{align}
which defines the electronic Hamiltonian for a molecular system with $N_{\mathrm{a}}$ atomic nuclei and $\eta$ electrons. The first term in Eq.~\eqref{eq:elec_hamil}, $h(\mathbf{r}) = \sum_i^{\eta} h_i(\mathbf{r})$, is the one-electron term which is the sum of the electronic kinetic energy and the interaction energy of the electrons with the nuclei with $\mathbf{r}$ being a position vector in real space, while the second term is the Coulomb two-electron interaction energy. Here $\nabla^2$ is the Laplacian with respect to the electronic coordinates, $Z_A$ is the $A^\text{th}$ nuclear charge, and $r_{iA} = |\mathbf{r}_i - \mathbf{R}_A|$ and $r_{ij} = |\mathbf{r}_i - \mathbf{r}_j|$ are the Euclidean distances between the $i^{\text{th}}$ electron and the $A^\text{th}$ nucleus, and between the $i^{\text{th}}$ and $j^{\text{th}}$ electrons, respectively. 

The molecular electronic Schr\"odinger equation is a function of the electronic coordinates $\mathbf{r}$ with a parametric dependence on the nuclear coordinates $\mathbf{R}$, e.g.,
\begin{equation} \label{eq:elec_schrod}
    H(\mathbf{r}, \mathbf{R}) |\Psi(\mathbf{r}, \mathbf{R})\rangle = E(\mathbf{R}) |\Psi(\mathbf{r}, \mathbf{R})\rangle. 
\end{equation}
The total energy within the Born-Oppenheimer approximation for fixed nuclear positions is given as (for simplicity we omit the explicit position dependence), 
\begin{equation}
    E_{\mathrm{tot}} = E + V_{\mathrm{nuc}} = E + \sum_{A>B}^{N_{\mathrm{a}}} \frac{Z_A Z_B}{R_{AB}},
\end{equation}
where $V_\mathrm{nuc}$ is the nuclear-nuclear repulsion energy with $R_{AB}=|\mathbf{R}_A-\mathbf{R}_B|$ the Euclidean distance between the $A^{\text{th}}$ and $B^{\text{th}}$ nucleus. 

Although the above provides a basis for molecular quantum mechanics and is sufficient for computing molecular energies, it is desirable to also be able to compute different molecular properties. Time-independent molecular properties can be expressed as gradients of the \textit{ab initio} electronic energy $E$ with respect to a suitable perturbation. For example, the first derivative of the energy with respect to an external electric or magnetic field evaluated at zero field strength yields the electric and magnetic dipole moments, respectively. Further, the first derivative of the energy with respect to the nuclear spin (internal magnetic field) yields the hyperfine coupling constants, which are important for multiple spectroscopy techniques such as nuclear magnetic resonance (NMR). Similarly, molecular forces are computed as gradients of the energy with respect to nuclear displacements. These molecular properties are summarized in Table \ref{fig:table_properties}. Although higher-order and mixed derivatives of the energy lead to additional properties, herein we will focus our attention on first order derivatives of the total energy. Obtaining analytic formulas for these gradients is a rich research area in classical quantum chemistry and we refer the interested reader to Refs.~\cite{yamaguchi2011analytic,helgaker1992,pulay2007analytical} for more background.

\begin{table}[tb]
\centering
  \begin{tabular}{ | l | r | r | }
  \hline
 Gradient & Perturbation & Property \\ [0.5ex]
    \hline
    $d E / d \mathbf{E}$ & electric field & electric dipole moment\\ \hline
    $d E / d \mathbf{B}$ & magnetic field  & magnetic dipole moment\\ \hline
    $d E / d \mathbf{I}$ & nuclear spin & hyperfine coupling constant\\ \hline
    $d E / d \mathbf{R}$ & nuclear displacement & nuclear forces\\ \hline
 
  \end{tabular}
  \caption{Examples of properties which can be computed as gradients of the total energy.}
    \label{fig:table_properties}
\end{table}

In this work, we will analyze two separate classes of methods for computing energy gradients; computing via the Hellmann-Feynman theorem \cite{Hellmann1937,Feynman1939}, and computing via higher-order finite difference techniques.
The Hellmann-Feynman theorem relates the energy derivative to the expectation value of the derivative of the Hamiltonian with respect to that same parameter
\begin{equation}\label{eq:Hellmann-Feynman}
    \frac{dE}{dx}=\bigg\langle\Psi\bigg|\frac{dH}{dx}\bigg|\Psi\bigg\rangle.
\end{equation}
Here, $\Psi$ is a normalized eigenstate of the Hamiltonian $H$, and the lower case $x$ represents a general parameter with respect to which derivatives are taken (e.g. a single nuclear coordinate $R_i$, an electric field $\mathbf{E}$, or another quantity in Table~\ref{fig:table_properties}).
In practice, the process of calculating the correct total derivative of the Hamiltonian $H$ can be challenging.
All explicit and implicit dependencies for the derivative have to be accounted for.
In the remainder of this section we detail the analytic form of these operators in second quantized atomic-centered basis sets, and in arbitrary plane wave basis sets.
However, neither of these methods are necessary to implement finite difference calculations.

\subsection{Force operators in second quantization for atomic-centered basis orbitals}\label{sec:force_operators_atomic_basis}

To obtain the force operators the Schrödinger equation, Eq.~\eqref{eq:elec_schrod}, needs to be solved. From the atomic-centered basis orbitals (AO) the Hartree-Fock approximation is typically invoked to obtain first a set of molecular orbitals (MO). However, these MOs yield no analytic form and depend on the set of AOs. Therefore, it is not straight forward to calculate a total derivative operator $\frac{dH}{dx}$ to allow force calculations through Eq.~\eqref{eq:Hellmann-Feynman}. However, through the relations of the AOs and MOs to each other, the derivatives can be calculated by the orbital connection theory of Helgaker~\cite{Helgaker1984Secondquantization}.

A Hamiltonian represented on a quantum computer is conceptually different from the Hamiltonian on a classical computer. The overlap integrals are all pre-computed beforehand in the given MO basis. The wavefunction on the quantum computer only gives the coefficients of all possible determinants. On a classical computer the MO basis is typically part of the wavefunction and not part of the Hamiltonian itself.

We use the following notational conventions. Lower case italics $\{p, q, r, s\}$ index general (either occupied or virtual) are used for the MOs. Lower case Greek letters $\{\mu, \nu, \lambda, \sigma\}$ index are used for the AOs. To distinguish vectors and tensors from their elements, they will be written in a bold typeface.

After a Hartree-Fock computation, the $\eta$-electron wave function is represented as a single Slater determinant, that is, an anti-symmetric product of spin orbitals $\{\phi_p(\mathbf{r})\}$. These spin orbitals are discretized over a set of basis functions $\{\chi_{\mu}(\mathbf{r})\}$, commonly Gaussian atomic orbitals or plane waves. The spin orbitals are expanded as
\begin{equation}
    \phi_p(\mathbf{r}) = \sum_{\mu} C_{\mu p} \chi_{\mu} (\mathbf{r})
\end{equation}
where $C_{\mu p}$ denotes an element of the molecular orbital (MO) coefficient matrix. Without loss of generality, we will only consider real-valued MO coefficients. The electronic Hamiltonian from Eq.~\eqref{eq:elec_hamil} can be cast in matrix form in the atomic orbital basis, with one-body integrals represented as
\begin{equation}
    h_{\mu\nu} = \int d \mathbf{r}_1 ~\chi^{*}_{\mu} (\mathbf{r}_1)h(\mathbf{r}_1) \chi_{\nu}(\mathbf{r}_1)  = \langle \chi_{\mu} | h | \chi_{\nu} \rangle,
\end{equation}
and two-body integrals
\begin{equation}
    g_{\mu\nu\lambda\sigma} = \iint d \mathbf{r}_1  d \mathbf{r}_2 ~ \chi^{*}_{\mu} (\mathbf{r}_1) \chi^*_{\lambda}(\mathbf{r}_2) \frac{1}{\mathbf{r}_{12}} \chi_{\nu}(\mathbf{r}_1) \chi_{\sigma}(\mathbf{r}_2) = \langle \chi_{\mu} \chi_{\lambda}| \frac{1}{\mathbf{r}_{12}} | \chi_{\nu} \chi_{\sigma} \rangle. 
\end{equation}
In the general case, the set of AO basis functions is not orthogonal. It is therefore necessary to consider their overlap matrix, 
\begin{equation}
    S_{\mu\nu} = \int d \mathbf{r}_1 ~\chi^{*}_{\mu} (\mathbf{r}_1)\chi_{\nu}(\mathbf{r}_1) = \langle \chi_{\mu} | \chi_{\nu} \rangle. 
\end{equation}
These three integrals and their total derivatives (the so-called ``skeleton'' or ``core'' derivative integrals) are the fundamental building blocks of the molecular gradients. Expressions for the total derivatives of these integrals with respect to an arbitrary parameter $x$ have been derived elsewhere and may be easily computed with most electronic structure software packages.

In the orthonormal MO basis, it is useful to introduce the \textit{second quantization} formalism, which is developed in terms of fermionic creation (annihilation) operators $a_{p}^{\dag}$ $(a_{q})$ that satisfy the anti-commutation relations, $\{a_{p}^{\dag}, a_{q}^{\dag}\} = \{a_{p}, a_{q}\} = 0, \, \text{and} \, \{a_{p}^{\dag}, a_{q}\} = \delta_{pq}$. In second quantization the electronic structure Hamiltonian Eq.~\eqref{eq:elec_hamil} in the MO basis is then given by,
\begin{equation}\label{eq:second_quant_H}
    H= \sum_{pq} h_{pq}a_{p}^{\dag}a_{q} + \sum_{pqrs} g_{pqrs}  a^{\dag}_{p} a^{\dag}_{r}a_{s}a_{q},
\end{equation}
with one- and two-body terms in the MO basis
\begin{equation}
    \begin{alignedat}{1}
    h_{pq} &= \sum_{\mu\nu} C_{\mu p} C_{\nu q} h_{\mu\nu}, \\
    g_{pqrs} &= \sum_{\mu\nu\lambda\sigma} C_{\mu p} C_{\nu q} C_{\lambda r} C_{\sigma s} g_{\mu\nu\lambda\sigma}.
    \end{alignedat}
\end{equation}
At times it is useful to consider the overlap matrix also in the MO basis, which is given by 
\begin{equation}
    S_{pq} = \sum_{\mu\nu} C_{\mu p} C_{\nu q} S_{\mu \nu},
\end{equation}
where the use of lower case italic and Greek subscripts distinguishes between the MO and AO representations, respectively. We note that the typical overlap matrix relation $S_{pq}=\delta_{pq}$ only holds at the reference configuration.

For the molecular electronic Hamiltonian, the energy is given by 
\begin{equation}
    E = \sum_{pq} \gamma_{pq} \sum_{\mu \nu} C_{\mu p} C_{\nu q} h_{\mu \nu} + \sum_{p q r s} \Gamma_{p q r s} \sum_{\mu \nu \lambda \sigma} C_{\mu p} C_{\nu q} C_{\lambda r} C_{\sigma s} g_{\mu \nu \lambda \sigma}\,,
\end{equation}
where $\gamma_{pq}$ and $\Gamma_{p q r s}$ are the matrix elements of the one- and two-body reduced density matrices (RDMs) respectively.

Energies from \emph{ab initio} calculations depend on several parameters: the one- and two-body AO integrals $\mathbf{h}$ and $\mathbf{g}$, the molecular orbital rotation matrix $\boldsymbol{\Theta}$, the set of determinant amplitudes $\mathbf{c}$, and any other parameters, which we denote as $\boldsymbol{\Omega}$. Given this, and using the chain rule, a general first derivative of the energy $E$ with respect to from any \emph{ab initio} calculation can be written as
\begin{equation} \label{eq:total_derivative}
    \frac{dE}{dx} = \frac{\partial E}{\partial x} + \frac{\partial E}{\partial \mathbf{h}}
    \frac{d\mathbf{h}}{dx} 
    + \frac{\partial E}{\partial \mathbf{g}} \frac{d\mathbf{g}}{dx}  + \frac{\partial E}{\partial \boldsymbol{\Theta}} \frac{d\boldsymbol{\Theta}}{dx} + \frac{\partial E}{\partial \mathbf{S}}\frac{d\mathbf{S}}{dx} +  \frac{\partial E}{\partial \mathbf{c}}\frac{d\mathbf{c}}{dx} + \frac{\partial E}{\partial \boldsymbol{\Omega}}\frac{d\boldsymbol{\Omega}}{dx}.
\end{equation}
The remaining challenge is to fill in explicit expressions for the above elements. As the exact energy is independent of the orbital rotational parameters $\boldsymbol{\Theta}$ and CI coefficients $\mathbf{c}$, the corresponding partial derivatives are identically zero,
\begin{align}
    \frac{\partial E}{\partial \boldsymbol{\Theta}} =
   \frac{\partial E}{\partial \mathbf{c}} \equiv 0.
\end{align}

Several of the partial derivatives of the exact energy are trivially evaluated
\begin{align}
   \left(\frac{\partial E}{\partial \mathbf{h}}\right)_{\mu\nu} =  \sum_{pq} \gamma_{pq} C_{\mu p} C_{\nu q},\qquad   \left(\frac{\partial E}{\partial \mathbf{g}}\right)_{\mu\nu\lambda\sigma} =  \sum_{pqrs}\Gamma_{pqrs} C_{\mu p} C_{\nu q} C_{\lambda r} C_{\sigma s}  \, .
\end{align}

Because there is no dependence on other parameters, e.g. $\boldsymbol{\Omega}$, the only remaining partial derivative of the energy is the one with respect to the overlap of the AO basis functions, 
\begin{equation} \label{eq:energy_grad}
   \frac{dE}{dx} = \sum_{pq}\sum_{\mu\nu} \gamma_{pq} C_{\mu p} C_{\nu q} \frac{d h_{\mu\nu}}{dx} + \sum_{pqrs}\sum_{\mu\nu\lambda\sigma} \Gamma_{pqrs} C_{\mu p} C_{\nu q} C_{\lambda r} C_{\sigma s} \frac{d g_{\mu\nu\lambda\sigma}}{dx} + \frac{\partial E}{\partial \mathbf{S}} \frac{d\mathbf{S}}{dx}\,.
\end{equation}
where the only terms that depend on $\mathbf{S}$ are the MO coefficients $\mathbf{C}$. 

With the density matrices $\gamma_{pq}$ and $\Gamma_{pqrs}$being given, the first two terms of Eq.~\eqref{eq:energy_grad} are easy to evaluate: they require the evaluation of atomic orbital {\em core} derivatives. 
The last term of Eq.~\eqref{eq:energy_grad} is a little more involved as we need to find an expression for $\partial E/ \partial\mathbf{S}$ in terms of the one- and two-body reduced density matrices. The core derivative overlap integrals $d\mathbf{S}/dx$ can be computed by most electronic structure packages. We obtain for the last term
\begin{equation}\label{eq:energy_weighted_dm}
    \frac{\partial E}{\partial \mathbf{S}}\frac{d \mathbf{S}}{d x} = \sum_{\eta \zeta} \sum_{pq} \gamma_{pq} \sum_{\mu \nu} h_{\mu \nu} \frac{\partial}{S_{\eta \zeta}} \left(C_{\mu p} C_{\nu q}\right) \frac{d S_{\eta \zeta}}{d x} + \sum_{\eta \zeta} \sum_{p q r s} \Gamma_{p q r s} \sum_{\mu \nu \lambda \sigma} g_{\mu \nu \lambda \sigma} \frac{\partial}{S_{\eta \zeta}} \left(C_{\mu p} C_{\nu q} C_{\lambda r} C_{\sigma s}\right) \frac{d S_{\eta \zeta}}{d x}
\end{equation}
The quantities in the above expression depend on the MO coefficients $\mathbf{C}$. Because the MO coefficients depend on $\mathbf{S}$ we need to derive the explicit expressions for $\partial \mathbf{C} / \partial \mathbf{S}$. The step-by-step derivations are presented in Appendix~\ref{app:mo_ao}, and we find 
\begin{equation}
    \frac{\partial C_{\mu p}}{\partial S_{\lambda \sigma}} \frac{dS_{\lambda \sigma}}{dx} = -\frac{1}{2} \sum_{q} C_{\mu q} C_{\lambda q} C_{\sigma p} \frac{dS_{\lambda \sigma}}{dx}.
\end{equation}
With these, after a derivation presented in Appendix~\ref{app:mo_ao}, we find for derivatives of the one- and two-body terms with respect to the overlap matrix $\mathbf{S}$ 
\begin{equation}
        \frac{\partial E}{\partial \mathbf{S}}\frac{d \mathbf{S}}{d x} = -\sum_{pqm} \gamma_{qm} h_{pm} \frac{dS_{pq}}{dx} - 2 \sum_{pq}\sum_{rst} \Gamma_{qrst} g_{prst} \frac{dS_{pq}}{dx}
\end{equation}

The final expression for the energy derivative, after reindexing, is given by
\begin{equation}
    \label{eq:gradient_general}
   \frac{dE}{dx} = \left(\sum_{pq} \gamma_{pq} \frac{dh_{pq}}{dx} + \sum_{pqrs} \Gamma_{pqrs} \frac{dg_{pqrs}}{dx} \right)
   - \left(\sum_{pqm} \gamma_{pq} h_{mq} \frac{dS_{mp}}{dx} + 2 \sum_{pqrst}\Gamma_{pqrs} g_{tqrs} \frac{dS_{tp}}{dx}\right)\,.
\end{equation}

To calculate the force on the $A^{\text{th}}$ nucleus, $\mathbf{F}_A = -\frac{dE}{d\mathbf{R}_A}$, we use the above expression to find the energy gradient at nuclear position $\mathbf{R}_A$.
Following the Hellmann-Feynman theorem Eq.~\eqref{eq:Hellmann-Feynman}, we then convert the problem of calculating the energy gradient to the problem of calculating the expectation value of a derivative operator. We find that in second quantization the derivative operator is given by
\begin{equation}\label{eq:der_operator}
    \frac{dH}{d\mathbf{R}_A}= \sum_{pq}a^{\dag}_pa_q\left[\frac{dh_{pq}}{d\mathbf{R}_A}-\sum_mh_{mq}\frac{dS_{mp}}{d\mathbf{R}_A}\right]+\sum_{pqrs}a^{\dag}_pa^{\dag}_r a_qa_s\left[\frac{dg_{pqrs}}{d\mathbf{R}_A}-2\sum_{t}g_{tqrs}\frac{dS_{tp}}{d\mathbf{R}_A}\right] \, .
\end{equation}
In this equation, the second terms in each bracket which include the derivative of the overlap matrix $\mathbf{S}$ correspond to the Pulay force \cite{Helgaker1984Secondquantization}.
For later reference, we write coefficients of this operator in the same form as the Hamiltonian
\begin{align}
    T^{(F_A)}_{pq}&=\left[
    \frac{dh_{pq}}{d\mathbf{R}_A}
    -\frac12\sum_m(h_{mq}\frac{dS_{pm}}{d\mathbf{R}_A}
    +h_{pm}\frac{dS_{mq}}{d\mathbf{R}_A})
\right],\\
V_{pqrs}^{(F_A)}&=\left[\frac{dg_{pqrs}}{d\mathbf{R}_A}
-\frac12\sum_{t}(g_{tqrs}\frac{dS_{pt}}{d\mathbf{R}_A}
+g_{ptrs}\frac{dS_{tq}}{d\mathbf{R}_A}
+g_{pqts}\frac{dS_{rt}}{d\mathbf{R}_A}
+g_{pqrt}\frac{dS_{ts}}{d\mathbf{R}_A})\right].
\label{eq:force_coefficients_def}
\end{align}

\subsection{Force operators in plane wave bases}\label{sec:force_operators_plane_waves}

Plane waves are one of the most common basis sets used to model condensed matter systems. They are a natural basis for periodic systems and are independent of the atomic positions. However, their drawback is that many plane waves are typically needed to describe the wavefunctions accurately. When defined on a cubic reciprocal lattice, the plane wave basis functions take the form,
\begin{equation}
\phi_\bf{p}\left(\mathbf{r}\right) = \sqrt{\frac{1}{\Omega}} e^{-i \, \mathbf{k}_{\bf{p}} \cdot \mathbf{r}} \, ,
\end{equation}
where $\Omega$ is the computational cell volume and the reciprocal lattice vector in three dimensions is defined as
\begin{equation}
\label{eq:G}
\mathbf{k}_\bf{p} = \frac{2 \pi \mathbf{p}}{\Omega^{1/3}} \, , \qquad \qquad
\mathbf{p} \in G \, , \qquad \qquad 
G = \left[-\frac{N^{1/3}-1}{2},\frac{N^{1/3}-1}{2}\right]^3 \subset \mathbb{Z}^3 \, ,
\end{equation}
with $N$ being the number of plane waves. The molecular integrals can be evaluated analytically for the case of plane wave basis functions, leading to the following representation of the second-quantized electronic structure Hamiltonian in Eq.~\eqref{eq:second_quant_H},
\begin{equation}
    H=\underbrace{\sum_{\bf{p}} \frac{|\mathbf{k}_{\bf{p}}|^{2}}{2}a_{\bf{p}}^{\dag}a_{\bf{p}}
    -\frac{4\pi}{\Omega}\sum_{\bf{p\neq q}}\sum_{ A}^{\eta}\left(Z_A\frac{e^{i\mathbf{k}_{\bf{q-p}}\cdot \mathbf{R}_A}}{\left|\mathbf{k}_{\bf{q-p}}\right|^2}\right)a^{\dag}_{\bf{p}} a_{\bf{q}}}_{\text{one-electron term}}
    \underbrace{+\frac{2\pi}{\Omega}\sum_{\bf{p\neq q,s\neq 0}}\frac{1}{\left|\mathbf{k}_{\bf{s}}\right|^2} a^{\dag}_{\bf{p}} a^{\dag}_{\bf{q}} a_{\bf{q+s}}a_{\bf{p-s}}}_{\text{two-electron term}}.\label{eq:second_quant_PW_ham}
\end{equation}
Here and in what follows, we have omitted the electron spin for simplicity. An equivalent expression in first quantization can be written down \cite{BabbushContinuum},
\begin{align}
    H =&\underbrace{\sum_{i=1}^{\eta}\sum_{\mathbf{p}\in G} \frac{\left | \mathbf{k}_{\bf{p}}\right|^2}{2} \ket{\mathbf{p}}\!\bra{\mathbf{p}}_i
   -\frac{4\pi}{\Omega}\sum_{A=1}^{N_{\mathrm{a}}}\sum_{i=1}^{\eta}\sum_{\mathbf{p}\neq \mathbf{q}\in G}\left(Z_A\frac{e^{i\mathbf{k}_{\mathbf{q-p}}\cdot \mathbf{R}_A}}{\left|\mathbf{{k}_{\mathbf{q-p}}}\right|^2}\right)\ket{\mathbf{p}}\!\bra{\mathbf{q}}_i}_{\text{one-electron term}} \\
    &\underbrace{+\frac{2 \pi}{\Omega} \sum_{i,j=1}^{\eta}\sum_{\mathbf{p,q}\in G} \sum_{\substack{\bf{s}\in G_0\\(\mathbf{p+s})\in G\\(\mathbf{q-s})\in G}}\frac{1}{\left| \mathbf{k}_{\bf{s}} \right|^2} \ket{\mathbf{p + s}}\!\bra{\mathbf{p}}_i \ket{\mathbf{q-s}}\!\bra{\mathbf{q}}_j}_{\text{two-electron term}},
\label{eq:first_quant_ham}
\end{align}
where $\ket{\mathbf{p}}\!\bra{\mathbf{q}}_j$ is a shorthand for $I_1\otimes\cdots\otimes\ket{\mathbf{p}}\!\bra{\mathbf{q}}_j\otimes\cdots\otimes I_{\eta}$ and $G_0$ is $G$ from Eq.~\eqref{eq:G} excluding the zero mode. 

While a number of papers have analyzed the viability of quantum algorithms for simulating chemistry in second quantization with plane waves~\cite{BabbushLow,BabbushSpectra,Kivlichan2019,Low2018,Su2020}, that approach faces some significant challenges. In particular, in second quantization the number of qubits required scales as the number of plane waves. This is a problem because often hundreds of thousands of plane waves might be required to obtain a suitable wavefunction accuracy. However, there have been proposals for fault-tolerant algorithms using plane waves in first quantization~\cite{BabbushContinuum, su2021fault}. In first quantization the number of qubits required scales only as the logarithm of the number of plane waves $N$ and linearly in the number of electrons, $\eta$. Algorithms have been demonstrated \cite{BabbushContinuum} for time-evolution or state preparation of molecular systems that scale only as
\begin{equation}
    \widetilde{\cal O}\left(\eta^{3} N^{1/3} \Omega^{-1/3}\right) \, .
\end{equation}
Due to the sublinear dependence on $N$, with these approaches one can conceivably perform simulations with millions of plane waves. 

An additional advantage of the plane wave basis is that the overlap matrix elements in the plane wave representation are reduced to 
\begin{equation}
    S_{\bf{pq}} = \frac{1}{\Omega}\int_{\Omega} d\bf{r} e^{i \bf{k_{q-p}} \cdot \bf{r}} = \delta_{\bf{pq}},
\end{equation}
and thus the overlap matrix contributions to the derivative operator from Eq.~\eqref{eq:der_operator} are identically zero. This suggests that representing the electronic structure Hamiltonian in first-quantized plane waves basis is a promising avenue for calculating energy derivatives of chemical systems.

In a plane wave basis, the only dependence of $H$ on the nuclear positions $\bf{R} _A$ is in the one-body term, which implies that (as expected for a non-atomic centered basis set) the force operator in plane waves is a strictly one-body operator.
(The same is true for other first-order derivatives that do not affect the electron-electron Coulomb, such as an applied electric or magnetic field.)
This operator may be further simply diagonalized by the fermionic fast Fourier transform (FFFT)~\cite{Verstraete2009,Ferris2014,BabbushLow}, in a similar manner to the potential term of the original Hamiltonian.
This is simplest to demonstrate in second quantization, so we will perform the calculation there first and then transform to our target first-quantized form.
Differentiating Eq.~\eqref{eq:second_quant_PW_ham} with respect to the $A^{\text{th}}$ nuclear co-ordinate $\mathbf{R}_A$ gives us
\begin{equation}
    \frac{d H}{d \mathbf{R}_A} = -\frac{4\pi i Z_A}{\Omega}\sum_{\bf{p\neq q}} \frac{\mathbf{k}_{\bf{q-p}} \, e^{i\bf{k_{q-p}}\cdot \bf{R}_A}}{\left| \bf{k_{q-p}}\right|^2} a^{\dag}_{\bf{p}} a_{\bf{q}} = \sum_{\bf{p \neq q}} f\left(\bf{q-p},\, A\right) e^{i\bf{k_{q-p}}\cdot \bf{R}_A} a_{\bf{p}}^\dagger a_{\bf{q}}, \qquad f\left(\bf{s},\, A\right) = -\frac{4 \pi i Z_A \bf{k}_{\bf{s}} }{\Omega \left| \bf{k_{s}}\right|^2} \,.
\end{equation}
Note that this is a vector-valued derivative, here $\bf{R}_A$ is the $3$-dimensional nuclear position vector (individual components of this vector may be obtained by taking individual components of the wavevector $\bf{k_{s}}$ in $f(\bf{s},A)$).
In second quantization, the FFFT performs the following single-particle rotation,
\begin{equation}
a^\dagger_{\bf{s}} = {\rm FFFT}^\dagger c^\dagger_{\bf{s}} \, {\rm FFFT} =  \sqrt{\frac{1}{N}}  \sum_{\bf{p}} c^\dagger_{\bf{p}} e^{- i \,\bf{k_s} \cdot \bf{r_p}},
\quad \quad \quad 
a_\bf{s} ={\rm FFFT}^\dagger c_\bf{s} \, {\rm FFFT} =  \sqrt{\frac{1}{N}}  \sum_{\bf{p}} c_{\bf{p}} e^{i \,\bf{k_s} \cdot \bf{r_p}}
\end{equation}
where $\bf{r_p} = \bf{p} (\Omega / N)^{1/3}$.
Under this transformation, the gradient of the electronic structure Hamiltonian becomes
\begin{align}
\frac{d H}{d \bf{R}_A} & = \sum_{\substack{\bf{p \neq q}}} f\left(\bf{q-p},\, A\right) e^{i \, \bf{k_{q-p}} \cdot \bf{R}_A} a^\dagger_{\bf{p}} a_{\bf{q}} \nonumber \\
& =  \sum_{\substack{\bf{p \neq q}}}f\left(\bf{q-p},\, A\right)  e^{i \, \bf{k_{q-p}} \cdot \bf{R}_A}
\left(\sqrt{\frac{1}{N}} \sum_{\bf{p}'} c_{\bf{p}'}^\dagger e^{- i \, \bf{k_p} \cdot \bf{r_{p'}}}\right)
\left(\sqrt{\frac{1}{N}} \sum_{\bf{q}'} c_{\bf{q}'} e^{i \, \bf{k_q} \cdot \bf{r_{q'}}}\right)\nonumber\\
& =  \frac{1}{N} \sum_{\substack{\bf{p \neq q}}} f\left(\bf{q-p},\, A\right)  e^{i \, \bf{k_{q-p}} \cdot \bf{R}_A} \sum_{\bf{p', q'}} c_{\bf{p}'}^\dagger c_{\bf{q}'} e^{i \, \bf{k_q} \cdot \bf{r_{q'- p'}} } e^{- i \,\bf{k_{p-q}} \cdot \bf{r_{p'}}}\nonumber\\
& =  \frac{1}{N} \sum_{\substack{\bf{p', q'}}} \sum_{\substack{\bf{p \neq q}}} f\left(\bf{q-p},\, A\right) e^{i \, \bf{k_{q-p}} \cdot \left(\bf{R}_A + \bf{r_{p'}}\right)}
\left(c_{\bf{p'}}^\dagger c_{\bf{q}'} e^{i \, \bf{k_q} \cdot \bf{r_{q'-p'}}} \right).
\end{align}
Recognizing that $(\bf{q - p})$ spans the full set of momentum vectors in our system due to aliasing, we can replace the sum over $\bf{q \neq p}$ and the indices $\bf{q-p}$ and $\bf{q}$ with a sum over $\bf{s} \neq 0$ and $\bf{p}$. Following this reindexing, our gradient operator diagonalizes immediately,
\begin{align}
\frac{d H}{d \bf{R}_A}  & =  \frac{1}{N} \sum_{\substack{\bf{p', q'}}}  \left(\sum_{\substack{\bf{s} \neq 0}} f\left(\bf{s},\, A\right)e^{i \, \bf{k_{S}} \cdot \left(\bf{R}_A + \bf{r_{p'}}\right)}\right)
\left( c_{\bf{p'}}^\dagger c_{\bf{q'}}\sum_{\bf{q}} e^{i \, \bf{k_q} \cdot \bf{r_{q'-p'}}} \right)\nonumber\\
& =   \sum_\bf{p} \left(\sum_{\substack{\bf{s} \neq 0}} f\left(\bf{s}, \, A\right) e^{i \, \bf{k_s} \cdot \left(\bf{R}_A + \bf{r_{p}}\right)}\right) c_{\bf{p}}^\dagger c_{\bf{p}} \nonumber\\
& =  - \frac{4 \pi i Z_A}{\Omega} \sum_{\bf{p}} \sum_{\substack{\bf{s} \neq 0}}  \frac{\bf{k_s} \, e^{i \, \bf{k_{s}} \cdot \left(\bf{R}_A + \bf{r_{p}}\right)}}{\left| \bf{k_s} \right|^2} c_{\bf{p}}^\dagger c_{\bf{p}}
\label{eq:pwd_u}
\end{align}
where we have used the fact that the summation grouped on the right side of the first equation is equal to zero unless $\bf{p' = q'}$. This is because the negative modes of $\bf{k_q}$ will have exactly the opposite phase as the positive modes of $\bf{k_q}$.

We now transform the derivative operator into a first-quantized representation.
In first quantization, we store our wavefunction by having a computational basis that encodes configurations of the electrons in $N$ basis functions such that a configuration is specified as $\ket{\phi_1 \phi_2 \cdots \phi_{\eta}}$ where each $\phi_j$ encodes the index of an occupied basis function.
Each $\phi_j$ may be specified in binary, making the space complexity only ${\cal O}(\eta \log N)$.
We can translate Eq.~\eqref{eq:pwd_u} into first quantization to give
\begin{align}
\label{eq:force_plane_wave_first}
     \frac{d H}{d \bf{R}_A} & = -\frac{4\pi i Z_A}{\Omega} \sum_{i=1}^{\eta} \sum_{\substack{\bf{p,q} \in G \\ \bf{p\neq q}}} \frac{\bf{k_{q-p}} \, e^{i\bf{k_{q-p}}\cdot \bf{R}_A}}{\left| \bf{k_{q-p}}\right|^2} \ket{\bf{p}}\!\bra{\bf{q}}_i 
      = {\rm QFT} \left( - \frac{4 \pi i Z_A}{\Omega} \sum_{i=1}^{\eta} \sum_{\bf{p,s} \in G_0} \frac{\bf{k_s} \, e^{i \, \bf{k_s} \cdot \left(\bf{R}_A - \bf{r_{p}}\right)}}{\left| \bf{k_s} \right|^2} \ket{\bf{p}}\!\bra{\bf{p}}_i \right){\rm QFT}^\dagger
\end{align}
where the QFT is the quantum Fourier transform with aliased frequencies (the first quantized version of the FFFT from Eq.~\eqref{eq:pwd_u}). The QFT can be implemented with Toffoli gate complexity $\widetilde{\cal O}(\eta)$~\cite{BabbushLow}.
Note that this guarantees that all force operators are mutually diagonal under the FFFT/QFT, which in turn implies that all force operators commute.

\subsection{Error tolerance for applications}
\label{sec:error_tolerance}

The most widely-used energy derivatives of a molecular system are nuclear forces.
The first application we consider is geometry optimization, where nuclear derivatives are used to find the geometry of the molecule with the lowest energy on the potential energy surface.
The second application is molecular dynamics, where the nuclear positions are propagated through time by a classical differential equation within the Born-Oppenheimer approximation.
At each time, the forces on the nuclei determine their next position. Both of these applications rely on the nuclear derivatives of the energy to repeatedly update the positions of the nuclei.
This is a process where a small error in each step can quickly accumulate. The tolerable error on the forces is an important parameter in the scaling of the quantum algorithms to calculate them.
In the next two subsections we investigate the error level that is acceptable.

\subsubsection{Geometry optimization}

The error tolerance of the forces for the geometric relaxation of a structure depends strongly on the system.
The geometries of systems with a rather steep potential energy surface can be determined with relative low accuracy of the forces.
For example, Gaussian sets the default thresholds for convergence of the maximum force to $0.9$~mHa/Å and the RMS of the error of single force component to $0.6$~mHa/Å~\cite{g16}.
However, for systems where forces are smaller because the potential energy surface is shallow, typically the geometries need to be determined by relaxing the atomic positions until the forces are one magnitude smaller~\cite{gross_bond-order_2012}. 

\subsubsection{Molecular dynamics}

\begin{figure}[tb]
    \centering
	\includegraphics[width=0.6\textwidth]{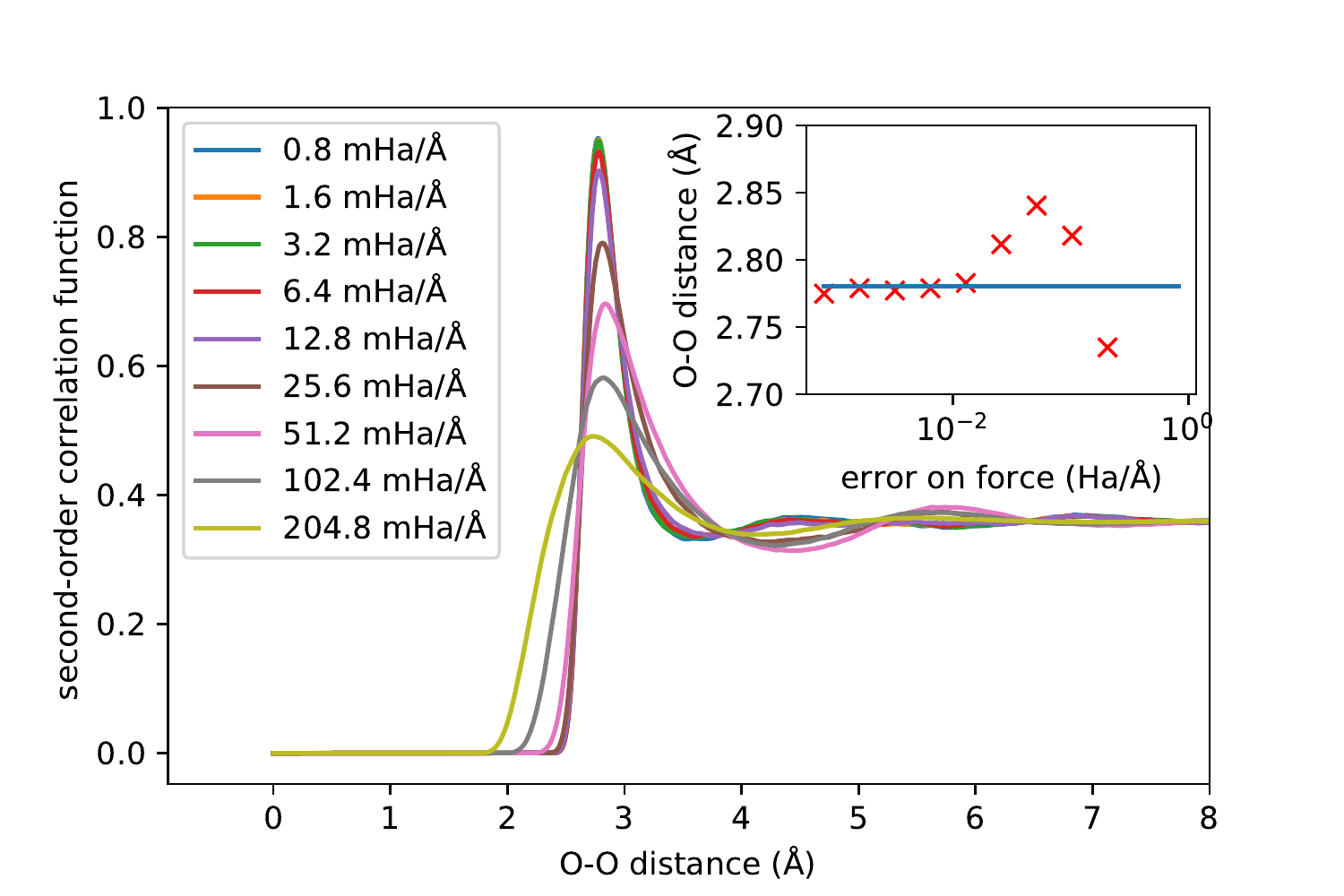}
\caption{The radial distribution function which is second-order correlation function of two oxygen atoms as function of their distance.
The MD run contained 216 water molecules in a periodic box with a temperature of 299~K.
The radial distribution function is plotted for different RMS errors on the force from $409.6$~mHa/Å to $0.8$~mHa/Å.
In the inset the peak positions of radial distribution function is plotted as function of the error on the force.}
\label{fig:rdf}
\end{figure}

Error bounds on forces required for MD simulations will again depend strongly on the system studied.
To find a simple baseline for a target accuracy for force components in this work, we focus on the required error tolerance for performing semi-classical MD simulations of a water system.
A quantum device would be used in this situation as a subroutine to provide accurate estimates of the classical potential, employing the TIP3P water model~\cite{jorgensen_1983}.
Here, the MD simulations were performed with Atomic Simulation Environment software package~\cite{ase-paper}.

As a proxy for simulation convergence, we study the $2$-particle radial distribution function 
\begin{equation}
    g^{(2)}(\mathbf{R}_1,\mathbf{R}_2)=\left(\frac{V}{N_a}\right)^2\frac{N_a(N_a-1)}{Z_{N_a}}\int_{V^{N_a-2}} \exp\big[-\beta E_\mathrm{tot}(\mathbf{R}_1,\ldots,\mathbf{R}_{N_a})\big]d\mathbf{R}_3\ldots d\mathbf{R}_{N_a},\label{eq:radial_distribution}
\end{equation}
of a $N_a = 3 \times 216$-atom system in a periodic box of volume $V$.
Here, $Z_{N_a}$ is the partition function of the $N_a$-particles, and $\beta=1/$299~K is the inverse temperature of the system.
As the radial distribution function $g^{(2)}$ is independent of translations and rotations of $\mathbf{R}_1$ and $\mathbf{R}_2$ about the origin, the data it contains can be found solely in the radial term:
\begin{equation}
    g^{(2)}(R)=\int_{V^2}\delta\big(\|\mathbf{R}_1-\mathbf{R}_2\|-R\big)g^{(2)}(\mathbf{R}_1,\mathbf{R}_2)d\mathbf{R}_1\,d\mathbf{R}_2
\end{equation}
To obtain a physically-relevant quantity, we isolate the pair distribution of the $N_a/3 = 216$ oxygen atoms and ignore the hydrogen atoms.
The $2$-particle radial distribution function is a macroscopic quantity and is often used to benchmark different water models.
The deviation from the ideal radial distribution function is a measure of the quality of the run.
We test for which size of errors in the force we can still reproduce the error free radial distribution function.
To achieve this we performed micro-canonical MD simulations with 36000 time steps, and at each 1~fs time step we added a random error term to the forces of water molecules. 
This error was sampled from a Gaussian distribution with a given RMS of the force error, with separate error terms drawn independently.
The radial distribution function of Eq.~\eqref{eq:radial_distribution} is averaged over all time steps of the simulation.
From these MD runs we can determine the size of errors for which we can still reproduce the error-free radial distribution function which we take as the ground truth.
Fig.~\ref{fig:rdf} shows that the radial distribution function rapidly converges to the error-free radial distribution function.
As a metric for convergence, we focus on the largest feature in the system (the peak around $2.8$~Å, and plot the error in the peak position (Fig.~\ref{fig:rdf} inset).
From this we conclude that for water the RMS error for molecular dynamics needs to be smaller than $6.4$~mHa/Å to reproduce macroscopic properties as the radial distribution function.
Our findings are in agreement with studies where {\em ab initio} MD simulations based on quantum Monte Carlo calculations were employed~\cite{luo_ab_2014}. 
Beyond radial distribution functions, quantities such as vibrational density of states are expected to require
higher precision of the force evaluation.~\cite{Marsalek2017Apr}

The geometry optimization seems to be more stringent by a order of magnitude on required accuracy of the forces than MD simulations. In MD simulations errors can average out. Later in this paper we consider the 2-norm as the parameter for accuracy of the forces. For the 2-norm the sum is taken over all forces and its components. Therefore, the 2-norm is a extensive quantity and depends on the number of atoms. We can convert the RMS error to the $2$-norm by multiplying RMS error with $\sqrt{3 N_a}$. For example, for the 216 water molecules we obtain a 2-norm of the error of $282.2$~mHa/Å.

\section{Computation of force vectors in NISQ}
\label{sec:nisq}

\begin{figure*}[tb]
        \centering
        \includegraphics[width=1\textwidth]{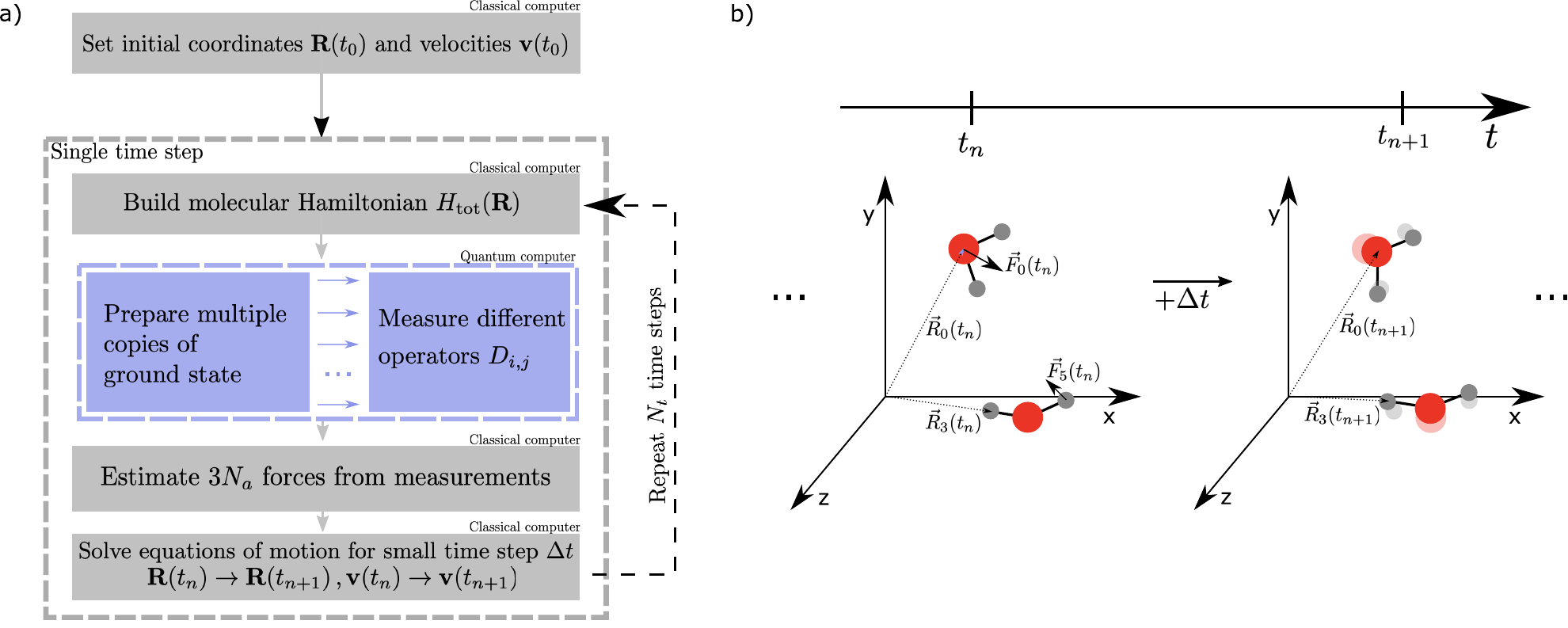}
\caption{Schematic representation of molecular dynamics enhanced by a NISQ device. (a) Flowchart highlighting the hybrid setup in which a NISQ device is used to calculate the forces, while a classical computer updates the nuclei coordinates $\mathbf{R}$, velocities $\mathbf{v}$ and the molecular Hamiltonian. A typical MD simulation requires $N_t=\mathcal{O}(10^6-10^9)$ time steps with a step size of $\Delta{t}=\mathcal{O}(10^{-15})s$ \cite{hollingsworth2018molecular,lazim2020advances}. (b) Example of a single update of the nuclei coordinates $\mathbf{R}$ of two water molecules with red balls and grey balls denoting oxygen and hydrogen atoms respectively. $\vec{F}_{i}$ denotes the three-dimensional force vector on the $i$-th atom. }
\label{fig:md_nisq}
\end{figure*}

To optimize quantum algorithms for NISQ quantum computers, we must reduce quantum circuit depth wherever practical.
Near-term proposals for quantum chemistry achieve this by preparing approximate ground states~\cite{Mcclean16Theory,QITE,Polla2021}, from which energies may be extracted by state tomography.
As long as the approximate state is variationally optimized within the active space considered on the device, the energy derivatives yielded by the Hellman-Feynman theorem are accurate for the variational energy; no further corrections need to be made to the methods outlined in Sec.~\ref{sec:energy_gradient}.
In this work, we assume access to the preparation of an initial state $|\psi\rangle$ that satisfies the Hellman-Feynman theorem Eq.~\eqref{eq:Hellmann-Feynman} ($\frac{d E}{d R_i}=\langle\psi|\frac{d H}{d R_i}|\psi\rangle$), and focus on the optimization of the measurement of this state to extract an estimate of $\langle\psi|\frac{d H}{d R_i}|\psi\rangle$.   While many approaches exist for reconstructing, the expectation, tomographic approaches have a major role in recent quantum computing works~\cite{Huggins19Efficient,Cotler20Quantum,Bonet20Nearly,Zhao20Fermionic,huang2020predicting}.
However, to the best of our knowledge no-one has optimized estimation methods for the measurement of vectors of operators prior to now.

State tomography in NISQ is complicated by the fact that simultaneous direct measurement of multiple operators is only possible in quantum mechanics when all operators mutually commute.  More broadly, the parameter estimation or partial tomographic protocols used to estimate a gradient consist of three key steps:
\begin{enumerate}
    \item Define a set of basis rotations $\{Q_j\}$ for which low-depth quantum circuits are known.
    \item For each basis rotation $Q_j$, prepare the state $|\psi\rangle$ $M_j$ times, apply the quantum circuit, and then destructively measures the system in the computational basis.
    \item Estimate the set $\{\langle\psi|\frac{d H}{d R_i}|\psi\rangle\}$ from the observed measurement data.
\end{enumerate}
In a NISQ cost model, the target is to reduce the total number $M=\sum_jM_j$ of preparations of $|\psi\rangle$, or experiment `shots', while targeting some error bound on the set of energy gradient estimates.

If the expectation value estimation (step 3 above) is linear and unbiased, the error on individual estimates $\langle\psi|\frac{d H}{d R_i}|\psi\rangle$ may be calculated by variance propagation.
Such an estimation corresponds to the decomposition of the gradient operator as a linear combination
\begin{equation}
    \frac{d H}{d R_i}=\sum_jD_{i,j},
\end{equation}
where the $D_{i,j}$ are operators that are diagonalized by the $j$th basis rotation $Q_j$.
If $Q_j$ diagonalizes $D_{i,j}$, $\langle\psi|D_{i,j}|\psi\rangle$ may be estimated by averaging over the $M_j$ destructive measurements taken in the $Q_j$ basis.
The variance in this estimation is given by
\begin{equation}
    \mathrm{Var}\Big[\langle\psi|D_{i,j}|\psi\rangle\Big]=\frac{\sigma_{i,j}^2}{M_i},\hspace{1cm}\sigma^2_{i,j}=\langle\psi|D^2_{i,j}|\psi\rangle - \langle\psi|D_{i,j}|\psi\rangle^2.
\end{equation}
As expectation values are linear, we have
\begin{equation}
    \frac{dE}{dR_i}=\bigg\langle\psi\bigg|\frac{d H}{d R_i}\bigg|\psi\bigg\rangle =\sum_j\langle\psi|D_{i,j}|\psi\rangle.
\end{equation}
Then, as each $D_{i,j}$ is measured independently, the variance of the estimation propagates in the usual way to the variance in an estimation of $\frac{dE}{dR_i}$
\begin{equation}
\epsilon_i^2=\mathrm{Var}\bigg[\frac{dE}{dR_i}\bigg]=\sum_j\frac{\sigma_{i,j}^2}{M_j}.
\end{equation}
These errors may be captured within a $3N_a$-dimensional error vector $\mathbf{\epsilon}$.
In practice estimates of $\sigma_{i,j}^2$ are not known in advance, making exact estimation of $\epsilon_i$ and subsequent parameter optimization difficult.
Instead, bounds on $\sigma_{i,j}$ are often substituted; we will introduce various such methods throughout this section.
Some of these bounds are in practice quite weak, which implies that fair comparison of the results described in this section may not be possible.

The general state tomography method described above leaves open a large number of parameters for optimization: the rotations $Q_j$, the shot allocation $M_j$, and the choice of operators $D_{i,j}$ in the decomposition of $\frac{dH}{dR_i}$.
In the following sections, we will describe and compare various methods that attempt to optimize these choices.
Complete optimization of each of these choices is not practical due to the sheer number of parameters and the (classical) cost of evaluating cost functions.
Optimizing basis rotations $Q_j$ to diagonalize multiple operators is in general an NP-hard problem~\cite{Verteletskyi19Measurement}.
Moreover, the lack of precise knowledge of $\sigma_{i,j}$ implies that the cost function may be difficult to estimate for the purposes of optimization.
However, various heuristic techniques are widely known, and many of these can be shown to achieve asymptotically optimal results.

\subsection{Basis rotation choices}
\label{sec:basis_rotations}

When choosing the set of basis rotations $Q_j$ for a state tomography protocol, one must try find operators $D_{i,j}$ that are diagonal in the $Q_j$ basis.
Calculating such operators is typically as difficult as simulating the circuit, so rotations $Q_j$ are typically chosen to be classically easy to simulate.
One must take further care that the $D_{i,j}$ are not exponentially difficult to express in order for step 3 in the general method above to be computationally feasible.
Drawing $Q_j$ from one of a few well-known sets of quantum circuits defined below typically solves this problem.

The commonly used Clifford circuits form the first example. These circuits preserve the Pauli group $\mathbb{P}^N=\{I,X,Y,Z\}^N$ modulo complex phases.
If $Q_j$ is a Clifford circuit, so is $Q_j^{\dag}$ and the algebra formed by the set $Q_j^{\dag}Z_nQ_j$ yields all possible operators $D_{i,j}$ that are diagonal in this basis.
Moreover, Clifford circuits can in principle be constructed to simultaneously diagonalize any set of mutually commuting elements of $\mathbb{P}^N$.
It is relatively easy to design a set of Clifford basis rotations $Q_j$ in this manner:
\begin{enumerate}
    \item Decompose all operators $\frac{dH}{dR_i}$ into a linear combination of Pauli operators following the Jordan-Wigner, Bravyi-Kitaev, or alternative fermion-to-qubit transformation.
    \item Subdivide the set $\mathcal{S}$ of Pauli operators that appear in at least one linear combination into commuting subsets $\mathcal{S}_j$ (i.e. so that all Pauli operators within each $\mathcal{S}_j$ commute).
    \item For each subset, find an appropriate basis rotation $Q_j$.
\end{enumerate}
A disadvantage to the above is that Clifford circuits to diagonalize mutually-commuting operators can be relatively deep~\cite{Yen19Measuring,Crawford21Efficient}.
This can be simplified by adding the requirement that the subsets $\mathcal{S}_j$ contain not commuting Pauli operators, but amenable Pauli operators.
Two Pauli operators are amenable if, on each qubit the tensor factor $I,X,Y,Z$ of both operators is the same or the tensor factor of at least one operator is the identity.
(For example, $YY$ and $YI$ are amenable, but $YY$ and $XX$ are not.)
Amenable Pauli operators can be mutually diagonalized by single-qubit basis rotations $Q_j$, making this a practical subdivision $\mathcal{S}_j$.
In general subdividing $\mathcal{S}$ into the minimum number of $\mathcal{S}_j$ is a known NP-hard problem~\cite{Verteletskyi19Measurement}, though relative success has been found in heuristics~\cite{Bonet20Nearly} or brute-force optimization methods~\cite{Crawford21Efficient}.
To date these methods have focused on minimizing the number of subsets that contain all Pauli elements that make up the fermionic $1$- and $2$-RDM, for which an $O(N^2)$ bound is known and has been achieved~\cite{Bonet20Nearly}.
However, most of these methods have not considered the subsequent allocation of measurements $M_j$ and the subsequent cost in wall-clock time to estimate one or more expectation values to a given accuracy.

A second set of circuits relevant for diagonalizing Hamiltonians and force operators in chemistry are Givens rotation circuits.
These correspond to evolution by a one-body fermionic operator
\begin{equation}
    Q_j=e^{i\sum_{n,m}Q_j^{n,m}a_n^{\dag}a_m},\label{eq:Givens}
\end{equation}
and are classically tractable to calculate as they map single creation and annihilation operators to each other.
\begin{equation}
    Q_ja_nQ_j^{\dag}=\sum_m[e^{iq_j}]_{n,m}a_m,\label{eq:sq_basis_rotation}
\end{equation}
where $q_j$ is the $N\times N$ Hermitian matrix with elements taken from Eq.~\eqref{eq:Givens}.
Givens rotation circuits are relatively low depth; an arbitrary Givens rotation may be implemented in depth $2N$ on a linear array using precisely $N^2$ two qubit gates~\cite{Kivlichan18Quantum, Google20Hartree}.
The above may be slightly generalized to the set of fermionic Gaussian unitaries~\cite{Zhao20Fermionic}
\begin{equation}
    Q_j=e^{\sum_{p,q}g_j^{p,q}\gamma_p\gamma_q},
\end{equation}
where $\gamma_m$ and $\gamma_n$ are anti-commuting Majorana operators
\begin{align}\label{eq:majorana_operators}
    \gamma_{2n}=a_n+a_n^\dagger && \gamma_{2n+1}=-\mathrm{i}(a_n-a_n^\dagger).
\end{align}
This strictly contains the set of Givens rotations, and also allows for Bogoliubov-style rotations between $a_n$ and $a_n^{\dag}$.

A low-cost method for constructing Givens rotation circuits $Q_j$ to target a two-body fermionic operator is to factorize the operator~\cite{Motta2018,Huggins19Efficient}.
Starting from the operator in its chemist formulation
\begin{align}
    A = \sum_{\sigma \in \{\uparrow,\downarrow\}}\sum_{p,q}T_{pq}a_{p,\sigma}^\dag a_{q,\sigma} + \underbrace{\sum_{\alpha,\beta\in \{\uparrow,\downarrow\}}\sum_{p,q,r,s}V_{pqrs}a_{p,\alpha}^\dag a_{q,\alpha} a_{r,\beta}^\dag a_{s,\beta}}_{=V}\,,
    \label{eq:chemist_notation}
\end{align}
we reshape the 4-rank tensor $V_{pqrs}$ into a 2-rank tensor $M_{(pq),(rs)}$. A direct diagonalization (or a Cholseky decomposition) of the flattened version of $V_{pqrs}$ yields
\begin{align}
    V&=\sum_{\ell=1}^{L} w_\ell\left(\sum_{\sigma \in \{\uparrow,\downarrow\}}\sum_{p,q=1}^{N/2} g^{(\ell)}_{pq} a^\dagger_{p,\sigma} a_{q,\sigma}\right)^2=\sum_{\ell=1}^{L} \left(\sum_{\sigma \in \{\uparrow,\downarrow\}}\sum_{p,q=1}^{N/2} W^{(\ell)}_{pq} a^\dagger_{p,\sigma} a_{q,\sigma}\right)^2=\sum_{\ell=1}^LW^{(\ell)2}
    \label{eq:single_factorized_hamiltonian}
\end{align}
with $g_{pq}^{(\ell)}$ and $w_\ell$ representing the eigenvectors and eigenvalues respectively. Further diagonalization of the squared single-body operators yields \cite{vonBurg2020}
\begin{align}\label{eq:double_low_rank_factorization}
    V=\sum_{\ell=1}^{L} U^{(\ell)} \left(\sum_{\sigma,\in \{\uparrow, \downarrow\}} \sum_{p=1}^{M_\ell} f_{p}^{(\ell)} n_{p,\sigma} \sum_{\sigma' \in \{\uparrow, \downarrow\}}\sum_{q=1}^{M_\ell} f_{q}^{(\ell)}  n_{q,\sigma'} \right) U^{(\ell)\dagger}\,,
\end{align}
with $f_{p}^{(\ell)}$ the eigenvalues of $W^{(\ell)}_{pq}$ and $U_\ell$ the unitaries performing the diagonalization, which can be expressed as a single-particle change of basis unitary
\begin{equation}\label{eq:basis_change_def}
U^{(\ell)} = \exp \left(-\sum_{p=1}^N \sum_{q=1}^{M_\ell} \kappa^{(\ell)}_{pq} \left(a^\dagger_p a_q - a^\dagger_q a_p\right)\right) \qquad \qquad U^{(\ell)} W^{(\ell)} U^{(\ell)\dagger} = \sum_{\sigma \in \{\uparrow, \downarrow\}}\sum_{p=1}^{M_\ell} f_p^{(\ell)} n_{p,\sigma}\,,
\end{equation}
and the $\kappa_{p,q}$ are obtained from the Givens rotation procedure in \cite{Kivlichan18Quantum}. Similary, the one-body fermionic operator may be diagonalized by a single Givens rotation $Q_j$, as one simply takes the rotation $q_j$ that diagonalizes the corresponding $N\times N$ one-body matrix (following Eq.~\eqref{eq:sq_basis_rotation}).

If the operator $A$ given in Eq.~\eqref{eq:chemist_notation} is the electronic structure Hamiltonian in a generic second-quantized basis, we have that $L = \widetilde{\cal O}(N)$ and $M < N$ and in some special cases $M_l = {\cal O}(\log N)$ \cite{Motta2018}.
This implies that one may estimate the expectation value of a Hamiltonian with $\widetilde{\mathcal{O}}(N)$ basis rotations $Q_j$.
As we show in Sec.~\ref{sec:block_encoding_second_quantization}, this extends to a bound on the number of Givens rotations required to factorize a single derivative operator
\begin{equation}
    \frac{dH}{dR_i}=\sum_{\ell=1}^LW^{(i,\ell)^2},\hspace{1cm} W^{(i,\ell)}=\sum_{\sigma \in \{\uparrow,\downarrow\}}\sum_{p,q=1}^{N/2} W^{(i,\ell)}_{pq} a^\dagger_{p,\sigma} a_{q,\sigma}.\label{eq:single_factorized_derivative}
\end{equation}
However, factorizations do not typically parallelize; the set of Givens rotation circuits that measure $\frac{dH}{dR_i}$ will not typically allow estimation of $\frac{dH}{dR_{i'}}$.
Moreover, it was found recently that the $L\sim\widetilde{\mathcal{O}}(N)$ scaling is relatively delicate; subtracting operators from one derivative will tend to yield an operator that is no longer low-rank~\cite{Rubin2021}.
This implies that it is likely not possible to significantly parallelize factorized methods, and the number of Givens rotations required to measure $3N_a$ derivative operators likely scales as $\widetilde{\mathcal{O}}(N_aN)$.

An obvious question to ask is whether the set of fermionic Gaussian unitaries and Clifford circuits intersect.
The answer to this question is yes: the intersection of these operators are the fermionic Gaussian Clifford unitaries, which are generated by the set of Majorana swap operators
\begin{equation}
    e^{\frac{\pi}{4}\gamma_i\gamma_j},\label{eq:Majorana_permutation}
\end{equation}
and correspond to the symmetric permutation group $\mathrm{Sym}(2N)$ (being permutations of indices of Majoranas).
For the sake of measurement, the effect of a Majorana permutation $Q_j$ is to pair the set of $2N$ Majorana operators: to choose a set of disjoint pairs $\{(\gamma_p,\gamma_q)\}$ containing all operators, and permute $p\rightarrow 2n$, $q\rightarrow 2n+1$ for some $n$.
This maps the Hermitian operator $i\gamma_p\gamma_q\rightarrow Z_n$, which implies that any linear combination of products of the pairs are diagonalized by $Q_j$.
This allows simultaneous measurement of $\binom{2N}{2} = \frac{N(N-1)}{2}$ linearly-independent $2$-body fermionic terms, which is optimal, and the basis for the best-known measurement schemes for the estimation of arbitrary $2$-body fermionic operators~\cite{Bonet20Nearly,Zhao20Fermionic}.

\subsection{Parallelized importance sampling}
\label{sec:NISQ_DirectMeasurements}

Once an optimal set of basis rotations $Q_j$ and operators $D_{i,j}$ have been chosen, it remains to allocate the number of shots $M_j$ to each $Q_j$.
Here, we target minimizing a given cost function $f(\{M_j\},\{\sigma_{i,j}\})$ while keeping the total number of measurements $M=\sum_j M_j$ constant (or vice-versa).
This may be achieved by Lagrangian methods~\cite{Rubin18Application, wecker2015progress}, this methodology being a form of importance sampling over the expectation values $\langle D_{i,j}\rangle$.
Such methods entail adding the total number of measurements as a constraint to the cost function with a Lagrangian multiplier $\lambda$, giving a Lagrangian
\begin{equation}
    \mathcal{L}=f(\{M_j\},\{\sigma_{i,j}\})+\lambda\Big(\sum_jM_j-M\Big).\label{eq:gen_lagrangian}
\end{equation}
The solution to the problem is then achieved by minimizing $\mathcal{L}$ with respect to all free parameters: $M_j$ and $\lambda$.
(See Appendix~\ref{app:ImportanceSamplingLagrange} for more details and explicit calculations of the optimizations used in the text.)
Crucially, $\sigma_{i,j}$ is typically not known {\em a priori}.
In principle $\sigma_{i,j}$ can be estimated during the expectation value estimation procedure, which could be used to adaptively optimize the distribution of the $M_j$.
However, typically in the literature a range of bounds $\bar{\sigma}_{i,j}>\sigma_{i,j}$ are used instead.
We will discuss the known bounds in detail in the next section.

In order to perform the above minimization procedure, we must define the cost function $f(\{M_j\},\{\sigma_{i,j}\})$.
This is complicated by the fact that we estimate $3N_a$ force components, and must combine the error on each into a single cost function.
This can be achieved by defining a norm on the error vector $\boldsymbol{\epsilon}=\mathbb{E}(\frac{dH}{d\mathbf{R}}-\widetilde{\frac{dH}{d\mathbf{R}}})$.
In Sec.~\ref{sec:error_tolerance}, we saw that the $2-$norm is a reasonable proxy to bound the error in molecular dynamics simulations.
An additional issue presents itself as we should take into account the covariance between different force components (assuming that we do not measure these independently).
However, this may be circumvented if we take the 2-norm squared as our cost function:
\begin{align}
     f(\{M_j\},\{\sigma_{i,j}\})=\mathbb{E}\left(\left\|\frac{\widetilde{d E}}{d \mathbf{R}} - \frac{{d E}}{d \mathbf{R}} \right\|_2^2 \right) & = \mathbb{E}\left(\sum_i\left(\frac{\widetilde{d E}}{d {R}_i} - \frac{{d E}}{d {R}_i}\right)^2\right) = \sum_i \epsilon_i^2=\sum_{i,j}\frac{\sigma_{i,j}^2}{M_j}.
     \label{eq:error_2_norm}
\end{align}
We finally write $f(\{M_j\},\{\sigma_{i,j}\})=\epsilon^2$, where $\epsilon$ is the RMS error in our final force vector.

The advantage of targeting the norm of the error vector for importance sampling is not just that we can allocate different numbers of shots to different gradient components depending on their relative need, but that we can account for basis rotations $Q_j$ that allow for multiple measurements.
Substituting Eq.~\eqref{eq:error_2_norm} into Eq.~\eqref{eq:gen_lagrangian} and replacing the true deviation $\sigma_{i,j}$ with our estimate $\bar{\sigma}_{i,j}$ yields the Lagrangian
\begin{equation}
    \mathcal{L}=\sum_{j}\frac{ \sum_i\bar{\sigma}_{i,j}^2}{M_j}+\lambda\Big(\sum_jM_j-M\Big)\,.
\end{equation}
Minimizing with respect to $M_j$ and solving for $M_j$ yields a shot allocation with respect to $\lambda$, which may be simplified by enforcing our constraint $\sum_jM_j=M$
\begin{equation}
    M_j=\frac{\sqrt{\sum_i\bar{\sigma}_{i,j}^2}}{\sqrt{\lambda}}=M\frac{\sqrt{\sum_i\bar{\sigma}_{i,j}^2}}{\sum_{j'}\sqrt{\sum_i\bar{\sigma}_{i,j
    '}}}.
\end{equation}
Re-substituting this into our definition of $\epsilon^2$ and solving for $M$ then achieves a relatively compact result,
\begin{equation}
    M\leq\epsilon^{-2}\Gamma_2^{(\mathrm{par})},\hspace{1cm}\Gamma_2^{(\mathrm{par})}=\left(\sum_j\sqrt{\sum_i\bar{\sigma}_{i,j}^2}\right)^2.~\label{eq:2norm_parallel}
\end{equation}
In Appendix \ref{App_FD_Lagrange_Mult} we repeat this calculation to find a bound on the measurement count required to estimate the error vector $\boldsymbol{\epsilon}$ to constant $1$-norm instead of $2$-norm.

It is instructive here to consider the effect of parallelization; what do we gain from the ability to use one basis rotation $Q_j$ to measure components $D_{i,j}$ of multiple force operators?
This is important as this ability is lost in schemes such as low-rank factorization, where basis rotations to diagonalize factors from $\frac{dH}{dR_i}$ and $\frac{dH}{dR_{i'}}$ cannot be made to easily overlap while keeping all operators low-rank~\cite{Rubin2021}.
This can be studied by replacing $M_j\rightarrow M_{i,j}$, and performing the same Lagrangian minimization as before.
The effect of this minimization can be immediately written down, as we are effectively losing the $i$ index from the second sum in Eq.~\eqref{eq:2norm_parallel} and replacing the $j$ index by a pair $(i,j)$.
Thus, we can write
\begin{equation}
    M\leq \epsilon^{-2}\Gamma_2^{(\mathrm{sep})},\hspace{1cm}\Gamma_2^{(\mathrm{sep})}=\left(\sum_{i,j}\bar{\sigma}_{i,j}\right)^2.\label{eq:2norm_serial}
\end{equation}
As $\sum_jx_j\geq\sqrt{\sum_jx_j^2}$, parallelization is clearly always favorable when possible (which we expect).
The gain in efficiency going from Eq.~\eqref{eq:2norm_serial} to Eq.~\eqref{eq:2norm_parallel} depends on how well parallel measurements can be grouped.
The case with the largest difference in efficiency is when a set of basis rotations can be chosen for all $3 N_a$ force operators such that the magnitude of all errors in each component are roughly equal for each rotation, i.e.\ $\bar{\sigma}_{i,j}\sim\bar{\sigma}_j$.
In this case, we have
\begin{equation}
    \Gamma_2^{(\mathrm{sep})}=3N_a\Gamma_2^{(\mathrm{par})}\,.
\end{equation}
However, in a real setting the asymptotic gain may be significantly smaller.

We can also consider the gain obtained from importance sampling in the parallel estimation case.
This will be useful to predict the improvement that might be gained from importance sampling in methods where this is not natively performed.
In the absence of importance sampling, we replace $M_j\rightarrow \frac{M}{N_r}$, where $N_r$ is the total number of basis rotations ($N_r=\sum_j 1$).
The 2-norm of the error in Eq.~\eqref{eq:error_2_norm} then becomes
\begin{equation}
    \epsilon^2\leq\frac{N_r}{M}\sum_{i,j}\bar{\sigma}_{i,j}^2,
\end{equation}
and rearranging yields
\begin{equation}
    M\leq \epsilon^{-2}\Gamma_2^{(\mathrm{par},\,\mathrm{n.i.})},\hspace{1cm}\Gamma_2^{(\mathrm{par},\,\mathrm{n.i.})}=N_r\sum_{i,j}\bar{\sigma}_{i,j}^2.\label{eq:2norm_no_importance_sampling}
\end{equation}
As one would expect, in the limit that $\bar{\sigma}_{i,j}=\bar{\sigma}$ Eq.~\eqref{eq:2norm_no_importance_sampling} and Eq.~\eqref{eq:2norm_parallel} are identical.
However, when $\sum_i\bar{\sigma}_{i,j}^2$ varies significantly as a function of $j$, the gain can be up to a factor of $N_r$; the number of basis rotations used.
As full tomography of the fermionic $2$-RDM requires $N_r\sim N_a^2$ (and naive tomography $N_r\sim N_a^4$), this can be a significant gain.

\subsection{Fermionic shadow tomography}
\label{sec:fermionic_shadow_tomography}

An alternative method for choosing basis rotations $Q_j$ and allocating shots $M_j$ is to choose them at random.
This idea has been recently formalized by the notion of classical shadows~\cite{huang2020predicting}.
Here, one considers the action of randomly drawing a basis rotation $Q_j$ from an ensemble $\mathcal{Q}$, measuring in the computational basis, and observing basis state $|b_j\rangle$.
One could in principle now invert $Q_j$ on the measured state $|b_j\rangle$ to give a new state $Q_j^{\dag}|b_j\rangle\langle b_j|Q_j$.
Assuming that the ensemble $\mathcal{Q}$ is informationally / tomographically complete (i.e. that every marginal of $\rho$ is measured by at least one element of $\mathcal{Q}$), the map
\begin{equation}
    \mathcal{M}:\rho\rightarrow \mathbb{E}\big[Q_j^{\dag}|b_j\rangle\langle b_j|Q_j\big]
\end{equation}
is invertible.
Moreover, the expectation value of the inverse map $\mathcal{M}^{-1}$ across sampled rotations $Q_j$ and consequently measured states $|b_j\rangle$ must be the initial state
\begin{equation}
    \mathbb{E}\Big[\mathcal{M}^{-1}\big(Q_j^{\dag}|b_j\rangle\langle b_j|Q_j\big)\Big]=\rho.
\end{equation}
Given a finite set of basis rotations $Q_j$ and subsequent measurements $|b_j\rangle$, this gives an estimator for $\frac{dE}{dR_i}=\mathrm{Trace}\!\left[|\psi\rangle\langle\psi|\frac{dH}{dR_i}\right]$~\cite{huang2020predicting,Zhao20Fermionic}
\begin{equation}
    \widehat{\frac{dE}{dR_i}}=\mathbb{E}\Bigg\{\mathrm{Trace}\bigg[\frac{dH}{dR_i}\mathcal{M}^{-1}\big(Q_j^{\dag}|b_j\rangle\langle b_j|Q_j\big)\bigg]\Bigg\}.
\end{equation}
The key advantages to this method are that the ensemble $\mathcal{Q}$ may be easier to design than a specific set of rotations $R_j$, and by averaging over the entire ensemble we may reduce the covariance between different terms.
To suppress the tails on the distribution of the estimator $\widehat{\frac{dE}{dR_i}}$ and achieve optimal scaling a median-of-means technique was originally used in \cite{huang2020predicting}, but it was shown in \cite{Zhao20Fermionic} that this is unnecessary for fermionic systems.
A provably optimal choice of $\mathcal{Q}$ to estimate arbitrary $2$-RDM elements is the ensemble $\mathcal{Q}_{\mathrm{FGU}}$ of fermionic Clifford Gaussian unitaries described in Sec.~\ref{sec:basis_rotations}.
We label the corresponding channel $\mathcal{M}_{\mathrm{FGU}}$.

The variance of the above estimator may be calculated by representing the force operator in the algebra generated by the Majorana operators $\gamma_{p}$ (Eq.~\eqref{eq:majorana_operators})
\begin{align}
    \frac{dH}{dR_i} = \sum_{k=1}^2 \sum_{\mu \in C(2N,2k)} f_{\mu}^{(i)}\Gamma_{\mu},\hspace{1cm} \Gamma_{\mu}=\mathrm{i}^k\gamma_{\mu_1}\dots\gamma_{\mu_{2k}},
\end{align}
where $C(2N,2k)$ is the set of all possible combinations of $2k$ elements drawn from $\{1,\ldots,2N\}$.
With this defined, the variance on the estimator constructed from a single choice of basis rotation $Q_j$ and observation of $|b_j\rangle$ is calculated in \cite{Zhao20Fermionic} to be
\begin{align}\label{eq_beta_fs}
    \mathrm{Var}_\mathrm{FS}\left(\widehat{\frac{dE}{dR_i}}\right)&=\sum_{k=1}^2 \sum_{\mu \in C(2N,2k)} \|f_{\mu}^{(i)} \Gamma_{\mu}\|_\mathrm{FGU}^2 - \mathrm{Tr}\left(F_i\rho\right)^2,
\end{align}
where here $\|\cdot\|_{\mathrm{FGU}}$ is the shadow norm~\cite{huang2020predicting} under the fermionic Gaussian Clifford ensemble
\begin{equation}
    \|O\|^2_{\mathrm{FGU}}=\max_{\mathrm{state}\;\rho}\bigg\{\mathbb{E}_{W\sim\mathcal{Q}_{\mathrm{FGU}}}\sum_{b\in\{0,1\}^n}\langle b|W\rho W^{\dag}|b\rangle\langle b|W\mathcal{M}^{-1}(O)W^{\dag}|b\rangle^2\bigg\}.
\end{equation}
The shadow norm of a product of $2k$ Majorana operators in this ensemble was calculated in \cite{Zhao20Fermionic} to be $\binom{2n}{2k}/\binom{n}{k}$.
Thus, by the central limit theorem, the variance of the estimator $\widehat{\frac{dE}{dR_i}}$ after $M$ different rotations $Q_j$ and measurements is
\begin{equation}
    \mathrm{Var}_\mathrm{FS}\left(\widehat{\frac{dE}{dR_i}}\right)=\frac{1}{M}\bigg\{\sum_{k=1}^2 \sum_{\mu \in C(2N,2k)} \binom{2N}{2k} \binom{N}{k}^{-1} |f_{\mu}^{(i)}|^2 - \mathrm{Tr}\left(\tfrac{dH}{dR_i}\rho\right)^2\bigg\}\,.\label{eq:fermionic_shadow_variance}
\end{equation}
Note that unlike other methods where one must take a bound on the variance of individual estimators, Eq.~\eqref{eq:fermionic_shadow_variance} is exact.
It is also significantly easier to calculate than the variance on other estimation methods, as it does not require access to higher-order correlators that come with expectation values of $H^2$.

We now extend the above analysis to estimate the error in the 2-norm $\epsilon_2$ (see  Eq.~\eqref{eq:error_2_norm}) of the force vector $\frac{dE}{d\mathbf{R}}$.
As we are drawing $M$ basis rotations from our distribution at random, this time we do not need to optimize the distribution of our measurements via importance sampling.
(Importance sampling over shadow tomography may be introduced by locally biasing the classical shadows~\cite{Bravyi19Classical}, which has been seen to yield significant improvements for Hamiltonian tomography~\cite{Zhao20Fermionic}.)
As we do not encounter covariances between terms when calculating the $2$-norm (see Sec.~\ref{sec:NISQ_DirectMeasurements}), we have immediately that
\begin{align}
    \epsilon_2^2 = \frac{1}{M}\sum_i\mathrm{Var}_\mathrm{FS}\left(\widehat{\frac{dE}{dR_i}}\right),
\end{align}
and so to bound $\epsilon_2^2<\epsilon^2$ requires that we set
\begin{align}
     M= \epsilon^{-2}\Gamma_2^{(\mathrm{FS})}, \hspace{1cm} \Gamma_2^{(\mathrm{FS})}=\sum_i  \mathrm{Var}_\mathrm{FS}\left(\widehat{\frac{dE}{dR_i}}\right)\,.
\end{align}

\subsection{Numerical results}
\label{sec:numerics_nisq}

We now attempt to summarize and estimate the cost of the different optimizations made in this section for NISQ tomography of force vectors.
The cost of estimating force vectors depends critically on the studied system.
In this section we consider two different scenarios; one where we make the assumption that our force operators are relatively uniformly distributed across our system (allowing us to make general asymptotic estimates), and a set of numerical cost estimates for hydrogen chains of varying length, calculated in STO-6G using localized orbitals, see App~\ref{app:localization}.
Although near optimal methods of grouping fermionic operators are known~\cite{Bonet20Nearly}, to simplify the results in this work we choose a naive grouping to compare the parallelized vs serial importance sampling discussed in Sec.~\ref{sec:NISQ_DirectMeasurements}.
That is, we consider a set of rotations $Q_j$ that measure independent Pauli terms $D_{i,j}=h_{i,j}P_j$, where $\frac{dH}{dR_i}=\sum_jh_{i,j}P_j$ following a Jordan-Wigner transformation.
We compare this in turn to results obtained for fermionic shadow tomography and to results for a low-rank factorization $\frac{dH}{dR_i}=\sum_lW^{(i,\ell)2}$ (Eq.~\eqref{eq:single_factorized_derivative}).

To avoid the need to calculate expectation values of four-body terms (which are required for exact variance estimates using the methods of Sec.~\ref{sec:NISQ_DirectMeasurements}), we use standard methods for upper bounding the variance contributions $\sigma_{i,j}^2\leq\bar{\sigma}_{i,j}^2$ instead.
Our variance bound for naive measurements is very simple; $\sigma_{i,j}\leq h_{i,j}$.
Note that this implies $\Gamma_2^{(\mathrm{sep})}$ is just the square of the induced 1-norm of the force operator (in a qubit representation).
By comparison, for the basis rotation grouping method, we take the worst-case bound on the variance for each given operator
\begin{align}\label{eq_beta_brg}
    \sigma^2_{i,\ell}\leq \frac{r^2_{i,\ell}}{4}\,,
\end{align}
where $r_{i,\ell}=|\lambda_{i,\ell}^{\mathrm{(max)}}-\lambda_{i,\ell}^{\mathrm{(min)}}|$ is the range of the spectrum of the corresponding factor $\sum_{\sigma,\sigma'}\sum_{pq}f_p^{(l)}f_q^{(l)}n_{p,\sigma} n_{q,\sigma'}$ constrained to the appropriate particle number sector.
In Fig.~\ref{fig:measurements_nisq}, we plot the resulting scaling coefficients $\Gamma$ for these techniques, alongside the classical fermionic shadow scaling coefficient from Sec.~\ref{sec:fermionic_shadow_tomography}.
As the bounds used for different methods differ in their tightness, it is not possible to make a direct comparison between the lines in different plots.
Instead, for comparison we plot the equivalent scaling factors for the Hamiltonians of the same system, allowing us to compare the cost of energy and derivative estimation.
As each line is approximately straight on a log-log plot, we may extract scaling coefficients by a linear fit, which we report in Tab.~\ref{tab:nisq_table}.
We contrast this in the same table with an asymptotic analysis under the assumption that all forces are the same magnitude.
We see that, the shift from separate to parallel force measurement in the naive case (left plot) decreases the cost of estimation by a factor $N_a^{0.2}<< N_a$.
When parallelized, we further observe that the cost to estimate the force operator to a target precision (in Ha/Å) is roughly the same the cost to estimate energies to a target precision (in Ha).
Similar results are found for the shadow tomography case (b).
As the error required on forces in molecular dynamics (in the given units) is roughly $20$ times larger than chemical accuracy (Sec.~\ref{sec:error_tolerance}), we suggest that a single-shot force estimation using shadow tomography may be already a factor $100$ cheaper than the corresponding energy estimation on a reasonably-sized system.
However, semi-classical molecular dynamics simulations typically requires millions of such estimations, which presents a significant additional multiplicative cost.
Similarly, we see that the scaling of the cost of estimating forces via basis rotation grouping (c) is smaller for force operators as compared to the Hamiltonian.
We note however that the here presented bounds are not tight (see \cite{Huggins19Efficient} for a comparison) and that it is an open question how these numbers change when using the explicit variances $\sigma_{i,j}$ of the ground state.

\begin{figure*}[tb]
        \centering
        \includegraphics[width=1\textwidth]{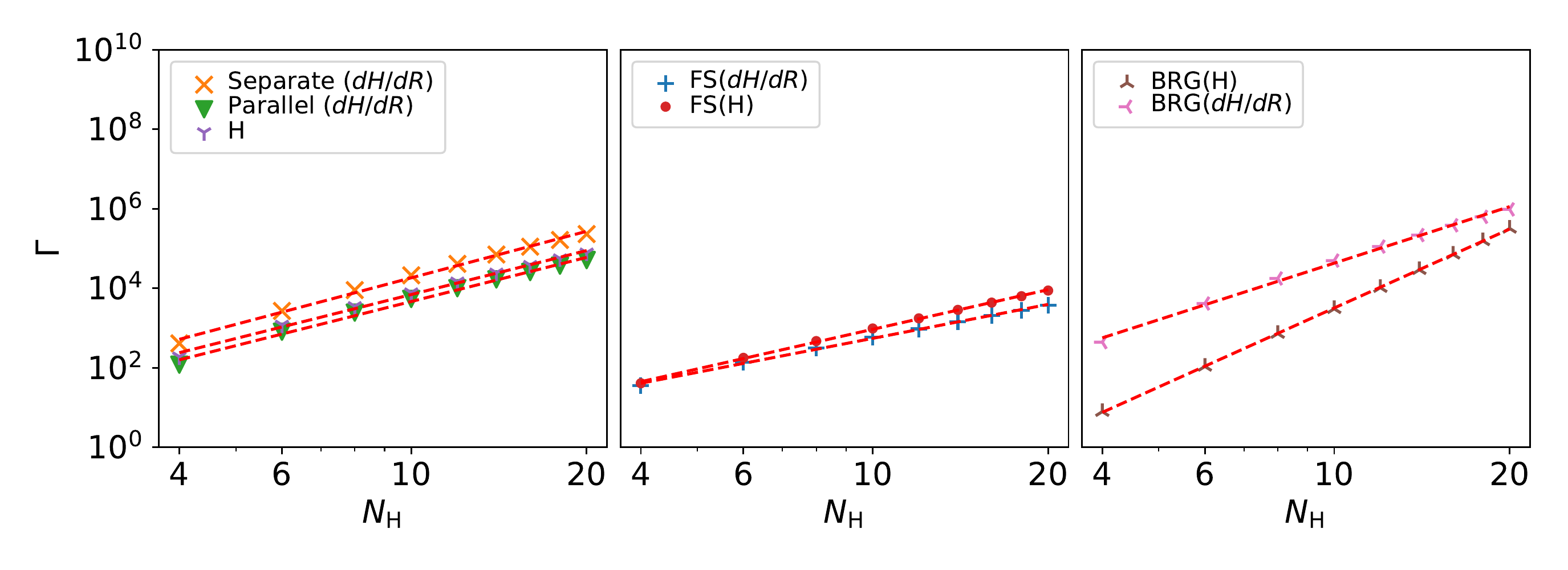}\llap{
  \parbox[b]{12.5in}{(a)\\\rule{0ex}{2.4in}
  }}\llap{
  \parbox[b]{8.4in}{(b)\\\rule{0ex}{2.4in}
  }}\llap{
  \parbox[b]{4.3in}{(c)\\\rule{0ex}{2.4in}
  }}
\caption{Numerical evaluation of upper bounds on the cost of force estimation for hydrogen chains up to $N_H=20$ atoms with spacing of $0.74084$\r{A} calculated in STO-6G using localized orbitals. We show the upper bounds of the Hamiltonian and the force operators using (a) Pauli measurements, (b) fermionic shadows and (c) the basis rotation grouping technique.}
\label{fig:measurements_nisq}
\end{figure*}

\begin{table}[]
    \centering
    \begin{tabular}{|C{4cm}|C{4cm}|C{3cm}|}
    \hline
    Method & Asymptotic scaling ($F_i=F, \forall i$) & Empirical scaling ($\ch{H}$ chain) \\
    \hline
    Separate measurements &  $\mathcal{O}(\frac{N_a^2 \beta(N_a)}{\epsilon^2})$ & $\mathcal{O}(\frac{N_a^{3.89\pm 0.1}}{\epsilon^2})$ \\
    \hline
    Parallel measurements & $\mathcal{O}(\frac{N_a \beta(N_a)}{\epsilon^2})$ & $\mathcal{O}(\frac{N_a^{3.69\pm 0.11}}{\epsilon^2})$\\
    \hline
    Fermionic shadows & $\mathcal{O}(\frac{N_a \beta_{\mathrm{FS}}(N_a)}{\epsilon^2})$ & $\mathcal{O}(\frac{N_a^{2.84\pm0.05}}{\epsilon^2})$\\
    \hline
    Basis rotation grouping & $\mathcal{O}(\frac{ N_a^2 \beta_{\mathrm{BRG}}(N_a)}{\epsilon^2})$ & $\mathcal{O}(\frac{N_a^{4.72\pm0.10}}{\epsilon^2})$ \\
    \hline
    \end{tabular}
    \caption{Comparison of the asymptotic scaling for the case that all force operators are equal ($F_i = F\, \forall i$) and the empirical scaling inferred from Fig.~\ref{fig:measurements_nisq}. The value $\beta$ refers to the one-norm of the coefficients of the force operator after a Jordan-Wigner transformation, $\beta_\mathrm{FS} = \mathrm{Var}_\mathrm{FS}(F)$ and $\beta_\mathrm{BRG}=(\sum_l \mathrm{Var}(F_l)^{1/2})^2$. }
    \label{tab:nisq_table}
\end{table}

Previous results~\cite{Obrien19Calculating, Mitarai2019, Sokolov:2021} have focused their analysis on the scalings of the quantum circuits required to compute gradients without considering the contribution coming from the properties of the gradient operator and without performing any optimization of the measurements. In contrast, here we have computed bounds on the number of measurements required while also optimizing the measurement number with three different methods. For this reason, performing a direct quantitative comparison between this work and previous results is not possible.

%===============================================================
\section{Block encoding Hamiltonians and force operators}

\input{BlockEncoding}
%===============================================================

\subsection{Numerical studies of the force operator in molecular systems}
\label{sec:numerics}
The Hamiltonian simulation cost depends on the representation of the Hamiltonian and can be, among other things, quantified by its spectral bandwidth, which we upper bound with a value $\lambda$.
For electronic structure Hamiltonians there exist numerical insights into the upper bounds for various systems and representation techniques and with respect to their scaling when increasing the system size or the number of orbitals~\cite{vonBurg2020, lee2021even, Berry19Qubitization}.
For the force operator, see Eq.~\eqref{eq:der_operator}, such results do not exist. 
In the following, we compare the values of $\lambda$ obtained for the force operator and the Hamiltonian using two methods, the sparse method~\cite{Berry19Qubitization} and the double low-rank factorization~\cite{vonBurg2020}.
The formula for the rescaling constant to block-encode an operator via the sparse method is
\begin{align}
    \lambda_1^{(\mathrm{sparse})}=\sum |T_{pq}+\sum_{r}V_{pqrr}|\,,&& \lambda_2^{(\mathrm{sparse})}=\frac{1}{2}\sum |V_{pqrs}| \quad \mathrm{and}\quad \lambda^{(\mathrm{sparse})}=\lambda_1^{(\mathrm{sparse})}+\lambda_2^{(\mathrm{sparse})}\,
    \label{eq:lambda_sparse}
\end{align}
where the coefficients $T_{pq}$ and $V_{pqrs}$ are taken from from Eq.~\eqref{eq:chemist_notation} for the Hamiltonian, and using $T^{(F_A)}_{pq}$ and $V^{(F_A)}_{pqrs}$ from Eq.~\eqref{eq:force_coefficients_def} for the force operators. 
These constants are induced $1$-norms of the Hamiltonian (or force) operator, which makes them highly dependent on the choice of molecular orbital basis.
As it has been observed that localizing the basis orbitals reduces the scaling for Hamiltonians in Refs.~\citenum{koridon2021orbital}, \citenum{lee2021even}, we also use localized orbitals for all of our calculations, see also App.~\ref{app:localization}. In Fig.~\ref{fig:sparse_cmo}, we provide the same results for canonical molecular orbitals. 

Instead of directly following the method used in Sec.~\ref{sec:block_encoding_second_quantization} to differentiate the double-low-rank-factorized Hamiltonian to prove a bound on the cost of block-encoding the derivative operator, in this section we directly block-encode the derivative operator instead.
This is practically easier, and should in principle yield a lower rescaling factor $\lambda$, however care should be taken in practice to suppress the systematic error this produces (as described in Sec.~\ref{sec:block_encoding_second_quantization}).
Within the double low-rank factorization the lambdas are given by
\begin{align}
    \lambda_1^{(\mathrm{DF})} = \sum_i |t_i|\,, && \lambda_2^{(\mathrm{DF})} = \frac{1}{4} \sum_l^L \left(\sum_p |f_p^{(l)}|\right)^2\quad \mathrm{and}\quad
    \lambda^{(\mathrm{DF})} = \lambda_1^{(\mathrm{DF})} + \lambda_2^{(\mathrm{DF})}\,,
    \label{eq:lambda_df}
\end{align}
with $t_i$ the eigenvalues of $T_{pq}+\sum_{r}V_{pqrr}$, cf. Eq.~\eqref{eq:double_low_rank_factorization} for the Hamiltonian, and the eigenvalues of $T^{(F_A)}_{pq}+\sum_rV^{(F_A)}_{pqrr}$ (Eq.~\eqref{eq:force_coefficients_def}) for the force operators.

As the force is a vector with $3N_{a}$ elements for a system with $N_{a}$ atoms, we find $3N_{a}$ lambdas. To compare the simulation cost of the Hamiltonian and the force operator, we use the mean of the lambdas of the force operators, 
\begin{align}
    \lambda_{\mathrm{F}}=\frac{\sum_{i=1}^{3N_{a}}\lambda_{\mathrm{F}_i}}{3N_a}\,,
    \label{eq:lambda_force_vector}
\end{align}
with $\lambda_{\mathrm{F}_i}$ the lambda of the $i$-th force operator. We note that, in contrast to electronic structure Hamiltonians, performing a double low-rank factorization of the force operator yields negative eigenvalues $w_l$ in Eq.~\eqref{eq:single_factorized_hamiltonian}. This can be accounted for in a block encoding by shifting the sign from the eigenvalue $w_l$ to the encoding itself, so we replace the eigenvalues by their absolute values $|w_l|$. For all shown plots, we use localized orbitals (LMOs) from a Hartee-Fock calculation, see Sec.~\ref{app:localization}. 

In the following we provide numerical results on one-dimensional chains of hydrogen atoms with spacings of $1.4$ Bohr radii ($0.74084$ \r{A}), calculated in STO-6G, a system often used to infer the scaling in the thermodynamic limit \cite{lee2021even, Berry19Qubitization}. Moreover, to provide a benchmark for systems closer to relevance for biological systems, we investigate the scaling of $\lambda^{(\mathrm{sparse})}$ and $\lambda^{(\mathrm{DF})}$ for water clusters of increasing size. Water clusters are prototypes of hydrogen bonded networks and therefore of fundamental interest, and thoroughly investigated benchmarks by experimental and theoretical methods~\cite{Keutsch:2001, Manna:2017}. Hydrogen bonds are regarded as the key to life in general, due to the directionality, cooperativity and versatility~\cite{Vladilo:2018}, being crucially involved in the chemical basis of the genetic code and all catalytic processes. The strength of hydrogen bonds is highly dependent on the local surrounding and requires accurate quantum treatment~\cite{Hus:2012}. As input for the calculations, we use the data from \cite{Rakshit:2019}, presented also in Fig.~\ref{fig:water_clusters}, which include the geometries of water clusters which were optimized by using a TTM2.1-F ab-initio based interaction potential together with subsequent MP2 calculations in the aug-cc-pVTZ basis.

To study the scaling of $\lambda^{(\mathrm{DF})}$ in the continuum limit, i.e. when increasing the basis set size, we follow \cite{lee2021even, Berry19Qubitization} and use $\ch{H4}$ on a square plaquette with a side length of $2$ Bohr radii. We first calculate the Hamiltonian and force operators using the cc-pvtz basis, which yields $N=56$ spatial orbitals. To study the effect of an increasing basis set we artificially truncate the two-body coefficients $V_{pqrs}$,
\begin{align}
    V_{pqrs}^{\mathrm{truncated}} = \sum_{pqrs=1}^{\tilde{N}} V_{pqrs}\,
\end{align}
for $10\leq\tilde{N}\leq 56$ and show the $\lambda^{(\mathrm{DF})}$ in Fig.~\ref{fig:lambda_h4_and_water_cluster}. 

In Fig.~\ref{fig:lambdas_hydrogen_chain}, we show $\lambda^{(\mathrm{sparse})}$ and $\lambda^{(\mathrm{DF})}$ of the force operators and Hamiltonian for the hydrogen chains of length between $N_{H}=10$ and $N_{H}=100$.
For both methods, we recognize that the force operator shows a better scaling compared to the scaling of the Hamiltonian, and indeed tends towards a constant or $\log(N_H)$ cost.
This makes sense, as we expect the force on a given nuclei in a $1D$ chain to scale logarithmically in the system size.

In Fig.~\ref{fig:lambda_h4_and_water_cluster}(a), we show the scaling of $\lambda^{(\mathrm{sparse})}$ and $\lambda^{(\mathrm{DF})}$ for the water clusters in STO-3G with increasing system size.
For technical details on the calculation we refer to Appendix~\ref{app:numerical_details_water_clusters}.
Similarly to the hydrogen chain, we find a significant better scaling of the average $\lambda$ of the force operators in comparison to the scaling of the Hamiltonian.

In Fig.~\ref{fig:lambda_h4_and_water_cluster}(b), we show the scaling of $\lambda^{(\mathrm{DF})}$ when increasing the basis set size for $\ch{H4}$.
In this case, the scaling of $\lambda^{(\mathrm{DF})}$ for the Hamiltonian and the force is comparable, which is contrary to the results in the previous section.
This makes sense: increasing the basis set increases the range of local energy densities on each atom, which suggests the spectrum of local operators such as derivatives should be increased also.

\begin{figure*}[tb]
    
        \centering
        \includegraphics[width=.5\textwidth]{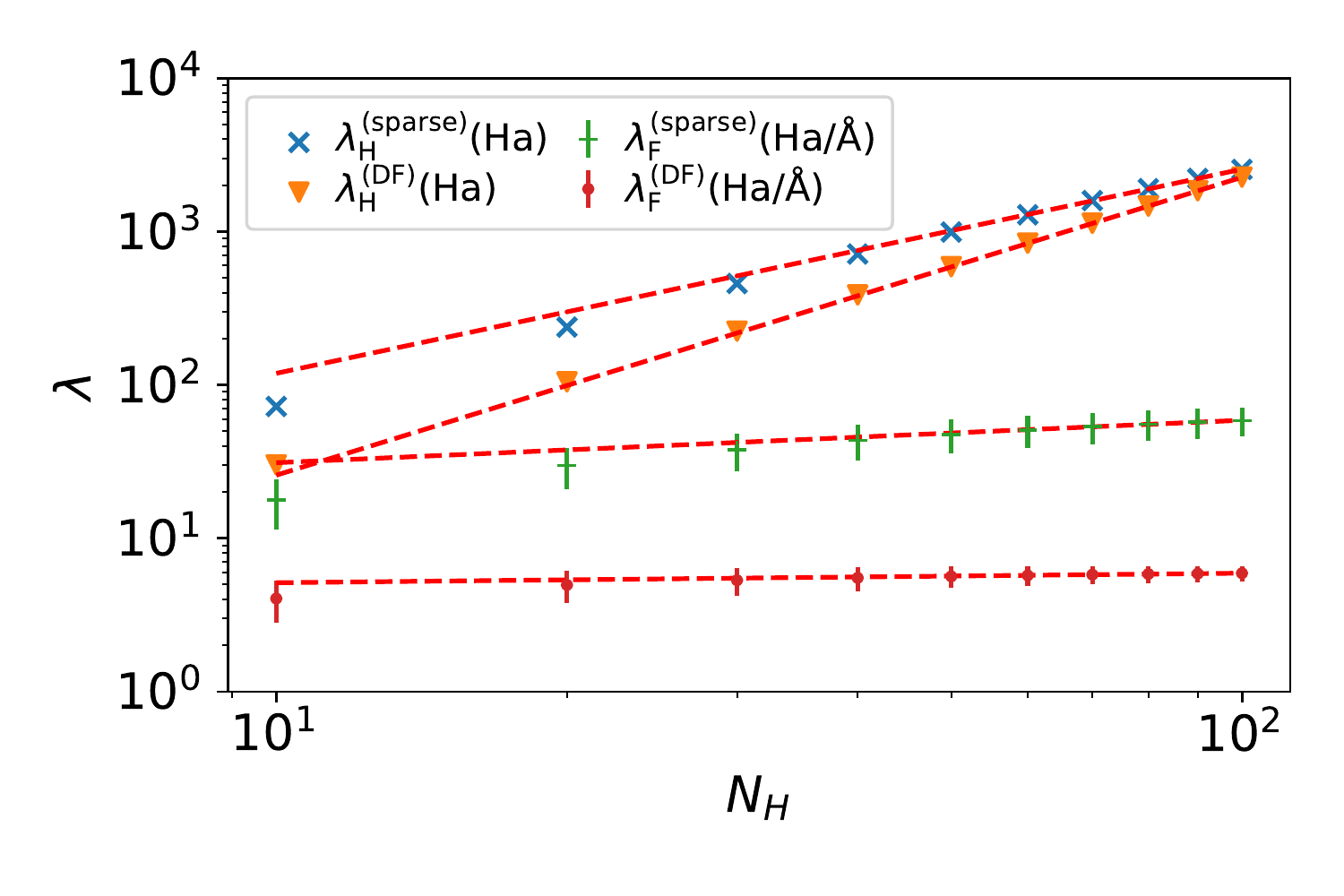}
\caption{Rescaling factors for the Hamiltonian $(\lambda_{\mathrm{H}})$ and mean rescaling factors for the force operators ($\lambda_{\mathrm{F}}$ - Eq.~\eqref{eq:lambda_force_vector}) using the sparse method with localized molecular orbitals, and a double low-rank factorization of the force operators (as opposed to differentiating the double-low-rank-factorized Hamiltonian described in Sec.~\ref{sec:block_encoding_second_quantization}) for one-dimensional hydrogen chains with up to $N_{H}=100$ atoms separated by 1.4 Bohr ($0.74084$ \r{A}).
All calculations were performed in STO-6G. The plot shows $\lambda_1+\lambda_2$ of the Hamiltonian and force operators, see Eq.~\eqref{eq:lambda_force_vector}. The fits suggest scalings of 
$\lambda^{\mathrm{(sparse)}}_{\mathrm{H}}=\mathcal{O}(N_H^{1.331\pm 0.012})$, 
$\lambda^{\mathrm{(sparse)}}_{\mathrm{F}}=\mathcal{O}(N_H^{0.277\pm 0.011})$, 
$\lambda^{\mathrm{(DF)}}_{\mathrm{H}}=\mathcal{O}(N_H^{1.942\pm 0.002})$ and 
$\lambda^{\mathrm{(DF)}}_{\mathrm{F}}=\mathcal{O}(N_H^{0.062\pm 0.004})$. Only force operators with $\lambda_{\mathrm{F}_i}\neq 0$ were included in this calculation. For all fits only the last five points were taken into account. The error bars show the standard deviation of the distribution of the lambdas of single force operators $\lambda_{\mathrm{F}_i}$.}
\label{fig:lambdas_hydrogen_chain}
\end{figure*}

\begin{figure*}[tb]
    \centering
    \begin{minipage}{0.49\textwidth}
        \centering
        \includegraphics[width=1\textwidth]{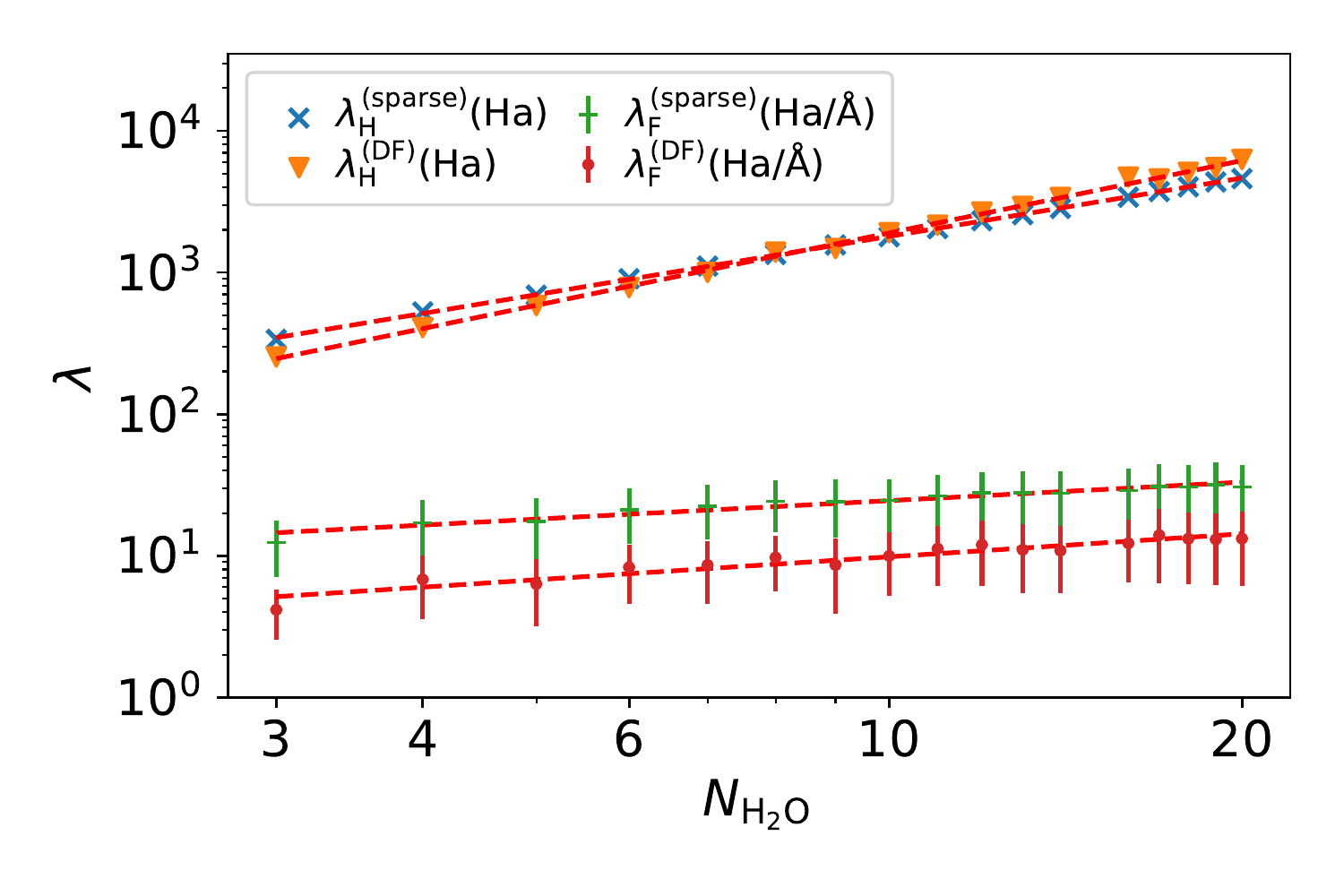}\llap{
  \parbox[b]{6.0in}{(a)\\\rule{0ex}{2.2in}
  }}
    \end{minipage}\hfill
    \begin{minipage}{0.49\textwidth}
        \centering
        \includegraphics[width=1\textwidth]{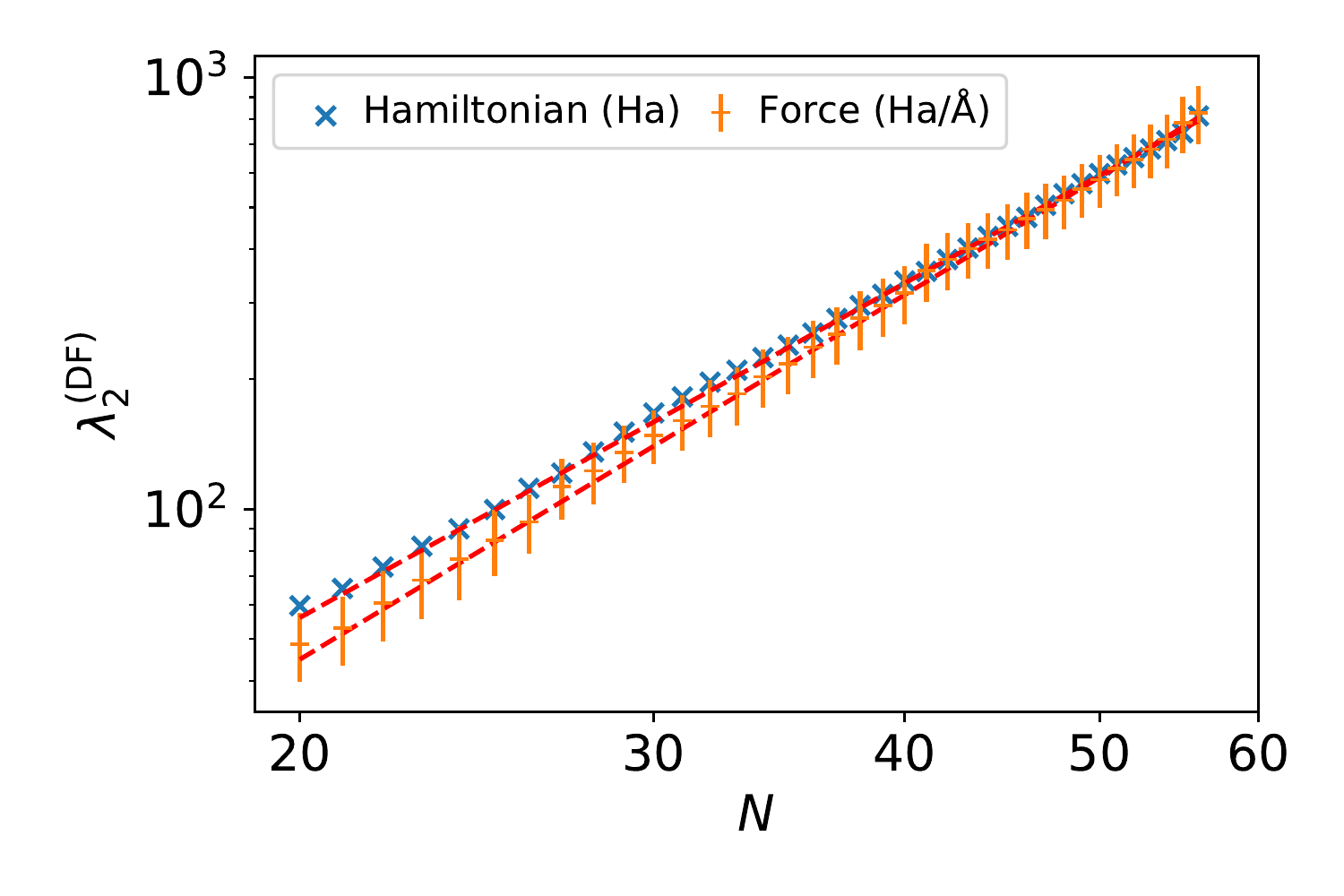}\llap{
  \parbox[b]{6.0in}{(b)\\\rule{0ex}{2.2in}
  }}
    \end{minipage}\hfill
\caption{(a) Rescaling factors $\lambda$ of the Hamiltonian and force operators for the water clusters presented in Fig.~\ref{fig:water_clusters}, for the sparse method (using localized molecular orbitals, see App.~\ref{app:localization}), and following a double low-rank factorization of the corresponding operators (as opposed to differentiating the double-low-rank-factorized Hamiltonian described in Sec.~\ref{sec:block_encoding_second_quantization}).
Fits suggest scalings of $\lambda^{\mathrm{(sparse)}}_{\mathrm{H}}=\mathcal{O}(N_{H2O}^{1.368\pm 0.005})$, $\lambda^{\mathrm{(sparse)}}_{\mathrm{F}}=\mathcal{O}(N_{H2O}^{0.434\pm 0.027})$, $\lambda^{\mathrm{(DF)}}_{\mathrm{H}}=\mathcal{O}(N_{H2O}^{1.697\pm 0.019})$ and $\lambda^{\mathrm{(DF)}}_{\mathrm{F}}=\mathcal{O}(N_{H2O}^{0.539\pm 0.042})$. (b) The rescaling factor of double-low-rank-factorized Hamiltonians and force operators with respect to the number of spin orbitals $N$ in an STO-6G basis for a $\ch{H4}$ plaquette with side length of $2$ Bohr radii.
Fits suggest scalings of $\lambda^{\mathrm{(DF)}}_{\mathrm{H}}=\mathcal{O}(N^{2.549\pm 0.014})$ and $\lambda^{\mathrm{(DF)}}_{\mathrm{F}}=\mathcal{O}(N^{2.749\pm 0.03})$.
As in Fig.~\ref{fig:lambdas_hydrogen_chain}, in both plots, the error bars show the standard deviation of the distribution of the lambdas of single force operators $\lambda_{\mathrm{F}_i}$.}
\label{fig:lambda_h4_and_water_cluster}
\end{figure*}

\section{Algorithms to compute forces in fault tolerance}\label{sec:FT}

We detail three different methods in this section by which one might estimate forces in a fault-tolerant setting.
We first study a semi-classical finite difference calculation, where the quantum computer is called as a subroutine to estimate energies at different positions $\mathbf{R}$.
We then study the direct estimation of forces as the expectation value of the derivative of the Hamiltonian via the Hellman-Feynman theorem, using the overlap estimation algorithm to make this estimation at the Heisenberg limit.
We finally study expectation value estimation of the Hamiltonian derivative by the new expectation value estimation algorithm of \cite{Huggins2021Nearly}, which re-encodes this as a separate derivative estimation problem that can be solved by the quantum derivative estimation algorithm of \cite{Gilyen_2017}.
We estimate the asymptotic cost for each method (up to logarithmic factors) for plane waves, hydrogen chains, and small water cluster systems.

\subsection{Numerical differentiation of the energy by higher order finite differences}\label{sec:Num_diff_FD}

Finite difference methods estimate gradients as a linear combination of the energies $E(\mathbf{R})$ from different atomic configurations $\mathbf{R}$.
Central finite differences formulas offer a quadratic advantage in the discretization error compared to backward and forward differences, so here we will focus only on the central difference forms. 
The simplest central finite difference formula is the first-order form, which requires energy calculations at two points $\mathbf{R}$:
\begin{equation}
    \frac{d E}{d R_i} = \frac{E(\mathbf{R} +\frac{dR\cdot \mathbf{v}_i}{2})-E(\mathbf{R}-\frac{dR\cdot \mathbf{v}_i}{2})}{dR}+\epsilon^{(1)}_{fd}.
\end{equation}
Here, $\epsilon^{(1)}_{fd}\sim dR^3$ is the error due to the finite difference approximation, and $\mathbf{v}_i$ is the unit vector along the component $R_i$.
This can be generalized to the degree-$2m$ (central) finite difference formula~\cite{Gilyen_2017}:
\begin{equation}\label{eq:higher_order_derivatives}
   \frac{d E}{d R_i}= \frac{1}{dR}\sum_{\substack{l=-m\\l\neq 0}}^m a^{(m)}_l E(\mathbf{R}+l\;dR\cdot \mathbf{v}_i) + \epsilon_{fd}^{(m)}:= \frac{dE^{(m)}}{dR_i} + \epsilon_{fd}^{(m)},
\end{equation}
where the coefficients $a^{(m)}_l$ are given by:
\begin{equation}
    a^{(m)}_l := \begin{cases}
    \frac{(-1)^{l-1}}{l} \frac{\binom{m}{|l|}}{\binom{m+|l|}{|l|}} & l \ne 0\\
    0 & \text{otherwise}
    \end{cases},
\end{equation}
and $\epsilon_{fd}^{(m)}\sim dR^{2m+1}$ is the degree-$2m$ finite difference error.

The easiest way to implement higher-order finite difference methods on a quantum device is to use the device as a subroutine that calculates the required energies $E(\mathbf{R}+l\,dR\cdot \mathbf{v}_i)$.
This then sets a competition between the finite-difference error and the error $\epsilon_{\mathrm{PE}}$ from phase estimation, as the latter scales inversely in $dR$.
Specifically, $\epsilon_{\mathrm{PE}}\sim dR^{-1}$ due to the division in Eq.~\eqref{eq:higher_order_derivatives}.
Optimizing higher-order finite difference methods on a quantum device then requires balancing these two sources of error by optimizing our choices of $dR$ and the degree $m$ of the finite-difference method used.
It also requires that we optimize our quantum circuits for finite difference estimation: as we are repeatedly estimating eigenvalues of similar (but not identical) quantum states, under certain conditions it may be preferable to reuse our state register in between phase estimation routines.
In this section we will optimize all the above, and calculate bounds on the asymptotical scaling of the resulting algorithm.

In Appendix~\ref{app:App_FD_Holevo_variance} we calculate the error $\epsilon_{\mathrm{PE}}$ due to quantum phase estimation on a single component $\frac{dE}{dR_i}$. We draw on results from \cite{BabbushSpectra,Berry2009,Higgins09Demonstrating}, and then propagate this error through Eq.~\eqref{eq:higher_order_derivatives}.
Then, in Appendix~\ref{App_FD_Lagrange_Mult}, we optimize our resource allocation to the estimation of different $E(\mathbf{R}+l\,dR\cdot\mathbf{v}_i)$ via Lagrangian techniques.
This yields a bound
\begin{equation}\label{eq_epsilon_PE}
    \epsilon_{\mathrm{PE}} \le \frac{\pi \lambda_{\mathrm{H}} 6^{3/2} m^{1/2}}{2dR\;T},
\end{equation}
where $\lambda_{\mathrm{H}}$ is the cost of block-encoding the Hamiltonian (Sec.~\ref{sec:block_encodings}).
In order for this bound to hold, we require that $\lambda_{\mathrm{H}}$ be the worst-case block-encoding cost over all points $\mathbf{R}+ldR\cdot\mathbf{v}_i$.

Equation~\eqref{eq_epsilon_PE} gives the first part of the total error for estimating the gradient with higher order finite differences. To complete the estimation of the error we need also to upper-bound the second component, $\epsilon_{\textrm{fd}}^{(m)}$. To do this, we invoke Theorem 24 of \cite{Gilyen_2017}.
This requires imposing the condition that the energy and its derivatives are bounded,
\begin{equation}\label{eq_bounding_Derivative}
    \left|\frac{d^kE}{dR_i^k} \right| \leq \mathfrak{e} c^{k} (k)^{\frac{k}{2}},~\forall~i,
\end{equation}
with $c$ a constant with units length$^{-1}$ and $\mathfrak{e}$ a constant with the units of energy.
When this condition is satisfied, the discretization error $\epsilon_{\textrm{fd}}^m$, for one component of the gradient can be upper bounded by:
\begin{equation}\label{eq:eq_error_approx_app}
    \epsilon_{fd}^{(m)}=\left|\frac{d E}{d {R_i}}   - \frac{dE^{(m)}}{dR_i}\right|\leq \frac{\mathfrak{e}}{dR} \sum_{k=2m+1}^{\infty} \left(8 \;dR\; c  \;m \right)^k 
    = \frac{\mathfrak{e}}{dR}\frac{(8 \;dR\; c\;  m)^{2m+1}}{1- (8 \;dR \;c  \;m )}.
\end{equation}
As the error shrinks exponentially with $m$, the value of $dR$ needed to ensure that the error is at most $\mathcal{O}(\epsilon)$ can be taken to be much larger as $m$ increases. Since the cost of phase estimation scales inversely with $dR$, we will see below that using higher order difference formulas will actually reduce the cost of phase estimation.

\subsubsection{Optimization of the finite difference step size}\label{App_FD_error_gyl}

It remains to optimize the remaining free parameters in our finite difference method.
Substituting Eq.~\eqref{eq:eq_error_approx_app}  and \eqref{eq_epsilon_PE} into  $\epsilon^2 = \epsilon_{\mathrm{PE}}^2 + (\epsilon_{fd}^{(m)})^2$,%\ref{eq_epsilon_fd_tot_simp},
we obtain: 
\begin{equation}
    \epsilon^2 = \epsilon_{\mathrm{PE}}^2 + (\epsilon_{fd}^{(m)})^2  \leq \left(\frac{\pi \lambda_{\mathrm{H}}}{2T} \frac{\Gamma}{dR}\right)^2+ \left(\mathfrak{e}\frac{(8 c m )^{2m+1} dR^{2 m}}{(1-8 mc \; dR)}\right)^2,
\end{equation}

If we simply make the assumption that $8\; m\; c\; dR\le 1/2$ and we combine the two errors quadratically,  then we have the following upper bound on the total error of the numerically estimated gradient:
\begin{align}
    \epsilon^2 = \epsilon^2_{\mathrm{PE}} +(\epsilon_{fd}^{(m)})^2  < \left(\frac{\pi\lambda_{\mathrm{H}}}{2T} \frac{6^{3/2}(m)^{1/2}}{dR}\right)^2+ \left(\mathfrak{2e}{(8 c m )^{2m+1} dR^{2m}}\right)^2.\label{eq_epsilon_tot_not_opt_app}
\end{align}

We can also find a second upper bound on the value of $dR$, by imposing that the $\epsilon_{fd}^2$ is at most $\epsilon^2/2$ in the upper bound of Eq.~\eqref{eq_epsilon_tot_not_opt_app} and obtain
\begin{equation}\label{eq:dRs}
 dR \leq \min\left\{\frac{1}{8cm }\left(\frac{\epsilon}{16\sqrt{2}\mathfrak{e} m  c}\right)^{\frac{1}{2m}},\frac{1}{16cm}\right\},
\end{equation}
whereas the second requirement is a direct consequence of the requirement that $8~dR~c m \le 1/2$. Substituting this into Eq.~\eqref{eq_epsilon_PE} will yield a comparable bound for the time.

One remaining point that needs to be addressed is the sensitivity of the estimation protocol on the ground state preparation.  In particular, the above discussion holds if the ground state preparation success probability per point is $1- \mathcal{O}(1/N_a)$.  This is because we assume in the above analysis that the probability of the algorithm projecting out of the groundspace at each of the $3N_a$ estimations is asymptotically negligible.  From the union bound our previous claim on the success probability implies this because $(1-\mathcal{O}(1/N_a))^{3N_a} \in \Omega(1)$.

While this may appear to be a substantial drawback of the finite difference method, the ground state will sometimes only need to be prepared once for the FD gradient estimation process.  This is because if the eigenvalue gap is large relative to the perturbation in the energy needed over our finite sized mesh then the ground states at two near by configurations will be nearly identical.  We show this fact in in~\app{robust} and specifically claim that for any $\Delta_1,\Delta_2\in \mathbb{R}^{3N_a}$ satisfying $\max(\|\Delta_1\|,\|\Delta_2\|)< (E_1(\mathbf{R}) -E_0(\mathbf{R}))/(4\max_{\mathbf{R}',i}\left\|\frac{dH(\mathbf{R}')}{dR_i}\right\|)$ (where the max over $\mathbf{R}'$ is taken over a ball of radius $\max(\|\Delta_1\|,\|\Delta_2\|)$ around $\mathbf{R}$), the difference between the approximate ground states $\ket{\tilde{\psi}_0(\mathbf{R}+\Delta)}$ and the true ground state at $\Delta=0$ is
\begin{equation}
|\braket{\tilde{\psi}_0(\mathbf{R}+\Delta_1)}{\tilde{\psi}_0(\mathbf{R} +\Delta_2)}| \ge 1 - 3\delta - \frac{6 \max\{\|\Delta_1\|,\|\Delta_2\|\} \max_{\mathbf{R}',i}\left\|\tfrac{dH(\mathbf{R}')}{dR_i}\right\|} {E_1(\mathbf{R}) - E_0(\mathbf{R})}
\end{equation}
where $E_i(\mathbf{R})$ refers to the $i^{\rm th}$ eigenvalue of $H(\mathbf{R})$ and 
\begin{equation}
\max\left(\left\|\ket{\tilde{\psi}_0(\mathbf{R}+\Delta_1)} -\ket{\psi_0(\mathbf{R} + \Delta_1)}\right\|,\left\|\ket{\tilde{\psi}_0(\mathbf{R}+\Delta_2)} -\ket{\psi_0(\mathbf{R} + \Delta_2)}\right\|\right) \le \delta.
\end{equation}
We will choose $\delta$ to be in $\Theta(\epsilon)$ since the logarithmic costs of making the state preparation accurate are negligible relative to the costs of the phase estimation step.  Under these assumptions we have that the $\delta$ contribution is negligible compared to the other two terms and so we can use the following lower bound
\begin{equation}
|\braket{\tilde{\psi}_0(\mathbf{R}+\Delta_1)}{\tilde{\psi}_0(\mathbf{R} +\Delta_2)}| \ge 1 - \frac{12 \max\{\|\Delta_1\|,\|\Delta_2\|\} \max_{\mathbf{R}',i}\left\|\tfrac{dH(\mathbf{R}')}{dR_i}\right\|} {E_1(\mathbf{R}) - E_0(\mathbf{R})}
\end{equation}

In order to ensure that the probability of failing to project is at most $1/3$ during the protocol it suffices to take the failure probability of each step to be $1/9mN_a$.  An appropriate probability can be met if we choose
\begin{equation}
    \max\{\|\Delta_i\|\}  \le \frac{{E_1(\mathbf{R}) - E_0(\mathbf{R})}}{ 36m N_a \max_{\mathbf{R}',i}\left\|\tfrac{dH(\mathbf{R}')}{dR_i}\right\|} \label{eq:Deltabd1}
\end{equation}

Then, as our choices of $\mathbf{y}$ in our finite difference algorithm are at max $mdR$ from $\mathbf{R}$, a sufficient choice of $|dR|$ to guarantee this scaling is
\begin{equation}
    \max_i\|\Delta_i\|  = m |dR| \Rightarrow  |dR| \le \frac{{E_1(\mathbf{R}) - E_0(\mathbf{R})}}{  36 m^2 N_a \max_{\mathbf{R}',i}\left\|\tfrac{dH(\mathbf{R}')}{dR_i}\right\|}:=\frac{\gamma}{  36 m^2 N_a \max_{\mathbf{R}',i}\left\|\tfrac{dH(\mathbf{R}')}{dR_i}\right\| }
\end{equation}

This potentially creates a problem since in order to ensure that the ground state remains stable over all the measurements (and thereby making the ground state projection cost additive rather than multiplicative) we need to have $\|\Delta_i\|$ constant as a function of $\epsilon$.  We can therefore choose $\Delta_i$ to be a constant by taking
\begin{equation}\label{eq:dRs2}
 dR \leq \min\left\{\frac{1}{8cm }\left(\frac{\epsilon}{16\sqrt{2}\mathfrak{e} m  c}\right)^{\frac{1}{2m}},\frac{1}{16cm}, \frac{\gamma}{  36 m^2 N_a\max_{\mathbf{R}',i}\left\|\tfrac{dH(\mathbf{R}')}{dR_i}\right\| } \right\}
\end{equation}

\subsubsection{Optimization of the finite difference order}
By substituting this expression for $dR$ in Eq.~\eqref{eq_epsilon_PE}, and inverting in $T$
we have that the total amount of time (in number of simulations) needed to compute a component of the gradient with error $\epsilon$ is bounded by:
\begin{align}
    T &\leq 48 \frac{\sqrt{2} \sqrt{6}\;\lambda_{\mathrm{H}}\; \pi\; c\sqrt{m}}{\epsilon}\times\max\left\{ \left(\frac{16\sqrt{2} \mathfrak{e} c m}{\epsilon}\right)^\frac{1}{2m}m^{3/2},2m^{3/2}, \frac{9 m^2 N_a \max_{\mathbf{R}',i}\left\|\tfrac{dH(\mathbf{R}')}{dR_i}\right\|}{2\gamma}\right\},\nonumber\\
    &\leq 48 \frac{\sqrt{2} \sqrt{6}\;\lambda_{\mathrm{H}}\; \pi\; cm^{5/2}}{\epsilon}\times\max\left\{ \left(\frac{16\sqrt{2} \mathfrak{e} c m}{\epsilon}\right)^\frac{1}{2m}, \frac{9  }{2} \left(1+ \frac{N_a\max_{\mathbf{R}',i}\left\|\tfrac{dH(\mathbf{R}')}{dR_i}\right\|}{\gamma c} \right)\right\}\label{eq:gen_T_equation}
\end{align}

After the above choices have been made, the only variable that we have access to, impacting the number of simulation steps needed in the phase estimation is $m$.
This is clearly optimized in Eq.~\eqref{eq:gen_T_equation} by setting the two terms in the curly brackets equal
\begin{equation}
    \frac{16\sqrt{2}\mathfrak{e}cm}{\epsilon}=\left(\frac{9}{2}\left(1 + \frac{N_a\max_{\mathbf{R}',i}\left\|\tfrac{dH(\mathbf{R}')}{dR_i}\right\|}{\gamma c}\right) \right)^{2m}
\end{equation}
The solution to this equation is given by a Lambert $W$-function and the asymptotics of the function reveals that it suffices to choose
\begin{equation}
    m \in \widetilde{\mathcal{O}} \left( \frac{\log(\frac{\mathfrak{e}c}{\epsilon})}{\log\left(\frac{9(1+ {N_a\max_{\mathbf{R}',i}\left\|\tfrac{dH(\mathbf{R}')}{dR_i}\right\|}/{\gamma c})}{2}\right)} \right)
\end{equation}

This leads to the conclusion that the query complexity of computing one of the components of the gradient within error $\epsilon$ is at most
\begin{equation}\label{eq_asympt_scaling_FD}
    T \in \widetilde{\mathcal{O}}\left(\frac{\lambda_{\mathrm{H}} c  m^{5/2} (1 +N_a\max_{\mathbf{R}',i}\left\|\tfrac{dH(\mathbf{R}')}{dR_i}\right\|/\gamma c)}{\epsilon } \right) \subseteq \widetilde{\mathcal{O}}\left(\frac{\lambda_{\mathrm{H}}   \log^{5/2}(\mathfrak{e}) (c +N_a\max_{\mathbf{R}',i}\left\|\tfrac{dH(\mathbf{R}')}{dR_i}\right\|/\gamma )}{\epsilon } \right).
\end{equation}

\subsubsection{Gradient vector estimation}
In the above we have considered the cost to estimate single components of the gradient vector.
It remains to convert our resulting expression for the query complexity for these single components (Eq.~\eqref{eq_asympt_scaling_FD}) into one for the error on the gradient vector.
This involves repeating this calculation for each of the $3N_a$ components of the gradient vector. If we want to ensure that the error in the gradient calculation (as quantified by the $2$-norm of the error vector) is at most $\epsilon$ it suffices to take
\begin{equation}
    \max_i\left|\frac{d E}{d {R_i}}   - \frac{dE^{(m)}}{dR_i}\right| \le \frac{\epsilon}{\sqrt{3N_a}},
\end{equation}
because
\begin{equation}
    \left\|\frac{d E}{d\mathbf{R}}   - \frac{dE^{(m)}}{d\mathbf{R}}\right\|_2 \le \sqrt{3N_a}\max_i\left|\frac{d E}{d {R_i}}   - \frac{dE^{(m)}}{dR_i}\right| \le \epsilon.
\end{equation}
(In principle this could be improved by importance sampling, however we do not believe this will provide a significant gain here.)
We then choose the error in individual gradient component estimations to be $\epsilon' \le \epsilon/\sqrt{3N_a}$.
The asymptotic cost of the overall number of queries made in the simulation, excluding the cost of initial state preparation, is then
\begin{align}
     T_{\rm{FD}}\!\left(\frac{dE}{d\mathbf{R}}\right) 
     &\in \widetilde{\mathcal{O}} \left(\frac{N_a^{3/2} \lambda_{\mathrm{H}}  \log^{5/2}(\mathfrak{e}) (c+N_a \max_{\mathbf{R}',i}\left\|\tfrac{dH(\mathbf{R}')}{dR_i}\right\|/\gamma) }{\epsilon}\right)~\label{eq:gradcalc}
\end{align}

If we define the initial state preparation process to be an algorithm that prepares a state using no queries (such as a Hartree-Fock state) and has a probability of success $a_0^2$ then the cost of implementing an approximate ground state projector using~\cite{Ge2017} is in

\begin{equation}
    T_{\rm{Prep}} \in \widetilde{\mathcal{O}}\left(\frac{\lambda_{\mathrm{H}} \log(1/\epsilon)}{|a_0|\gamma}\right).
\end{equation}
Recall that state preparation is only needed on average $O(1)$ times, due to the choice of $\|\Delta_i\|$ made in Eq.~\eqref{eq:Deltabd1}. 
This implies that the total query complexity of using the finite difference method is 

\begin{align}
     T_{\rm{FD-Tot}}\!\left(\frac{dE}{d\mathbf{R}}\right) 
     &\in \widetilde{\mathcal{O}} \left(\lambda_{\mathrm{H}} \left(\frac{N_a^{3/2}   \log^{5/2}(\mathfrak{e}) (c+N_a\max_{\mathbf{R}',i}\left\|\tfrac{dH(\mathbf{R}')}{dR_i}\right\|/\gamma) }{\epsilon}+\frac{\log(1/\epsilon)}{|a_0|\gamma}\right)\right)~\label{eq:gradcalc2}
\end{align}

The above cost is given in the number of oracle calls to a block encoding of the Hamiltonian (Sec.~\ref{sec:block_encodings}).
To convert this into wall-clock time, we have to multiply this by the corresponding oracle cost in first or second quantization.
The results of~\cite{vonBurg2020} show that the PREPARE operation for second quantized simulations of systems in a non-planewave basis can be implemented using $\widetilde{\mathcal{O}}(N^2)$ Toffoli gates; however, with the use of tensor hypercontraction this becomes $\widetilde{\mathcal{O}}(N)$ asymptotically \cite{lee2021even}.  The SELECT operation can be implemented in $\widetilde{\mathcal{O}}(N)$ operations, where $N$ is the number of spin orbitals used in each simulation~\cite{lee2021even}.  Thus the overall cost of the walk operator is in $\widetilde{\mathcal{O}}(N)$, and we have
\begin{equation}
    {\rm ToffCount}_{2\text{nd}} \in \widetilde{\mathcal{O}}\left(\lambda_{\mathrm{H}} N \left(\frac{N_a^{3/2}   \log^{5/2}(\mathfrak{e}) (c+N_a\max_{\mathbf{R}',i}\left\|\tfrac{dH(\mathbf{R}')}{dR_i}\right\|/\gamma ) }{\epsilon}+\frac{\log(1/\epsilon)}{|a_0|\gamma}\right)\right)\label{eq:Toff2}
    \end{equation}
The scaling of this algorithm depends on the scaling of the value of $\lambda_{\mathrm{H}}$. In first quantization, the cost of the simulation using qubitization is based on the discussion in~\sec{BEfirst} is in~\cite{su2021fault}
\begin{equation}
    {\rm ToffCount}_{1\text{st}} \in \widetilde{\mathcal{O}}\left(\lambda_{\mathrm{H}} \eta \log(N) \left(\frac{N_a^{3/2}   \log^{5/2}(\mathfrak{e}) (c+N_a\max_{\mathbf{R}',i}\left\|\tfrac{dH(\mathbf{R}')}{dR_i}\right\|/\gamma ) }{\epsilon}+\frac{\log(1/\epsilon)}{|a_0|\gamma}\right)\right)
    \end{equation}
Here the result arises from the fact that the select operation incurs a cost that is in $O(
\eta \log(N))$ and the prepare operation has a cost that is in $O( \eta \log^2(N) + \log^3(N))$.

\begin{table}[tb]
    \begin{tabular}{|c|c|c|}
    \hline
    System & Analytic Scaling & Empirical Scaling\\
    \hline
    First quantized plane waves & $\widetilde{\mathcal{O}}\left(\frac{N^{2/3}N_a^{17/6}(\mathbb{E}(|Z_i|))^{1/3} }
    {\epsilon}\right)$ & -- \\
    Second quantized plane waves & $\widetilde{\mathcal{O}}\left(\frac{N^{3} N_a^{3/2}  }{\epsilon}\right)$ & --\\
    \hline
    Hydrogen Chains (Sparse) & -- & $\widetilde{\mathcal{O}}\left(\frac{N_a^{3.833} }{\epsilon}\right)$ \\
    Hydrogen Chains (Double Factorization) & -- & $\widetilde{\mathcal{O}}\left(\frac{N_a^{4.442} }{\epsilon}\right)$\\
    Water Clusters (Sparse) & -- & $\widetilde{\mathcal{O}}\left(\frac{N_a^{3.868} }{\epsilon}\right)$ \\
    Water Clusters (Double Factorization) & -- & $\widetilde{\mathcal{O}}\left(\frac{ N_a^{4.197} }{\epsilon}\right)$\\
    \hline
    \end{tabular}
    \caption{Scaling of finite difference gradient estimation for a system with $N_a$ atoms and $N$ orbitals and $\eta$ electrons within the Born-Oppenheimer approximation for constant $c$ and $\mathfrak{e}$ as well as $\epsilon$ error as measured by the $2$-norm of the gradient vector.  Numerical estimates follow from the data in Sec.~\ref{sec:numerics}.  The analytic results for plane waves follows from taking the volume of the unit cell $\Omega$ to be proportional to the number of plane waves $N$.  Note that for the Hydrogen chains and water clusters the we choose $N\propto N_a$ because a minimal basis is used within these simulations. The cost of state preparation is ignored in these scalings, which is true if $|a_0|\gamma \gg \epsilon$.}
    \label{tab:FDscale}
\end{table}

As a final point, we can consider the case where we make no attempt to reuse the quantum state evaluated for the finite difference formulas.  The advantage of this approach is that we do not need to reduce $dR$ to ensure that the groundspace is sufficiently stable under perturbations of the Hamiltonian across the different points of interest.  However, it will often be less efficient due to the many state preparation steps needed and also it will require a set of different state.  We assume below that each such state has success probability at least $|a_i|^2$.  Repeating the same analysis given above leads to a number of oracle queries that scales as
\begin{align}
     T_{\rm{FD}}\!\left(\frac{dE}{d\mathbf{R}}\right) 
     &\in \widetilde{\mathcal{O}} \left(N_a^{3/2} \lambda_{\mathrm{H}} \left[\frac{  \log^{5/2}(\mathfrak{e}) c }{\epsilon}+\frac{\log(1/\epsilon)}{\min_i|a_i|\gamma}\right]\right).~\label{eq:gradcalc3}
\end{align}
Combining Eq.~\eqref{eq:gradcalc3} with Eq.~\eqref{eq:gradcalc2} yields a joint expression
\begin{equation}\label{eq:gradcalc5}
    T_{\rm{FD}}\!\left(\frac{dE}{d\mathbf{R}}\right)\in\widetilde{\mathcal{O}}\left(\lambda_{\mathrm{H}}N_a^{3/2}\left[\frac{\log^{5/2}(\mathfrak{e})c}{\epsilon}+\min\left\{\frac{\log(1/\epsilon)}{\min_i|a_i|\gamma},\;\frac{N_a\max_{\mathbf{R}',i}\left\|\tfrac{dH(\mathbf{R}')}{dR_i}\right\|\log^{5/2}(\mathfrak{e})}{\epsilon\gamma}\right\}\right]\right).
\end{equation}

We examine this scaling in Sec.~\ref{sec:numerics} and the results of the scalings given either empirically or analytically is given in \tab{FDscale}.  In all cases we assume that $\min_i |a_i| \gamma \gg \epsilon$.  While these choices are made for simplicity, they do privilege the state preparation method given here because it means that we do not need to artificially reduce the value of $dR$ to ensure that the state preparation costs are additive.  This point is important to consider when comparing these costs with the methods in the following sections.

The analytic results given in~\tab{FDscale} are straightforward to verify.  In the plane-wave basis, it is difficult to reason about ionic systems because of the periodic nature of the basis functions leading to infinite energy.  For this reason we focus on the case where the number of electrons in the system is equal to the nuclear charge.  This implies that if we define $Z_i$ to be the charge of nucleus $i$ then the number of electrons can be bounded above by
\begin{equation}
    \eta = \sum_{i=1}^{N_a} Z_i \le N_a \mathbb{E}(|Z_i|)
\end{equation}
Next the value of $\lambda_{\mathrm{H}}$ for such systems is given by Eq. 25 of~\cite{su2021fault} to be in $\mathcal{O}(N^{2/3}\eta^{1/3})$ in the limit where the number of plane-waves is much greater than the number of electrons.
Thus 
\begin{equation}
    \lambda_{\mathrm{H}}\in O \left( N^{2/3} N_a^{1/3} (\mathbb{E}(|Z_i|))^{1/3} \right)
\end{equation}
This implies that the cost of performing this process within the required error (under the assumptions on $\|\frac{dH}{dR_i}\|, c,\mathfrak{e}$ made above) is
\begin{equation}
    {\rm ToffCount}_{1st} \in \widetilde{\mathcal{O}}\left(\frac{N^{2/3}N_a^{17/6} (\mathbb{E}(|Z_i|))^{1/3}}{\epsilon}\right)
    \end{equation}

The case of second quantized planewave simulations is much simpler to analyze.
From Eq.~(53) of~\cite{Babbush2019SYK} we find that $\lambda_{\mathrm{H}} \in \widetilde{\mathcal{O}}(N^2)$ when we take $\Omega \propto N$.
Eq.~\eqref{eq:Toff2} then gives us that the Toffoli count scales as $\widetilde{\mathcal{O}}(N^3 N_a^{3/2} /\epsilon)$.
This completes our justification of the analytic scalings given in~\tab{FDscale}.

\subsubsection{Considerations on system dependent constants}\label{App_Grad_c_constant}

One of the biggest disadvantages about using high-order 
%divided
finite
difference formula is that their performance depends on the higher-order derivatives of the energy.  Specifically, Eq.~\eqref{eq_bounding_Derivative} gives the requirements that these constants much satisfy for each of the derivatives.  While a simple expression for $c$ and $\mathfrak{e}$ cannot be easily extracted from such an expression, lower bounds can be extracted from the requirements.  
Using Eq.~\eqref{eq_bounding_Derivative} we have from substituting $k=1$ into the expression that if $\ket{\psi_0}$ is the true ground state (after choosing the global phases of the eigenstates such that their derivatives are orthogonal to themselves) we then have that
\begin{equation}
     \left\langle\psi_0\left| \frac{dH}{dR_i} \right|\psi_0\right\rangle \le \mathfrak{e} c.
\end{equation}
From this expression we see that while the product of the two constants is lower bounded by the expected derivative in the ground state, that does not determine either of the two constants that constitute it.  However, we can see that dimensionally $\mathfrak{e}$ can be thought of as an energy scale and $c$ can be thought of an inverse length scale for the system.

A reasonable approximation may be made for the constant $c$ in the case of first-quantized plane waves.
If we assume that the eigenvalue gap is large and the high-order derivatives of the Hamiltonian dominate the expectation value of the energy then it is straight forward to see from perturbative arguments that
\begin{equation}
     \frac{d^kE}{dR_i^k}\sim \left\langle\psi_0\left| \frac{d^kH}{dR_i^k} \right|\psi_0\right\rangle  \le \mathfrak{e} c^k k^{k/2}.
\end{equation}
Thus we have in this restricted case that 
\begin{equation}
c \ge   \frac{1}{\sqrt{k}}\left(\frac{\left\langle\psi_0\left| \frac{d^kH}{dR_i^k} \right|\psi_0\right\rangle}{\mathfrak{e}}\right)^{1/k}~\forall~k 
\end{equation}
If we choose $\mathfrak{e} = \bra{\psi_0} H \ket{\psi_0}$ to be the ground state energy then it follows that any such scaling that satisfies this must obey
\begin{equation}
c \gtrapprox  \sup_k \left(\frac{1}{\sqrt{k}}\left(\frac{\left\langle\psi_0\left| \frac{d^kH}{dR_i^k} \right|\psi_0\right\rangle}{\bra{\psi_0} H \ket{\psi_0}}\right)^{1/k}\right)  
\end{equation}
While the derivative with respect to the coordinate of the Hamiltonian in general is difficult to compute, we can argue about it in plane wave bases which have an analytical expression as illustrated by Eq.~\eqref{eq:pwd_u}.  Specifically, we have that if $R_i=R_{l,\alpha}$ is one of the position coordinates of a nuclei with charge $Z_l$ we have (using the same notation as Sec.~\ref{sec:BEfirst}) that
\begin{equation}
\frac{d^kH}{dR_{l,a}^k}=- \frac{4 \pi Z_l}{\Omega} \sum_{\bf{p}} \sum_{\substack{\bf{s} \neq 0}}  \frac{(ik_s^{\alpha}[i])^k \, e^{i \, \bf{k_{s}} \cdot \left(\bf{R}_l + \bf{r_{p}}\right)}}{\left| \bf{k_s} \right|^2} c_{\bf{p}}^\dagger c_{\bf{p}}
\end{equation}
Terms in this result scale like $\mathcal{O}(N^{k/3})$, but as the phase $e^{i\mathbf{k}_{\mathbf{s}}\cdot(\mathbf{R}_l+\mathbf{r}_{\mathbf{p}})}$ oscillates rapidly as we change $\mathbf{p}$ and $\mathbf{s}$, the final sum over both variables can be as small as $o(1)$.
This suggests that physical circumstances might exist where $c\in o(1)$, but it could in other cases scale as $c\in\mathcal{O}(N^{1/3})$ or potentially higher.
Numerical explorations will likely be needed to better understand the scaling of $\mathfrak{e}$ and $c$.

\subsection{Expectation value estimation of the force operator at the Heisenberg limit with overlap estimation}\label{sec:HeisenbergLimit}

On a fault-tolerant quantum device one can estimate expectation values at the Heisenberg limit scaling as $\epsilon^{-1}$ with the unbiased error $\epsilon$ (which is tight~\cite{Giovannetti04Quantum,Higgins09Demonstrating}).
This allows us to estimate gradients via the Hellman-Feynman theorem, similarly to how one would achieve this in a NISQ setting (Sec.~\ref{sec:nisq}).
The Heisenberg limit was originally achieved by the overlap estimation algorithm (OEA) of \cite{Knill06Optimal}.
In our case the expectation values that we wish to measure are real-valued, enabling a slight optimization over the original algorithm.
In this section we perform an asymptotic cost analysis of our implementation of the OEA on a block-encoded derivative operator, including optimizing the preparation and reflection subroutines for the purposes of derivative estimation, and compare it to previous methods.

The OEA aims to perform unbiased estimation of $\langle\psi|U|\psi\rangle$ for an arbitrary unitary operator $U$ and state $|\psi\rangle$.
It does so by calling a subroutine, the amplitude estimation algorithm (AEA), which estimates $|\langle\psi|U|\psi\rangle|$ via phase estimation of the Szegedy walk unitary
\begin{equation}
\mathcal{S}=\mathcal{R}U\mathcal{R}U^{\dag},\;\mathcal{R}=I-2|\psi\rangle\langle\psi|.\label{eq:Szegedy_def}
\end{equation}
The eigenvalues of this walk operator, constrained to the subspace spanned by $\ket{\psi}$ and $U \ket{\psi}$ provide us with all the information we need to estimate the amplitude.
To see this, let 
\begin{equation}
    P_\psi = \ket{\psi}\!\bra{\psi} + (1 - \ketbra{\psi}{\psi})U\ketbra{\psi}{\psi} U^\dagger(1 - \ketbra{\psi}{\psi}).
\end{equation}
We have from Jordan's Lemma that the action of $\mathcal{S}$ can be broken into a direct sum of irreducible two-dimensional subspaces. Therefore $P_\psi \mathcal{S} P_\psi = \mathcal{S} P_\psi$.  It then follows that for any $\ket{\phi}$ such that $P_\psi \ket{\phi} = \ket{\phi}$ that phase estimation will return an eigenvalue of the form $e^{\pm i\theta}$ for some angle $\theta \in [-\pi,\pi]$.  Further analysis in~\cite{Szegedy2004} shows that this eigenphase satisfies $\cos(\theta)=2|\langle U\rangle|^2-1$. This implies that if $\hat{\theta}$ is our estimate of theta, then
\begin{equation}
   \widehat{|\langle {U} \rangle|} = \sqrt{\frac{\cos(\hat{\theta}) +1}{2}}  = |\cos(\hat\theta/2)| =|\cos((\theta + \Delta )/2)|= |\cos(\theta/2)\cos(\Delta/2)-\sin(\theta/2)\sin(\Delta/2) | ,
\end{equation}
where $\Delta = \hat{\theta} - \theta$ and $\mathbb{V}(\hat{\theta}) = \mathbb{E}(\Delta^2)$ using the phase estimation method reviewed in~\app{App_FD_Holevo_variance} which is unbiased.
Under the assumption that $\sqrt{\mathbb{V}(\theta)}/\theta \ll 1$, the fourth moment of the distribution of $\hat{\theta}$ is negligible and that the distribution of $\hat{\theta}$ has negligible support over the branch cut
\begin{align}
    \mathbb{V}(\widehat{|\langle {U} \rangle|}) &= \cos^2(\theta/2)\mathbb{V}(\cos(\Delta/2)) + \sin^2(\theta/2)\mathbb{V}(\sin(\Delta/2)) \lesssim \mathbb{V}(\sin((\Delta/2))|)\nonumber\\
    &\approx \mathbb{V}(\hat{\theta})/4. \label{eq:derivUnc}
\end{align}
Under the above assumptions, an estimator $\hat{\theta}$ of the eigenphases of $\mathcal{S}$ in the support of the input state $\ket{\psi}$ that is approximately unbiased and has variance of $\epsilon_{\theta}^2$ can be achieved using the techniques of \app{App_FD_Holevo_variance} in approximately $\frac{\pi}{2\epsilon_{\theta}}$ calls to $\mathcal{S}$.
This variance in $\hat{\theta}$ then propagates to an estimate of $|\langle U\rangle|$ with a standard deviation that is approximately bounded by $\frac{\epsilon_{\theta}}{2}$ from Eq.~\eqref{eq:derivUnc}.

To learn the sign of $\langle U\rangle$, we use the fact that
\begin{equation}
    \langle +|\langle\psi|\mathrm{c-}U|\psi\rangle|+\rangle=\frac{1}{2}\Big(1+\mathrm({\langle\psi|U|\psi\rangle})\Big).\label{eq:cU}
\end{equation}
We can then estimate the expectation value using the fact that, in our context $\bra{\psi} U \ket{\psi}$ is real valued.
(Note that this would not be the case if one encoded $U=e^{-itdH/dR_i}$ as was originally proposed in \cite{Knill06Optimal}.)
That $\bra{\psi} U \ket{\psi}$ is real implies $|\langle +|\langle\psi|\mathrm{c-}U|\psi\rangle|+\rangle|=\langle +|\langle\psi|\mathrm{c-}U|\psi\rangle|+\rangle$, and thus
\begin{equation}
    \bra{\psi}U \ket{\psi} = \frac{1}{2} \Big(4\big((1+ \bra{\psi} U \ket{\psi})/2\big)^2 -(\bra{\psi} U \ket{\psi})^2 -1  \Big) ,
\end{equation}
and we can find an estimator for the overlap as
\begin{equation}
    \widehat{ \langle U \rangle} =\frac{1}{2}\left( 4\widehat{\left(\frac{1 + \langle U \rangle }{2}\right)}^2  - \widehat{|\langle U \rangle|}^2 -1 \right).
\end{equation}
Note that one could in principle estimate $\langle\psi|U|\psi\rangle$ directly from amplitude estimation of $c-U$. 

The variance of the estimate of $\widehat{\langle U \rangle}$ is then found from the additive property of variance of independent variables:
\begin{equation}
\mathbb{V}(\widehat{ \langle U \rangle}) =
    4\mathbb{V}\left( \widehat{\left(\frac{1 + \langle U \rangle }{2}\right) }^2\right)+\frac{1}{4}\mathbb{V}\left( \widehat{|\langle U \rangle|}^2 \right) .
\end{equation}
We have from Jensen's inequality that for any convex function $\Phi$ and random variable $X$ with variance $\sigma^2$,
\begin{equation}
    \mathbb{E}(\Phi(X)) - \Phi(\mathbb{E}(X)) \le \frac{\sigma^2 \sup_X \Phi''(X)}{2}.\label{eq:tightJensen}
\end{equation}
Since $\Phi(X) = X^2$ is a convex function with second derivative $1$ for all $X\in [-1,1]$ we have from Eq.~\eqref{eq:tightJensen} that
\begin{equation}
    \mathbb{V}(\widehat{ \langle U \rangle}) \le
    4\mathbb{V}\left( \widehat{\left(\frac{1 + \langle U \rangle }{2}\right) }\right)+\frac{1}{4}\mathbb{V}\left( \widehat{|\langle U \rangle|} \right) ,
\end{equation}
Specifically, we will choose $\mathbb{V}(\widehat{\langle U \rangle}) \le \epsilon_{\rm OEA}^2$.  This motivates the following choices of the variance targets for the two operations
\begin{equation}
    \mathbb{V}\left( \widehat{\left(\frac{1 + \langle U \rangle }{2}\right) }\right) \le \frac{\epsilon_{\rm OEA}^2}{8}, \qquad \mathbb{V}\left( \widehat{|\langle U \rangle|} \right) \le 2\epsilon_{\rm OEA}^2\label{eq:varTargets}
\end{equation}

We now estimate the complexity of computing the requisite estimators within the variance requirements of Eq.~\eqref{eq:varTargets}.
Our circuit primitives and phase estimation routine for the amplitude estimation algorithm (including control) are given in Fig.~\ref{fig:FT_circuits}.
We note that control can be added to $\mathcal{S}$ by either controlling the implementations of $U$ and $U^{\dag}$ or by controlling the implementation of $R$.
We assume that the cost to implement controlled-$\mathcal{S}$ compared to the cost to implement $\mathcal{S}$ without control is negligible in the number of Toffoli gates required (which is typically the case).
We make one further improvement to our phase estimation routine in the same vein as \cite{BabbushSpectra}: instead of performing phase estimation by controlling $\mathcal{S}^{2^n}$ by the $n$th qubit in the QPE register, we implement the unitary
\begin{equation}
    |0\rangle\langle 0|\mathcal{S}^{-2^{n-1}}+|1\rangle\langle 1|\mathcal{S}^{2^{n-1}}.\label{eq:faster_QPE_unitary}
\end{equation}
This may be achieved by realizing that
\begin{equation}
    \mathcal{R}\mathcal{S}^{k}\mathcal{R}=\mathcal{R}\Big(\mathcal{U}^{\dag}\mathcal{R}\mathcal{U}\mathcal{R}\Big)^k\mathcal{R}=\Big(\mathcal{R}\mathcal{U}^{\dag}\mathcal{R}\mathcal{U}\Big)^k=\Big(\mathcal{U}^{\dag}\mathcal{R}\mathcal{U}\mathcal{R}\Big)^{\dag k}=\mathcal{S}^{-k},
\end{equation}
and so one need only append $\mathcal{S}^{2^n}$ by an application of $\mathcal{R}$ on either side, controlled by the $n$th qubit in the QPE register being in the $|0\rangle$ state to implement Eq.~\eqref{eq:faster_QPE_unitary}.
This technique halves the number of applications of $\mathcal{S}$ required, and significantly reduces the control overhead.
With this implemented, using the variance targets in Eq.~\eqref{eq:varTargets} for our two separate rounds of amplitude estimation will achieve an estimate of $\langle\psi|U|\psi\rangle$ with standard deviation approximately bounded by $\epsilon_\text{OEA}$, using
\begin{equation}
N_\text{calls}=\frac{\sqrt{2}\pi}{\epsilon_{\mathrm{OEA}}} + \frac{\pi}{2\sqrt{2}\epsilon_{\mathrm{OEA}}}= \frac{5\pi}{2\sqrt{2} \epsilon_{\rm OEA}}
\end{equation}
calls to a (controlled) circuit implementation of $\mathcal{S}$.

\begin{figure}
    \centering
    \includegraphics[width=0.6\textwidth]{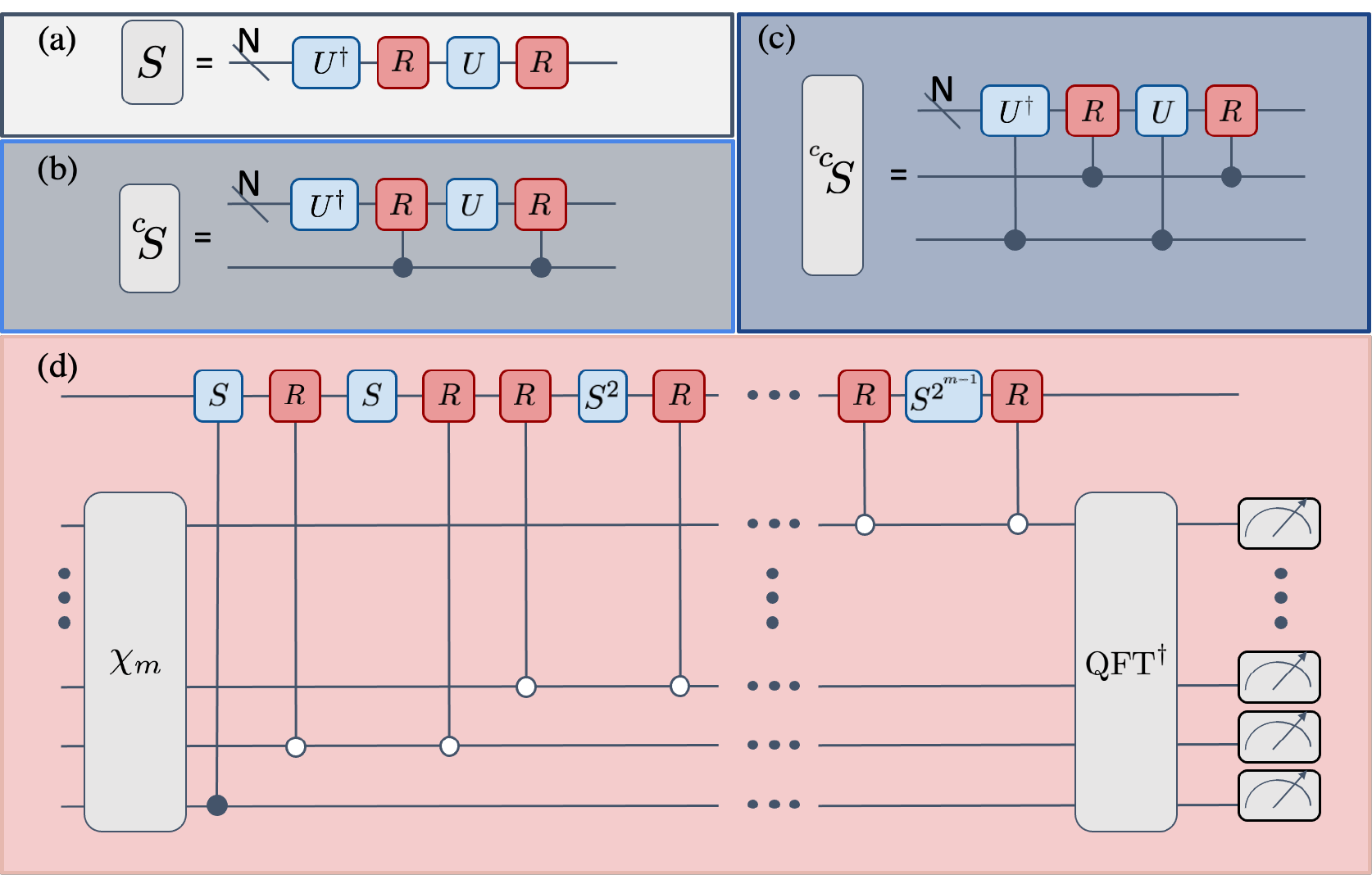}
    \caption{FT circuits for the overlap estimation algorithm. (a) a circuit implementation of the quantum walk unitary $\mathcal{S}$ without control. (b) a circuit implementation of the quantum walk unitary $\mathcal{S}$ controlled by a single qubit. (c) A circuit implementation of the quantum walk unitary $\mathcal{S}$ controlled by two qubits, required to estimate the sign of $\langle U\rangle$. (d) Algorithm for the quantum phase estimation routine, sped up by a factor $2$ by improvements discussed in the text. Note that for the purposes of estimating the sign of $\langle U\rangle$, $\mathcal{S}$ should be replaced by $c-\mathcal{S}$ in the QPE circuit (d) and an additional ancilla qubit added. The unitary $\chi_m$ denotes a preparation of a QPE register state to estimate at the Heisenberg limit~\cite{BabbushSpectra}. Black dots denote a controlled operation implemented when the control qubit is in the $|1\rangle$ state, while white dots denote a controlled operation implemented when the control qubit is in the $|0\rangle$ state.}
    \label{fig:FT_circuits}
\end{figure}

To convert the cost of the OEA to the cost of gradient estimation, we need to account for the fact that $\frac{d H}{d R_i}$ is not a unitary operator.
This is an issue, as the OEA explicitly requires unitary $U$.
This issue may be resolved by block encoding $\frac{d H}{d R_i}$ (Sec.~\ref{sec:block_encodings}): adding additional qubits and finding a unitary operator $U_i$ on the larger Hilbert space such that
\begin{equation}
    \left\|(\bra{0}\otimes \openone)
U_i(\ket{0}\otimes \openone)
-(\bra{0}\otimes \openone) \left( \begin{array}{cc}
    \frac{1}{\lambda_{\mathrm{F}_i}}\frac{dH}{dR_i} & U_{01}\\
    U_{10} & U_{11} \end{array}\right)(\ket{0}\otimes \openone)\right\|_{\infty}\leq\epsilon_{\mathrm{BE}} \hspace{0.5cm}\rightarrow\hspace{0.5cm} \left|\langle\psi|\langle 0|U_i|0\rangle|\psi\rangle-\frac{1}{\lambda_{\mathrm{F}_i}}\frac{d E}{dR_i}\right|\leq\epsilon_{\mathrm{BE}},\label{eq:def_block_encoding_derivative}
\end{equation}
where $\epsilon_{\mathrm{BE}}$ is the block-encoding error.
This block encoding of $\frac{dH}{dR_i}$ requires rescaling by some constant $\lambda_{\mathrm{F}_i}$ for $U_i$ to be unitary.
We may propagate this rescaling to the error in estimating $\frac{dE}{dR_i}$; an error $\epsilon_{\mathrm{OEA}}$ in the estimation of $\langle\psi|U_i|\psi\rangle$.  
However, one complication involved in doing so arises from the fact that many of the estimation errors given above are stated in terms of the variance of the estimate.  In contrast, errors such as the block-encoding errors are deterministic errors.
To solve this, we can use the Chebyshev inequality to say that for an estimate $\mathbf{y}$ with variance $\sigma^2$
\begin{equation}
{\rm Pr}(|\mathbf{y} -\mathbb{E}(\mathbf{y})| \ge k \sigma) \le \frac{1}{k^2}.    
\end{equation}
Thus, up to a constant factor, we can with probability of failure in $\mathcal{O}(1)$, replace the uncertainty in our estimates with their standard deviation.
From the above discussion, we then
yields an estimate of $\tfrac{dE}{dR_i},$ denoted $\widehat{\tfrac{dE}{dR_i}}$ such that with constant probability of failure $\le 1/4$ (i.e. $k=2$ above)
\begin{equation}
  \left|\widehat{\frac{dE}{dR_i}} - \frac{dE}{dR_i}  \right|  \le \lambda_{\mathrm{F}_i}(\epsilon_{{\rm OEA},i} + \epsilon_{\rm BE} + \epsilon_{\mathcal{R}})
\end{equation}
in the estimation of $\frac{dE}{dR_i}$.
Here, $\epsilon_{\mathcal{R}}$ is the error in implementing the reflection operator $\mathcal{R}=I-2|\psi\rangle\langle\psi|$, and we index the errors in the overlap estimation and block encoding by $i$ (as these will be different for different gradient operators).
To bound the error in $\frac{dE}{dR_i}$ by some $\epsilon_i$, we can require the error in both the block encoding and overlap estimation be bounded as $\epsilon_{\mathrm{BE}},\epsilon_{\mathrm{OEA}},\epsilon_{\mathcal{R}}\leq\frac{\epsilon_i}{3\lambda_{\mathrm{F}_i}}$.
The bound on $\epsilon_{\mathrm{OEA}}$ then gives us the cost of executing the gradient estimation algorithm
\begin{equation}
    T_{\mathrm{grad-OEA},i}=\frac{15\pi\lambda_{\mathrm{F}_i}}{\sqrt{2}\epsilon_i}(T_{F,i}+T_{R})+T_{P,i}+\mathcal{O}\big(\log[\lambda_{\mathrm{F}_i}/\epsilon_i]\big)\label{eq:cost_oea},
\end{equation}
where $T_R$ is the cost of implementing the reflection about $|\psi\rangle$, and $T_{P,i}$ is the cost of preparing $|\psi\rangle$ for the $i$th gradient estimation.
(We will re-use our system register between gradient estimation applications to significantly reduce the cost of $T_{P,i}$ for $i>1$.)
Note that from the Chernoff bound, the probability of failure  of estimating the component of the gradient within the required error tolerance can be reduced from $1/4$ to $\delta$ at cost that is in $\mathcal{O}(\log(1/\delta))$.
We allow here for the fact that when performing multiple different gradient estimations we may have different block encodings (see Sec.~\ref{sec:block_encodings}) and state preparation methods (see Sec.~\ref{sec:state_prep_and_reflection}). However, we assume that the reflection oracle (see Sec.~\ref{sec:state_prep_and_reflection}) is implemented the same for each estimation.

\subsubsection{Parallel estimation and importance sampling}\label{sec:FT_Importance_Sampling}

The cost estimation in Eq.~\eqref{eq:cost_oea} is the cost of estimating a single gradient element.
To expand this cost to a vector of gradients $\frac{dE}{dR_i}$, we sum individual terms
\begin{equation}
    T_{\mathrm{grad-OEA}}=\sum_iT_{\mathrm{grad-OEA},i}.
\end{equation}
Following the methods developed in Sec.~\ref{sec:NISQ_estimation}, we may select $\epsilon_i$ to optimize the $1$- or $2$-norm of the gradient vector $\mathbf{\epsilon}$ by importance sampling, where the resources are allocated optimally, spending more time on the components that require a higher accuracy.
The final term in Eq.~\eqref{eq:cost_oea} is negligible and the second term is not affected by our choice of $\epsilon_i$ (we will investigate the cost of state preparation in the following section). 
Thus, we consider the problem of minimizing the sum of the first terms:
\begin{equation}
    \sum_i\frac{15\pi\lambda_{\mathrm{F}_i}}{\sqrt{2}\epsilon_i}(T_{F,i}+T_{R}),
\end{equation}
with either a fixed $1$-norm $\sum_i\epsilon_i$ or $2$-norm $\sum_i\epsilon_i^2$.
This can be solved by the same Lagrangian methods as in previous sections.
Solving for a fixed $1$-norm yields the condition
\begin{equation}
    \epsilon_i=\frac{\epsilon_{\mathrm{grad-OEA}}\Big[\lambda_{\mathrm{F}_i}(T_{F,i}+T_{R})\Big]^{\frac{1}{2}}}{\sum_{i'}\Big[\lambda_{i'}(T_{F,i'}+T_{R})\Big]^{\frac{1}{2}}}.
\end{equation}
The Cauchy-Schwarz inequality and the Chernoff bound can then be used to bound the total cost of the gradient evaluation within error $\epsilon_{\rm grad-OEA}$ with probability of failure $1/4$ using a total cost of 
\begin{align}
    T_{\mathrm{grad-OEA}}&=\frac{90\pi\log(12N_a)}{\sqrt{2}\epsilon_{\mathrm{grad-OEA}}}\bigg[\sum_i(T_{F,i}+T_{R})^{\frac{1}{2}}\lambda_{\mathrm{F}_i}^{\frac{1}{2}}\bigg]^2+\sum_iT_{P,i}+\sum_i\mathcal{O}\Big(\log\big[\lambda_{\mathrm{F}_i}^{\frac{1}{2}}(T_{F}+T_R)^{\frac{1}{2}}\epsilon^{-1}_{\mathrm{grad-OEA}}\big]\Big)\nonumber\\
    &\in \widetilde{\mathcal{O}}\left( \frac{1}{\epsilon_{\rm grad - OEA}}\left( \sum_i \lambda_{\mathrm{F}_i} \right)\left(\sum_i T_{F,i} + T_R \right) + \sum_i T_{P,i} \right).
\end{align}
If we promise that $\lambda_{\mathrm{F}_i} \le \lambda_{\mathrm{F}}$ and $T_{F,i}\le T_{F}$, this becomes
\begin{equation}
    T_{\mathrm{grad-OEA}}\in\widetilde{\mathcal{O}}\Big(\epsilon_{\mathrm{grad-OEA}}^{-1}(T_{F}+T_R)\lambda_{\mathrm{F}}N_a^2+\sum_iT_{P,i}\Big).
\end{equation}
We can repeat the same arguments for the case where error $\epsilon_{\rm grad-OEA}$ is desired in the $2$-norm.
\begin{equation}
    \epsilon_i=\frac{\epsilon_{\mathrm{grad-OEA}}\Big[\lambda_{\mathrm{F}_i}(T_{F,i}+T_R)\Big]^{\frac{1}{3}}}{\Big[\sum_{i'}\lambda_{i'}^{2/3}(T_{F,i'}+T_R)^{2/3}\Big]^{\frac{1}{3}}},
\end{equation}
for a total cost
\begin{align}
    T_{\mathrm{grad-OEA}}&=\frac{90\pi\log(12N_a)}{\sqrt{2}\epsilon_{\mathrm{grad-OEA}}}\bigg[\sum_i(T_{F,i}+T_R)^{\frac{2}{3}}\lambda_{\mathrm{F}_i}^{\frac{2}{3}}\bigg]^{\frac{3}{2}}+\sum_iT_{P,i}+\sum_i\mathcal{O}\Big(\log\big[\lambda_{\mathrm{F}_i}^{\frac{2}{3}}(T_{F}+T_R)^{\frac{2}{3}}\epsilon^{-1}_{\mathrm{grad-OEA}}\big]\Big)\nonumber\\
    &\in \mathcal{O}\left(\frac{1}{\epsilon_{\mathrm{grad-OEA}}}\bigg[\sum_i(T_{F,i}+T_R)^{\frac{4}{3}}\bigg]^{3/4} \bigg[\sum_i\lambda_{\mathrm{F}_i}^{\frac{4}{3}}\bigg]^{\frac{3}{4}}+\sum_iT_{P,i}+\sum_i\mathcal{O}\Big(\log\big[\lambda_{\mathrm{F}_i}^{\frac{2}{3}}(T_{F}+T_R)^{\frac{2}{3}}\epsilon^{-1}_{\mathrm{grad-OEA}}\big]\Big)\right),
\end{align}
and in the case where all forces are equal cost and $\lambda_{\mathrm{F}_i} \le \lambda_{\mathrm{F}}$,
\begin{equation}
    T_{\mathrm{grad-OEA}}\in\widetilde{\mathcal{O}}\Big(\epsilon_{\mathrm{grad-OEA}}^{-1}(T_{F}+T_R)\lambda_{\mathrm{F}} N_a^{\frac{3}{2}}+\sum_iT_{P,i}\Big).\label{eq:TgradOEA_equal_lambda}
\end{equation}

The cost of implementing the overlap estimation algorithm now depends on the circuit costs $T_F$, $T_R$, and $T_{P,i}$, and on the rescaling factor $\lambda_{\mathrm{F}}$.
Asymptotic costs for $T_F$ and $\lambda_{\mathrm{F}}$ were calculated in Sec.~\ref{sec:block_encodings}.
It remains to calculate $T_R$ and $T_{P,i}$ in order to obtain a final costing of the algorithm.

\subsubsection{Implementation of preparation and reflection unitaries}\label{sec:state_prep_and_reflection}

As part of the expectation value estimation algorithm, we need to implement a reflection $I-2|\psi\rangle\langle\psi|$ about the ground state $|\psi\rangle$.
In lieu of other access to $|\psi\rangle$, this reflection can be constructed by an inequality test on its energy $E$~\cite{Lin20Near}.
If we have some $\mu$ such that $E<\mu<E+\gamma$, where $\gamma$ is the gap between the ground and first excited state with $E$ denoting the ground state energy, and we assume the ground state is non-degenerate, we have that we can define an operator of the form
\begin{equation}
    \mathrm{sign}[H-\mu]:= \sum_j \mathrm{sign}(E_j-\mu)\ketbra{E_j}{E_j} = I-2|\psi\rangle\langle\psi|.
\end{equation}
One can approximate the sign function above by approximating it by as a sum of exponentials $\sum_tc_t\exp(iHt)$ and block encoding this using LCU techniques~\cite{Ge2017}, or by making a polynomial approximation and applying the quantum singular value transformation~\cite{Lin20Near}.
In this work we consider the latter.
Quantum signal processing requires a block encoding of $H-\mu$
\begin{equation}
    U_H=\left(\begin{array}{cc}\frac{1}{\lambda_{\mathrm{H}}}(H-\mu) & \cdot \\ \cdot & \cdot \end{array}\right),\label{eq:block_encoding_Hminusmu}
\end{equation}
where $\lambda_{\mathrm{H}}$ is an appropriate rescaling factor; this is at most $\mu$ larger than a corresponding block-encoding of $H$.
The technique also requires an approximation to the sign function that is within $\epsilon$ on the eigenvalues of $\frac{1}{\lambda_{\mathrm{H}}}(H-\mu)$.
Near $\mu$ this must break down for a finite polynomial series, so it is crucial that $\mu$ be chosen near the middle of the gap $(E,E+\gamma)$.
Known polynomial approximations to the sign function that achieve these requirements exist, with polynomial degree $d=\mathcal{O}(\frac{\lambda_{\mathrm{H}}}{\gamma}\log(\frac{1}{\epsilon}))$~\cite{Low2019-ls}.
Let us specify the smallest such constant by defining $c_{\mathrm{sgn}}$ such that
\begin{equation}
    d \le c_{\mathrm{sgn}}\lambda_{\mathrm{H}}\gamma^{-1}\log\big(\epsilon^{-1}\big)
\end{equation}
Given this sign function approximation and an exact block encoding of $H-
\mu$, we can construct a block encoding of $\mathrm{sign}[H-\mu]$ via quantum phase estimation with exactly $d$ repetitions of the block encoding (or its inverse), $d$ single-qubit $z$-rotations and $2d$ $N$-qubit Toffoli gates~\cite{Lin20Near,Gilyen2019QuantumArithmetics}.
For this to have an error less than $\epsilon_r$, we require the block encoding of $H$ to be implemented with error less than $\epsilon_r/d$.
Let the number of Toffoli gates needed to implement the block encoding of the $H$ to an error $\epsilon_{\mathrm{BE}}\geq\frac{\epsilon_r}{c_{\mathrm{sgn}}\gamma^{-1}\log(\epsilon^{-1})\lambda_{\mathrm{H}}}$ be $T_{H}$ (For many block encodings $T_H$ will be logarithmically dependent on $\epsilon_{\mathrm{BE}}$, and so the rescaling here will not affect any $\widetilde{O}$ scaling results). Further, assume that we are provided a phase gradient state of the form $\frac{1}{\sqrt{\lceil \frac{2\pi}{\delta_z} \rceil}}\sum_{j=0}^{\lceil \frac{2\pi}{\delta_z} \rceil-1} e^{2\pi i k/ \lceil \frac{2\pi}{\delta_z} \rceil} \ket{k}$ where $\delta_z$ is the required precision for the single-qubit $z$-rotations.  The number of Toffoli gates needed to perform these using controlled swaps on the above phase gradient state is
\begin{equation}
    T_R \in \widetilde{\mathcal{O}} \Big( \big(2N + \log(\delta_z^{-1}) + T_{H}\big)\lambda_{\mathrm{H}}c_{\mathrm{sgn}}\gamma^{-1}\log\big(\epsilon^{-1}\big)\label{eq:TR_form}\Big),
\end{equation}
Here we resort to $\widetilde{\mathcal{O}}(\cdot)$ notation to simplify the result given the doubly-logarithmic terms that arise from using the lower bound on $\epsilon_{\rm BE}$.  

We construct the reflection operator $1-2|\psi\rangle\langle\psi|$ to prepare the ground state $|\psi\rangle$ from an initial state $|\phi\rangle$ through the use of fixed-point amplitude amplification~\cite{Yoder14Fixed,Gilyen2019QuantumArithmetics}.
This requires we have access to a preparation unitary $U_{\phi}$ that prepares $|\phi\rangle$, and a bound $0<a_0\leq|\langle\phi|\psi\rangle|$.
The fixed point amplitude amplification algorithm then prepares an approximation to $|\psi\rangle$ with error $\epsilon_A$ with $\frac{\log(2/\epsilon_A)}{a_0}$ calls to $U_p$ and $1-2|\phi\rangle\langle\phi|$.
A call to the reflection $1-2|\phi\rangle\langle\phi|$ can be implemented by
\begin{equation}
    1-2|\phi\rangle\langle\phi| = U_\phi\big(1-2\ketbra{0}{0}\big) U_\phi^\dagger,
\end{equation}
which requires two queries to $U_\phi$ and a single application of a reflection about $|0\rangle$ which can be implemented using $2N$ Toffoli gates.
Therefore, we have
\begin{equation}
  \label{eq:initial_state_prep_detailed}
    T_{P,i=1}\in \widetilde{\mathcal{O}}\left(\frac{\log(1/\epsilon_A)}{a_0}(T_R+T_{\phi}+N)\right),
\end{equation}
where $T_{\phi}$ is the cost of implementing $U_{\phi}$.

This does not require an accurate estimate of $|\langle\phi|\psi\rangle|$, but just that $a_0 \le |\langle\phi|\psi\rangle|$.
In many cases a bound $a_0$ will be known that is relatively tight, in which case the scaling of this algorithm is optimal~\cite{Yoder14Fixed}.
When the bound is loose, it may be improved by performing initial rounds of amplitude amplification followed by direct estimation of $a_0$ on the device with little overhead~\cite{Brassard_2002}.

As we are repeating the overlap estimation algorithm multiple times on the same initial state, we can recycle the system register between each iteration, reducing the dependence of our asymptotic cost on the initial state overlap.
Let us define the action of our block encoding $U_i$ of $\frac{dH}{dR_i}$ on our ground state $|\psi\rangle$ as
\begin{equation}
    U_i|\psi\rangle = \tfrac{1}{\lambda_{\mathrm{F}_i}}\tfrac{dE}{dR_i}|\psi\rangle + \sqrt{1-\big[\tfrac{1}{\lambda_{\mathrm{F}_i}}\tfrac{dE}{dR_i}\big]^2}|\chi_i\rangle,\rightarrow |\chi_i\rangle = \frac{U_i-\tfrac{1}{\lambda_{\mathrm{F}_i}}\tfrac{dE}{dR_i}}{\sqrt{1-\big[\tfrac{1}{\lambda_{\mathrm{F}_i}}\tfrac{dE}{dR_i}\big]^2}}|\psi\rangle
\end{equation}
and then we have the well-known result that the Szegedy walk operator $\mathcal{S}_i$ in Eq.~\eqref{eq:Szegedy_def} is block-diagonal on the two-dimensional subspace spanned by $|\psi\rangle$ and $|\chi_i\rangle$, where it takes the form
\begin{equation}
    S_i=\exp\Big[iY\sin^{-1}(\theta_i)\Big],\;\;\;\theta_i = \tfrac{2}{\lambda_{\mathrm{F}_i}}\tfrac{dE}{dR_i}\sqrt{1-\big[\tfrac{1}{\lambda_{\mathrm{F}_i}}\tfrac{dE}{dR_i}\big]^2}.
\end{equation}
Regardless of the value of $\tfrac{1}{\lambda_{\mathrm{F}_i}}\tfrac{dE}{dR_i}$, the eigenstates of $S_i$ are simply $|s_i^{\pm}\rangle = \frac{1}{\sqrt{2}}(|\psi\rangle\pm i|\chi_i\rangle)$.
This implies that following each implementation of the overlap estimation algorithm, the overlap of the system register with the ground state is always $\langle s_i^{\pm}|\psi\rangle=\frac{1}{\sqrt{2}}$.
This probability can be raised to $1$ by a round of amplitude amplification, but this is not practical here because we do not have easy access to a reflection about $|s_i^{\pm}\rangle$.
This is especially true when $\tfrac{1}{\lambda_{\mathrm{F}_i}}\tfrac{dE}{dR_i}$ is very small (implying the gap between the eigenvalues of $S_i$ is quite small), which we expect to be the case.
Instead, one may measure the reflection operator $1-2|\psi\rangle\langle\psi|$ via a Hadamard test, which entails performing this operator conditional on a control qubit prepared in the $|+\rangle$ state, and measuring the control qubit in the $X$-basis.
In terms of Toffoli gates, this costs exactly $T_R$.
With $50\%$ probability, the control qubit is flipped to the $|-\rangle$ state and the system register is prepared in $|\psi\rangle$ (up to the error given by the reflection operator).
When the control qubit remains in the $|+\rangle$ state, we know that our system register is now in the state $|\chi_i\rangle$.
To avoid throwing away this state, we consider the operation of $U_i^{\dag}$ on $|\chi_i\rangle$.
If we define
\begin{equation}
    U_i^{\dag}|\psi\rangle = \tfrac{1}{\lambda_{\mathrm{F}_i}}\tfrac{dE}{dR_i}|\psi\rangle + \sqrt{1-\big[\tfrac{1}{\lambda_{\mathrm{F}_i}}\tfrac{dE}{dR_i}\big]^2}|\chi_i^{\dag}\rangle,\rightarrow |\chi_i^{\dag}\rangle = \frac{U_i^{\dag}-\tfrac{1}{\lambda_{\mathrm{F}_i}}\tfrac{dE}{dR_i}}{\sqrt{1-\big[\tfrac{1}{\lambda_{\mathrm{F}_i}}\tfrac{dE}{dR_i}\big]^2}}|\psi\rangle,
\end{equation}
we have
\begin{equation}
    U^{\dag}_i|\chi_i\rangle =\sqrt{1-\big[\tfrac{1}{\lambda_{\mathrm{F}_i}}\tfrac{dE}{dR_i}\big]^2}|\psi\rangle -\tfrac{1}{\lambda_{\mathrm{F}_i}}\tfrac{dE}{dR_i}|\chi_i^{\dag}\rangle,\;\;\;U_i|\chi^{\dag}_i\rangle =\sqrt{1-\big[\tfrac{1}{\lambda_{\mathrm{F}_i}}\tfrac{dE}{dR_i}\big]^2}|\psi\rangle -\tfrac{1}{\lambda_{\mathrm{F}_i}}\tfrac{dE}{dR_i}|\chi_i\rangle.
\end{equation}
Thus, if our initial measurement of the reflection operator fails, we may iterate performing either $U_i^{\dag}$ or $U_i$ to the system register, followed by re-measuring the reflection operator $1-2|\psi\rangle\langle\psi|$ till it reports a success.
The probability of this failing till the $n$th round (for $j>1$) and then succeeding is $p_n=\frac{1}{2}(1-[\frac{1}{\lambda_{\mathrm{F}_i}}\frac{dE}{dR_i}]^2)[\frac{1}{\lambda_{\mathrm{F}_i}}\frac{dE}{dR_i}]^{2(n-1)}$, and so the average Toffoli gate cost of the operation can be calculated to be
\begin{equation}
    T_{P,i>1}=\Bigg(\frac{1}{2}+\sum_{n=1}^{\infty}(n+1)p_n\Bigg)T_R=\Bigg(1+\frac{1}{2}\frac{1}{1-\big[\tfrac{1}{\lambda_{\mathrm{F}_i}}\tfrac{dE}{dR_i}\big]^2}\Bigg)T_R.
\end{equation}
This cost approaches its minimal value of $\frac{3}{2}T_R$ as $\frac{dE}{dR_i}\rightarrow 0$, and diverges as $\frac{dE}{dR_i}\rightarrow\lambda_{\mathrm{F}_i}$.
This is curious as it is the opposite behaviour that one finds when considering the cost of estimating expectation values in NISQ.
The divergence can be avoided by artificially increasing $\lambda_{\mathrm{F}_i}$, though we expect in practice that this will never be required.
Assuming $\frac{1}{\lambda}\frac{dE}{dR_i}\leq 0.5$, we have $T_{P,i>1}<2T_R$.
This yields a total preparation cost of
\begin{align}
    \sum_iT_{P,i}&
    \in \widetilde{\mathcal{O}}\left(\bigg(N_a+\frac{\log(1/\epsilon_A)}{a_0}\bigg)T_R+\frac{\log(1/\epsilon_A)}{a_0}(T_{\phi}+N)\right).\label{eq:TP_full}
\end{align}
The dependence on $a_0^{-1}$ in this equation is additive to the dependence on $N_a$.
For practical applications, we expect $T_R>>T_{\phi}+N$, and so the last term can be neglected.

\subsubsection{Total costing of the overlap estimation algorithm}
\label{sec:total_cost_overlap}
The full cost of the overlap estimation algorithm can be found by substituting Eq.~\eqref{eq:TP_full} into Eq.~\eqref{eq:TgradOEA_equal_lambda}.
This yields an expression in terms of the costs $T_R$ and $T_{F}$ to implement the reflection algorithm and block-encoded force operator respectively, as well as the rescaling factor $\lambda_{\mathrm{F}}$ for the block encoding of the force operator, and some other physical parameters.
In general the first term of Eq.~\eqref{eq:TP_full} will dominate the second (we assume $T_R\gg T_{\phi}+N$, so we drop it for simplicity, yielding
\begin{equation}
    T_{\mathrm{grad-OEA}}\in\widetilde{\mathcal{O}}\Big(\epsilon^{-1}_{\mathrm{grad-OEA}}\lambda_{\mathrm{F}}N_a^{3/2}T_{F}+\big[\epsilon^{-1}_{\mathrm{grad-OEA}}\lambda_{\mathrm{F}}N_a^{3/2}+2N_a+a_0^{-1}\big]T_R\Big).
\end{equation}
The last two terms in the square brackets come from the cost of state preparation.
We see that the middle term is completely dominated by the first term in the square brackets.
In practice, we expect the final term to be similarly dominated.
Then, in Eq.~\ref{eq:TR_form} we expect $T_H>>2N,\log(\delta_z^{-1})$.
Ignoring constants that do not scale with the system size, we can then write
\begin{equation}
  \label{eq:their_T_R}
    T_R\in\widetilde{\mathcal{O}}\Big(T_{H}\lambda_{H}\gamma^{-1}\Big),
\end{equation}
which yields a simplified asymptotic estimate
\begin{equation}
    T_{\mathrm{grad-OEA}}\in\widetilde{\mathcal{O}}\Big(\epsilon^{-1}_{\mathrm{grad-OEA}}\lambda_{\mathrm{F}}N_a^{3/2}\big[T_{F}+\lambda_{\mathrm{H}}\gamma^{-1} T_{H}\big]\Big).\label{eq:asymptotic_simplified_fh}
\end{equation}
As our methods for implementing the block-encoded derivative operator are either similar to or better than our methods for implementing the block-encoded Hamiltonian, we assume
\begin{equation}
    T_F\in\mathcal{O}(T_H).\label{eq:TFleqTH}
\end{equation}
This implies the second term in Eq.~\eqref{eq:asymptotic_simplified_fh} will be the dominant contribution to an asymptotic resource estimate due the additional multiplicative factors.
We tabulate the estimated scalings for the overlap estimation algorithm in Tab.~\ref{tab:HFscale}.

Let us now compare the asymptotic cost of gradient estimation via the overlap estimation algorithm to the result of the finite difference method, Eq.~\eqref{eq:gradcalc2} and Eq.~\eqref{eq:gradcalc3}, including terms assumed to be constant in Tab.~\ref{tab:FDscale} and Tab.~\ref{tab:HFscale}.
As we do not have tight lower bounds on either method, a direct universal comparison is not possible.
In place of this we compare the asymptotic upper bounds.
Given the assumption in Eq.~\eqref{eq:TFleqTH}, all terms in these bounds depend linearly on $\lambda_{\mathrm{H}}$ and $T_H$ so it is relevant to compare
\begin{equation}
    \tau_{\mathrm{OEA}}=\epsilon^{-1}\gamma^{-1}\lambda_{\mathrm{F}}
\end{equation}
against the three different potentially dominant terms in the finite difference method Eq.~\eqref{eq:gradcalc5}
\begin{equation}
    \tau_{\mathrm{FD,1}}=\epsilon^{-1}c\log^{5/2}(\mathfrak{e}),\hspace{0.5cm}\tau_{\mathrm{FD,2}}=\epsilon^{-1}N_a\gamma^{-1}\max_{\mathbf{R}',i}\left\|\tfrac{dH(\mathbf{R}')}{dR_i}\right\|,\hspace{0.5cm}\tau_{\mathrm{FD,3}}=\log(1/\epsilon)(\min_i|a_i|)^{-1}\gamma^{-1}\log^{5/2}(\mathfrak{e}).
\end{equation}
Note that all of the above terms have units of inverse energy (as we have dropped the universal dependence on $\lambda_{\mathrm{H}}$).
Eq.~\eqref{eq:gradcalc5} chooses between two different algorithms depending on which is faster, so we should compare $\tau_{\mathrm{OEA}}$ to the best-scaling algorithm.
The appropriate quantity to compare to is
\begin{equation}
    \max\Big(\tau_{\mathrm{FD},1},\min\big(\tau_{\mathrm{FD,2}},\tau_{\mathrm{FD,3}}\big)\Big)\label{eq:minmaxtaus}
\end{equation}
We have from Sec.~\ref{App_Grad_c_constant} that $\mathfrak{e}c\leq\frac{dE}{dR_i}\leq\max_{\mathbf{R}',i}\left\|\tfrac{dH(\mathbf{R}')}{dR_i}\right\|$.
Then, as we assume $\mathfrak{e}\sim\langle\psi_0|H|\psi_0\rangle$, we have $\mathfrak{e}>>\gamma$, which implies $c<<\max_{\mathbf{R}',i}\left\|\tfrac{dH(\mathbf{R}')}{dR_i}\right\|\gamma^{-1}$.
This in turn implies $\tau_{\mathrm{FD,2}}>> \tau_{\mathrm{FD,1}}$.
Let us now consider the comparison between $\tau_{\mathrm{OEA}}$ and these three terms in turn.

The comparison between $\tau_{\mathrm{OEA}}$ and $\tau_{\mathrm{FD,1}}$ is stark.
Here, the relative comparison is between $\gamma^{-1}\lambda_{\mathrm{F}}$ and $c\log(\mathfrak{e})$ respectively.
We expect $c\in o(1)$ (Sec.~\ref{App_Grad_c_constant}), and we have only a logarithmic dependence on $\mathfrak{e}$.
By contrast, $\lambda_{\mathrm{F}}$ can be quite large and scale badly in the system size (see Sec.~\ref{sec:block_encodings}, in particular Eqs.~\eqref{eq:lambdaF_THC},~\eqref{eq:lambdaF_low_rank} and~\eqref{eq:lambdaF_planewaves}).
The gap $\gamma$ to the first excited state varies in molecular and condensed matter systems, but it can also be quite small.
This implies that the best-scaling $\tau_{\mathrm{OEA}}$ (e.g. for double-factorized force operators in systems with large gaps) scale at best equally with the corresponding $\tau_{\mathrm{FD,1}}$.
This is a somewhat surprising result: $\tau_{\mathrm{FD,1}}$ contains the cost of performing phase estimation which is typically expected to dominate the finite difference algorithm, while $\tau_{\mathrm{FD,2}}$ and $\tau_{\mathrm{FD,3}}$ detail the cost of state recycling or state re-preparation.
The reason behind this is that the overlap estimation algorithm contains a sub-routine (the reflection operation) that uses exactly the same circuitry as phase estimation, and the OEA requires then repeating this sub-routine $\epsilon^{-1}$ times.
This multiplicative factor is interesting as for NISQ methods the cost of derivative estimation is completely independent of the Hamiltonian $1$-norm.

The comparison between $\tau_{\mathrm{FD},2}$ and $\tau_{\mathrm{OEA}}$ is more complex: here the rescaling factor $\lambda_{\mathrm{F}}$ competes with $N_a\max_{\mathbf{R}',i}\left\|\tfrac{dH(\mathbf{R}')}{dR_i}\right\|$.
For a fixed $\mathbf{R}$, the spectral norm $\|\frac{dH(\mathbf{R})}{dR_i}\|$ is a lower bound for $\lambda_{\mathrm{F}_i(\mathbf{R})}$ (being the rescaling factor for a block encoding of $\frac{dH}{dR_i}$ at point $\mathbf{R}$), though many block encodings are not this efficient.
For example, in Hydrogen chains we expect $\|\frac{dH(\mathbf{R})}{dR_i}\|$ to be roughly independent of the system size $N_a$.
This is matched by a double-factorized encoding, but a sparse block encoding has $\lambda_{\mathrm{F}}\sim N_a^2$ (Fig.~\ref{fig:lambdas_hydrogen_chain}).
On the other hand, the finite difference algorithm requires that we take the maximum $\|\frac{dH(\mathbf{R})}{dR_i}\|$ across the range of points used for finite difference, which significantly grows this cost.
Assuming that an efficient block encoding is chosen so that $\lambda_{\mathrm{F}}\sim\|\frac{dH(\mathbf{R})}{dR_i}\|$, we have $\tau_{\mathrm{FD},2}\lesssim N_a \tau_{\mathrm{OEA}}$.
This implies that when the cost of state re-preparation is bad enough that the recycling algorithm becomes preferred, the overlap estimation algorithm will likely become favourable again.

We finally compare $\tau_{\mathrm{FD,3}}$ to $\tau_{\mathrm{OEA}}$.
Ignoring logarithmic factors, the comparison here is between $\epsilon^{-1}\lambda_{\mathrm{F}}$ and $(\min_i|a_i|)^{-1}$.
This is highly dependent on the initial state preparation, of which only few studies have been made (e.g. \cite{Tubman2018Postponing}).
In lieu of exact statements that can be made here, we note that the above comparison suggests
\begin{equation}
    \tau_{\mathrm{FD,1}}<\tau_{\mathrm{OEA}}<\tau_{\mathrm{FD,2}}.
\end{equation}
Then, investigating Eq.~\eqref{eq:minmaxtaus}, whenever $\tau_{\mathrm{FD,3}}<\tau_{\mathrm{FD,1}}$ the finite difference algorithm will likely be asymptotically faster, while whenever $\tau_{\mathrm{FD,3}}>\tau_{\mathrm{FD,1}}$ the overlap estimation algorithm may have an advantage.
This can be summarized in the following statement: whenever repeat state preparation becomes the dominant cost for estimating energies at any point of the finite difference method, the overlap estimation algorithm likely becomes a more efficient algorithm.
This is because as the overlap estimation algorithm can efficiently recycle states with unit probability, while the finite difference method has to pay a significant cost in its choice of step size $dR$ to guarantee the same.
For applications in chemistry it is typical to assume that state preparation is not the dominant cost of the calculation (due to the relative size of the target error $\epsilon$ to the gap $\gamma$).
This suggests that for typical use-cases the finite difference method may indeed be an optimal choice.

\begin{table}[tb]
    \begin{tabular}{|c|c|c|}
    \hline
    System & Analytic Scaling & Empirical Scaling\\
    \hline
    First quantized plane waves & $\widetilde{\mathcal{O}}\left(\frac{N^{4/3}N_a^{23/6}(\mathbb{E}(|Z_i|))^{4/3} }
    {\epsilon}\right)$ & -- \\
    Second quantized plane waves & $\widetilde{\mathcal{O}}\left(\frac{N^{3} N_a^{5/2}  }{\epsilon}\right)$ & --\\
    \hline
    Hydrogen Chains (Sparse) & -- & $\widetilde{\mathcal{O}}\left(\frac{N_a^{4.11} }{\epsilon}\right)$ \\
    Hydrogen Chains (Double Factorization) & -- & $\widetilde{\mathcal{O}}\left(\frac{N_a^{4.50} }{\epsilon}\right)$\\
    Water Clusters (Sparse) & -- & $\widetilde{\mathcal{O}}\left(\frac{N_a^{4.30} }{\epsilon}\right)$ \\
    Water Clusters (Double Factorization) & -- & $\widetilde{\mathcal{O}}\left(\frac{ N_a^{4.74} }{\epsilon}\right)$\\
    \hline
    \end{tabular}
    \caption{Scaling of the overlap estimation algorithm for a system with $N_a$ atoms and $N$ orbitals and $\eta$ electrons within the Born-Oppenheimer approximation, assuming that $c_{\mathrm{sgn}},\gamma\in\Theta(1)$ and that the second term of Eq.~\eqref{eq:asymptotic_simplified_fh} dominates, and using an error $\epsilon$ as measured by the $2$-norm of the gradient vector. Scalings are calculated by multiplying the results of Tab.~\ref{tab:FDscale} by the relevant $\lambda_{\mathrm{F}}$ quantity.}
    \label{tab:HFscale}
\end{table}

\subsection{Simultaneous estimation of force operators using a gradient-based expectation value estimation}

Previously in this work we have considered using a quantum computer as an energy-estimation subroutine to evaluate forces via higher-order finite difference methods.
One might wonder whether a speedup can be gained from performing the difference estimation on the quantum device itself.
This was originally considered in \cite{Jordan_2005}, and later significantly improved by \cite{Gilyen_2017}, which proposed a quantum algorithm for evaluating the gradient of a function encoded as an expectation value.
However, in order to apply this algorithm to estimate the gradient, the coherent evaluation of the function at a superposition of points is required.
The cost of applying this approach directly to estimate the energy gradients as a function of nuclear co-ordinates appears prohibitive, especially for second-quantized Hamiltonians employing basis sets that are, themselves, dependent on the positions of the nuclei.
Recently however, in \cite{Huggins2021Nearly} some of the authors of this work proposed a strategy for measuring a collection of expectation values that applies the gradient algorithm of \cite{Gilyen_2017} to a simple and easy to implement auxiliary function
\begin{equation}
  f(\bm{\theta}) =  -\frac{1}{2}{\textrm{Im}}\Bigg[\bra{\psi}\prod_{j=1}^{M} e^{-2 i \theta_j O_j}\ket{\psi}\Bigg] + \frac{1}{2},\label{eq:new_gradient_function}
\end{equation}
where the \(\ket{\bm{\theta}}\) encode \(M\) binary numbers representing coordinates in a \(M\)-dimensional box centered on the origin.
This function can be encoded in a parameterized quantum circuit by performing a Hadamard test for the imaginary component of the unitary
\begin{equation}\label{eq:Utheta}
  U_{\bm{\theta}} = \sum_{\bm{\theta}} \Bigg[ \ketbra{\bm{\theta}}{\bm{\theta}} \otimes \prod_{j=1}^M e^{-2 i \theta_j O_j}\Bigg].
\end{equation}
Substituting $O_j:=\frac{1}{\lambda_{\mathrm{F}_j}}\frac{dH}{dR_j}$ into Eq.~\eqref{eq:new_gradient_function}, and assuming that $|\psi\rangle$ is the ground state of $H$ (at fixed $\mathbf{R}$) yields a function $f(\bm{\theta})$ that is different from the energy --- $f(\bm{\theta})\neq E(\mathbf{R}+\bm{\theta})$ --- but that satisfies
\begin{equation}
    \frac{df}{d\theta_i}\bigg|_{\bm{\theta}=0}=\frac{1}{\lambda_{F_i}}\frac{dE}{dR_i}\bigg|_{\mathbf{R}}.
\end{equation}
Crucially, implementing a quantumly-controlled version of this auxiliary function is more straightforward than a quantumly-controlled energy as a function of the nuclear coordinates, essentially because the auxiliary function is only required to reproduce the correct gradient at a single point.
Implementing this as an oracle in the gradient-based approach of \cite{Gilyen_2017} allows for multiple expectation values to be simultaneously estimated at the Heisenberg limit with some degree of parallelization.
In this section we will estimate the asymptotic Toffoli cost of implementing this oracle in the gradient based algorithm, and compare it to the methods developed in previous sections.

We begin by reviewing the essential features of the gradient-based expectation value estimation algorithm.
We do not present a complete review of the quantum algorithm for the gradient presented in \cite{Gilyen_2017}, nor the extension of this to estimate expectation values given in \cite{Huggins2021Nearly}.
Instead, we highlight those details that are necessary to describe the modifications required for our purposes and refer the reader to~\cite{Gilyen_2017} and \cite{Huggins2021Nearly} for a comprehensive description.
The gradient-based expectation value estimation algorithm targets the evaluation of the expectation values of \(M\) operators \(\{O_1, O_2, ..., O_M\}\) (with \(\left\|O_j\right\| \leq 1\) for all \(j\)) with respect to the state \(\ket{\psi}\) to within a precision \(\epsilon_{\textrm{grad}}\) in the infinity-norm.
It does this by invoking the gradient estimation algorithm of \cite{Gilyen_2017} as a subroutine, which requires constructing a quantum oracle of function $f$.
This is done so in terms of a probability oracle \(U_p\),
\begin{equation}
  \label{eq:probability_oracle}
  U_p: \ket{x}\ket{0} \rightarrow \ket{x} \otimes \bigg((\sqrt{f(x)} \ket{\psi_{\textrm{good}}(x)}\ket{1} + \sqrt{1 - f(x)}\ket{\psi_{\textrm{bad}}(x)}\ket{0}\bigg),
\end{equation}
where \(\ket{x}\) denotes a collection of registers encoding binary
approximations to the inputs to \(f\) and \(\ket{\psi_{\textrm{good}}(x)}\) and
\(\ket{\psi_{\textrm{bad}}(x)}\) are arbitrary normalized states.
In our case, $x=\bm{\theta}$, which requires a number of ancilla qubits scaling as $\mathcal{\tilde{O}}(M \log(\epsilon^{-1}_{\mathrm{grad}}))$ for storage.

One of the main results in \cite{Huggins2021Nearly} is the construction of a probability oracle $U_p$ for the function \(f\) using a single query to the state preparation unitary \(U_\psi\).
This probability oracle acts on the input register (the register containing the input state \(\ket{x}\) above), the system register (the Hilbert space containing \(\ket{\psi}\)), and one additional ancilla qubit.
The result of \cite{Gilyen_2017} allows us to solve our estimation using \(\mathcal{\widetilde{O}}(\sqrt{M}/\epsilon_{\textrm{grad}})\) calls to the resulting phase oracle $U_p$.
It was further shown in \cite{Huggins2021Nearly} that this may be achieved while upper-bounding the maximum time evolution required to implement a call to the probability oracle  ($\|\bm{\theta}\|_{\infty}$) in the following way,
$\|\bm{\theta}\|_{\infty} \leq t_{\max}\in\widetilde{\mathcal{O}}(\log(\epsilon^{-1}_{\mathrm{grad}})/\sqrt{M})$.
The failure probability of the algorithm of \cite{Gilyen_2017} is bounded only by \(1/3\), but this can be reduced to \(\delta\) using the Chernoff bound at a cost that is logarithmic in \(\delta\).

There are several details to address in order to optimally apply this gradient-based approach to the estimation of forces, and to estimate the asymptotic cost of the final result.
Firstly, it is desirable to take advantage of the fact that we can implement the reflection operator \(I - 2\ketbra{\psi}{\psi}\) more efficiently than \(U_{\psi}\) (by a factor $a_0$). 
We will modify the parallel expectation value algorithm to remove all but one of the calls to the state preparation unitary \(U_\psi\), replacing this component of the algorithm with a comparable number of reflections about the ground state.
This allows us to obtain a benefit similar to the one we found for the serial approach to overlap estimation presented in \sec{HeisenbergLimit} (see specifically \sec{state_prep_and_reflection} for the implementation of the reflection operator).
Secondly, we need to account for the cost of implementing the initial state preparation step and the controlled time-evolution by the force operators in terms of the number of calls to the appropriate block-encodings. 
Thirdly, we need to translate the error bounds from \cite{Huggins2021Nearly} to our case, accounting for the normalization of the block-encoded force operators, the various sources of error in the state preparation, reflection, and time evolution subroutines, and the fact that we desire a error in the 2-norm rather than the infinity norm.

\subsubsection{Replacement of the intermediate state preparations with ground state reflections}

We now aim to explain how we can replace all but one of the calls to the state preparation unitary $U_{\psi}$ and its conjugate by the more affordable reflection about the ground state, thereby avoiding most of the dependence on the overlap of the initial state with the ground state (\(a_0\)).
To do so we require some intermediary details of the gradient-based expectation value algorithm, which we recall from \cite{Gilyen_2017} and \cite{Huggins2021Nearly}.
The quantum algorithm for the gradient makes use of the probability oracle of Eq.~\eq{probability_oracle} by using the walk operator from amplitude amplification to convert the probability of a good state into a phase, 
\begin{equation}
    \label{eq:walk_operator}
    G_U = (2\Pi_1 - I)U_p^\dagger(2\Pi_2 - I) U_p,
\end{equation}
where \(2 \Pi_1 - I\) is a reflection about the zero state of the system register and the ancilla qubits of the probability oracle, and \((2 \Pi_2 - I)\) is a reflection about the zero state of the specific ancilla qubit that flags the good and bad states in the probability oracle.
The eigenvalues of $G_U$ are of the form $\exp(\pm i \sin^{-1}(\sqrt{f(x)}))$, where $f(x)$ is the probability of observing a good outcome as given by~\eqref{eq:probability_oracle}.

The gradient algorithm of \cite{Gilyen_2017} uses a quantum singular value transformation to convert from $G_U$ to a block encoding of the form $(\bra{0} \otimes I) {\rm QSVT}(G_U) (\ket{0} \otimes I) = C$ wherein $C$ (within some target error $\epsilon$) has eigenvalues $e^{-i r f(x)}$ within the eigenspace of $G_U$ that is in the support of $\ket{x}$.  The block encoded operation can be viewed as a fractional query to the phase oracle, and the block encoding procedure can be implemented using either Linear Combination of Unitaries methods or the more ancilla efficient methods of~\cite{Gilyen2019QuantumArithmetics} (which use controlled queries to $G_U$ interspersed with single qubit rotations).

By querying such fractional phase oracles at a superposition of carefully chosen points, the gradient algorithm uses phase kickback and the quantum Fourier transform to obtain a high-order finite difference approximation to the gradient.
For our purposes, the crucial detail is that the only interaction that the gradient algorithm has with the system register is through the action of (one of several)  walk operators (or their conjugates) that take the same form as the one from Eq.~\eq{walk_operator} (with some varying number of ancilla qubits or additional single-qubit gates acting on the ancilla).
In the rest of this section, we will show how one may shift all calls to the preparation unitary $U_{\psi}$ within a circuit-based implementation of $G_U$ to the edges of the unitary, such that $G_U$ takes the form
\begin{equation}
    G_U=U_{\psi}^{\dag}G'_UU_{\psi}.
\end{equation}
The reader may confirm that all other sub-routines $G_A$ used in the gradient estimation algorithm of \cite{Gilyen_2017} take the same form, making it possible to write
\begin{equation}
    G_A=U_{\psi}^{\dag}G'_AU_{\psi},
\end{equation}
where $G'_A$ may be implemented without calls to $U_{\psi}$.
This implies that we may concatenate any combination of these sub-routines and eliminate the product $U_{\psi}U_{\psi}^{\dag}$ that forms, i.e.
\begin{equation}
\prod_AG_A = U_{\psi}^{\dag}\prod_AG'_AU_{\psi}.
\end{equation}
This will reduce the total number of preparations required to $1$ (the final call to $U_{\psi}^{\dag}$ may be discarded as we do not use the system register afterwards).

The queries to $U_{\psi}$ used in each iteration of $G_U$ can be similarly removed using a basis transformation (at the price of requiring two queries to transform in and out of the basis of $\ket{\psi}$). Specifically, we consider the walk operator as a product of two components, \(2 \Pi_1 - I\) and \(U_p^\dagger (2 \Pi_2 - I) U_p\).
We may define the projector
\begin{equation}
\Pi_1' = U_\psi \Pi_1 U_\psi^\dagger,
\end{equation}
so that \(2 \Pi_1' - I\) is a reflection operator about the \(\ket{\psi}\) state of the system register together with the \(\ket{0}\) state of the ancilla qubits associated with \(\Pi_1\).
Then, the reader may confirm that the probability oracle $U_p$ constructed by \cite{Huggins2021Nearly} takes the form
\begin{equation}
  U_p = U_p' U_\psi,
\end{equation}
where \(U_p'\) contains no calls to $U_{\psi}$, but instead queries the force unitary $U_{\bm{\theta}}$. 
Now we can rewrite our expression for the reflections by adding the following resolutions of the identity to the two components of $G_U$:
\begin{align}
  2 \Pi_1 - I &= U_\psi^\dagger U_\psi \big( 2 \Pi_1 - I \big) U_\psi^\dagger U_\psi = U_\psi^\dagger \big( 2 \Pi_1' - I \big) U_\psi \\
  U_p^\dagger (2 \Pi_2 - I) U_p &= U_\psi^\dagger U_\psi U_p^\dagger (2 \Pi_2 - I) U_p U_\psi^\dagger U_\psi = U_\psi^\dagger  U_p^{'\dagger} (2 \Pi_2 - I) U_p' U_\psi.
\end{align}
This yields our new form of $G_U$,
\begin{equation}
    G_U=U_{\psi}^{\dag}(2\Pi'_1-I)U_p^{'\dag}(2\Pi_2-I)U'_pU_{\psi},
\end{equation}
and in the above notation,
\begin{equation}
    G'_U=(2\Pi'_1-I)U_p^{'\dag}(2\Pi_2-I)U_p',
\end{equation}
as required.
More generally, the reader may see by a careful inspection of \cite{Huggins2021Nearly} that the gradient-based expectation value estimation algorithm only interacts with the system register via the application of the components $2\Pi_1-I$ and $U_p^{\dag}(2\Pi_2-I)U_p$, interleaving them with other operators that act as the identity on the system register.
Therefore, if we consider the whole sequence of steps that comprise the algorithm, we can see that all of the intermediate state preparation unitary calls can be removed
by cancellation with their adjacent conjugates. 
Furthermore, the system register is discarded at the end of the expectation value algorithm,
so the final call to \(U_\psi^\dagger\) can also be removed.

Next with these pieces in place we can discuss the cost of the algorithm.  We take the cost here, as in previous sections, to be the total number of non-Clifford operations needed for each of the operations carried out.
Noting that the cost of the reflection $(2\Pi_2-I)$ is negligible (being a simple Pauli-$Z$ gate), we can calculate a total cost for the gradient-based phase estimation algorithm of
\begin{equation}
  T_{\textrm{simultaneous}} = \mathcal{\tilde{O}}(T_P + Q (T_R + T_{\textrm{time}})),\label{eq:total_time_cost_definition}
\end{equation}
where $Q$ denotes the number of repetitions of $G_U$ and any related unitary, and $T_P$, $T_R$, and $T_{\mathrm{time}}$ denote the cost of implementing $U_{\psi}$, $(2\Pi'_1-I)$, and $U'_p$ respectively; we will estimate these in the next sections.

\subsubsection{Asymptotic costs of reflection, force evolution, and preparation}

We now provide asymptotic upper bounds on the scaling of  $T_P$, $T_R$, $T_{\mathrm{time}}$ defined in the previous section.
The cost $T_R$ of reflecting about our ground state $|\psi\rangle$ depends on the precision to which we need to reflect since the a measurement of the groundspace is not generally known.
To calculate the cost of the reflection, let \(\epsilon_R\) be a desired bound on the total contribution to the error from implementing the reflection operators, and let \(R_{\textrm{prep}}\) denote the number of calls to the reflection operator used by the initial state preparation circuit.
We can achieved the desired overall bound on the error by implementing the \(Q/2 - 1 + R_{\textrm{prep}}\) reflection operators each to within a precision \(\frac{\epsilon_R}{Q/2 - 1 + R_{\textrm{prep}}}\) in the spectral norm.
The Toffoli count for block encoding the Hamiltonian (required by the construction of the reflection operator) and implementing the reflections both scale logarithmically with the desired precision.
Assuming that \(Q\) and \(R_{\textrm{prep}}\) both scale at most polynomially with the system size and any other relevant parameters of the problem (we shall later determine that they satisfy this assumption), we can implement each reflection about \(\ket{\psi}\) for a cost \(T_R\) which scales identically (up to logarithmic factors) to the case of serial amplitude estimation (see Eq.~\eq{their_T_R}):
\begin{equation}
\label{eq:our_T_R}
T_R\in\widetilde{\mathcal{O}}\Big(T_{H}\lambda_{H}\gamma^{-1}\Big),
\end{equation}
where $\gamma$ is a lower bound on the eigenvalue gap between the groundstate and the rest of the spectrum.

We now address the cost of the initial state preparation $T_P$.
We can consider the cost for this step given in Eq.~\eq{initial_state_prep_detailed} of \sec{state_prep_and_reflection}, where it was addressed in the context of serial expectation value estimation with the overlap estimation algorithm. 
We can also make the same simplifying assumption later used in \sec{total_cost_overlap}, namely that \(T_R \gg T_\phi + N\).  Here $T_{\phi}$ is the Toffoli-count for implementing the state preparation unitary $U_\phi$, and will often be negligible in cases where an elementary ansatz such as a Hartree-Fock state is used for state preparation.

The contribution to the overall error from imperfectly implementing the reflection operators is already accounted for in the section above, provided that the initial state preparation step requires a number of calls to the reflection operator \(I - \ketbra{\psi}{\psi}\) that grows at most polynomially with the problem parameters. 
This conditions hold, as discussed in \sec{state_prep_and_reflection}.
The only remaining source of error is the error from amplitude estimation, which only appears logarithmically in the cost of the state preparation step.
Making these simplifications and substitutions, and neglecting all logarithmic factors, we find that (similarly to the initial preparation in Sec.~\ref{sec:state_prep_and_reflection}), the initial state preparation step requires a number of Toffoli gates scaling as
\begin{equation}
  \label{eq:our_T_P}
  T_P = \mathcal{\tilde{O}}(a_0^{-1}T_R).
\end{equation}

The cost $T_{\mathrm{time}}$ of implementing $U_p$ is dominated by the cost of implementing the subroutine $U_{\bm{\theta}}$ (Eq.~(\ref{eq:Utheta})).
Interestingly, it suffices here to take a naive implementation of $U_{\bm{\theta}}$ to demonstrate a speedup; i.e. we aim to implement the $e^{2i\theta_iO_i}$ in series for each $O_j=\frac{1}{\lambda_{F_i}}\frac{dH}{dR_i}$.
As in \sec{block_encodings}, for the \(i\)th force operator, we let \(T_{F,i}\) denote the cost of implementing the circuit for the block encoding and we let \(\lambda_{\mathrm{F}_i}\) denote the associated rescaling factor. 
For simplicity, we actually consider the case where the operators are all block-encoded with the same rescaling factor \(\lambda_{\mathrm{F}}\), where \(\lambda_{\mathrm{F}_i} \leq \lambda_{\mathrm{F}}\) for all \(i\).
These block encodings are also block encodings of the operators \(\frac{1}{\lambda_{\mathrm{F}}} \frac{d H}{d R_i}\), and we have \(\left\|\frac{1}{\lambda_{\mathrm{F}}} \frac{d H}{d R_i}\right\| \leq 1\) for all \(i\).
As discussed in \cite{Huggins2021Nearly}, for each of the \(Q\) calls to \(U_p'\), we perform time evolution by each of the observables for \(\mathcal{O}(\log(\epsilon_{\textrm{grad}}^{-1}))\) different lengths of time, with a maximum length of time for any such evolution \(\|\bm{\theta}\|_{\infty}\leq t_{\max} = \mathcal{\tilde{O}}(\log(\epsilon_{\textrm{grad}}^{-1})/\sqrt{3 N_a})\).
We can use optimal Hamiltonian simulation algorithms to accomplish each of these evolutions to within an accuracy \(\epsilon'\) using \(\mathcal{O}(t_{\max} + \log(1/\epsilon'))\) calls to the block encoding for the appropriate force operator~\cite{Low2019-ls, Gilyen2019QuantumArithmetics}.
Let \(T_F\) be an upper bound on the cost of implementing any of the block-encoded force operators to within a precision \(\epsilon''\).

Now we wish to determine the asymptotic scaling of the Toffoli count for performing the time evolution steps while contributing an error of at most \(\epsilon_{\textrm{time}}\) to the final output of the gradient-based
estimation algorithm.
The complexity of optimal Hamiltonian simulation depends logarithmically on \(\epsilon'\) and we make the further assumption that the complexity of implementing the block-encoded force operator depends logarithmically on \(\epsilon''\).
The total number of time evolution steps performed, and the total number of calls to the block encoding for each force operator, both scale linearly with \(Q N_a\) if we neglect logarithmic factors.
Therefore, if we wish to bound the total error that arises from implementing time evolution by the force operators to \(\epsilon_{\textrm{time}}\) it suffices to set \(\epsilon'\) and \(\epsilon''\) polynomially smaller than
\(\epsilon_{\textrm{time}}\).
When we neglect the logarithmic factors, we then see that we can implement all of the time evolution steps required by a single query to \(U_p'\) using a number of Toffoli gates which that as
\begin{equation}
  \label{eq:T_time}
  T_{\textrm{time}} = \mathcal{\tilde{O}}(\sqrt{N_a}T_F).
\end{equation}

\subsubsection{Total cost for gradient-based estimation}

Now we are ready to determine the overall cost for the force estimation using the gradient-based approach to measurement.
We aim to measure the \(3 N_a\) force operators to within a precision
\(\epsilon\) in the 2-norm. 
We can achieve this by setting our desired precision in the infinity norm to
\(\epsilon/\sqrt{3N_a}\). 
For simplicity, we assume that each of the force operators is block-encoded with
a rescaling factor \(\lambda_{\mathrm{F}_i} \leq \lambda_{\mathrm{F}}\).
Then we satisfy the normalization conditions of the gradient-based estimation
algorithm by considering the task of estimating of the \(3 N_a\) operators
\(O_i=\frac{1}{\lambda_{\mathrm{F}}} \frac{d H}{d R_i}\) to within a precision
\(\frac{\epsilon}{\lambda_{\mathrm{F}} \sqrt{3 N_a}}\). 
We need to consider four different sources of error. 
Let \(\epsilon_{\textrm{grad}}\) denote the error from the gradient-based
estimation algorithm itself. 
Let \(\epsilon_R\) denote the error from the reflection operator about the
ground state that we shall use as a subroutine. 
Let \(\epsilon_{P}\) denote the error from the initial ground state
preparations step. 
Finally, let \(\epsilon_{\textrm{time}}\) denote the total error contributed by our
implementations of the various time evolution operators for each of the forces. 
In order to achieve our desired error in the 2-norm, we require that each of
these components of the overall error be bounded above by \(\frac{\epsilon}{4\lambda_{\mathrm{F}} \sqrt{3 N_a}}\).
(The cost of $\epsilon_{R},\epsilon_{P},$ and $\epsilon_{\mathrm{time}}$ are already accounted for in our choice of $T_R, T_P,$ and $T_{\mathrm{time}}$ respectively.)
Taking the above bound \(\epsilon_{\textrm{grad}}\leq\frac{\epsilon}{4\lambda_{\mathrm{F}} \sqrt{3 N_a}}\), and substituting this into the above analysis, we find that we need 
\begin{equation}
  \label{eq:Q_count}
  Q = \mathcal{\tilde{O}}(N_a \lambda_{\mathrm{F}} / \epsilon)
\end{equation}
queries to the probability oracle \(U_p'\).

The gradient-based estimation algorithm requires a number of Toffoli gates which
is dominated by the cost of the initial state preparation, the \(Q\) queries to
\(U_p'\), and the associated reflection operators.
Subsituting Eqs.~\eq{Q_count},~\eq{our_T_R},~\eq{our_T_P}, and~\eq{T_time} into Eq.~\eqref{eq:total_time_cost_definition}, we find that
\begin{equation}
  T_{\textrm{simultaneous}} = \mathcal{\tilde{O}}(a_0^{-1}T_H\lambda_{\mathrm{H}}\gamma^{-1} + N_a \lambda_{\mathrm{F}} \epsilon^{-1}(T_H\lambda_{\mathrm{H}}\gamma^{-1} + \sqrt{N_a}T_F)).
\end{equation}
As in the case of serial expectation value estimation, we expect that the cost
of the initial state preparation will be negligible (as it does not depend on
\(\epsilon^{-1}\)) in practice and we can therefore drop the first term. We also
assume that \(T_F = \mathcal{O}(T_H)\), i.e., that the cost of implementing the
block-encoded derivative operators scales no worse than the cost of implementing
the block-encoded Hamiltonian. Making these simplifications, we have
\begin{equation}
  T_{\textrm{simultaneous}} = \mathcal{\tilde{O}}(T_H N_a \lambda_{\mathrm{F}} \epsilon^{-1} (\lambda_{\mathrm{H}}\gamma^{-1} + \sqrt{N_a})).
\end{equation}
The norm of the Hamiltonian should scale at least linearly with each added
nuclei, therefore we expect that \(N_a = \mathcal{O}(\lambda_{\mathrm{H}})\). We can use
this to perform a final simplification, yielding
\begin{equation}
  T_{\textrm{simultaneous}} = \mathcal{\tilde{O}}(T_H \lambda_{\mathrm{H}} \gamma^{-1} N_a \lambda_{\mathrm{F}} \epsilon^{-1}).\label{eq:GBEVEA_final_scaling}
\end{equation}
This improves strictly over the bound for the overlap estimation algorithm (Eq.~(\ref{eq:asymptotic_simplified_fh})) by an asymptotic factor of $N_a^{1/2}$.
This can be understood as a direct square root speedup from the parallelization of the estimation of multiple gradients.
It was proven in \cite{Huggins2021Nearly} that, for certain cases, gradient-based estimation achieves the optimal (up to logarithmic factors) scaling in the number of calls to the state preparation unitary $U_{\psi}$. This suggests that further improvements in our case may be difficult and would necessarily come from exploiting additional structure not present in the general case.
The comparison to the finite difference method (Sec.~\ref{sec:Num_diff_FD}) is slightly more complex: following the comparison made in Sec.~\ref{sec:total_cost_overlap}, the relevant comparison is between $\tau_{\mathrm{FD},1}=\epsilon^{-1}c\log^{5/2}(\mathfrak{e})$ and
\begin{equation}
    \tau_{\mathrm{GBEV}}:= N_a^{-1/2}\epsilon^{-1}\gamma^{-1}\lambda_F = \tau_{\mathrm{OEA}}N_a^{-1/2}.
\end{equation}
It was argued in Sec.~\ref{sec:total_cost_overlap} that there exist systems for which $\gamma^{-1}$ and $\lambda_F$ are both constant (or scale logarithmically) in the system size (the latter based on the numerical analysis in Sec.~\ref{sec:block_encodings}).
In this case, the gradient-based expectation value estimation algorithm would be asymptotically faster by a factor $N_a^{1/2}$, even assuming the higher order gradient bound $c$ is constant.
However, for systems with a particularly small gap, or for which the block encoding of the force operator is large, the finite difference algorithm may be optimal.
We give asymptotic scalings of the systems that we studied in this work in Tab.~\ref{tab:GBEVEAscale}, following the analysis used to generate Sec.~\ref{tab:FDscale}.

\begin{table}[tb]
    \begin{tabular}{|c|c|c|}
    \hline
    System & Analytic Scaling & Empirical Scaling\\
    \hline
    First quantized plane waves & $\widetilde{\mathcal{O}}\left(\frac{N^{4/3}N_a^{10/3}(\mathbb{E}(|Z_i|))^{4/3} }
    {\epsilon}\right)$ & -- \\
    Second quantized plane waves & $\widetilde{\mathcal{O}}\left(\frac{N^{3} N_a^{2}  }{\epsilon}\right)$ & --\\
    \hline
    Hydrogen Chains (Sparse) & -- & $\widetilde{\mathcal{O}}\left(\frac{N_a^{3.61} }{\epsilon}\right)$ \\
    Hydrogen Chains (Double Factorization) & -- & $\widetilde{\mathcal{O}}\left(\frac{N_a^{4.00} }{\epsilon}\right)$\\
    Water Clusters (Sparse) & -- & $\widetilde{\mathcal{O}}\left(\frac{N_a^{3.80} }{\epsilon}\right)$ \\
    Water Clusters (Double Factorization) & -- & $\widetilde{\mathcal{O}}\left(\frac{ N_a^{4.24} }{\epsilon}\right)$\\
    \hline
    \end{tabular}
    \caption{Scaling of the gradient-based expectation value algorithm for a system with $N_a$ atoms and $N$ orbitals and $\eta$ electrons within the Born-Oppenheimer approximation assuming the results of Eq.~(\ref{eq:GBEVEA_final_scaling}), and using an error $\epsilon$ as measured by the $2$-norm of the gradient vector. Scalings are calculated by multiplying the results of Tab.~\ref{tab:FDscale} by the relevant $\lambda_{\mathrm{F}}N_a^{-1/2}$ quantity. We assume $\gamma \in \Theta(1)$.}
    \label{tab:GBEVEAscale}
\end{table}

\section{Conclusion}
\label{sec:conclusion}

In this work we performed an in-depth study of the cost of estimating forces using current NISQ and future fault-tolerant devices.
We optimized strategies for estimating forces on a NISQ device, and determined that for some state-of-the-art methods, estimating the entire force vector to a fixed error (as measured by the $2$-norm of the error vector) may be roughly equal or a lower cost than estimating the energy of the system.
This suggests that force vectors would effectively come `for free' in a NISQ electronic structure calculation.
We note however that as NISQ methods tend to be variational, which requires calculating the electronic energy repeatedly to minimize it with respect to variational parameters, it does not immediately follow that force estimation is easier than energy estimation for a system.

We further present, optimize and cost several algorithms for calculating forces within a fault-tolerant quantum computing framework.  Specifically, we propose a method based on finite difference algorithms for computing gradients as well as another method based on estimating expectation values of force operators.  We find that in situations where the expected cost of state preparation is low, the finite difference algorithms provide asymptotically better scaling.  In contrast, the gradient operator approach provides better scaling if the cost of state preparation is high.  In the best case scenario, the scaling of these algorithms for first quantized plane waves are a factor of $N_a^{3/2}$ worse than the cost of the best first quantized simulations.

The cost estimates for block-encoding forces for atomic-centered basis orbitals in second quantization is more complex because it requires a numerical study.  We find that in the examples considered, block encodings of the gradient operators does not translate into a relative reduction in cost for estimating forces compared to estimating energies for the electronic structure problem.
This is because all methods studied are dominated by the cost of Hamiltonian simulation, either directly or as part of a reflection around the ground state.
As such a reflection is a crucial part of targeting the ground state of an electronic structure problem, we do not expect this problem to be easily circumvented for direct force estimation on a quantum device.

Although we estimate that the cost of reaching chemical accuracy for forces ($6.4$~mHa/Å error per force component) is comparable to that of energies ($1.6$~mHa), the bottleneck of practical MD simulations is that they typically require on the order of $10^6-10^9$ unique force calculations \cite{hollingsworth2018molecular}.
Wall-clock time estimates for single point energy calculations on fault-tolerant devices in the beyond-classical regime are typically on the order of multiple hours; repeating this millions of times to perform a MD simulation is not practical.
This suggests that to enhance MD simulations with quantum computers, new approaches need to be found.
However, our methods look promising for applications that do not require so many repeated force or derivative calculations, e.g.~geometry optimization and dipole estimation appear more practical following this work.

We suggest some directions that we believe may be fruitful for future research.
A clear target for further research in force optimization is to improve the bounds on the variance that we obtained in the NISQ section (e.g. by implementing locally-biased classical shadow techniques~\cite{Bravyi19Classical}), though we do not believe that this will be a sufficient improvement to achieve beyond-classical molecular dynamics simulations in the NISQ era.
It is also of critical importance to study how generalizable error tolerances are for forces; it is unclear whether our result of $6.4$~mHa/Å applies to systems other than water.
In a similar vein, it would be of interest to determine the error tolerances required for other first-order derivative quantities, for example see Table~\ref{fig:table_properties}; our methods extend immediately to these systems, and easier targets than molecular dynamics are likely to exist.

Moreover, instead of performing a semi-classical simulation where one extracts force information from a quantum device to perform a molecular dynamical step, it should be possible to perform a fully-quantum simulation -  which is performed entirely on the device.  This is relevant because one of the major drivers of complexity is the $\tilde{\mathcal{O}}(1/\epsilon)$ cost of computing the gradient.  A more practical approach would entail either simulating the differential equation governing the dynamics on the quantum computer itself (while retaining a classical description of the nuclear motion within the Born-Oppenheimer approximation), or by a fully-quantum simulation that treats the nuclei and electrons on a similar footing such as a first quantized simulation.
Alternatively, one can consider taking a small number of measurements from a quantum device and fitting a semi-empirical model to the results.
This could be done using a semi-classical molecular mechanical model, for example machine-learning techniques that take a few high-accuracy and many low-accuracy data points to generate high-accuracy force-fields~\cite{Bogojeski20Quantum}.
Equivalently, one could consider reusing points from finite difference calculations, or sparsely sampling gradients to reduce these constant factor overheads overhead.

The larger point from this work is that it has identified a new problem within chemistry simulation that needs further research.  In this sense, our work should be seen similar to early work in simulating groundstate energies, such as~\cite{Wecker2014,Reiher2017}.  These earlier works may not have yielded results that are as practical as later methods such as~\cite{Berry19Qubitization,lee2021even} but they identified the existence of an important problem that needs to be optimized to realize the full potential of quantum computation within that domain.  From that perspective, this work represents a first step towards the discovery of the first practical methods for simulating MD on quantum computers.

\begin{center} \textbf{Acknowledgements} \end{center}

The authors thank Jarrod McClean, Jordi Tura, Kianna Wan, and Clemens Utschig-Utschig for helpful discussions. NW worked on this project under a research grant from Google Quantum AI and was partially supported by the ``Embedding Quantum Computing into Many-body Frameworks for Strongly Correlated  Molecular and Materials Systems'' project, which is funded by the U.S. Department of Energy, Office of Science, Office of Basic Energy Sciences (BES), the Division of Chemical Sciences, Geosciences, and Biosciences. DWB worked on this project under a sponsored research agreement with Google Quantum AI. DWB is also supported by Australian Research Council Discovery Projects DP190102633 and DP210101367.

\bibliography{force,babbush_mendeley}

%%%%%%%%%%%%%%%%%%%%%%%%%%%%%%%%%%%%%%%%%%%%%%%%%%%%%%%%%%%%%%%%%%%%%%%%%%%%%%%%%%%

\appendix

\section{Dependence of the molecular orbital coefficients on the atomic orbital overlap coefficients}
\label{app:mo_ao}

Because only the MO coefficients $\mathbf{C}$ carry the dependence on $\mathbf{S}$, to derive explicit expressions for these terms, we will need to derive expressions for $\partial \mathbf{C} / \partial \mathbf{S}$. Herein, we will assume that the MO coefficients are real-valued. We begin by writing a general parameterization\cite{Maurice1998-ju} of the molecular orbitals 
\begin{equation}
    \mathbf{C}\left(\mathbf{C}^0, \mathbf{S}, \mathbf{\Theta}\right) = \mathbf{C}^0 \mathbf{M}(\mathbf{S}) \mathbf{U}(\mathbf{\Theta})
\end{equation}
where $\mathbf{C}^0$ are the (unperturbed) reference orbitals---here taken to be from the solution of the Hartree-Fock equations, $\mathbf{M}$ ensures the orthogonalization of the metric $\mathbf{S}$, and the Hartree-Fock energy is minimized by the unitary matrix $\mathbf{U}$, which is a function of orbital rotation parameters $\mathbf{\Theta}$. Because $\mathbf{U}$ is unitary, it may take the form

\begin{equation}
    \mathbf{U}(\mathbf{\Theta}) = \exp{[\mathbf{T}(\mathbf{\Theta}})] 
\end{equation}

provided that $\mathbf{T}$ is anti-symmetric. The elements of $\mathbf{T}(\mathbf{\Theta})$ are given by

\begin{equation}
    \mathbf{T}(\mathbf{\Theta})_{mn} = \Theta_{ia}^{*} \delta_{mi} \delta_{na} - \Theta_{ai} \delta_{ma} \delta_{ni}
\end{equation}

which, again, is anti-symmetric and therefore ensures unitarity of $\mathbf{U}$. Therefore, we can write the expression for $\mathbf{U}$ more clearly as

\begin{equation}
   \mathbf{U}(\mathbf{\Theta}) = \exp{\begin{bmatrix}\mathbf{\quad 0}&\mathbf{\Theta}^{\dagger}\\\mathbf{-\Theta}&  \mathbf{0}\end{bmatrix}}
\end{equation}

$\mathbf{\Theta}$ is a ($\mathrm{virt} \times \mathrm{occ}$)-dimensional matrix where the elements $\Theta_{ai}$ are the orbital rotation angles that minimize the Hartree-Fock energy. When canonical Hartree-Fock orbitals are used as the reference $\mathbf{C}^0$, the rotation angles $\Theta_{ai}$ are zero---and $\mathbf{U}(\mathbf{\Theta})$ is the identity matrix---by virtue of the variational optimization of the orbitals.

An expression for $\mathbf{M}$ can be determined from the Hartree-Fock orthogonality condition

\begin{equation}
    \mathbf{C}^{\dagger} \mathbf{S} \mathbf{C} = \mathbf{U}^{\dagger} \mathbf{M}^{\dagger} \mathbf{C}^{0\dagger} \mathbf{S} \mathbf{C}^{0} \mathbf{M} \mathbf{U} = \mathbf{I}
\end{equation}

$\mathbf{U}$ is unitary and may be removed by left and right multiplication by $\mathbf{U}$ and $\mathbf{U}^{\dagger}$, respectively. Requiring that $\mathbf{M}$ is both symmetric and invertible allows us to end up with the expression

\begin{equation}
    \mathbf{M} = \left(\mathbf{C}^{0\dagger} \mathbf{S} \mathbf{C}^{0}\right)^{-1/2}
\end{equation}

which is by no means a unique solution for $\mathbf{M}$. Putting the previous results together, the final parameterization of the molecular orbitals yields the expression

\begin{equation}
    \mathbf{C} = \mathbf{C}^{0} \left(\mathbf{C}^{0\dagger} \mathbf{S} \mathbf{C}^{0}\right)^{-1/2} \exp{\begin{bmatrix}\mathbf{\quad 0}&\mathbf{\Theta}^{\dagger}\\\mathbf{-\Theta}&  \mathbf{0}\end{bmatrix}}
\end{equation}

where the elements $C_{\mu p}$ are given as 

\begin{equation}
    C_{\mu p} = \sum_{qr} C^{0}_{\mu q} \left[ \sum_{\lambda \sigma} C^{0}_{\lambda q} S_{\lambda \sigma} C^{0}_{\sigma r}\right]^{-1/2} U_{rp}
\end{equation}

For a reference set of orbitals, $\mathbf{C}^{0}$, the parameterization of the orbitals depends on $\mathbf{S}$ and $\mathbf{\Theta}$. We will want to consider partial derivatives of the orbitals with respect to these quantities. First we will consider derivatives of the orbitals with respect to the atomic orbital overlap matrix $\mathbf{S}$

\begin{equation}
    \frac{\partial C_{\mu p}}{\partial S_{\lambda \sigma}} = -\frac{1}{2} \sum_{q} C^{0}_{\mu q} C^{0}_{\lambda r} C^{0}_{\sigma r} \left[ \sum_{\lambda \sigma} C^{0}_{\lambda q} S_{\lambda \sigma} C^{0}_{\sigma r}\right]^{-3/2} U_{rp}
\end{equation}

Without loss of generality, we can consider the case where the MO coefficients are small perturbations of the canonical Hartree-Fock orbitals. Therefore we consider the limit where $\mathbf{C}^{0} \to \mathbf{C}$ and $\mathbf{\Theta} \to \mathbf{0}$, such that $\mathbf{U} \to \mathbf{I}$. This allows us to set $\mathbf{M} = \mathbf{I}$ and $\mathbf{U} = \mathbf{I}$, respectively, when evaluating expressions after differentiation. Therefore, the above expression simplifies to

\begin{align}
    \frac{\partial C_{\mu p}}{\partial S_{\lambda \sigma}}\biggr\rvert_{\mathbf{C}^{0} = \mathbf{C}, \mathbf{\Theta} = \mathbf{0}}  &= -\frac{1}{2} \sum_{q} C_{\mu q} C_{\lambda r} C_{\sigma r} \delta_{qr} \delta_{rp} \\
    & = -\frac{1}{2} \sum_{q} C_{\mu q} C_{\lambda q} C_{\sigma p}
\end{align}
such that
\begin{equation}
    \frac{\partial C_{\mu p}}{\partial S_{\lambda \sigma}} \frac{dS_{\lambda \sigma}}{dx} = -\frac{1}{2} \sum_{q} C_{\mu q} C_{\lambda q} C_{\sigma p} \frac{dS_{\lambda \sigma}}{dx}
\end{equation}

\iffalse
For derivatives of the MO coefficients with respect to orbital rotations, we have 

\begin{equation}
 \frac{\partial C_{\mu p}}{\partial \Theta_{a i}} = \sum_{qr} C^{0}_{\mu q} \left[ \sum_{\lambda \sigma} C^{0}_{\lambda q} S_{\lambda \sigma} C^{0}_{\sigma r}\right]^{-1/2} U_{rp} \left(\delta_{ri} \delta_{pa} - \delta_{ra} \delta_{pi}\right)
\end{equation}

\begin{equation}
    \frac{\partial C_{\mu p}}{\partial \Theta_{a i}}\biggr\rvert_{\mathbf{C}^{0} = \mathbf{C}, \mathbf{\Theta} = \mathbf{0}} = C_{\mu i} \delta_{a p} -  C_{\mu a} \delta_{i p}
\end{equation}

such that

\begin{equation}
    \frac{\partial C_{\mu p}}{\partial \Theta_{a i}}\Theta_{ai}^{[x]} = \left(C_{\mu i} \delta_{a p} -  C_{\mu a} \delta_{i p}\right)\Theta_{ai}^{[x]}
\end{equation}
\fi

with these, we are ready for handling general derivatives of \emph{ab initio} energies involving the overlap, as in Eq.~\eqref{eq:energy_weighted_dm}.

For the one body terms:
\begin{equation}
\begin{alignedat}{1}
  \sum_{\lambda \sigma} \frac{\partial}{\partial S_{\lambda \sigma}} \left[\sum_{qm} \gamma_{qm} h_{qm} \right] \frac{dS_{\lambda \sigma}}{dx} &= \sum_{\lambda \sigma} \frac{\partial}{\partial S_{\lambda \sigma}} \left[ \sum_{qm} \gamma_{qm} \sum_{\mu\nu} C_{\mu q} C_{\nu m} h_{\mu \nu} \right] \frac{dS_{\lambda \sigma}}{dx} \\ 
  &= \sum_{qm} \gamma_{qm} \sum_{\mu \nu \lambda \sigma} \left[ \frac{\partial C_{\mu q}}{\partial S_{\lambda \sigma}} C_{\nu m} h_{\mu\nu} + C_{\mu q} \frac{\partial C_{\nu m}}{\partial S_{\lambda \sigma}} h_{\mu \nu}\right]\frac{dS_{\lambda \sigma}}{dx} \\
  &= \sum_{qm} \gamma_{qm} \sum_{\mu \nu \lambda \sigma} \left[-\frac{1}{2} \sum_{p} C_{\mu p} C_{\lambda p} C_{\sigma q} C_{\nu m} h_{\mu\nu} -\frac{1}{2} \sum_{p} C_{\mu q} C_{\nu p} C_{\lambda p} C_{\sigma m}  h_{\mu\nu}\right]\frac{dS_{\lambda \sigma}}{dx} \\
  &= -\frac{1}{2} \sum_{pqm} \gamma_{qm} h_{pm} \frac{dS_{pq}}{dx} -\frac{1}{2} \sum_{pqm} \gamma_{qm} h_{qp} \frac{dS_{pm}}{dx} \\
  &= -\sum_{pqm} \gamma_{qm} h_{pm} \frac{dS_{pq}}{dx} %\\
\end{alignedat}
\end{equation}
The two body terms are slightly more complicated. For the two body terms, we can rearrange to isolate the terms depending on the overlap
\begin{equation}  \label{eq:dtwodoverlap}
    \begin{alignedat}{1}
    \sum_{\alpha\beta} \frac{\partial}{\partial S_{\alpha \beta}} \left[ \sum_{pqrs} \Gamma_{pqrs} g_{pqrs} \right] \frac{dS_{\alpha \beta}}{dx} &= \sum_{\alpha\beta} \sum_{pqrs} \Gamma_{pqrs} \left[ \frac{\partial}{\partial S_{\alpha \beta}} \sum_{\mu\nu\lambda\sigma} C_{\mu p} C_{\nu q} C_{\lambda r} C_{\sigma s} g_{\mu \nu  \lambda \sigma} \right] \frac{dS_{\alpha \beta}}{dx} 
    \end{alignedat}
\end{equation}
next, we can focus on evaluating the terms in brackets above:
\begin{equation}
    \begin{alignedat}{1}
    \frac{\partial}{\partial S_{\alpha \beta}} \sum_{\mu\nu\lambda\sigma} \frac{\partial}{\partial S_{\alpha \beta}} C_{\mu p} C_{\nu q} C_{\lambda r} C_{\sigma s} g_{\mu \nu \lambda \sigma} &= \quad \sum_{\mu\nu\lambda\sigma} \frac{\partial C_{\mu p}}{\partial S_{\alpha \beta}}  C_{\nu q} C_{\lambda r} C_{\sigma s} g_{\mu \nu  \lambda \sigma} + \sum_{\mu\nu\lambda\sigma} C_{\mu p} \frac{\partial C_{\nu q}}{\partial S_{\alpha \beta}}  C_{\lambda r} C_{\sigma s} g_{\mu \nu \lambda \sigma} \\
    &\quad + \sum_{\mu\nu\lambda\sigma} C_{\mu p} C_{\nu q} \frac{\partial C_{\lambda r}}{\partial S_{\alpha \beta}}  C_{\sigma s} g_{\mu \nu \lambda \sigma} + \sum_{\mu\nu\lambda\sigma} C_{\mu p} C_{\nu q} C_{\lambda r} \frac{\partial C_{\sigma s}}{\partial S_{\alpha \beta}}  g_{\mu \nu \lambda \sigma} \\
    &= -\frac{1}{2}\sum_{\mu t} C_{\mu t} C_{\alpha t} C_{\beta p} g_{\mu q rs} -\frac{1}{2}\sum_{\nu t} C_{\nu t} C_{\alpha t} C_{\beta q} g_{p \nu rs} \\
    & \quad -\frac{1}{2}\sum_{\lambda t} C_{\lambda t} C_{\alpha t} C_{\beta r}  g_{pq  \lambda s} -\frac{1}{2}\sum_{\sigma t} C_{\sigma t} C_{\alpha t} C_{\beta s}  g_{pq r \sigma} \\
    &= -\frac{1}{2}\sum_{t} C_{\alpha t} C_{\beta p} g_{t q  rs} - \frac{1}{2}\sum_{t}  C_{\alpha t} C_{\beta q} g_{p t rs} \\
    & \quad -\frac{1}{2}\sum_{t} C_{\alpha t} C_{\beta r}  g_{pq t s} - \frac{1}{2}\sum_{t}  C_{\alpha t} C_{\beta s}  g_{pq r t} \\
    \end{alignedat}
\end{equation}
so, upon plugging back in to \eq{dtwodoverlap}
\begin{equation}
    \begin{alignedat}{1}
&\sum_{\alpha\beta} \sum_{pqrs} \Gamma_{pqrs} \left[ \frac{\partial}{\partial S_{\alpha \beta}} \sum_{\mu\nu\lambda\sigma} C_{\mu p} C_{\nu q} C_{\lambda r} C_{\sigma s} g_{\mu \nu  \lambda \sigma} \right] \frac{dS_{\alpha \beta}}{dx} \nonumber\\
&\qquad= -\frac{1}{2} \sum_{\alpha\beta} \Biggr[ \sum_{pqrst} \Gamma_{pqrs} C_{\alpha t} C_{\beta p} g_{t q rs} + \sum_{pqrst} \Gamma_{pqrs}  C_{\alpha t} C_{\beta q} g_{p t  rs} + \sum_{pqrst} \Gamma_{pqrs} C_{\alpha t} C_{\beta r}  g_{pq  t s} + \sum_{pqrst} \Gamma_{pqrs}  C_{\alpha t} C_{\beta s}  g_{pq  r t} \Biggr] \frac{dS_{\alpha \beta}}{dx} \\
    &= -2 \sum_{pq} \frac{dS_{pq}}{dx} \sum_{rst} \Gamma_{qrst} g_{prst}
    \end{alignedat}
\end{equation}

where in the last line we used the permutational symmetries of the real two-electron integrals and the fact that indices $\{p,q,r,s,t\}$ are dummy indices. Putting it all together 
\begin{equation}
\frac{\partial E}{\partial \mathbf{S}} \frac{d\mathbf{S}}{dx} = -\sum_{pqm} \gamma_{qm} h_{pm} \frac{dS_{pq}}{dx} -2 \sum_{pq} \frac{dS_{pq}}{dx} \sum_{rst} g_{prst} \Gamma_{qrst}
\end{equation}
which, reindexing and collecting terms by RDM, gives the final expression
\begin{equation}
  \frac{dE}{dx} = \sum_{pq} \gamma_{pq} \left[\frac{dh_{pq}}{dx} - \sum_m h_{mq} \frac{dS_{mp}}{dx}\right] + \sum_{pqrs} \Gamma_{pqrs} \left[ \frac{dg_{pqrs}}{dx} -2 \sum_{t}g_{tqrs} \frac{dS_{tp}}{dx}\right]
\end{equation}

\section{Additional calculations for NISQ force estimation}\label{sec:NISQ_estimation}

In this appendix we provide some additional calculations and numerics relevant to the estimation of forces on NISQ quantum computers that were not considered in the main text.

\subsection{Importance sampling to optimize 1-norm estimation}\label{app:ImportanceSamplingLagrange}

In the main text we studied the effect of importance sampling to minimize the $2$-norm of the error in the force vector.
In this section we repeat this analysis to minimize the $1$-norm instead.
Let given that we have a $D$-component vector, our aim is to minimize the total number of shots needed to ensure that the estimate of the gradient is $\epsilon-$close to the true value of the gradient with respect to the vector $1$-norm.  That is to say, we wish to minimize the expectation value over data collected of the one-norm of the difference between the true gradient and this operation, which we bound using Jensen's inequality as
\begin{equation}
    \mathbb{E}\left(\left\|\frac{\widetilde{d E}}{d \mathbf{R}} - \frac{{d E}}{d \mathbf{R}} \right\|_1 \right)=\mathbb{E}\left(\sum_i \sqrt{\left(\frac{\widetilde{d E}}{d {R}_i} - \frac{{d E}}{d {R}_i}\right)^2}  \right)\le    \left(\sum_i \sqrt{\mathbb{E}\left(\frac{\widetilde{d E}}{d {R}_i} - \frac{{d E}}{d {R}_i}\right)^2}  \right)
\end{equation}
Let us assume that our estimator for the component of the gradient is unbiased with variance $\sigma_{F,i}^2$
\begin{equation}
    \mathbb{E}\left(\left\|\frac{\widetilde{d E}}{d \mathbf{R}} - \frac{{d E}}{d \mathbf{R}} \right\|_1 \right) \le \sum_i \sigma_{F,i} 
\end{equation}
Next let us assume that each shot of $\frac{\widetilde{d E}}{d \mathbf{R}}$ is estimated by using importance sampling to stochastically sample terms from the gradient operator.  Let us assume that the variance yielded from the importance sampling distribution for component $i$ is at most $|\alpha_i|^2$.

Now assume that the operator computed has an absolute value of at most $|a_i|$ when computing the gradient.  This implies, if we use the sample mean as the estimator of $\sigma_{F,i}$ using $M_i$ shots then the standard deviation satisfies
\begin{equation}
    \sigma_{F,i} = \frac{|\alpha_i|}{\sqrt{M_i}}.
\end{equation}
with $M$ being our target for $\sum_i M_i$.  The first problem that we wish to tackle is the problem of distributing our $M$ measurements over these different possibility in a way that minimizes the error.
This fortunately, can be expressed as a simple constrained optimization problem with Lagrange multiplier $\lambda$ wherein the optimal solution is found when the derivative of the Lagrangian
\begin{equation}
    \mathcal{L}=\sum_i \frac{|\alpha_i|}{\sqrt{M_i}} + \lambda (\sum_i M_i -M),
\end{equation}
with respect to all parameters: $M_1,\ldots,M_{D}$ and $\lambda$, is equal to $0$.
Setting this condition gives us $D$ non-trivial equations of the form
\begin{equation}
    M_k =  \left(\frac{|\alpha_k|}{2\lambda} \right)^{2/3}~\forall~k.
\end{equation}
The Lagrange multiplier is, as yet, unspecified.  We can constrain it by demanding that the expected error in the $1$-norm is bounded above by $\epsilon$.  Specifically, we have that
\begin{equation}
    \sum_k \frac{|\alpha_k|}{\sqrt{M_k}} = (2 \lambda)^{1/3} \left(\sum_k {|\alpha_k|^{2/3}} \right) =\epsilon
\end{equation}
Therefore we can choose the Lagrange multiplier to satisfy
\begin{equation}
    \lambda = \frac{\epsilon^3}{2\left(\sum_k {|\alpha_k|^{2/3}} \right)^3}~\Rightarrow~M_k = \frac{|\alpha_k|^{2/3} (\sum_k |\alpha_k|^{2/3})^2}{\epsilon^2}
\end{equation}
This further implies that the total number of shots, as a function of the individual standard deviations of the individual terms is 
\begin{equation}
    M = \frac{(\sum_k |\alpha_k|^{2/3})^3}{\epsilon^2}.
\end{equation}
Next, let us assume that for each shot of a particular component of the force vector we use an importance sampling procedure wherein for each of the $Q_{ij}$ terms that need to be measured to estimate the gradient.  We then choose our probability such that
\begin{equation}
    Pr_{ij} = \frac{\|Q_{ij}\|}{\sum_{j} \|Q_{ij}\|}
\end{equation}
and further consider the operator $Y_{ij} = Q_{ij} / Pr_{ij}$
choose to measure, for fixed $i$, $Q_{ij}$ with probability $Pr_{ij}$.  Thus $\sum_j \bra{\psi} Q_{ij} \ket{\psi} = \sum_{j} \bra{\psi} Y_{ij} \ket{\psi} Pr_{ij}$.  The population variance of this estimate is then
\begin{align}
    \mathbb{E}(Y_{ij}^2) - \mathbb{E}(Y_{ij})^2 &= \sum_j \frac{(\bra{\psi} Q_{ij} \ket{\psi})^2}{Pr_{ij}} - \left(\sum_j \bra{\psi} Q_{ij} \ket{\psi}\right)^2\nonumber\\
    &\le \left( \frac{\sum_{k} \|Q_{ik}\|\bra{\psi} Q_{ij} \ket{\psi}}{\|Q_{ij}\|}\right)^2\nonumber\\
    &\le (\sum_k \|Q_{ik}\|)^2
\end{align}
Thus the overall number of measurements the $Q_{ij}$ needed to reduce the variance of the one-norm of the gradient to epsilon is at most
\begin{equation}
    M \le \frac{(\sum_{i} (\sum_j \|Q_{ij}\|)^{2/3})^3}{\epsilon^2}.
\end{equation}

\subsection{Localization of molecular orbitals}
\label{app:localization}
Canonical molecular orbitals (CMOs) from a Hartree-Fock calculation are only one of many possibilities to express the Hamiltonian and the force operators - any other orthonormal molecular orbitals might be used too. A prominemt example for another possible choice are localized molecular orbitals (LMOs). In contrast to CMOs, which typically have support over the entire molecule, LMOs are restricted to a certain spatial region of the molecule \cite{edmiston1963localized,pipek1989fast,foster1960canonical}. Given a set of CMOs, one finds LMOs by applying a unitary transformation $U$ , 

\begin{align*}
\tilde{\phi}_q(r)=\sum_p \phi_p(r) U_{pq}\,.
\end{align*}

The unitary $U$ is altered with the goal to optimize the locality of the molecular orbitals. Different schemes with different cost functions exists, such as Edmiston-Ruedenberg (ER) \cite{edmiston1963localized}, Pipek-Mezey (PM) \cite{pipek1989fast} or Foster-Boys (FB) \cite{foster1960canonical}. 
In \cite{koridon2021orbital}, it was observed that LMOs in comparison to CMOs have a positive, lowering, effect on the 1-norm of the qubit Hamiltonian. As the force operators force operators include terms with $\mathcal{O}(1/r^2)$ while the Coloumb potential is $\mathcal{O}(1/r)$, we expect that using LMOs will also lower the 1-norms of the force operators. As different localization schemes have similar effects on the 1-norm of the Hamiltonian \cite{koridon2021orbital}, we restrict ourselves to the ER localization scheme in the following. Using ER, we aim to maximize
\begin{align}
	\mathcal{L}_\mathrm{ER} = \sum_{p}\int\int |\tilde{\phi}_p(r_1)|^2 \frac{1}{r_{12}}|\tilde{\phi}_p(r_2)|^2 \mathrm{dr_1}\mathrm{dr_2}
\end{align}
within the occupied/virtual molecular orbitals in order to find localized molecular orbitals. We then use the new set of LMOs to express the Hamiltonian and the force operators. We note that the fermionic shadow tomography scheme and the basis rotation grouping strategy, see Sec.~\ref{sec:nisq}, are not affected by any localization.  The basis rotation grouping scheme is based on a diagonalization, which is independent from the basis used, while the variance of the estimator of the fermionic shadow tomography depends  on the 2-norm, see Eq.~\eqref{eq:fermionic_shadow_variance}, which is preserved under any unitary transformation.

\section{Additional details on force estimation by numerical differentiation}\label{Appendix_Finite_differences}

\subsection{Comparison with first order finite differences}
The first order finite difference formula is:

\begin{equation}
    \frac{d E}{d R_i} = \frac{E(R +\frac{dR\cdot v_i}{2})-E(R-\frac{dR\cdot v_i}{2})}{dR}+\epsilon_{fd} 
\end{equation}

as in the general case, the total error for this will be the result of the error due to the quantum phase estimation and the discretization error due to the first order finite difference approximation:

By propagating the statistical error, due to quantum phase estimation, where we assume each QPE error is upper bounded by Eq.~\eqref{eq_error_PE_H_var}, we can write:
\begin{equation}
    \epsilon_{\mathrm{PE}}^{(m=1)} \lesssim \sqrt{2} \frac{\lambda_{\mathrm{H}} \pi}{2 T\;dR},
\end{equation}
where $T$ is the number of oracle calls used in each QPE evaluation of the energy and $\lambda_{\mathrm{H}} \lesssim \norm{H}_1$.      
On the other hand, if we assume that the energy is a well behaved function, then $\epsilon_{fd}$ is:
\begin{equation}
    \epsilon_{fd} = \mathfrak{k} \sum^{\infty}_{i=1} \frac{d^{2i} E(R)}{d R^{2i}} \frac{dR^{2i}}{2i!} \approx \mathfrak{k} \mathfrak{a} \times dR^2,
\end{equation}
with $\mathfrak{k}$ a constant factor with the dimension of an energy over a length cubed, and $\mathfrak{a}$ a dimensionless constant which upper-bounds the derivatives at $R$.
So the total error is:
\begin{equation}
    \epsilon_{\mathrm{1-tot}} \approx \sqrt{2} \frac{\lambda_{\mathrm{H}} \pi}{2 T\;dR} + \mathfrak{ka} dR^2
\end{equation}

This has the same $\tilde{O}$ scaling as the upper bound on the error in Eq.~\eqref{eq_epsilon_tot_not_opt_app} for $m=1$. In fact, for $m=1\rightarrow \Gamma=\sqrt{2}$ Eq.~\eqref{eq_epsilon_tot_not_opt_app} is:
\begin{equation}
      \epsilon = \epsilon_{\mathrm{PE}} +\epsilon_{fd}^{(m=1)}  < \frac{\pi\lambda_{\mathrm{H}}}{2T} \frac{6^{3/2}}{dR}+ \mathfrak{2e}{(8 c )^{3} dR^{2}}.
\end{equation}
The two have the same scalings in terms of $dR$, $\lambda_{\mathrm{H}}$ and $T$, and they are, in fact, identical up to constant factors in the second term.

\subsection{The error contribution from quantum phase estimation and the Holevo variance}\label{app:App_FD_Holevo_variance}

The error in phase estimation protocols is typically quantified by the Holevo variance of a phase $\phi$~\cite{Berry2009}, which is given by
\begin{equation}
V_H(\phi) = \frac{1}{|\mathbb{E}(e^{-i (\phi - \phi_{\rm true})})|^2}-1,    
\end{equation}
where $\phi$ is the estimated eigenphase and $\phi_{\rm true}$ is the true value of the eigenphase.
Specifically, it can be shown that there exists a phase estimation protocol that applies the fundamental unitary $T_l$ times and yields an estimate such that: 
\begin{align}
    \mathbb{E}(\phi-\phi_{\rm true}) =0, \; \mathrm{and } & \; V_H(\phi)=\tan^2\left(\frac{\pi}{2 T_l+1}\right)\approx \left(\frac{\pi}{2 T_l}\right)^2,
\end{align}
where the first expression means that $(\phi-\phi_{\rm true})$ is a random variable with zero mean~\cite{Berry2009, BabbushSpectra}. 
The purpose behind this estimate is to address the problem that the variance of a phase depends on the branch chosen for the phase as the Holevo variance is independent of the branch chosen.  Because of these differences, the properties that one would expect of a variance do not necessarily always hold here.  Specifically, we would like to apply the additive property of the variance to ensure that the variance of the sum of the energy estimated at each point is at most the sum of the variances.  We have that if for any $\phi_j$ the fourth order momentum of the unbiased random variable $(\phi_j-\phi_{j,\rm true})$
\begin{align}
         \mathcal{K}_j&:=\mathbb{E}((\phi_j-\phi_{j,\rm true})^4) \gg \sum_{k=3}^\infty \frac{\mathbb{E}(\phi_j-\phi_{j,\rm true})^{2k}}{2k!},
\end{align}
Then using the fact that the phase estimation protocol used is unbiased we have from Taylor's theorem that:
\begin{align}
    \mathbb{V}_H(\phi_j) = \frac{1}{(1-\frac{1}{2}\mathbb{E}(\phi_j - \phi_{j,true})^2)(1-\frac{1}{2}\mathbb{E}(\phi_j - \phi_{j,true})^2) + O(\mathcal{K}_j)}-1 = \mathbb{V}(\phi_j) + O(\mathcal{K}_j).
\end{align}
Thus we have from the additivity of variance that for independent random variables $\phi_i \in [-\pi, \pi)$ the variance of the sum of the random variables is
\begin{equation}
    \mathbb{V}(\sum_{j=1}^m \phi_j) = \sum_{j=1}^m \mathbb{V}(\phi_j) \le \sum_{j=1}^m V_H(\phi_j) + O(m \max_j \mathcal{K}_j).
\end{equation}
We therefore have that for such estimates, in the limit where the fourth and higher moments of the angular distribution are negligible then the statistical errors from phase estimation add in quadrature.

We obtain then that the variance of the estimates yielded by phase estimation can be written as:
\begin{equation}
    \mathbb{V}(\phi_l) = \frac{\pi^2}{4T_l^2} + O(\max_l \mathcal{K}_l).
\end{equation}
If we apply this to the problem of finding the eigenvalue of the energy with the quantum phase estimation, we have that
the phases returned from applying phase estimation on the qubitization walk oracle are: $\pm \cos^{-1}(E_j/\lambda_{\mathrm{H}})$, which is to say that the energy estimates needed for the gradient evaluation are from Eq.~\eqref{eq:higher_order_derivatives} 
propagating the IID variances of the $2m$ estimations we obtain the total mean square error $\epsilon_{\mathrm{PE}}$: 
\begin{equation}\label{eq_Epsilon_PE1_app}
    \epsilon_{\mathrm{PE}}^2\lesssim \pi^2 dR^{-2}\lambda_{\mathrm{H}}^{2}\sum_l|a_l^{(m)}|^2 (2\;T_l)^{-2}.
\end{equation}
It is equivalent to take Eq.~(23) from \cite{BabbushSpectra}:
\begin{equation}\label{eq_error_PE_H_var}
    \epsilon_{\mathrm{PE}}(l)^2 \lesssim \lambda_{\mathrm{H}}^2 \left(\frac{\pi^2}{4 T_l^2} + (\epsilon_{\rm{prep}} +\pi^2 \epsilon_{\rm{QFT}})^2 \right),
\end{equation}
neglect the contribution to the errors $\epsilon_{\rm prep}$ and $\epsilon_{\rm QFT}$ and then propagate assuming the covariances are zero to obtain Eq.~\eqref{eq_Epsilon_PE1_app}. $\epsilon_{\rm prep}$ and $\epsilon_{\rm QFT}$ are expected to be considerably smaller than the $1/T_l$ contribution, though they cannot be neglected in general, their consideration would require to find numerical solutions for estimating the scaling of the error.

\subsection{Robustness of phase estimation success probability}\label{app:robust}
The aim of this section is to provide a result that can be used in order to argue how the success probability for phase estimation.  This result is especially important for the finite difference algorithm discussed above because this result will show that if the eigenvalue gap is sufficiently large then with high probability all the values used for the state preparation can be taken from the same ground state.  If this is not true, then in the worst case scenario the ground state preparation will need to be applied for each perturbed Hamiltonian evaluated.  The result is stated below in the following Lemma, and its proof is a direct consequence of matrix calculus and elementary norm inequalities.
\begin{lemma}
Let $H:\mathbb{R} \mapsto \mathbb{C}^{2^n \times 2^n}$ be a parameterized family of Hamiltonians and let $\ket{\psi_i(\mathbf{R})}$ be the $i^{\rm th}$ eigenstate of $H(\mathbf{R})$ evaluated at parameters $\mathbf{R}\in \mathbb{R}^{3N_a}$ and let $E_i(\mathbf{R})$ be the corresponding eigenvalue.  If $\ket{\tilde{\psi}_0(\mathbf{R})}$ is known such that $\left\|\ket{\tilde{\psi}_0(\mathbf{R})} - \ket{{\psi}_0(\mathbf{R})} \right\| \le \delta$ and $\Delta\in \mathbb{R}^{3N_a}$ such that $\|\Delta\|< (E_1(\mathbf{R}) -E_0(\mathbf{R}))/(4\max_{\mathbf{R}',i}\left\|\tfrac{dH(\mathbf{R}')}{dR_i}\right\|)$ then
$$
|\braket{\tilde{\psi}_0(\mathbf{R})}{{\psi_0(\mathbf{R}+\Delta)}}| \ge 1 - \delta -  \frac{2\|\Delta\| \max_{\mathbf{R}',i}\left\|\tfrac{dH(\mathbf{R}')}{dR_i}\right\|}{E_1(\mathbf{R}) - E_0(\mathbf{R})}.
$$
Here, the maximum is taken over all points in a ball of radius $\|\Delta\|$ around $\mathbf{R}$ and over all derivative components $\frac{d}{dR_i}$.
\end{lemma}
\begin{proof}
First  we have from the fundamental theorem of calculus, the assumption that $E_j(\mathbf{R}+s) > E_0(\mathbf{R}+s)$ that 
\begin{align}
    \ket{\psi(\mathbf{R}+\Delta)} &= \ket{\psi(\mathbf{R})} + \int_0^1 \frac{d\ket{\psi(\mathbf{R}+s\Delta)}}{ds} \mathrm{d}s\nonumber\\
    &= \ket{\psi(\mathbf{R})} + \int_0^1 \sum_{j\ne 0} \frac{\ket{\psi_j(\mathbf{R}+s\Delta)} \bra{\psi_j(\mathbf{R}+s\Delta)} \frac{dH(\mathbf{R}+s\Delta)}{ds} \ket{\psi(\mathbf{R}+s\Delta)}}{E_j(\mathbf{R}+s\Delta) - E_0(\mathbf{R}+s\Delta)}\mathrm{d}s
\end{align}
Next we have from the Cauchy-Schwarz inequality
\begin{align}
    |\braket{\psi(\mathbf{R})}{\psi(\mathbf{R}+\Delta)}|&\ge 1 - \left|\int_0^1 \sum_{j\ne 0} \frac{\braket{\psi(\mathbf{R})}{\psi_j(\mathbf{R}+s\Delta)} \bra{\psi_j(\mathbf{R}+s\Delta)} \tfrac{dH(\mathbf{R}+s\Delta)}{ds} \ket{\psi(\mathbf{R}+s\Delta)}}{E_j(\mathbf{R}+s\Delta) - E_0(\mathbf{R}+s\Delta)}\mathrm{d}s \right|\nonumber\\
    &\ge 1 - \int_0^1\left| \sum_{j\ne 0} \frac{\braket{\psi(\mathbf{R})}{\psi_j(\mathbf{R}+s\Delta)} \bra{\psi_j(\mathbf{R}+s\Delta)} \tfrac{dH(\mathbf{R}+s\Delta)}{ds} \ket{\psi(\mathbf{R}+s\Delta)}}{E_j(\mathbf{R}+s\Delta) - E_0(\mathbf{R}+s\Delta)} \right|\mathrm{d}s\nonumber\\
    &\ge 1 - \int_0^1\sqrt{\left| \sum_{j\ne 0} \frac{\bra{\psi(\mathbf{R}+s\Delta)} \tfrac{dH(\mathbf{R}+s\Delta)}{ds}\ket{\psi_j(\mathbf{R}+s\Delta)}\bra{\psi_j(\mathbf{R}+s\Delta)} \tfrac{dH(\mathbf{R}+s\Delta)}{ds} \ket{\psi(\mathbf{R}+s\Delta)}}{(E_1(\mathbf{R}+s\Delta) - E_0(\mathbf{R}+s\Delta))^2} \right|}\mathrm{d}s\nonumber\\
    &\ge 1 - \int_0^1 \frac{\|\tfrac{dH(\mathbf{R}+s\Delta)}{ds}\|}{E_1(\mathbf{R}+s\Delta) - E_0(\mathbf{R}+s\Delta)}\mathrm{d}s\nonumber\\
    &\ge 1 -  \frac{\|\Delta\| \max_s\|\tfrac{dH(\mathbf{R}+s\Delta)}{ds}\|}{\min_s|E_1(\mathbf{R}+s\Delta) - E_0(\mathbf{R}+s\Delta)|} \ge 1 -  \frac{\|\Delta\| \max_s\|\tfrac{dH(\mathbf{R}+s\Delta)}{ds}\|}{E_1(\mathbf{R}) - E_0(\mathbf{R}) -2\|\Delta\| \max_s\| \tfrac{dH(\mathbf{R}+s\Delta)}{ds}\|}\nonumber\\
    &\ge 1 -  \frac{2\|\Delta\| \max_s\|\tfrac{dH(\mathbf{R}+s\Delta)}{ds}\|}{E_1(\mathbf{R}) - E_0(\mathbf{R})}\ge 1 -  \frac{2\|\Delta\| \max_{\mathbf{R}',i}\left\|\tfrac{dH(\mathbf{R}')}{dR_i}\right\|}{E_1(\mathbf{R}) - E_0(\mathbf{R})}
\end{align}
where the last inequality only follows under  the assumption that $\max_s\| \tfrac{d}{ds} H(\mathbf{R}+s\|\Delta\|)\| \le (E_1(\mathbf{R}) - E_0(\mathbf{R}))/4\Delta$.

Finally if $\|\ket{\tilde{\psi}_0(\mathbf{R}+\Delta)} - \ket{{\psi}_0(\mathbf{R}+\Delta)} \| \le \delta$ then there exists $\delta'\le \delta$ and $\ket{\phi}$ such that $\ket{\tilde{\psi}_0(\mathbf{R}+\Delta)} = \ket{\psi_0(\mathbf{R}+\Delta)} + \delta' \ket{\phi}$.
We therefore have from the reverse triangle inequality that under the assumptions of the lemma

\begin{align}
    |\braket{\tilde{\psi}_0(\mathbf{R}+\Delta)}{\psi_0(\mathbf{R})}| \ge  1 - \delta -  \frac{2\|\Delta\| \max_s\|\tfrac{dH(\mathbf{R}+s\Delta)}{ds}\|}{E_1(\mathbf{R}) - E_0(\mathbf{R})}\ge 1 -  \delta-\frac{2\|\Delta\| \max_{\mathbf{R},i}\left\|\frac{dH(\mathbf{R})}{dR_i}\right\|}{E_1(\mathbf{R}) - E_0(\mathbf{R})}
\end{align}
\end{proof}
This shows that provided the eigenvalue gap is significant relative to the norm of the  derivative operator then the reduction in success probability for phase estimation is minimal.  This means that in practice, small shifts in the Hamiltonian will not meaningfully impact the overlaps provided these gap assumptions are met.

\begin{corollary}
Assume that the displacements $\Delta_1\in\mathbb{R}^{3N_a}$ and $\Delta_2 \in \mathbb{R}^{3N_a}$ individually satisfy the assumptions of the previous lemma we then have that
$$
|\braket{\tilde{\psi}_0(x+\Delta_1)}{\tilde{\psi}_0(x +\Delta_2)}| \ge 1 - 3\delta - \frac{6 \max\{\|\Delta_1\|,\|\Delta_2\|\} \max_{\mathbf{R}',i}\left\|\tfrac{dH(\mathbf{R}')}{dR_i}\right\|}{E_1(x) - E_0(x)}
$$
\end{corollary}
\begin{proof}
Under the assumptions of the previous theorem we have that there exist $\delta_1$ and $\delta_2$ that are bounded above in absolute value by $\delta + \frac{4 \max\{\|\Delta_1\|,\|\Delta_2\|\} \max_{\mathbf{R}',i}\left\|\tfrac{dH(\mathbf{R}')}{dR_i}\right\|}{E_1(x) - E_0(x)}$ and $\ket{\phi_1},\ket{\phi_2}$ such that
\begin{align}
    |\braket{\tilde{\psi}_0(x+\Delta_1)}{\tilde{\psi}_0(x +\Delta_2)}| &= |1 +\delta_2 \braket{\tilde{\psi}_0(x+\Delta_1)}{\phi_2}+\delta_1 \braket{\phi_1}{\tilde{\psi}_0(x+\Delta_2)} + \delta_1\delta_2\braket{\tilde{\psi}_0(x+\Delta_1)}{\tilde{\psi}_0(x+\Delta_2)}|\nonumber\\
    &\ge 1 - 3\max\{|\delta_1|,|\delta_2|\}.
\end{align}
The result then immediately follows by the assumed bounds on $\delta_1,\delta_2$.
\end{proof}

\subsection{Optimal T for numerical differentiation of the energy  with Lagrange Multipliers}\label{App_FD_Lagrange_Mult}

In order to find the optimal time to achieve a given error in the evaluation of the gradient with finite differences, we need to minimize $T=\sum_lT_l$, with the condition 
$\epsilon^2_{\mathrm{PE}} = \pi^2dR^{-2}\lambda_{\mathrm{H}}^{2}\sum_l|a_l^{(m)}|^2(2T_l)^{-2}$.
\begin{align}
    \mathcal{L} &=\sum_{l=-m}^m T_l + \lambda \left( \epsilon_{\mathrm{PE}}^2 - \tilde{\epsilon}_{\mathrm{PE}}^2\right) \\
    \mathcal{L} &=\sum_{l=-m}^m T_l + \lambda \left( \pi^2 dR^{-2} \lambda_{\mathrm{H}}^2 \sum_{l=-m}^m |a_l^{(m)}|^2 \left(2 T_l\right)^{-2} -\tilde{\epsilon}_{\mathrm{PE}}^2 \right).
\end{align}
We have then:
\begin{equation}
    \frac{\partial \mathcal{L}}{\partial T_l} = 1 - 2^{-1} \lambda \left(\pi^2 dR^{-2} \lambda_{\mathrm{H}}^2 |a_l^{(m)}|^2  T_l^{-3}\right) = 0.
\end{equation}
Solving in $T_l$ we have:
\begin{equation}\label{eq_T_l_Lagrange}
    T_l = \left(\lambda 2^{-1} \pi^2 dR^{-2} \lambda_{\mathrm{H}}^2 |a_l^{(m)}|^2 \right)^{\frac{1}{3}}, 
\end{equation}
substituting this back in Eq.~\eqref{eq_Epsilon_PE1_app} with equal sign, which is our condition for the Lagrangian, and solving in lambda, we have:
\begin{align}
    \epsilon_{\mathrm{PE}}^2 &= \pi^2dR^{-2}\lambda_{\mathrm{H}}^{2} \sum_{l=m}^m \left( |a_l^{(m)}|^2 2^{-2} \lambda^{-\frac{2}{3}} 2^\frac{2}{3} \pi^{-\frac{4}{3}} dr^\frac{4}{3} \lambda_{\mathrm{H}}^{-\frac{4}{3}}|a_l^{(m)}|^{-\frac{4}{3}}\right) \\ 
    \epsilon_{\mathrm{PE}}^2 &= \pi^{\frac{2}{3}} dR^{-\frac{2}{3}} \lambda_{\mathrm{H}}^{\frac{2}{3}} \lambda^{-\frac{2}{3}} 2^{-\frac{4}{3}} \sum_{l=m}^m |a_l^{(m)}|^\frac{2}{3}\\
    \lambda &= \epsilon_{\mathrm{PE}}^{-3} \; \pi \;  dR^{-1}\; \lambda_{\mathrm{H}} 2^{-2} \left(\sum_{l=m}^m |a_l^{(m)}|^\frac{2}{3} \right)^\frac{3}{2}
\end{align}
and substituting this in Eq.~\eqref{eq_T_l_Lagrange}
\begin{equation}
    T_l = \epsilon_{\rm{PE}}^{-1}\;\pi \; dR^{-1} \lambda_{\mathrm{H}} 2^{-1} |a_l^{(m)}|^\frac{2}{3} \left( \sum_{l=m}^m |a_l^{(m)}|^\frac{2}{3} \right)^\frac{1}{2}
\end{equation}
Summing over $l$ and using the shorthand $\Gamma$ from Eq.~\eqref{eq_Gamma}:
\begin{align}
    T =\sum_{l=-m}^m T_l =  \epsilon_{\rm{PE}}^{-1}\;\pi \; dR^{-1} \lambda_{\mathrm{H}} 2^{-1} \sum_{l=-m}^m |a_l^{(m)}|^\frac{2}{3} \left( \sum_{l=m}^m |a_l^{(m)}|^\frac{2}{3} \right)^\frac{1}{2} =\\
    \epsilon_{\rm{PE}}^{-1}\;\pi \; dR^{-1} \lambda_{\mathrm{H}} 2^{-1} \Gamma^\frac{3}{2}\Gamma^\frac{1}{3}= \epsilon_{\rm{PE}}^{-1}\;\pi \; dR^{-1} \lambda_{\mathrm{H}} 2^{-1} \Gamma^\frac{5}{6}\leq \epsilon_{\rm{PE}}^{-1}\;\pi \; dR^{-1} \lambda_{\mathrm{H}} 2^{-1} \Gamma, 
\end{align}
In summary, using the Lagrange multipliers method, solving in $T_l$ and summing over all $l$ we obtain that the total time necessary to obtain a target statistical error $\epsilon_{\mathrm{PE}}$:    
\begin{align}
    T=\sum_lT_l&\le \pi\;\epsilon_{\mathrm{PE}}^{-1}\;dR^{-1}\;\lambda_{\mathrm{H}} \;2^{-1} \left[\sum_{l=-m}^m|a_l^{(m)}|^{2/3}\right]^{3/2}.
\end{align}

So that we can write:
\begin{equation}\label{App_T_scaling_Lagrange_M}
    T<\epsilon_{\rm{PE}}^{-1}\;\pi \; dR^{-1} \lambda_{\mathrm{H}} 2^{-1} \Gamma
\end{equation}
and inverting Eq.~\eqref{App_T_scaling_Lagrange_M}:
\begin{equation}
    \epsilon_{\mathrm{PE}} \le \frac{\pi \lambda_{\mathrm{H}} \;\Gamma}{2dR\;T} < \frac{\pi \lambda_{\mathrm{H}} 6^{3/2} m^{1/2}}{2dR\;T}.
\end{equation}
It remains to place a bound on $\Gamma$.
We have that
\begin{align}
    |a_l^{(m)}| &= \frac{1}{|l|} \frac{\binom{m}{|l|}}{\binom{m+|l|}{|l|}} = \frac{1}{|l|}\frac{\binom{m}{|l|}}{\binom{m+1}{|l|}}\left(\frac{\binom{m+1}{|l|}}{\binom{m+|l|}{|l|}} \right) = \frac{1}{|l|} \left(\frac{m!}{(m-|l|)!} \frac{(m+1-|l|)!}{(m+1)!} \right)\left(\frac{\binom{m+1}{|l|}}{\binom{m+|l|}{|l|}} \right)\nonumber\\
    &\le \frac{1}{|l|}\left(\frac{m}{m+1} \right)\left(\frac{\binom{m+1}{|l|}}{\binom{m+|l|}{|l|}} \right)\le \frac{1}{|l|}\left(\frac{m}{m+1} \right).
\end{align}

and taking a continuum approximation yields
\begin{equation}\label{eq_Gamma}
    \Gamma^{\frac{2}{3}}=
    \sum_{l=-m}^m |a_l^{(m)}|^\frac{2}{3} =2 \sum_{l=1}^m |a_l^{(m)}|^\frac{2}{3} \le 2 \sum_{l=1}^m |l|^{-\frac{2}{3}} \le 2 \left( 1 + \int_1^{m-1} |l|^{-\frac{2}{3}} \mathrm{d}l \right)= 6(m-1)^{\frac{1}{3}}-4 < 6m^\frac{1}{3}.
\end{equation}

\section{Momentum state preparation for first-quantized plane-wave simulations}
\label{app:first_quantized_force_prepare}

In order to block-encode the force operator in first quantization in a plane wave basis, we require the ability to prepare a momentum state proportional to
\begin{equation}
\label{eq:momentum_state}
	\sum_{\nu\neq0}\frac{\sqrt{|\nu_x|}}{\norm{\nu}}\ket{\nu}.
\end{equation}
The challenge for this preparation is that there are $\Theta(N)$ values of $\nu$ in this state which means that the Toffoli gate complexity scales like $\Theta(N)$ if we prepare it naively.
Instead, we will describe a circuit implementation which has an asymptotic cost of $\mathcal{O}(\mathrm{polylog}(N,1/\epsilon))$, where $\epsilon$ is the error of the block-encoded operator.

Note that this is different from performing Hamiltonian simulation in first quantization \cite{BabbushContinuum,su2021fault}. There, our goal is to prepare a state proportional to $\sum_{\nu\neq0}\frac{1}{\norm{\nu}}\ket{\nu}$, and we achieve that by grouping the set of momenta into nested boxes
\begin{equation}
	\left\{\nu\ \left|\ \abs{\nu_x}<2^{\mu-1}\land\abs{\nu_y}<2^{\mu-1}\land\abs{\nu_z}<2^{\mu-1}\land
	\left(\abs{\nu_x}\geq 2^{\mu}\lor\abs{\nu_y}\geq 2^{\mu}\lor\abs{\nu_z}\geq 2^{\mu}\right)\right.\right\}.
\end{equation}
The number of boxes scales only like $\mathcal{O}(\log N)$, so a superposition over all nested boxes can be prepared efficiently. Meanwhile, the momenta within each nested box do not vary significantly in size, which means that the coefficients are nearly uniformly distributed and a superposition over those momenta can also be efficiently prepared. However, we now have the additional weighting of $\sqrt{\nu_x}$ for block-encoding the derivative operator, whose value changes dramatically within each box. As a result, the preparation subroutine would succeed with a very small probability which would need to be amplified with significant overhead.

\input{preparation}

\section{Numerical studies of the force operator in molecular systems}
\label{app:numerical}

\subsection{Details on numerical calculations of water clusters}
\label{app:numerical_details_water_clusters}
When calculating the Hamiltonian and the force operators of the water clusters, shown in Fig.~\ref{fig:water_clusters}, in STO-3G, we recognize that the matricized two-body coefficients $M_{(pq),(rs)}$ have degenereted eigenspaces. The resulting eigenvectors are thus only defined up to some gauge. While this is not a practical issue for applying any of the methods, repeated calculations will lead to slightly different lambdas. To avoid this ambiguity, we break the symmetry by shifting the positions of the atoms, $x_i\rightarrow x_i+\mathcal{N}(0,\sigma)$ with $\sigma=10^{-6}$.  

For the largest water cluster we study here, with $N_{H2O}=20$ water molecules, the spatial two-body-integrals of the Hamiltonian and the force operators have a dimension of $(140,140,140,140)$ and consequently the matrix $M_{(pq),(rs)}$ is of dimension $(19600,19600)$. As a direct diagonalization of a matrix of such size is getting computationally expensive, we instead use Scipy's implementation of Lanczos algorithm to find the most important eigenvalues and eigenvectors. To find the error of this approximation, we calculate the expectation value of the Hamiltonian and the force operator with respect to the CISD wavefunction found under the reconstructed ERIs. In Fig.~\ref{fig:errors_cisd}(a), we show the error on the CISD energy per atom $\epsilon_\mathrm{E}^{\mathrm{CISD}} = (E_{\mathrm{
full}}-E_{\mathrm{lanczos}}^k)/(3N_{H2O})$. We find that keeping the $4N$ most important eigenvalues and eigenvecotors yields an error below $50 \mu Ha$ per atom. In Fig.~\ref{fig:errors_cisd}(b), we show the error on the expectation value of the force operator with respect to the CISD wavefunction. As the force is a vector of dimension $3\times 3\times N_{H2O}$, we use $ \epsilon_\mathrm{F}^{\mathrm{CISD}} = ||\vec{F}_\mathrm{full}-\vec{F}_\mathrm{Lanczos}||$ as error. As visible from the plots, keeping $4N$ eigenvalues, as for the Hamiltonian, is a reasonable choice for the error on the force vector. 

\begin{figure*}
    \centering
    \begin{minipage}{0.49\textwidth}
        \centering
        \includegraphics[width=1\textwidth]{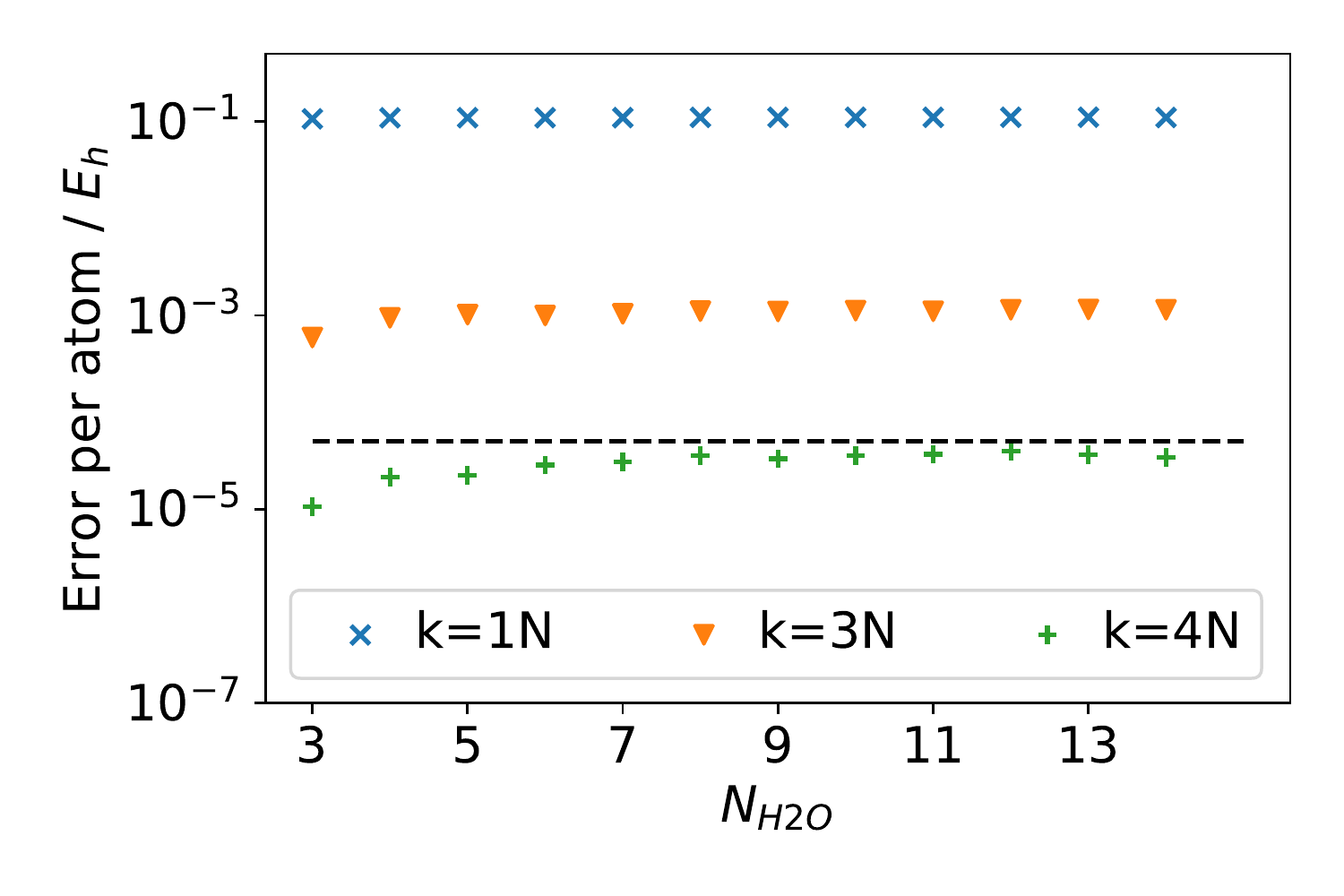}\llap{
  \parbox[b]{6.0in}{(a)\\\rule{0ex}{2.2in}
  }}
    \end{minipage}
    \begin{minipage}{0.49\textwidth}
        \centering
        \includegraphics[width=1\textwidth]{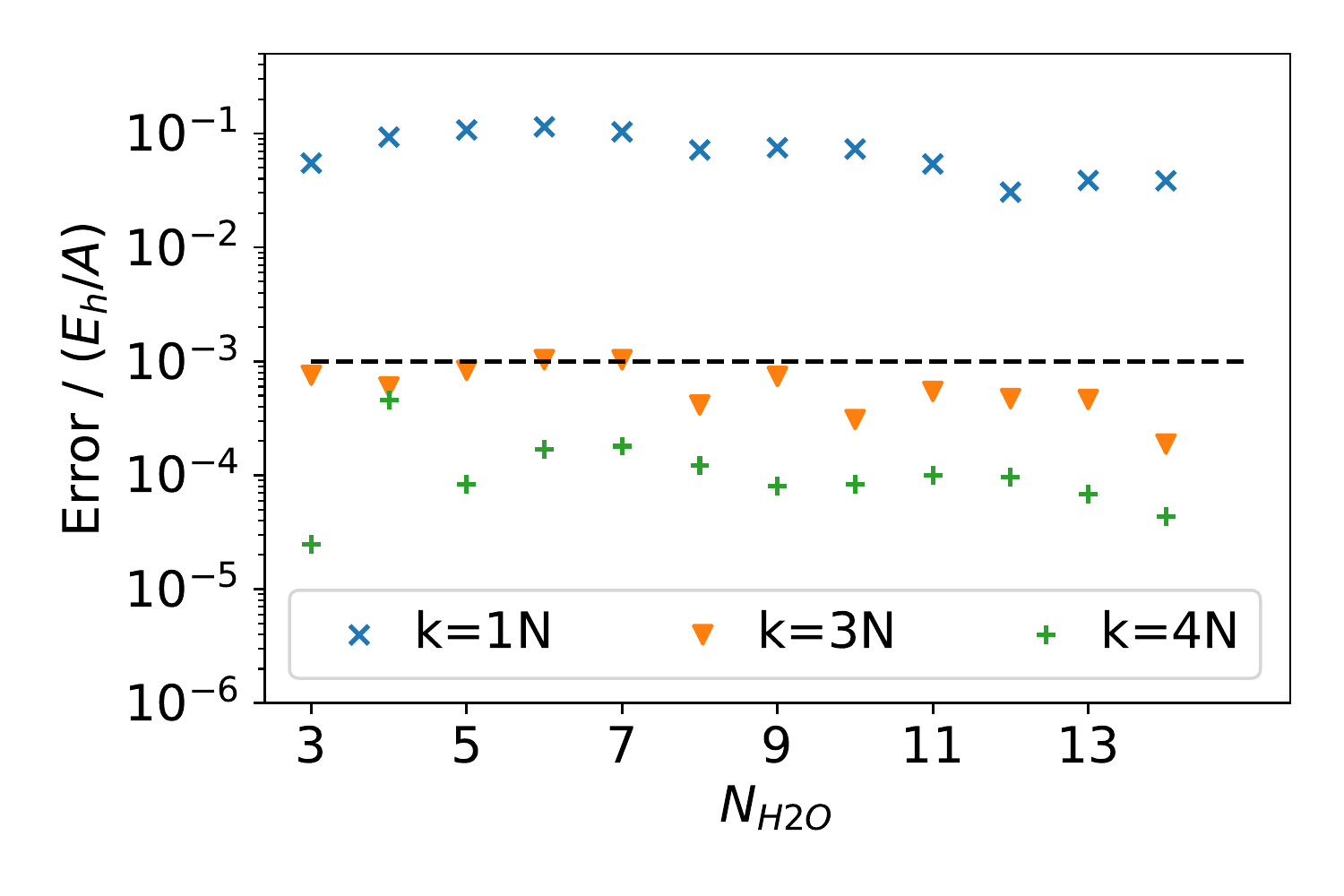}\llap{
  \parbox[b]{6.0in}{(b)\\\rule{0ex}{2.2in}
  }}
    \end{minipage}\hfill
\caption{The error from using Lanczos algorithm to decompose the matricized ERI of the water clusters on the (a) expectation value of Hamiltonian and (b) the force operators with respect to to the CISD wavefunction}
\label{fig:errors_cisd}
\end{figure*}

\begin{figure}
\centering
\begin{tabular}{cccc}
  \includegraphics[width=35mm]{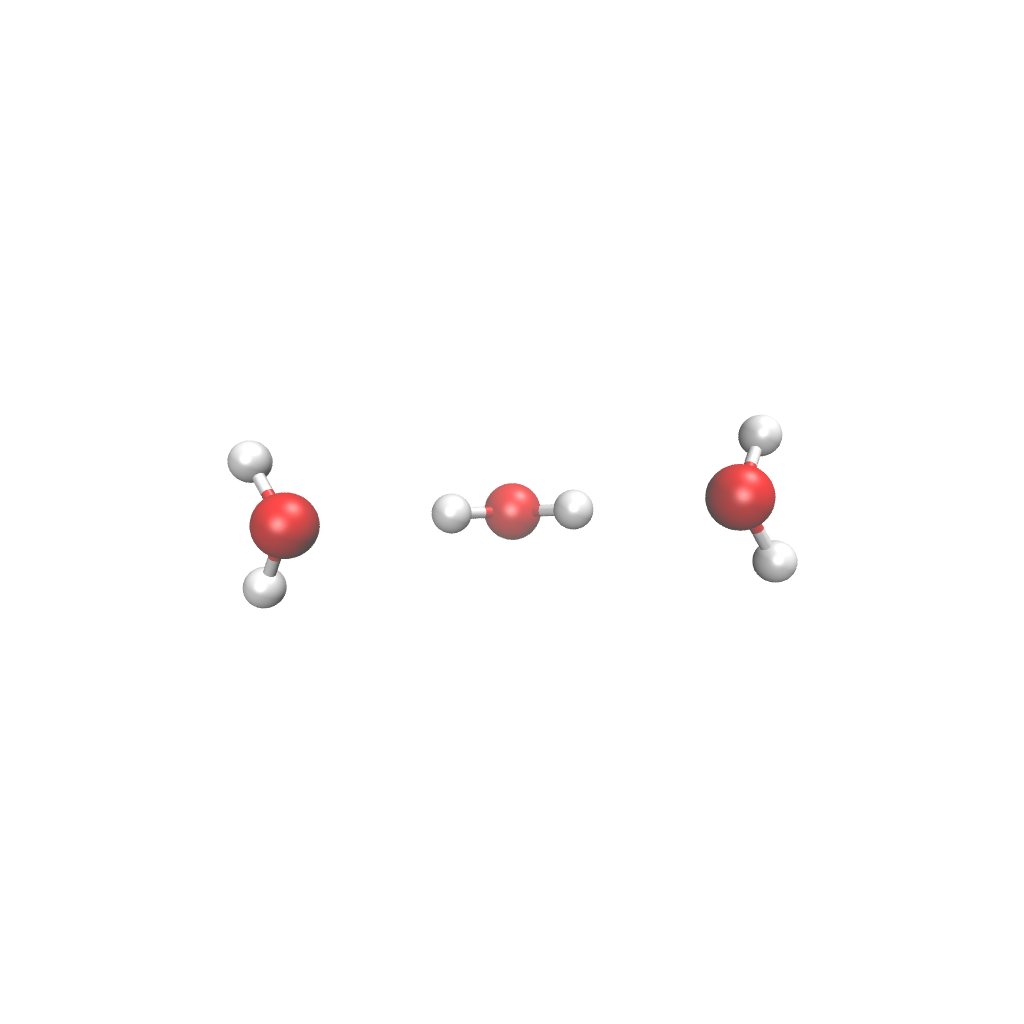} &   \includegraphics[width=35mm]{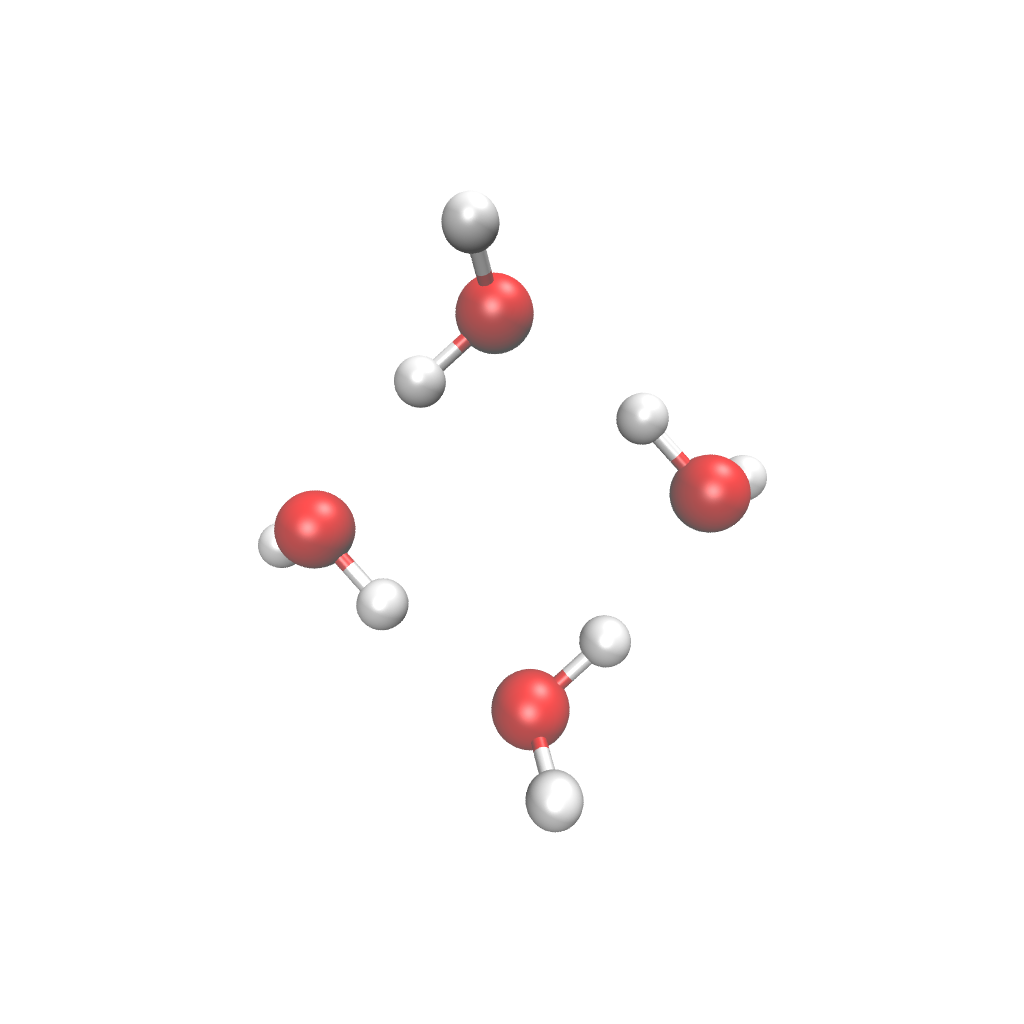} &   \includegraphics[width=35mm]{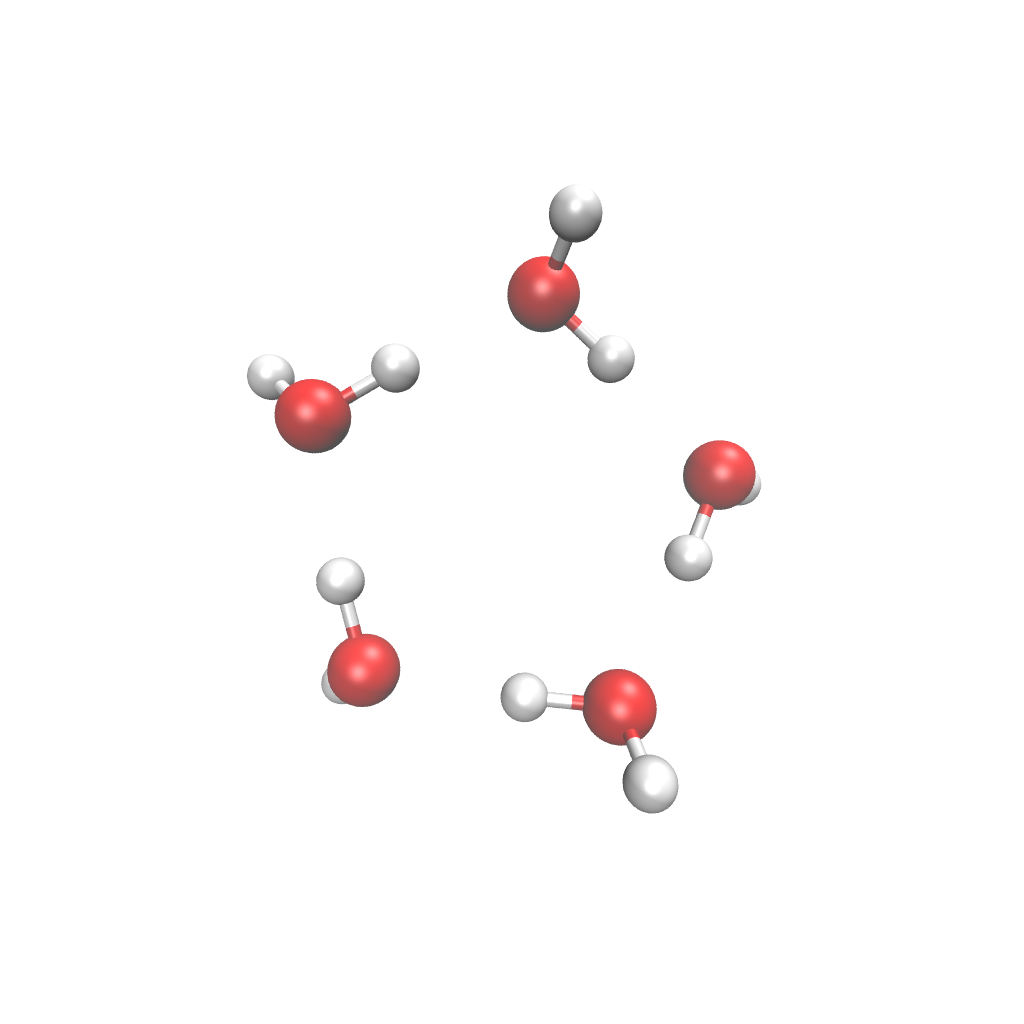} &   \includegraphics[width=35mm]{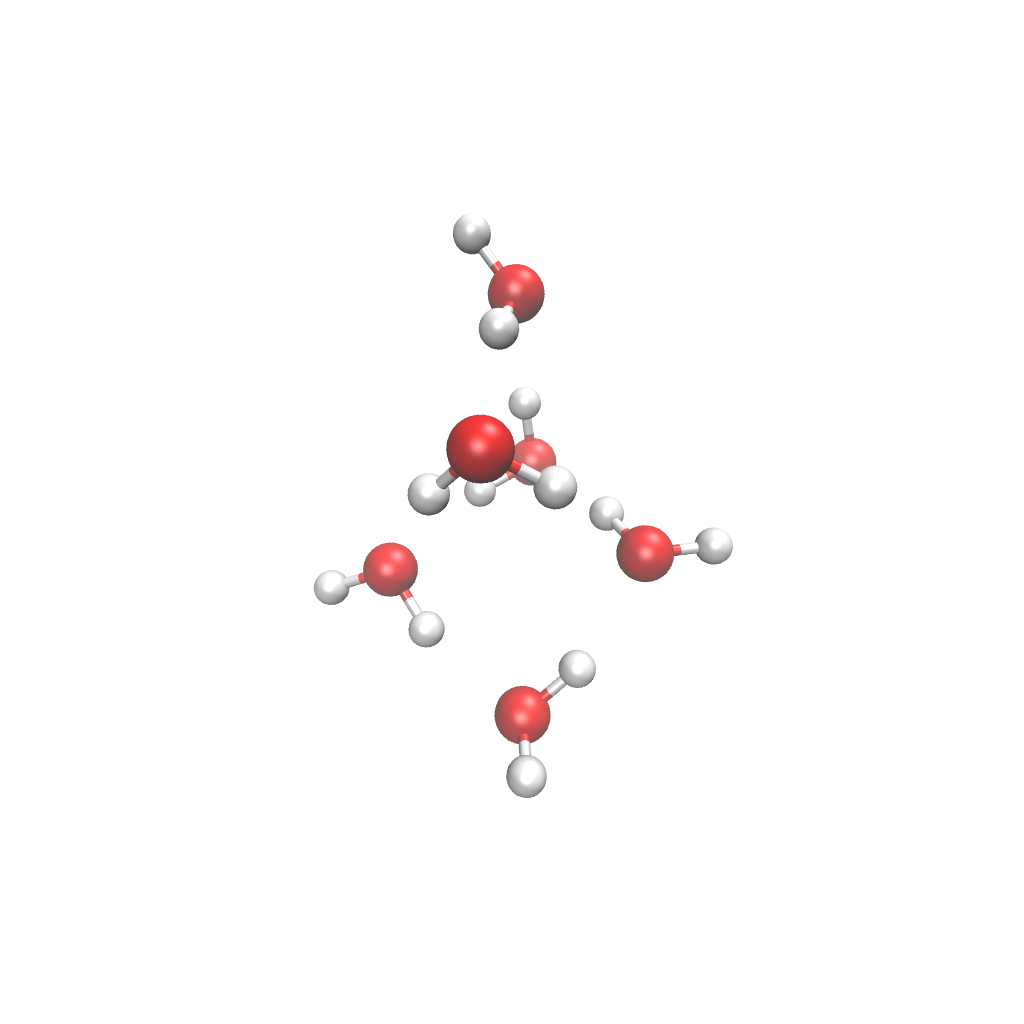} \\
(N=3)  & (N=4)  & (N=5) & (N=6)  \\[6pt]
  \includegraphics[width=35mm]{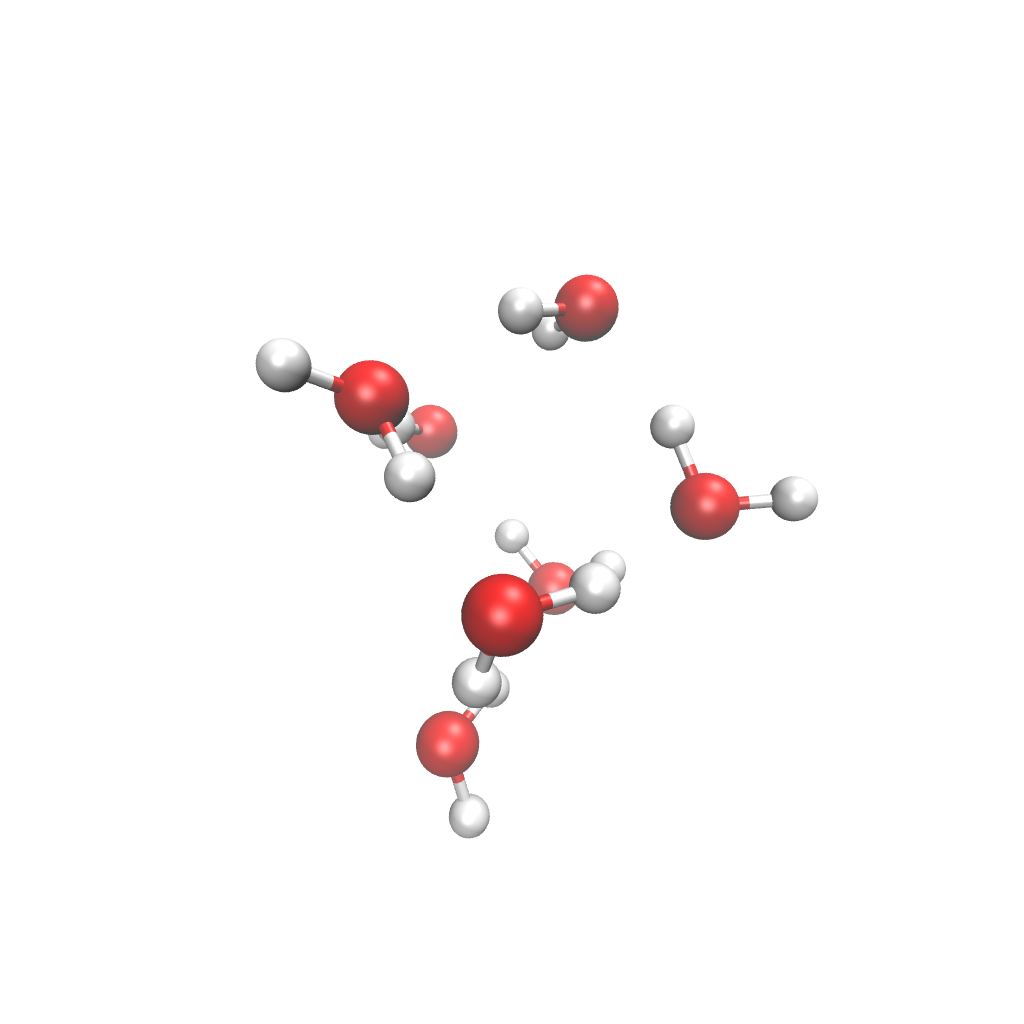} &   \includegraphics[width=35mm]{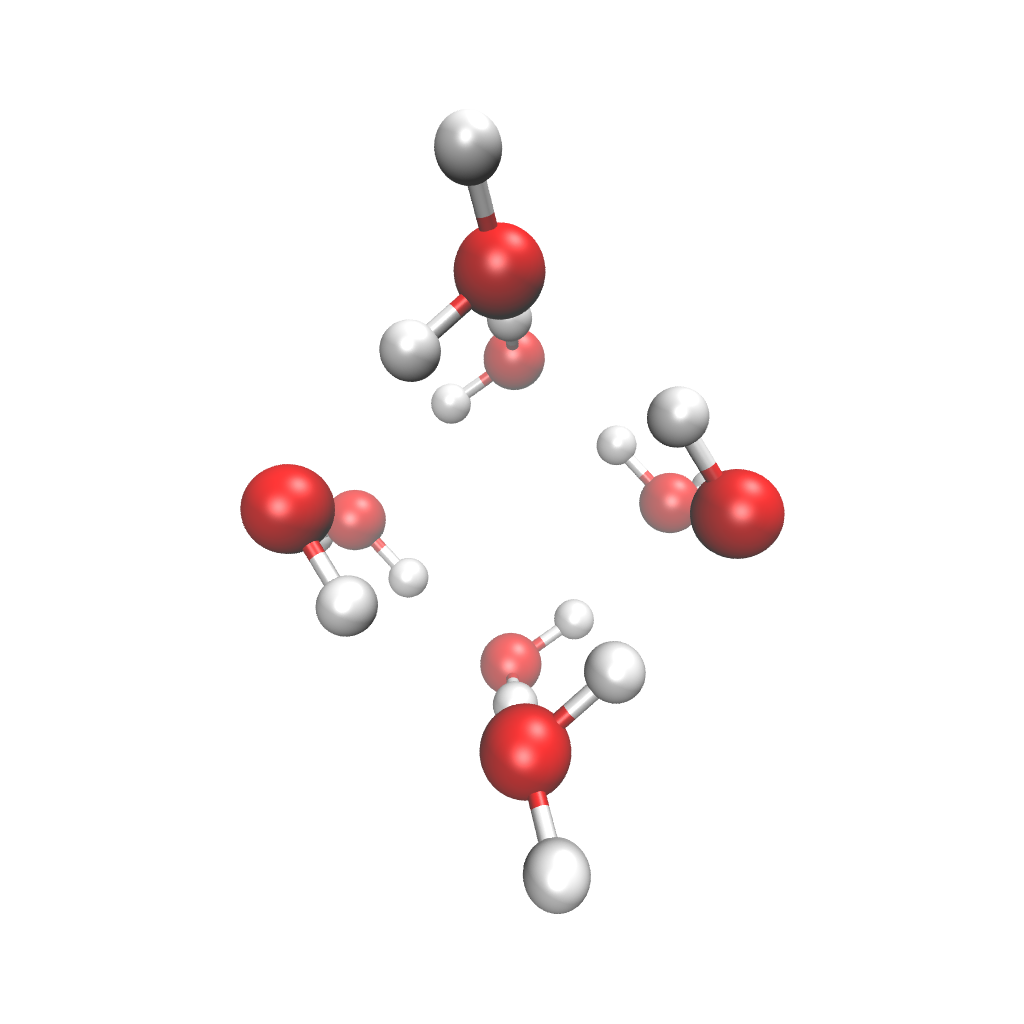} &   \includegraphics[width=35mm]{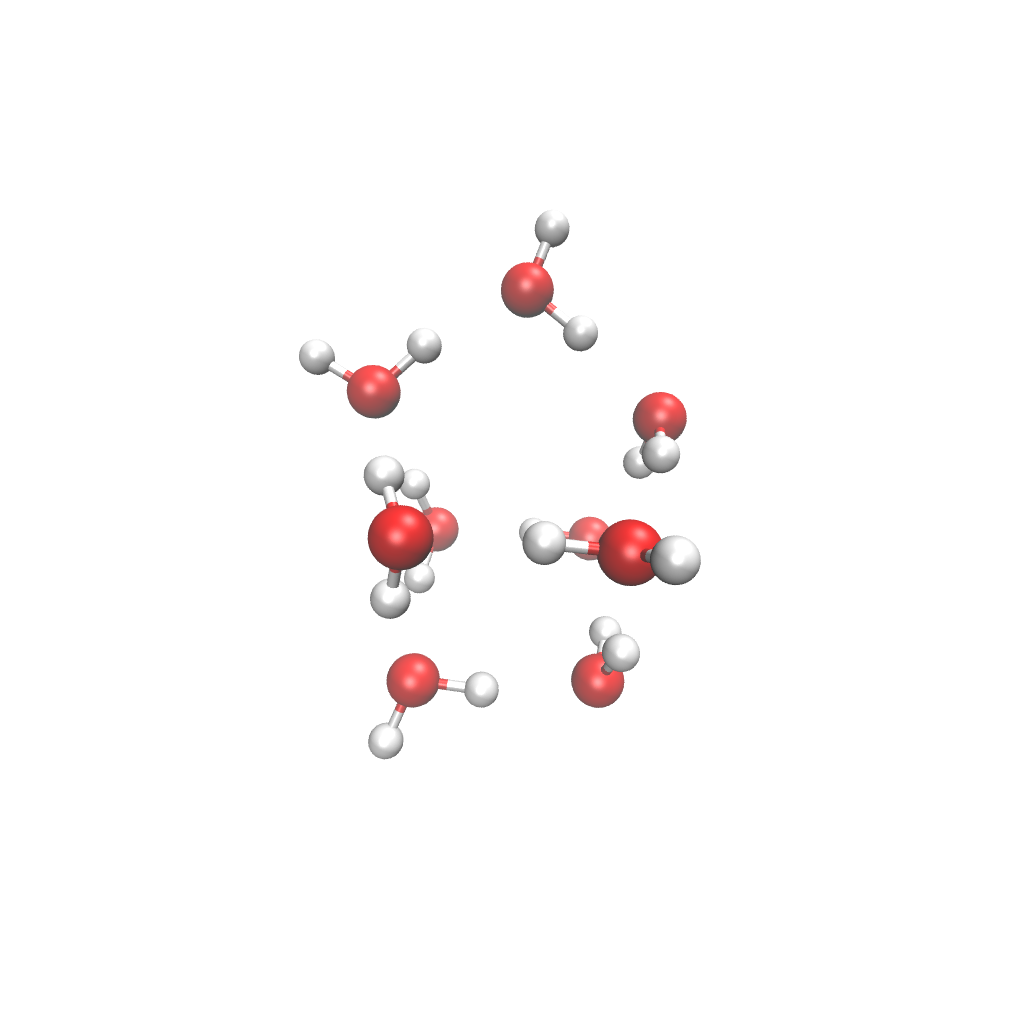} &   \includegraphics[width=35mm]{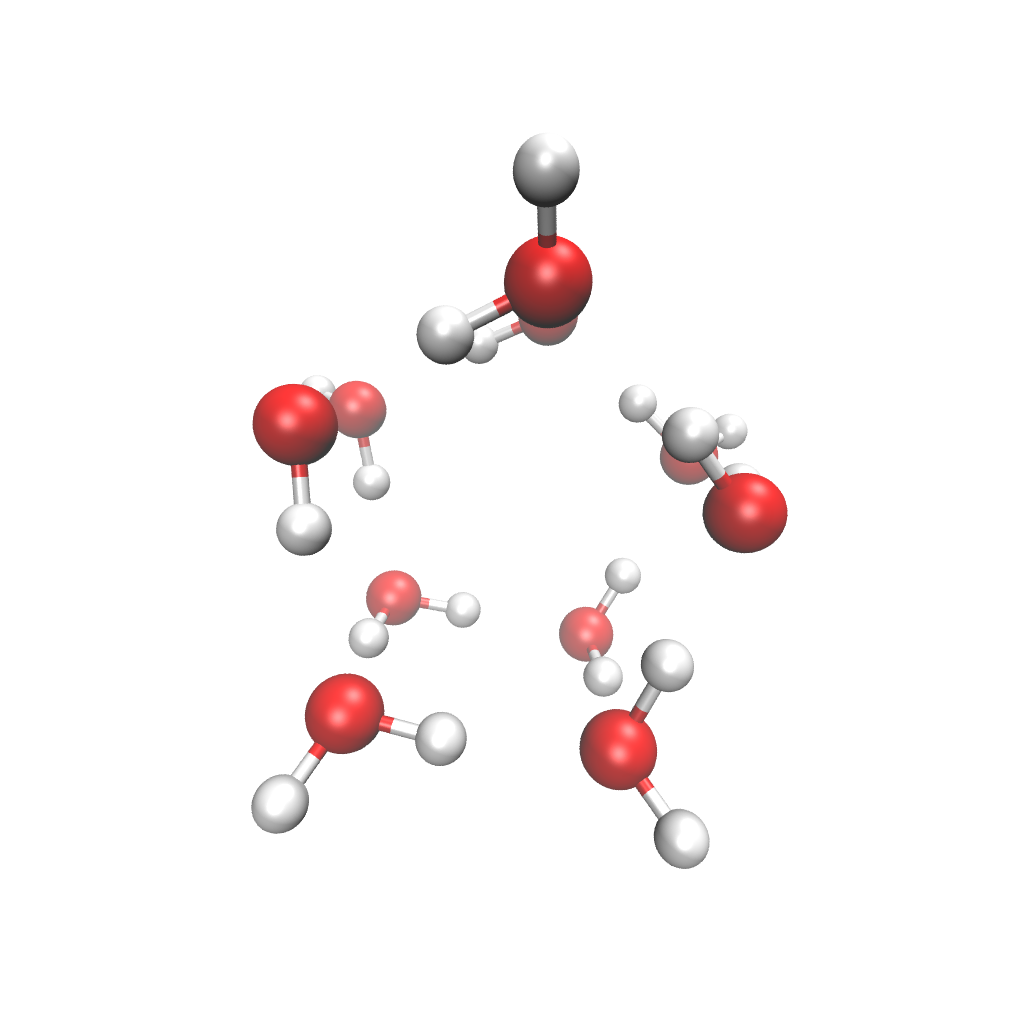} \\
(N=7)  & (N=8)  & (N=9) & (N=10)   \\[6pt]
  \includegraphics[width=35mm]{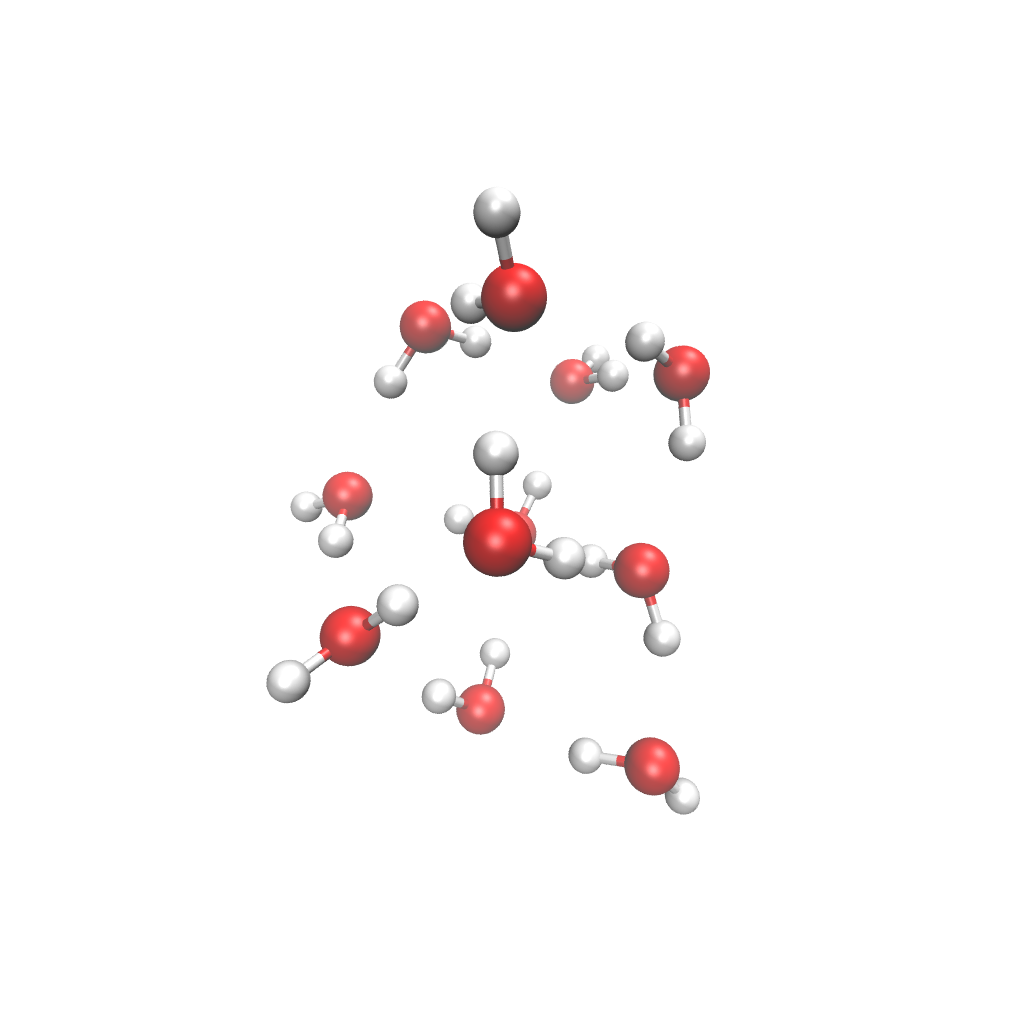} &   \includegraphics[width=35mm]{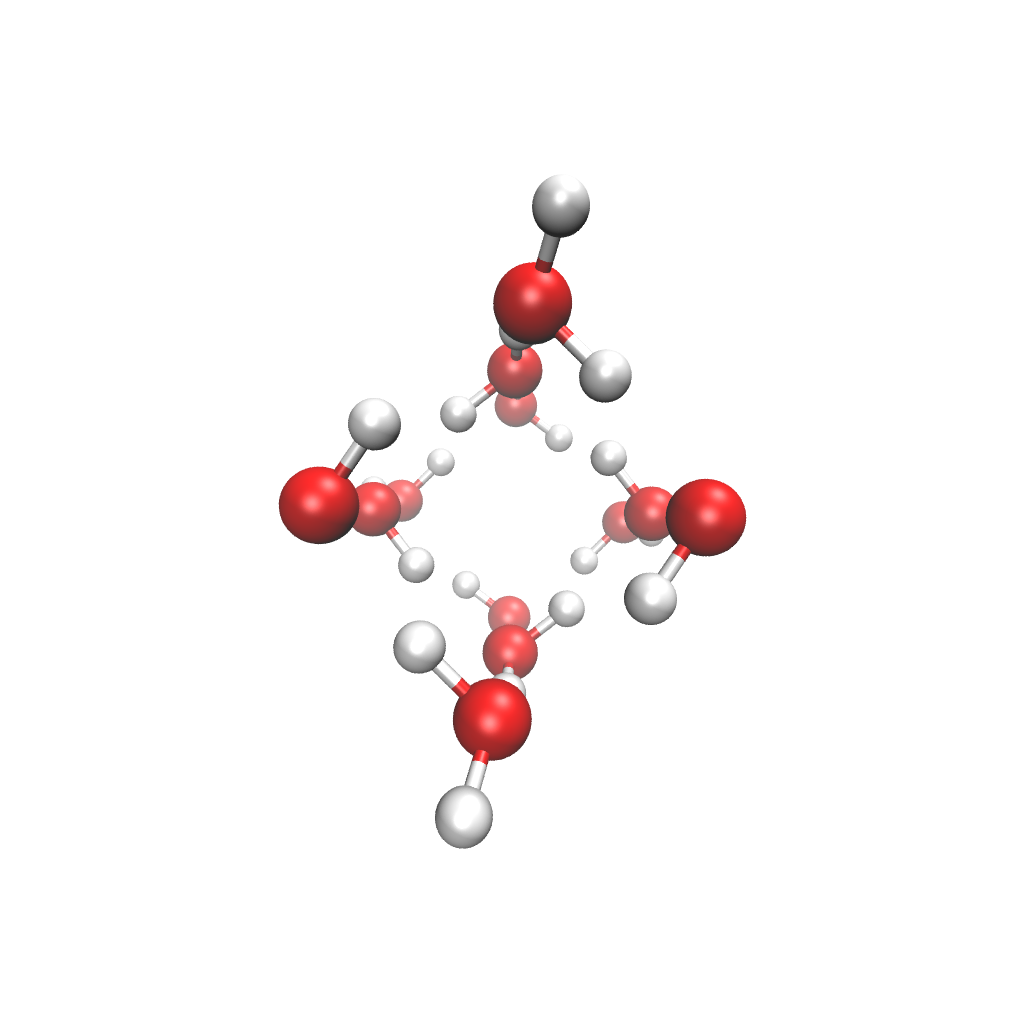} &   \includegraphics[width=35mm]{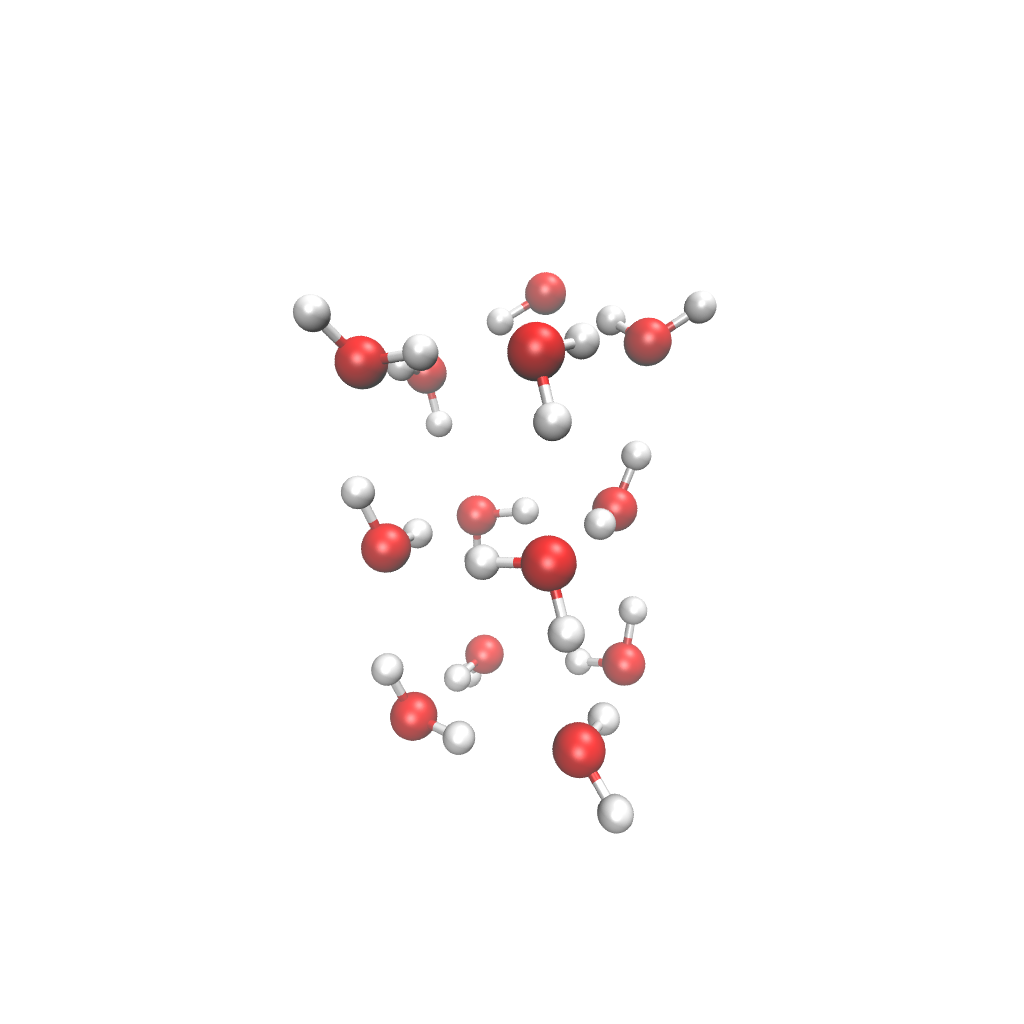} &   \includegraphics[width=35mm]{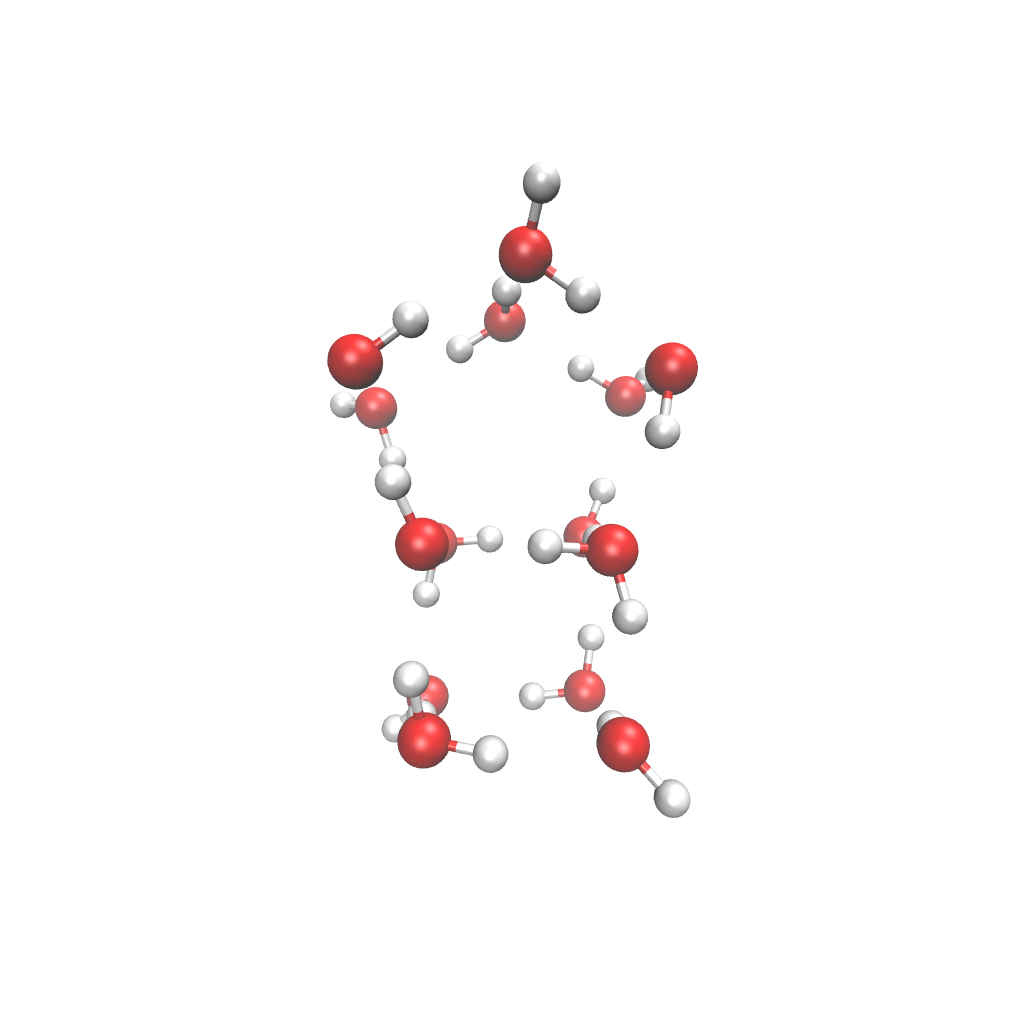} \\
(N=11)  & (N=12)  & (N=13) & (N=14)    \\[6pt]
  \includegraphics[width=35mm]{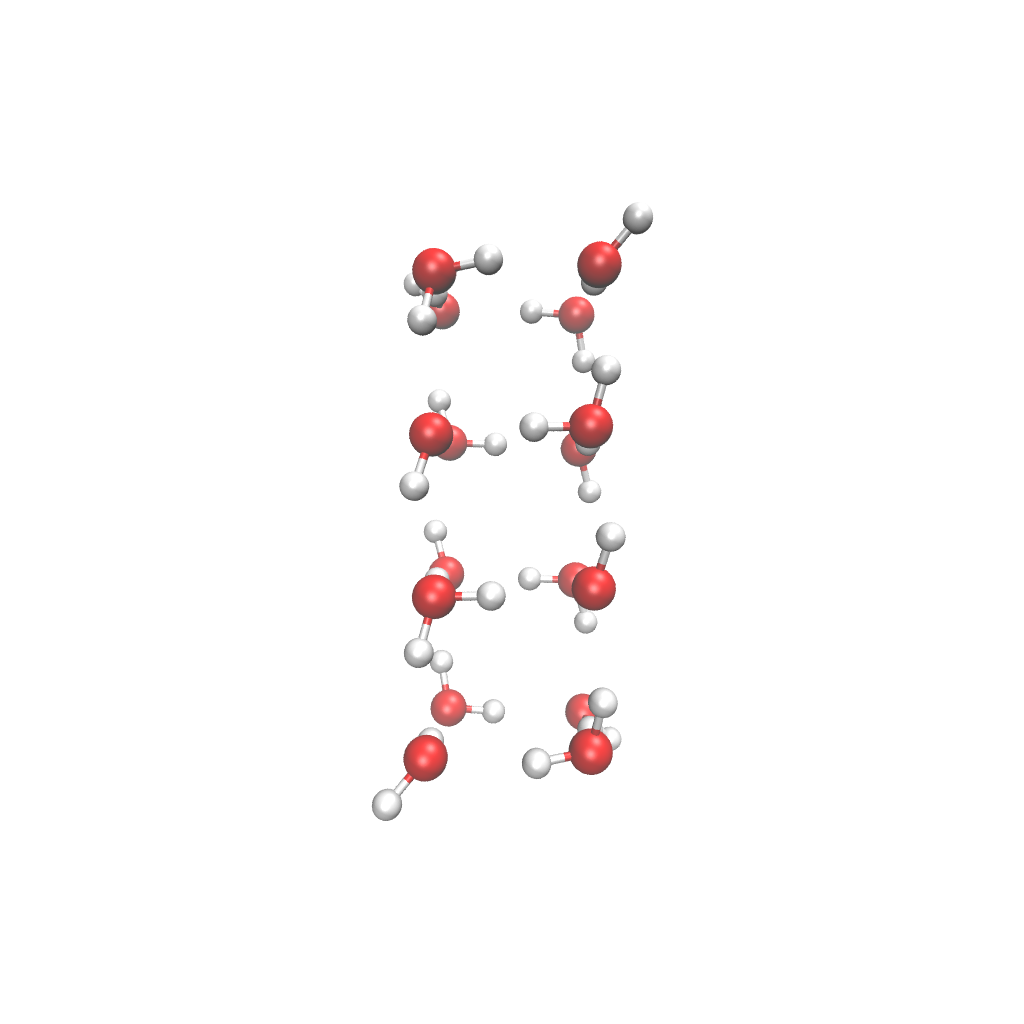} &   \includegraphics[width=35mm]{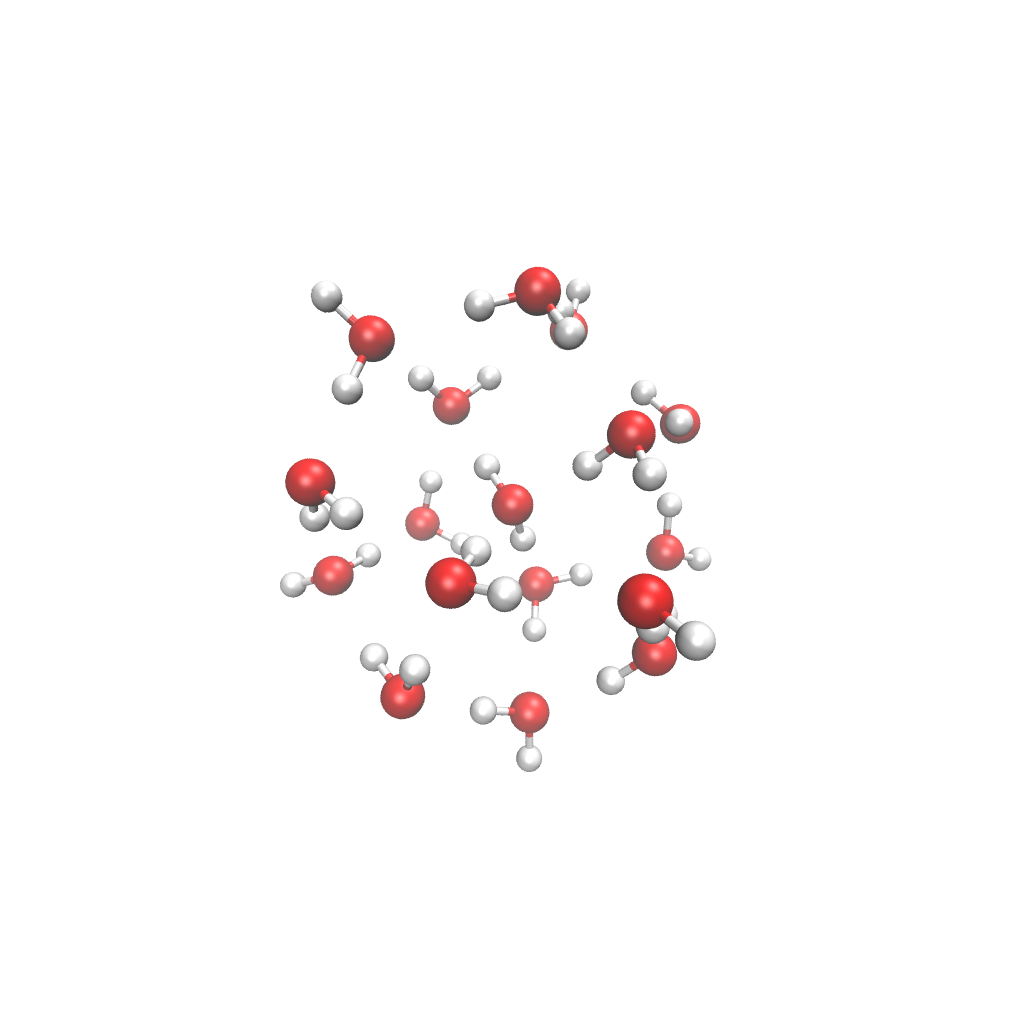} &   \includegraphics[width=35mm]{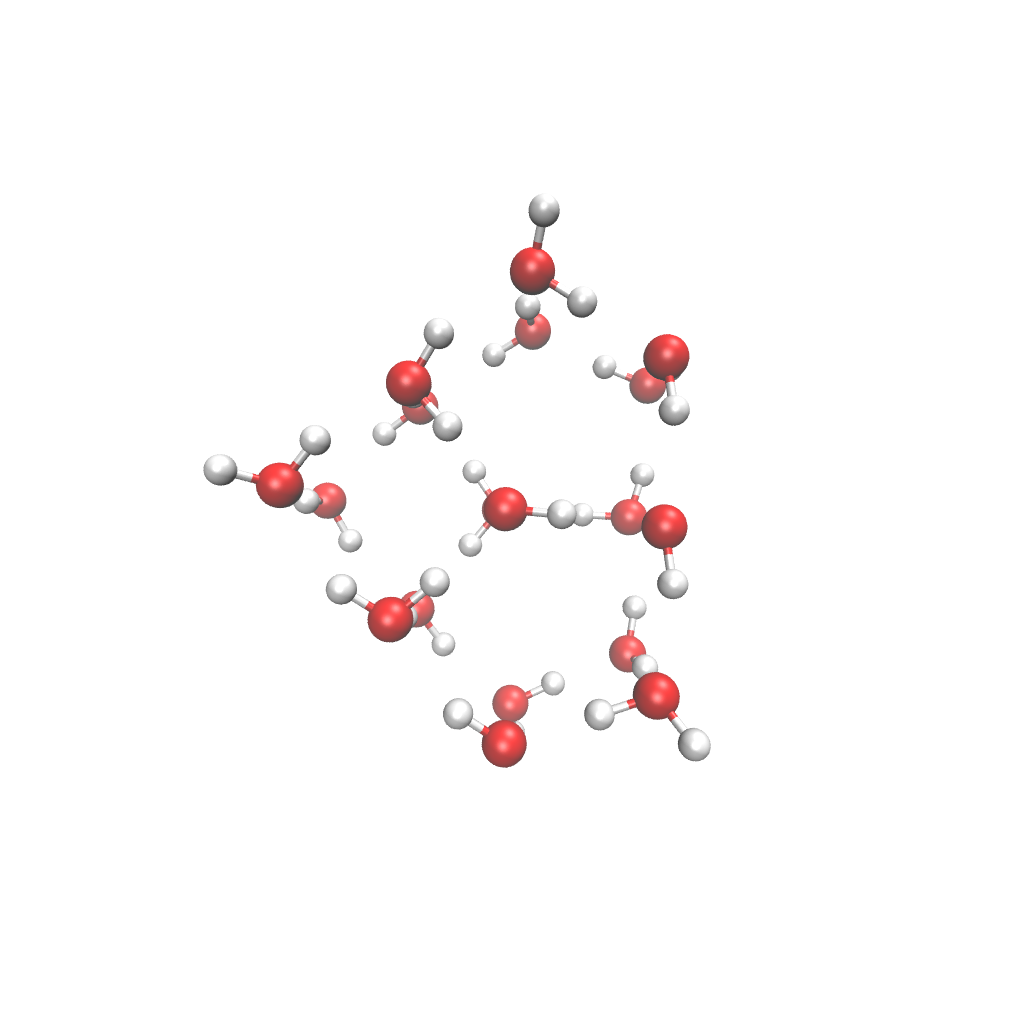} &   \includegraphics[width=35mm]{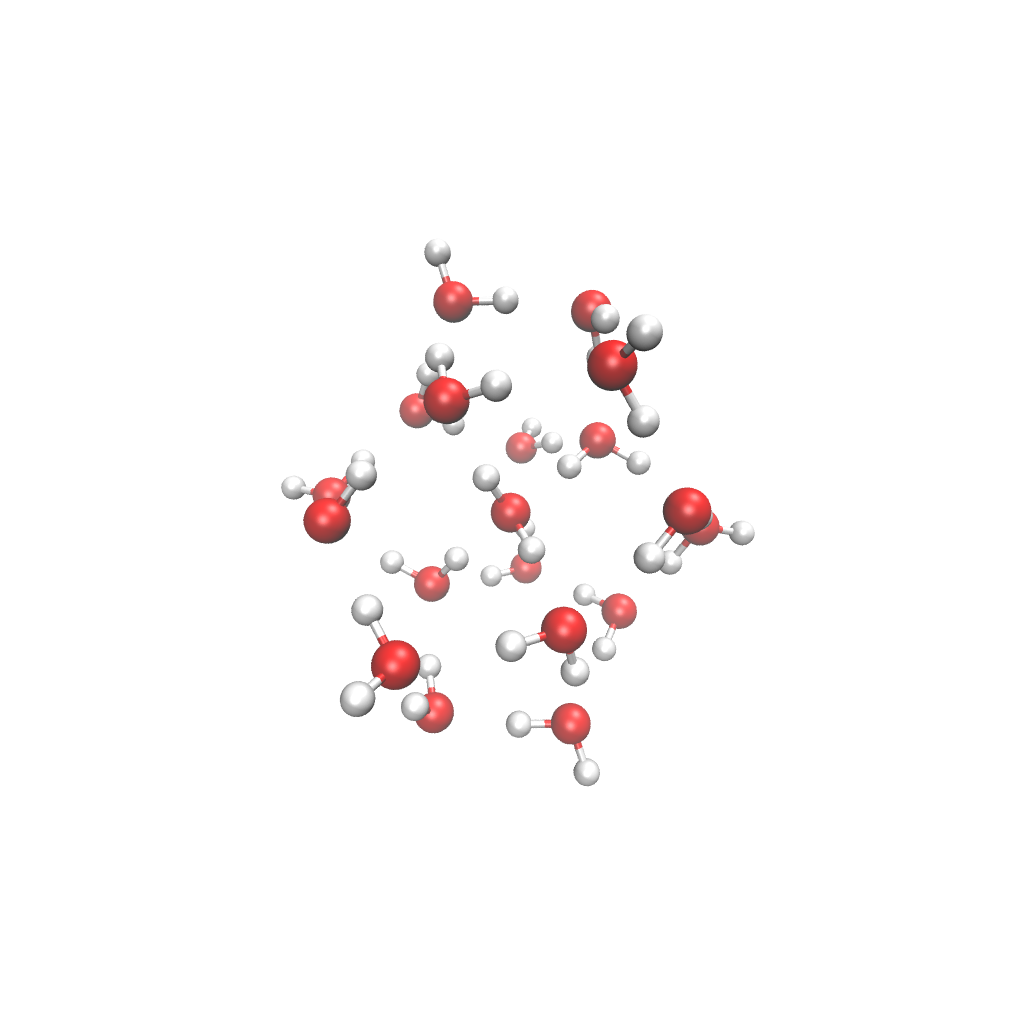} \\
(N=16)  & (N=17)  & (N=18) & (N=19)   \\[6pt]
  \includegraphics[width=35mm]{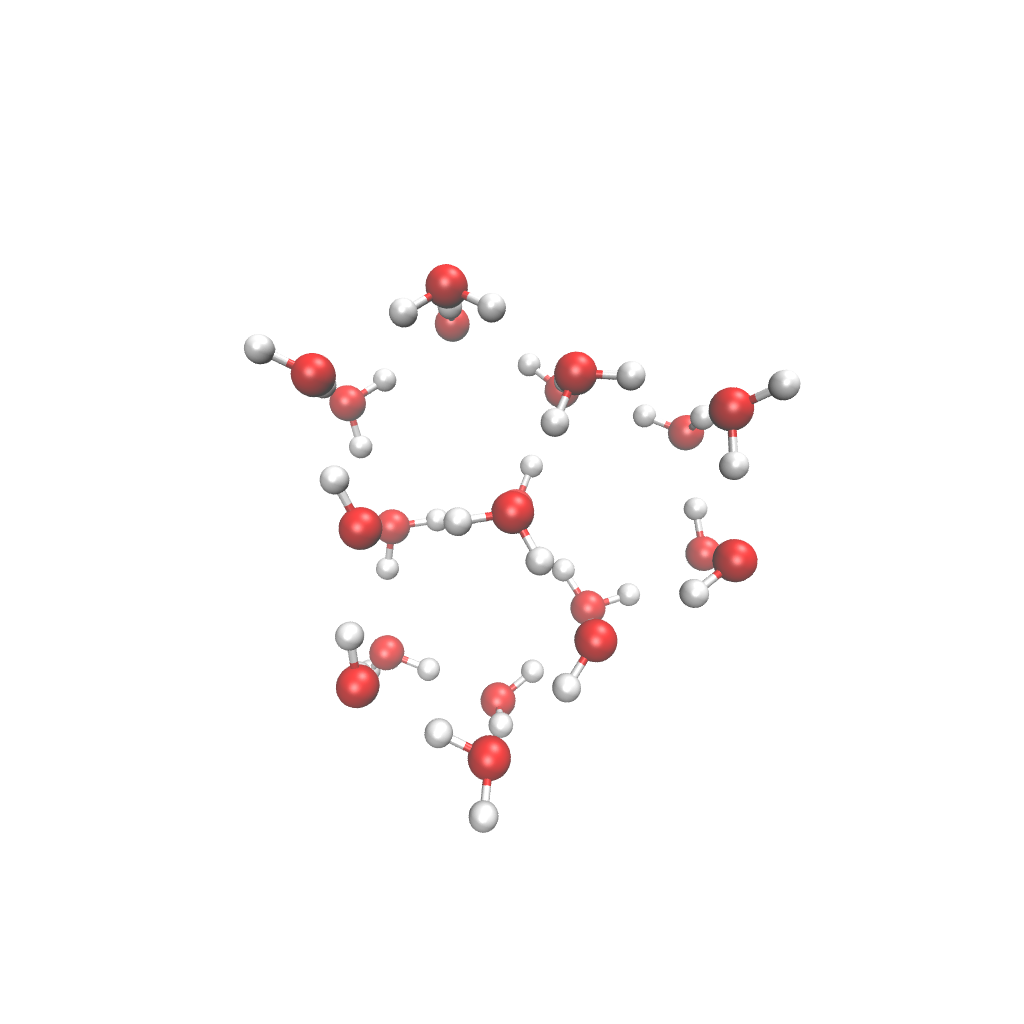} & &  & \\
 (N=20)  &   &  &
\end{tabular}
\caption{Structures of the water clusters used in Sec.~\ref{sec:numerics}. Red and white balls represent oxygen and hydrogen respectively.  Geometries from \cite{Rakshit:2019}. For visualization, the visual molecular dynamics  software package (VMD) was used \cite{HUMP96}.}
\label{fig:water_clusters}
\end{figure}

\subsection{Structure of individual force operators}
Any fermionic operator can be represented after a Jordan-Wigner transformation by an operator defined in qubit space, $A= c_{i}P_{i}$,
where $P_i$ is a $n$-qubit Pauli operator and $c_i$ its coefficient. In Fig.~\ref{fig:pauli_structure}, we show the distribution of the Pauli coefficients $c_i$ after a Jordan-Wigner transformation of the Hamiltonian and the force operators for a hydrogen chain with $6$ hydrogen and a water molecule. For both system, the operators were obtained from canonical Hartree-Fock orbitals in the STO-6G basis.

\subsection{Additional calculations using canonical molecular orbitals}
In Fig.~\ref{fig:sparse_cmo}, we provide additional results on $\lambda^\mathrm{(sparse)}$ for the Hamiltonian and the force operators using CMOs instead of LMOs. For both systems, we find a worse scaling compared to the scaling obtained by using LMOs. 
\begin{figure*}
    \centering
    \begin{minipage}{0.49\textwidth}
        \centering
        \includegraphics[width=1\textwidth]{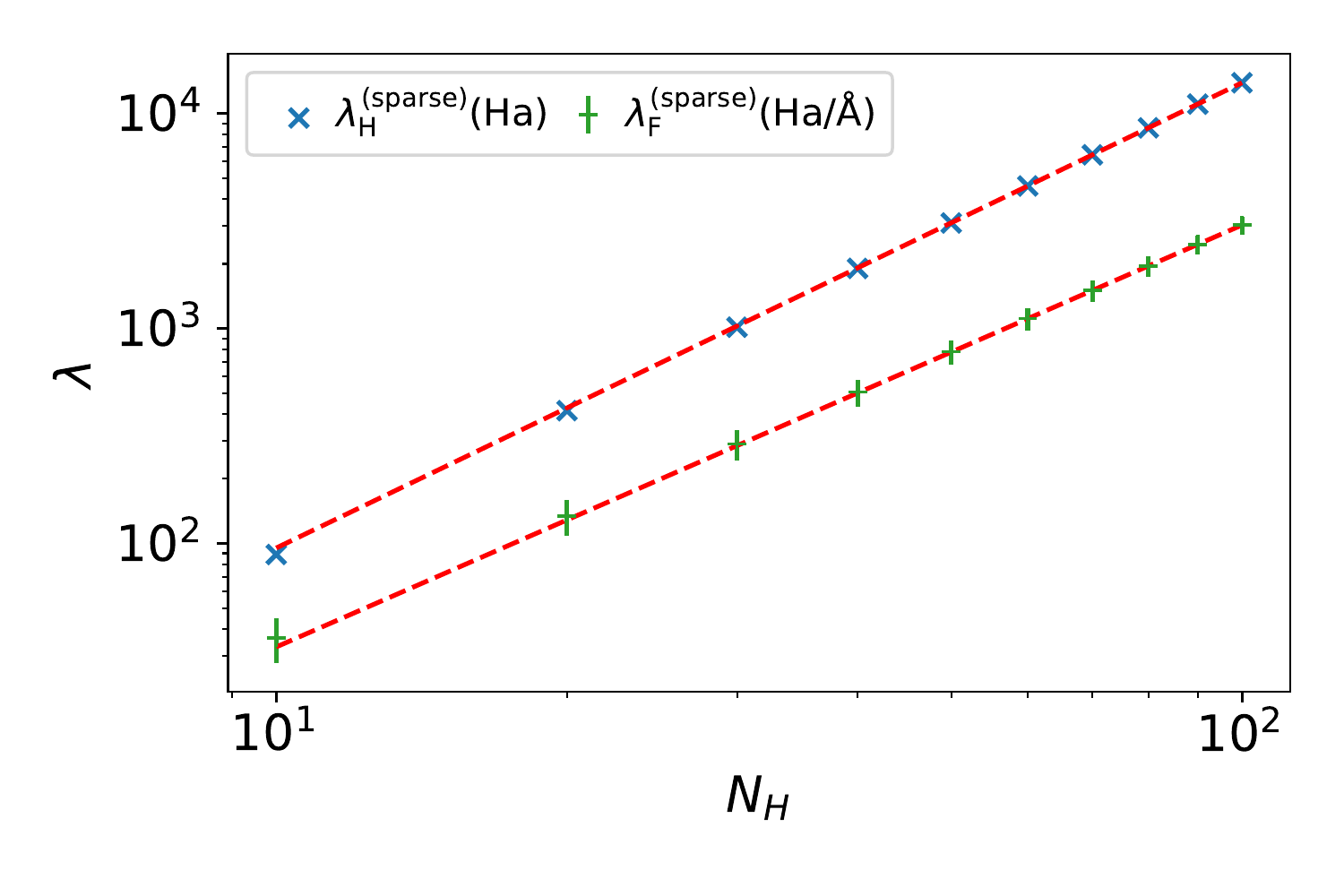}\llap{
  \parbox[b]{6.0in}{(a)\\\rule{0ex}{2.2in}
  }}
    \end{minipage}
    \begin{minipage}{0.49\textwidth}
        \centering
        \includegraphics[width=1\textwidth]{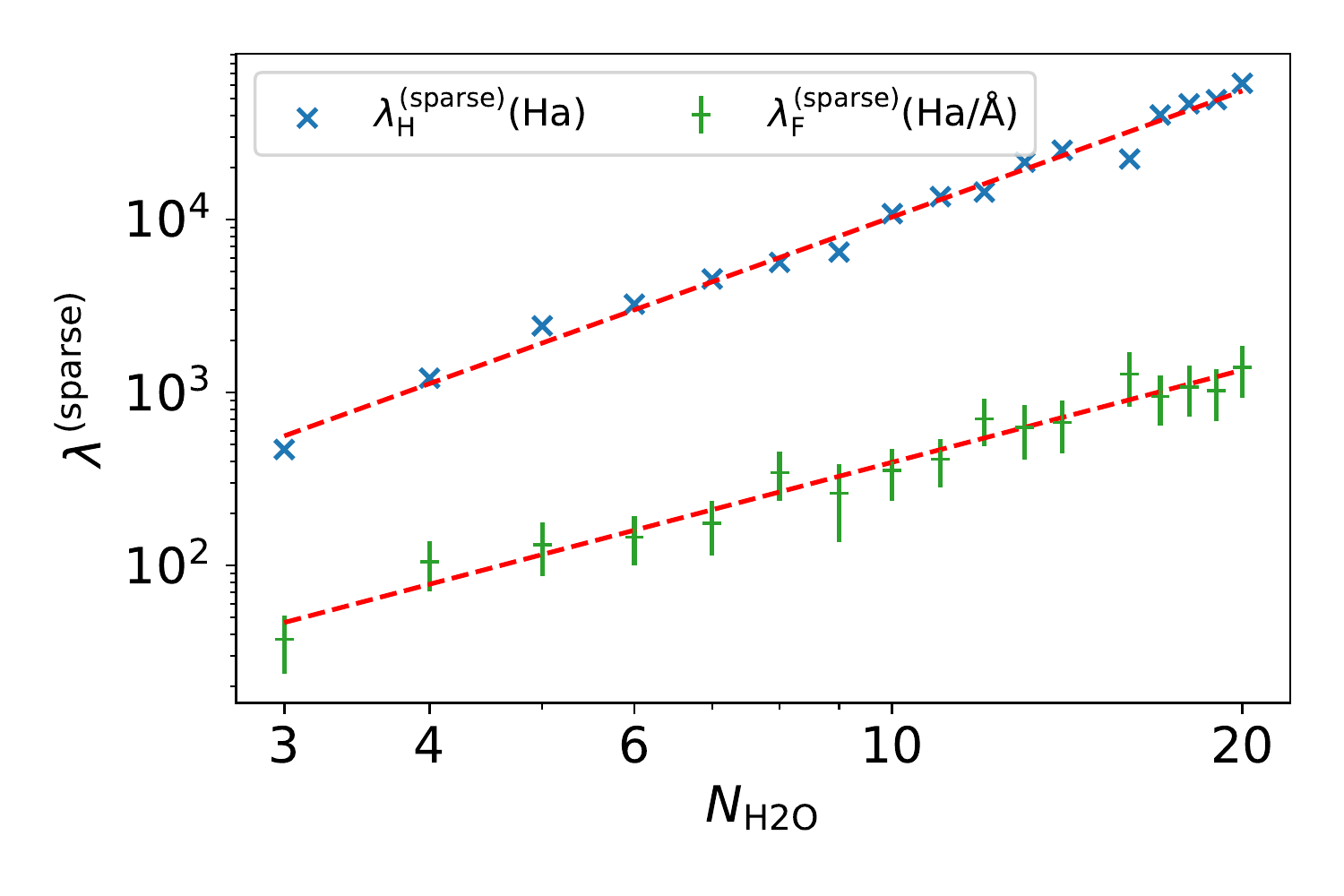}\llap{
  \parbox[b]{6.0in}{(b)\\\rule{0ex}{2.2in}
  }}
    \end{minipage}\hfill
\caption{Same plots as in Fig.~\ref{fig:lambdas_hydrogen_chain} and Fig.~\ref{fig:lambda_h4_and_water_cluster}(a) respectively, but using CMOs instead of LMOs. The fits suggest scalings of $\lambda^{\mathrm{(sparse)}}_{\mathrm{H}}=\mathcal{O}(N_H^{2.165\pm 0.001})$ and 
$\lambda^{\mathrm{(sparse)}}_{\mathrm{F}}=\mathcal{O}(N_H^{1.963\pm 0.001})$ for the hydrogen chain (a), and $\lambda^{\mathrm{(sparse)}}_{\mathrm{H}}=\mathcal{O}(N_{H2O}^{2.421\pm 0.065})$ and
$\lambda^{\mathrm{(sparse)}}_{\mathrm{F}}=\mathcal{O}(N_{H2O}^{1.772\pm 0.086})$ for the water clusters (b).}
\label{fig:sparse_cmo}
\end{figure*}
\begin{figure*}
    \centering
    \begin{minipage}{0.49\textwidth}
        \centering
        \includegraphics[width=1\textwidth]{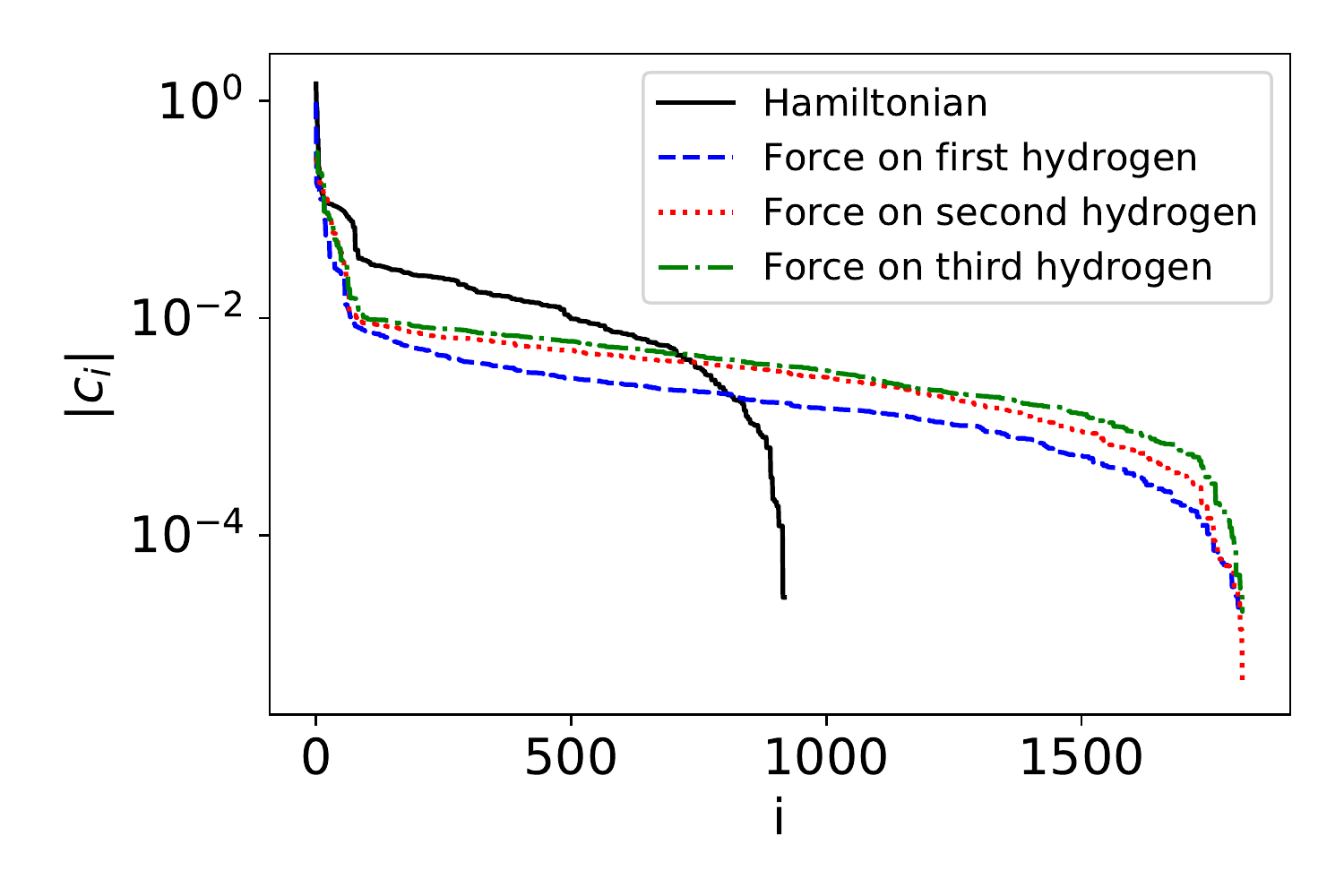}
    \end{minipage}
    \begin{minipage}{0.49\textwidth}
        \centering
        \includegraphics[width=1\textwidth]{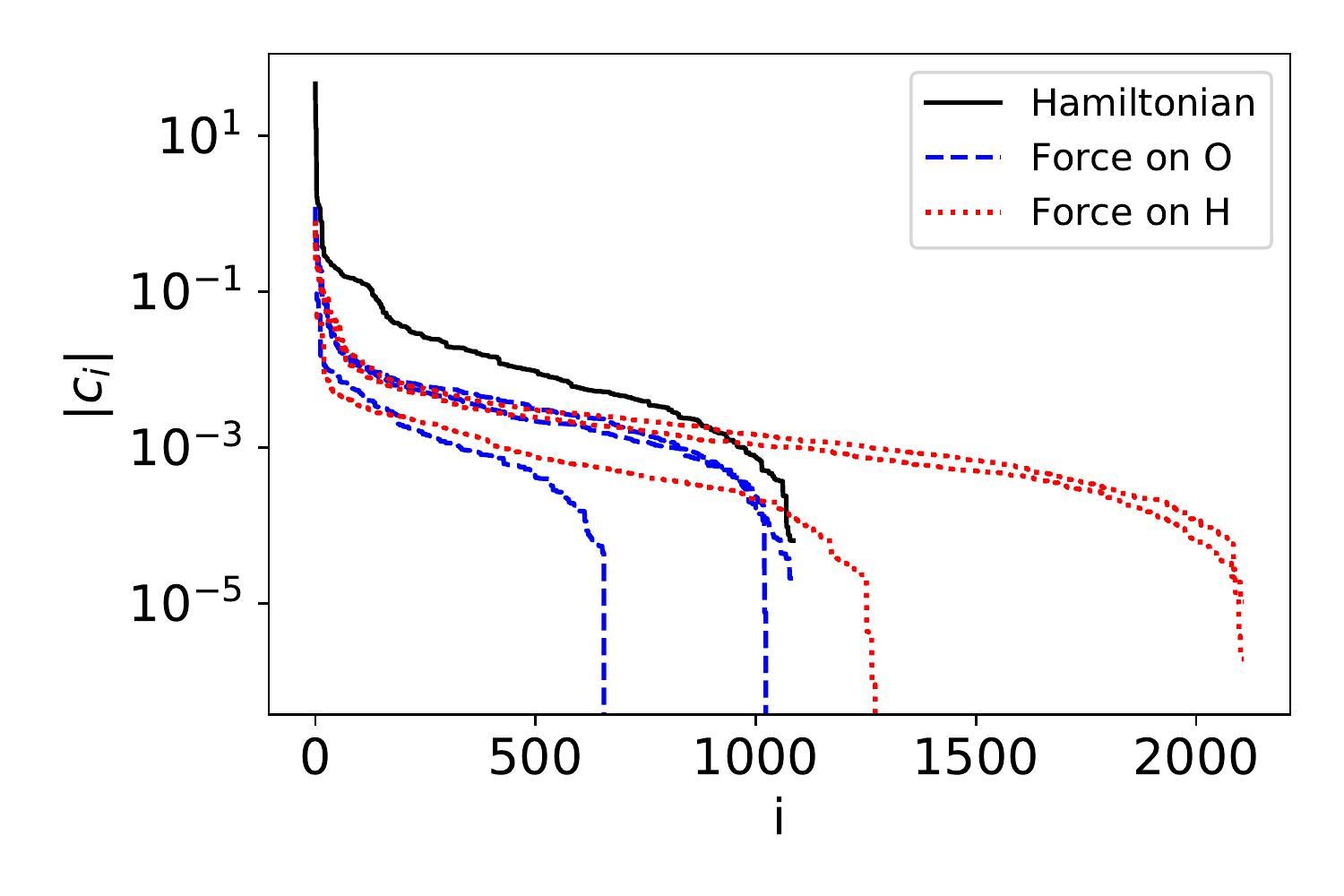}
    \end{minipage}\hfill
\caption{Distribution of the absolute values of Pauli weights $|c_i|$ of the Hamiltonian and force operators after a Jordan-Wigner transformation. For the calculations, canonical Hartree-Fock orbitals were used. (a) An one-dimensional chain of $6$ hydrogen atoms with spacing of $0.74084$\r{A}. Only forces in direction of the connecting line are shown. (b) For a single water molecule, with the geometry obtained from \cite{johnson2020nist}.}
\label{fig:pauli_structure}
\end{figure*}

\end{document}

%% file: BlockEncoding.tex
\subsection{General block encoding of the Hamiltonian and force operator}~\label{sec:block_encodings}

In the next section, we will present several methods to calculate energy gradients in a fault-tolerant quantum computing section.
These methods all require access to block encodings of the Hamiltonian and derivative operators.
Block encoding is a method of encoding a non-unitary operator on a quantum device; one block-encodes an operator $O$ by adding additional qubits to a system and finding a unitary $U$ that contains $O$ in the block corresponding to the $|0\rangle$ state of the additional qubits.
That is, if
\begin{equation}
    U=\left(\begin{array}{ccc}u_0&u_1&\ldots\\u_1^{\dag}&\ddots & \\\vdots & &\end{array}\right),
\end{equation}
where the subscript on $u_i$ denotes the computational basis state of the additional qubits, we require
\begin{equation}
    \left\|u_i-\frac{1}{\lambda_O}O\right\|_{\infty}\leq\epsilon.\label{eq:block_encoding_def}
\end{equation}
This encoding typically requires the rescaling factor $\lambda_O \gg 1$, as sub-blocks of a unitary matrix must have eigenvalues with norm $\leq 1$.
These rescaling factors typically then appear as multiplicative costs for implementing a given algorithm.
For example, one can calculate immediately given Eq.~\eqref{eq:block_encoding_def} that
\begin{equation}
    \left|\langle\psi|\otimes\langle 0|U|0\rangle\otimes|\psi\rangle - \frac{1}{\lambda_O}\langle\psi|O|\psi\rangle\right|\leq \epsilon.
\end{equation}
In this work, we use $\lambda_H$ and $\lambda_{\mathrm{F}_i}$ respectively to refer to the rescaling factors for a Hamiltonian and force operator (or $\lambda_{\mathrm{F}}$ to refer to the rescaling factor averaged across all forces).
The cost of implementing these block encodings as quantum circuits is stated as $T_H$ and $T_{F,i}$ respectively (with $T_F$ the cost averaged over all force components again).
In this section, we give circuit implementations of these block encodings, and the corresponding $T_H$, $T_{F,i}$, $\lambda_H$ and $\lambda_{F,i}$ costs.
These differ depending on the type of basis set used to solve the electronic structure problem, and whether the problem is solved in first or second quantization; we will give results for multiple commonly considered systems.

The most commonly used methods to block-encode an arbitrary operator $O$ are linear combination of unitaries (LCU) methods~\cite{Childs2012}.
We give a brief overview of this class of techniques here.
LCU methods involve writing $O$ as a linear combination of unitary operators
\begin{equation}
    O=\sum_ih_iP_i,\;\;\; P_i^{\dag}P_i=I,\;\;\; h_i\geq 0.\label{eq:LCU_sum}
\end{equation}
The condition that $h_i$ be real and positive can be accounted for by placing any complex phase on the operator $P_i$ (as $P_ie^{i\phi}$ is unitary whenever $P_i$ is unitary).
The LCU block encoding then requires two oracles; a PREPARE oracle that constructs the state
\begin{equation}
    \mathrm{PREPARE}|0\rangle=\frac{1}{\sqrt{\sum_ih_i}}\sum_i\sqrt{h_i}|i\rangle
\end{equation}
on an additional LCU register, and the SELECT oracle defined by
\begin{equation}
    \mathrm{SELECT}=\sum_i|i\rangle\langle i|\otimes P_i.
\end{equation}
One can then confirm that given states $|\psi_1\rangle,|\psi_2\rangle$ on the system register,
\begin{equation}
    \langle\psi_1|\otimes \langle 0|\mathrm{PREPARE}^{\dag}\,\mathrm{SELECT}\,\mathrm{PREPARE}|0\rangle\otimes |\psi_2\rangle = \frac{1}{\sum_ih_i}\langle\psi_1|O|\psi_2\rangle,
\end{equation}
as required for this to be a block encoding, with a rescaling factor
\begin{equation}
    \lambda = \sum_ih_i.
\end{equation}
One has a great degree of freedom in choosing the decomposition (Eq.~\eqref{eq:LCU_sum}) and constructing the PREPARE and SELECT oracles.
Optimizing these is a critical requirement to lower fault-tolerant quantum costs.

\subsection{Block encoding in second quantization with local basis sets}\label{sec:block_encoding_second_quantization}

The critical factor in block-encoding Hamiltonians or derivative operators with a local basis is the requirement to transfer all of the data about the operator onto the device.
A generic two-body fermionic operator requires $O(N^4)$ pieces of information to describe, and in the absence of additional structure this bounds the cost of block-encoding below.
Luckily, the electronic structure problem has additional structure that can be used to reduce this cost.
As discussed briefly in Sec.~\ref{sec:basis_rotations}, the two-body electron tensor may be factorized as the product of lower-order tensors that may be expressed with fewer indices.
Two of the most well-known factorization methods are low-rank factorization~\cite{Berry19Qubitization} and tensor hypercontraction~\cite{lee2021even}.
Both of these methods incur additional truncation errors, but numerically these have been shown to converge at a much lower cost than the $O(N^4)$ bound above.

Tensor hypercontraction holds the current record as the best scaling method for block encoding Hamiltonians.
We can write the THC form of the Coulomb operator as
\begin{equation}
\label{eq:thc}
V \approx G = \frac{1}{2} \sum_{\alpha, \beta \in \{\uparrow, \downarrow\}} \sum_{p,q,r,s=1}^{N/2} G_{pqrs} a^\dagger_{p,\alpha} a_{q,\alpha} a^\dagger_{r,\beta} a_{s,\beta}
\qquad \qquad
G_{pqrs} =
\sum_{\mu, \nu = 1}^{M} \chi_{p}^{(\mu)} \chi_{q}^{(\mu)} Z_{\mu\nu} \chi_{r}^{(\nu)}  \chi_{s}^{(\nu)} 
\end{equation}
where $\chi$ is an $M \times N$ dimension matrix and $\zeta$ is an $M \times M$ dimension matrix.
Crucially, for the electronic structure problem it is possible to truncate $M = {\cal O}(N \textrm{polylog}(1/\epsilon_{\mathrm{THC}}))$ while achieving a truncation error $\epsilon_{\mathrm{THC}}$.
Ref.~\cite{lee2021even} demonstrated a method for block encoding $H / \lambda_H$ using a quantum walk with gate complexity $T_H\in\widetilde{\cal O}(M) = \widetilde{\cal O}(N)$.

We now argue that the same method can be used to block encode the force operator, with the same scaling. One can represent the derivative of the THC Hamiltonian with respect to nuclear coordinates as
\begin{align}
\tfrac{d}{dR_i}G_{pqrs}
&=
\sum_{\mu, \nu = 1}^{M} (\tfrac{d}{dR_i}\chi_{p}^{(\mu)}) \chi_{q}^{(\mu)} \zeta_{\mu \nu} \chi_{r}^{(\nu)}  \chi_{s}^{(\nu)} +
(\tfrac{d}{dR_i}\chi_{q}^{(\mu)}) \chi_{p}^{(\mu)} \zeta_{\mu\nu} \chi_{r}^{(\nu)}  \chi_{s}^{(\nu)}\\ 
&+
\chi_{p}^{(\mu)} \chi_{q}^{(\mu)} \zeta_{\mu\nu} (\tfrac{d}{dR_i}\chi_{r}^{(\nu)})  \chi_{s}^{(\nu)} 
+
\chi_{p}^{(\mu)} \chi_{q}^{(\mu)} \zeta_{\mu\nu} (\tfrac{d}{dR_i}\chi_{s}^{(\nu)})  \chi_{r}^{(\nu)} 
+
\chi_{p}^{(\mu)} \chi_{q}^{(\mu)} 
(\tfrac{d}{dR_i}\zeta_{\mu\nu})
\chi_{r}^{(\nu)}  \chi_{s}^{(\nu)},\nonumber
\end{align}
We see that the dimension of these operators is not increased however there are now four types of tensors in the expression ($\chi$, $\zeta$, $\tfrac{d\chi}{dR_i}$ and $\tfrac{d\zeta}{dR_i}$) rather than just two ($\chi$, $\zeta$). However, in order to pursue the same strategy for realizing the block encoding that is employed by \cite{lee2021even}, it is necessary to yet define two more types of tensor $\chi^{(+)} = \frac{d\chi}{d \mathbf{R}} + \chi$ and $\chi^{(-)} = \frac{d\chi}{dR_i} - \chi$. With these definitions we can rewrite the expression as
\begin{align}
\tfrac{d}{dR_i}G_{pqrs} & =\frac{1}{2}\sum_{\mu,\nu=1}^M (\chi_p^{(\mu)})^{(+)}(\chi_q^{(\mu)})^{(+)}\zeta_{\mu\nu}\chi_r^{(\nu)}\chi_s^{(\nu)} - (\chi_p^{(\mu)})^{(-)}(\chi_q^{(\mu)})^{(-)}\zeta_{\mu\nu}\chi_r^{(\nu)}\chi_s^{(\nu)}\\
&+\chi_p^{(\mu)}\chi_q^{(\mu)}\zeta_{\mu\nu}(\chi_r^{(\nu)})^{(+)}(\chi_s^{(\nu)})^{(+)} - \chi_p^{(\mu)}\chi_q^{(\mu)}\zeta_{\mu\nu}(\chi_r^{(\nu)})^{(-)}(\chi_s^{(\nu)})^{(-)} + 2 \chi_{p}^{(\mu)} \chi_{q}^{(\mu)} 
(\frac{d}{dR_i}\zeta_{\mu\nu})
\chi_{r}^{(\nu)}  \chi_{s}^{(\nu)}. \nonumber
\end{align}
We see that this expression also involves five tensors, ($\chi$, $\zeta$, $\chi^{(+)}$, $\chi^{(-)}$, $\frac{d}{dR_i}\zeta$). There are now five terms instead of one but we see that within each of these terms, the two tensors with the $\mu$ index are always the same and the two tensors with the $\nu$ index are always the same. This is a critical requirement for the following step so that creation and annihilation operators transform in the same way. Associated with each of the three tensors related to $\chi$, is the following projection of the fermionic ladder operators into three different larger auxiliary bases:
\begin{align}
\label{eq:projection}
c^\dagger_{\mu, \sigma} & = \sum_{p=1}^{N/2} \chi_{p}^{(\mu)} a^\dagger_{p, \sigma}
\qquad \qquad \qquad
c_{\mu, \sigma} = \sum_{p=1}^{N/2} \chi_{p}^{(\mu)} a_{p, \sigma} \\
c^{\dagger,(+)}_{\mu, \sigma} & = \sum_{p=1}^{N/2} (\chi_{p}^{(\mu)})^{(+)} a^\dagger_{p, \sigma}
\qquad \qquad
c_{\mu, \sigma}^{(+)} = \sum_{p=1}^{N/2} (\chi_{p}^{(\mu)})^{(+)} a_{p, \sigma} \\
c^{\dagger,(-)}_{\mu, \sigma} & = \sum_{p=1}^{N/2} (\chi_{p}^{(\mu)})^{(-)} a^\dagger_{p, \sigma}
\qquad \qquad
c_{\mu, \sigma}^{(-)} = \sum_{p=1}^{N/2} (\chi_{p}^{(\mu)})^{(-)} a_{p, \sigma}
\end{align}
where these creation and annihilation operators act on a larger space of $2 M$ spin-orbitals rather than $N$ spin-orbitals.
Without loss of generality, we are taking the $\chi^{(\mu)}$, $(\chi^{(\mu)})^{(+)}$ and $(\chi^{(\mu)})^{(-)}$ to be normalized vectors for each $\mu$ (because constant factors can be absorbed into $\zeta$). Using this, we rewrite the full force operator as
\begin{align}
\frac{d}{dR_i}G &=
\frac{1}{4} \sum_{\alpha, \beta \in \{\uparrow, \downarrow\}} 
\sum_{\mu, \nu = 1}^{M}  \left(\left(
n_{\mu,\alpha}^{(+)}
n_{\nu,\alpha}
- n_{\mu,\alpha}^{(-)}
n_{\nu,\alpha}
+ n_{\mu,\alpha}
n_{\nu,\alpha}^{(+)}
- n_{\mu,\alpha}
n_{\nu,\alpha}^{(-)}\right)  \zeta_{\mu\nu} + 2 n_{\mu,\alpha}
n_{\nu,\alpha} \frac{d}{dR_i}\zeta_{\mu\nu}\right)
\end{align}
where
\begin{align}
n_{\mu,\alpha} = c^{\dagger}_{\mu,\alpha} c_{\mu,\alpha}
\qquad \qquad
n_{\mu,\alpha}^{(+)} = c^{\dagger,(+)}_{\mu,\alpha} c_{\mu,\alpha}^{(+)}
\qquad \qquad
n_{\mu,\alpha}^{(-)} = c^{\dagger,(-)}_{\mu,\alpha} c_{\mu,\alpha}^{(-)} \, 
\end{align}
are the number operators in the auxillary bases defined earlier. Having reduced the Hamiltonian to this diagonal form, it is now clear that one can use the same block encoding strategy as that defined in \cite{lee2021even} with minimal additional modification. In particular, because there are now five different terms, we will need several additional ancilla bits that flag which term we mean to block encode. We will need a slightly larger QROM because we now need to access five tensors rather than just four. However, the $\zeta$ tensors which are dimension $M \times M$ tend to dominate the cost since the $\chi$ tensors only have dimension $M \times N$. Furthermore, the Toffoli cost of the QROM increases only as the square root of the database size. Thus, the quantum walk will have a Toffoli cost that is increased by roughly a factor of $T_F\sim \sqrt{2}T_H$.
The rescaling factor should now be defined as
\begin{equation}\label{eq:lambdaF_THC}
\lambda_{\mathrm{F}_i} = \sum_{\mu,\nu} \left( \left| \zeta_{\mu\nu}\right | + \frac{1}{2} \left| \frac{d}{dR_i}\zeta_{\mu\nu} \right | \right).
\end{equation}
Using this method, $\lambda_{\mathrm{F}_i}$ will be larger than the original $\lambda_H$ by at least a factor of two. We do not believe $\lambda_{\mathrm{F}_i}$ will be asymptotically larger than $\lambda_H$ because we would expect the $\frac{d}{dR_i}\zeta$ values to generally be smaller than the $\zeta$ values. Thus, we have shown that the block encodings of Ref.~\cite{lee2021even} also work for the force operator, with the same asymptotic complexity.

Our expectation is that this procedure would be suboptimal and that a better approach would be to use \emph{exactly} the same procedure as \cite{lee2021even} in order to block encode the force operator.
By this, we mean that one could apply THC directly to the coefficients $G_{pqrs}$, rather than take the derivatives of the THC tensors.
Then, one ends up with a single $\chi$ and a single $\zeta$ which compress the force operator, rather than a mix of the $\chi$ and $\zeta$ compressing the original Hamiltonian and their derivatives.
Doing this will likely lead to even smaller values of $\lambda_{\mathrm{F}_i}$ than for the Hamiltonian simulation (perhaps even asymptotically so).
However, it is more difficult to show analytically that $M$ would be as small or smaller than the original $M$.
We expect this to be the case due to the sparseness of the derivative operator; studying this numerically (or analytically) would be a clear target for future work.
One additional complication that presents itself here is that the truncation error from factorizing the derivative operator is different to the truncation error in the Hamiltonian, which implies that the forces estimated by this approach do not quite match the energy manifold of the THC-factorized Hamiltonian.
(This error is avoided when differentiating the THC-factorized Hamiltonian, as the truncation here is identical.)
This implies that a lower tolerance for the truncation error on both the Hamiltonian and force operator may be necessary to make the systematic error here negligible; this would need to be properly bounded in any future work.

Although the THC algorithm is expected to be more efficient, a very similar argument can be made for the low rank block encodings of Ref.~\cite{Berry19Qubitization}.
The form of the low-rank Hamiltonian is given in Sec.~\ref{sec:basis_rotations}, in Eq.~\eqref{eq:double_low_rank_factorization}.
Ref.~\cite{Berry19Qubitization} shows that one can use this form of the Hamiltonian in conjunction with qubitization \cite{Low2016} to block encode $H / \lambda_H$ with complexity scaling as $T_H\sim {\cal O}(\sqrt{N L})$ and in this case 
\begin{equation}
\lambda_H = \sum_{\ell=1}^L \sum_{p,q=1}^{M_\ell} \left| f_p^{(\ell)} f_q^{(\ell)} \right| \, =\sum_{l=1}^L\Bigg[\sum_{p=1}^{M_l}f_p^{(l)}\Bigg]^2
\end{equation}
Let us now extend this factorization to the derivative.
We define $V$ as in Eq.~\eqref{eq:double_low_rank_factorization}; $V=\sum_{l=1}^LW^{(l)}W^{(l)\dag}$.
Then, we can write the derivative as a difference of squares:
\begin{align}
    \frac{dV}{dR_i}&=\sum_{l=1}^L\bigg[\frac{dW^{(l)}}{dR_i}W^{(l)\dag}+W^{(l)}\frac{dW^{(l)\dag}}{dR_i}\bigg]\\
    &=\frac{1}{2}\sum_{l=1}^L\Bigg\{\bigg[W^{(l)}+\frac{dW^{(l)}}{dR_i}\bigg]\bigg[W^{(l)}+\frac{dW^{(l)}}{dR_i}\bigg]^{\dag} - \bigg[W^{(l)}-\frac{dW^{(l)}}{dR_i}\bigg]\bigg[W^{(l)}-\frac{dW^{(l)}}{dR_i}\bigg]^{\dag}\Bigg\}.
\end{align}
The $W^{(l)}$ (Eq.~\ref{eq:single_factorized_derivative}) and their derivatives are one-body operators, so their sum is a one-body operator and can be block diagonalized following Ref.~\cite{Berry19Qubitization}.
This gives a constant factor overhead and requires a single additional qubit to flag between two choices of basis rotations (in place of Eq.~\eqref{eq:basis_change_def})
\begin{equation}
    U^{(l+)}\bigg[W^{(l)}+\frac{dW^{(l)}}{dR_i}\bigg]U^{(l+)}=\sum_{\sigma\in\{\uparrow,\downarrow\}}\sum_{p=1}^{M_l}f_p^{(l+)}n_{p,\sigma},\hspace{0.5cm} U^{(l-)}\bigg[W^{(l)}-\frac{dW^{(l)}}{dR_i}\bigg]U^{(l-)}=\sum_{\sigma\in\{\uparrow,\downarrow\}}\sum_{p=1}^{M_l}f_p^{(l-)}n_{p,\sigma}.
\end{equation}
The corresponding scaling factor of the block encoding is in this case
\begin{equation}\label{eq:lambdaF_low_rank}
    \lambda_{\mathrm{F}_i}=\sum_{l=1}^L\frac{1}{2}\Bigg\{\Bigg[\sum_{p=1}^{M_l}f_p^{(l+)}\Bigg]^2+\Bigg[\sum_{p=1}^{M_l}f_p^{(l-)}\Bigg]^2\Bigg\}.
\end{equation}
Assuming that $\frac{dW^{(l)}}{dR_i}$ is a similar scale to $W^{(l)}$, we can expect that $f_p^{(l\pm)}\sim 2f_p^{(l)}$, and that $\lambda_{\mathrm{F}}\sim 4\lambda_H$.
However, as noted before for the THC derivative factorization, we expect these scalings to be improved significantly by a direct factorization of the derivative operator instead of differentiating the factorization.
This may yet lead to an asymptotic improvement, but with the same systematic bias as described for the THC method that comes from a different truncation error between the Hamiltonian and force operators.

\subsection{Block encoding in first quantization in a plane wave basis}
\label{sec:BEfirst}
The relatively simple structure of the electronic structure Hamiltonian in a plane wave basis implies that the cost to block encode it is far below the $O(N^4)$ bound for an arbitrary two-body operator.
Block-encodings of the first-quantized plane wave Hamiltonian are already well-known~\cite{su2021fault}, with a number of gates required scaling as
\begin{equation}
    T_H\in\widetilde{\mathcal{O}}((\log N)^2\eta),
\end{equation}
where $\eta$ is the number of electrons in the system, and $N$ is the number of plane waves.
The corresponding scaling factor is similarly known to be~\cite{BabbushContinuum}
\begin{equation}
    \lambda_{H}\in \mathcal{O}(\eta N^{2/3}/\Omega^{2/3}).
\end{equation}

We now consider how to block-encode the force operator in its first-quantized plane wave form.
We first rewrite Eq.~\eqref{eq:force_plane_wave_first} slightly to place it in a linear combination of unitaries form
\begin{equation}
    \frac{dH}{dR_{l,\alpha}}=\frac{4\pi Z_l}{\Omega}\sum_{j=1}^{\eta}\sum_{\nu\in G_0}\frac{|k^{\alpha}_{\nu}|}{\|\mathbf{k}_{\nu}\|^2}\Bigg(\textrm{QFT}\sum_{p\in G_0}\,-i\,\mathrm{sign}(k_{\nu}^{\alpha})e^{i\mathbf{k}_{\nu}\cdot(\mathbf{R}_l-\mathbf{r}_p)}|p\rangle\langle p|_j\;\textrm{QFT}^{\dag}\Bigg).
\end{equation}
Here, we split our indexing of the $3N_a$-dimensional nuclear co-ordinate vector $\mathbf{R}$ into the nuclei index $l$ and the position index $\alpha$, and write $\mathbf{R}_l=(R_{l,x},R_{l,y},R_{l,z})$ for the $3$-dimensional position of nuclei $l$. Then, we write $k_{\nu}^{\alpha}$ for the projection of the $3$-dimensional vector $\mathbf{k}_{\nu}$ onto the unit vector in the direction of $R_{l,\alpha}$.
To see that the operator inside each bracket is unitary, recall that $|p\rangle\langle p|_j$ projects the $j$th electron register onto the $p$th basis state, but acts as the identity on all other electron registers.
This implies that $|p\rangle\langle p|_j^{\dag}|q\rangle\langle q|_j=\delta_{p,q}|p\rangle\langle p|_j$, and $\sum_{p}|p\rangle\langle p|_j=I$, as required.
This implies that our PREPARE operator requires us to prepare the state
\begin{equation}
    \sqrt{\frac{4\pi Z_l}{\lambda_{\mathrm{F}_{l,\alpha}}\Omega}}\sum_{j,\nu}\frac{\sqrt{|k_{\nu}^{\alpha}|}}{\|\mathbf{k}_{\nu}\|}|\nu,j\rangle,
\end{equation}
with a rescaling constant
\begin{align}
    \lambda_{\mathrm{F}_{l,\alpha}}=\frac{4\pi\eta Z_l}{\Omega}\sum_{\nu\in G_0}\frac{|\mathbf{k}_{\nu}\cdot \mathbf{R}_l|}{\|\mathbf{k}_{\nu}\|^2\|\mathbf{R}_l\|}
    &= \mathcal{O} \left(\frac{\eta Z_l}{\Omega}\int_0^{2\pi}d\phi\int_0^{\pi}d\theta\int_0^{N^{1/3}}dr\, \frac{\Omega^{1/3}r\left|\cos\theta \right|}{r^2}r^2\sin\theta\right)\nonumber\\
    &={\cal O}(\eta N^{2/3}/\Omega^{2/3}).\label{eq:lambdaF_planewaves}
\end{align}
In \app{first_quantized_force_prepare}, we describe a circuit for the PREPARE operator that generates an initial approximation via a nested boxes approach, using inequality testing and a single amplitude amplification round to achieve $\epsilon$ precision at a cost $\mathcal{O}(\log N\,\log(N/\epsilon))$.
The SELECT operator for our block encoding takes the form
\begin{equation}
    \mathrm{SELECT}=-i \,\mathrm{QFT}\sum_{\nu,j}\mathrm{sign}(\mathbf{k}_{\nu}^{\alpha})\sum_p e^{i\mathbf{k}_{\nu}\cdot (\mathbf{R}_l-\mathbf{r}_p)}|\nu,j\rangle\langle\nu,j| \otimes |p\rangle\langle p|_j\mathrm{QFT}^{\dag}.
\end{equation}
Normally the QFT operation would be performed on each of the $\eta$ electron registers.
However, we only need to act on the momentum register $j$.
We therefore first swap the $j$th electron register into the $0$th electron register conditional on $|j\rangle$.
That can be achieved with $\eta$ Toffolis for unary iteration \cite{BabbushSpectra}, together with $\eta\log N$ Toffolis for controlled swaps.
At the end we invert the procedure to replace the register.
The QFT on this register can be implemented with gate complexity $\mathcal{O}((\log N)^2)$.

To implement the phase $e^{i\mathbf{k}_{\nu}\cdot (\mathbf{R}_l-\mathbf{r}_p)}$, one would first subtract $\mathbf{r}_p$ from the classically given value $\mathbf{R}_l$.
The number of bits needed for $\mathbf{R}_l$ can be found in the case of energy estimation to be $\mathcal{O}(\log(N/\epsilon))$, using Lemma 2 of \cite{su2021fault}.
Here we are estimating force rather than energy, but the problem is sufficiently similar that the same scaling is needed.
The complexity of the subtraction is therefore $\mathcal{O}(\log(N/\epsilon))$.
The components of both $\mathbf{R}_l$ and $\mathbf{r}_p$ would be given in two's complement, then one would convert the difference to signed integers.

The phasing can then be performed by using a phase gradient state \cite{Kitaev2002,GidneyAdder}.
Let us assume that the $|\nu\rangle$ register contains three separated components $|\nu_x,\nu_y,\nu_z\rangle$ written in binary form; $k_{\nu}^{\alpha}=\frac{2\pi \nu_{\alpha}}{\Omega^{1/3}}$ for $\alpha=x,y,z$.
Then each bit of each component of $\nu$ can be used to control addition of a component of $\mathbf{R}_l-\mathbf{r}_p$ into the phase gradient register.
There are sign bits of both, since they are given as signed integers.
The product of the sign bits can be used to give an overall sign, and this sign can be used to control addition versus subtraction of $\mathbf{R}_l-\mathbf{r}_p$ without further non-Clifford cost.
Since there are $\mathcal{O}(\log N)$ bits of $\nu$ and $\mathcal{O}(\log(N/\epsilon))$ bits of $\mathbf{R}_l-\mathbf{r}_p$, the overall cost is $\mathcal{O}(\log N \log(N/\epsilon))$ for the phasing.
This is similar to the complexity of the state preparation.

Taking into account the $\mathcal{O}((\log N)^2)$ complexity of the QFT and the $\mathcal{O}(\eta\log N)$ complexity of swapping the momentum register, the complexity to block encode the derivative operator is
\begin{equation}
    T_F=\mathcal{O}\big[\eta \log N +\log N \log(N/\epsilon)\big].
\end{equation}

%% file: preparation.tex
To address this issue, we consider grouping the momenta into nested boxes for $(\nu_x,\nu_y)$ alone:
\begin{align}
    C_{\mu} := \left\{(\nu_y,\nu_z)\ \left|\ \abs{\nu_y}<2^{\mu-1}\land\abs{\nu_z}<2^{\mu-1}\land
	\left(\abs{\nu_y}\geq 2^{\mu-2}\lor\abs{\nu_z}\geq 2^{\mu-2}\right)\right.\right\}.
\end{align}
We now describe a quantum circuit that prepares the momentum state in \eq{momentum_state}. We start by preparing
\begin{equation}
    \frac{1}{\sqrt{2^{n+1}}}\sum_{\mu=2}^{n}2^{\mu/2}\ket{\mu}.
\end{equation}
The method for preparation of this state was given in Section IIB of \cite{Su2021}, and has Toffoli complexity linear in $n$ (given a $\ket{T}$ state that is used catalytically to implement controlled Hadamards).
Here we have allowed this state to be subnormalised to account for failure of this preparation.
Given $\mu$, we aim to prepare $y,z$ values within box $C_\mu$.
This can be achieved by preparing the state $\ket{\mu}$ in unary, and performing controlled Hadamards on the $y$ and $z$ registers with complexity $\mathcal{O}(\log N)$.
We encode $y,z$ as signed integers, and also eliminate the negative zero as a failure with complexity $\mathcal{O}(\log N)$.
That gives
\begin{equation}
    \frac{1}{\sqrt{2^{n+1}}}\sum_{\mu=2}^{n}2^{\mu/2}\frac 1{2^{\mu}}\sum_{\nu_y,\nu_z=-(2^{\mu-1}-1)}^{2^{\mu-1}-1}\ket{\mu,\nu_y,\nu_z}.
\end{equation}
Next, we can eliminate the inner squares in $\nu_y,\nu_z$ using inequality tests (with complexity $\mathcal{O}(\log N)$) to give
\begin{equation}
    \frac{1}{\sqrt{2^{n+1}}}\sum_{\mu=2}^{n}2^{-\mu/2}\sum_{\substack{\nu_y,\nu_z=-(2^{\mu-1}-1)\\ |\nu_y|,|\nu_z|\ge 2^{\mu-2}}}^{2^{\mu-1}-1}\ket{\mu,\nu_y,\nu_z}.
\end{equation}
For $\nu_x$ we simply create an equal superposition using Hadamards (without any nesting of boxes), and flag the negative zero (complexity $\mathcal{O}(\log N)$), to give
\begin{equation}
    \frac{1}{\sqrt{2^{2n+1}}}\sum_{\nu_x=-(2^{n-1}-1)}^{2^{n-1}-1}\sum_{\mu=2}^{n}2^{-\mu/2}\sum_{(\nu_y,\nu_z) \in C_{\mu}}\ket{\mu,\nu_x,\nu_y,\nu_z}.
\end{equation}

Next, we aim to show that
\begin{equation}
\frac{\sqrt{|\nu_x|}}{\sqrt{\nu_x^2+\nu_y^2+\nu_z^2}} \le 2^{-(\mu-1)/2} ,
\end{equation}
so that we can perform inequality testing for the state preparation.
To maximise this expression, we need to choose $\nu_y,\nu_z$ such that $\nu_y^2+\nu_z^2$ is minimised for a given box $C_\mu$.
The minimum is obtained for $|\nu_y|=2^{\mu-2}$ and $|\nu_z|=0$ or $|\nu_z|=2^{\mu-2}$ and $|\nu_y|=0$ (the middle of a side of an inner square), in which case $\nu_y^2+\nu_z^2=(2^{\mu-2})^2$.
The maximum as a function of $\nu_x$ is then for $\nu_x=2^{\mu-2}$.
That gives us an upper bound of
\begin{equation}
\frac{\sqrt{|\nu_x|}}{\sqrt{\nu_x^2+\nu_y^2+\nu_z^2}} \le \frac{2^{(\mu-2)/2}}{\sqrt{2^{2\mu-4}+2^{2\mu-4}}} = 2^{-(\mu-1)/2} ,
\end{equation}
as required.
Therefore, we now prepare another superposition over $M$ values of $m$ and test the inequality
\begin{equation}
\frac{|\nu_x|}{\nu_x^2+\nu_y^2+\nu_z^2} > 2^{-(\mu-1)} \frac{m}{M}.
\end{equation}
In order to avoid divisions, this inequality test can be rewritten as
\begin{equation}
2^{\mu-1} |\nu_x|M > m(\nu_x^2+\nu_y^2+\nu_z^2).
\end{equation}
This inequality test can be evaluated using only additions and multiplications.

The Toffoli complexity of the multiplications is the product of the number of bits, and we have $\mathcal{O}(\log N)$ bits for the components of $\nu$.
Using Lemma 3 of \cite{su2021fault} we can give the number of bits for $m$ as
$\mathcal{O}(\log (N/\epsilon))$ for the case of energy estimation.
The case of force estimation is sufficiently similar that the scaling is the same.
The squares of components of $\nu$ will have complexity $\mathcal{O}((\log N)^2)$, and
the multiplication by $m$ will have complexity $\mathcal{O}(\log N\log (N/\epsilon))$.
The additions have linear complexity in the number of bits, so the total complexity is $\mathcal{O}(\log N\log (N/\epsilon))$ Toffolis for the complete inequality test.

The resulting state, with success of the inequality test, can be written as
\begin{equation}
    \frac{1}{\sqrt{2^{2n+1}}}\sum_{\nu_x=-(2^{n-1}-1)}^{2^{n-1}-1}\sum_{\mu=2}^{n}2^{-\mu/2}\sum_{(\nu_y,\nu_z) \in C_{\mu}}\ket{\mu,\nu_x,\nu_y,\nu_z} \frac 1{\sqrt{M}} \sum_{m=0}^{Q-1} \ket{m},
\end{equation}
with
\begin{equation}
Q = \left\lceil \frac{2^{\mu-1}M |\nu_x|}{\|\nu\|^2} \right\rceil .
\end{equation}

For $M$ sufficiently large, one can verify that the above state is proportional to the momentum state we want to prepare. Indeed, for a fixed choice of $\mu$ and $\nu$, the corresponding amplitude is proportional to
\begin{equation}
\frac{1}{\sqrt{2^{2n+1}}} 2^{-\mu/2}
\frac{2^{(\mu-1)/2}\sqrt{|\nu_x|}}{\|\nu\|} = \frac{1}{2^{n+1}} 
\frac{\sqrt{|\nu_x|}}{\|\nu\|}
\end{equation}
up to the error due to discretization. Now, the probability of success has the asymptotic scaling of
\begin{align}
\frac{1}{2^{2n+1}}
\sum_{\nu_x,\nu_y,\nu_z=-(2^{n-1}-1)}^{2^{n-1}-1} \frac{|\nu_x|}{\|\nu\|^2} &\approx\frac{1}{2^{2n+1}} \int_{-2^{n-1}}^{2^{n-1}} d\nu_z \int_{-2^{n-1}}^{2^{n-1}} d\nu_y \int_{-2^{n-1}}^{2^{n-1}} d\nu_x \frac{|\nu_x|}{\nu_x^2+\nu_y^2+\nu_z^2} \nonumber \\
&= \frac 12 \int_{0}^1 dz \int_{0}^1 dy \int_{0}^1 dx \frac{x}{x^2+y^2+z^2} \nonumber \\
&= \frac 14 \int_{0}^1 dz \int_{0}^1 dy \log\left( 1+\frac{1}{y^2+z^2} \right) \nonumber \\
&= \frac 14\int_{0}^1 dz \left[ -2 z \arccot(z) + 2 \sqrt{1 + z^2}  \arccot\sqrt{1 + z^2} + 
 \log\left(1 + \frac 1{1 + z^2}\right)  \right] \nonumber \\
&= \frac 14\left[ -1-\frac{\pi}2 + 2\sqrt{2} \arctan\left( \frac 1{\sqrt 2} \right) + \log(3/2)\right] \nonumber \\
& \quad + \frac 12 \int_{0}^1 dz \left[ \sqrt{1 + z^2}  \arccot\sqrt{1 + z^2} \right].
\end{align}
The integral can be computed to arbitrary precision numerically, giving
\begin{equation}
\int_{0}^1 dz \left[ \sqrt{1 + z^2}  \arccot\sqrt{1 + z^2} \right] \approx
0.8197551759351555.
\end{equation}
That then gives the probability of success as approximately
\begin{equation}
\frac{1}{2^{2n+2}}
\sum_{\nu_x,\nu_y,\nu_z=-(2^{n-1}-1)}^{2^{n-1}-1} \frac{|\nu_x|}{\|\nu\|^2} \approx
0.3037546589794463.
\end{equation}
To boost the success probability close to unity, we can use just one step of amplitude amplification because the probability is greater than $1/4$.
The dominant complexity is involved in the squares and multiplications in the inequality test, and the amplitude amplification only gives a multiplicative factor, for a total complexity of $\mathcal{O}(\log N\log (N/\epsilon))$ Toffolis.